\documentclass[final,notitlepage,11pt]{article}
\usepackage{amsfonts}
\usepackage{adjustbox}
\usepackage{amssymb}
\usepackage{algorithm}
\usepackage{algorithmic}
\usepackage{bm}
\usepackage{booktabs,tabularx}
\usepackage{caption}
\usepackage{enumitem}
\usepackage{multirow}
\usepackage{graphicx}
\usepackage{amsmath}
\usepackage{latexsym}
\usepackage{rotating}
\usepackage{lscape}
\usepackage{mathrsfs}
\usepackage{hyperref}
\usepackage[svgnames]{xcolor}

\graphicspath{ {./images/} }
\hypersetup{
colorlinks=true,
linkcolor=DarkRed,
citecolor=NavyBlue}
\allowdisplaybreaks
\RequirePackage[round,authoryear,longnamesfirst]{natbib}

\newtheorem{theorem}{Theorem}[section]
\newtheorem{ass}{Assumption}
\newtheorem{assumptionalt}{Assumption}[ass]
\newenvironment{assp}[1]{
  \renewcommand\theassumptionalt{#1}
  \assumptionalt
}{\endassumptionalt}

\newtheorem{definition}{Definition}

\newtheorem{lemma}[theorem]{Lemma}

\DeclareMathOperator*{\argmin}{arg\,min}

\newenvironment{proof}[1][Proof]{\textbf{#1.} }{\  \rule{0.5em}{0.5em}}
\numberwithin{equation}{section}
\numberwithin{definition}{section}
\newcommand{\RNum}[1]{\uppercase\expandafter{\romannumeral #1\relax}}
\DeclareMathOperator*{\plim}{plim}
\allowdisplaybreaks
\begin{document}

\title{Panel Data Models with Time-Varying Latent Group Structures\thanks{%
Phillips acknowledges support from a Lee Kong Chian Fellowship at SMU, the
Kelly Fund at the University of Auckland, and the NSF under Grant No. SES
18-50860. Su gratefully acknowledges the support from the National Natural
Science Foundation of China under Grant No. 72133002. Address correspondence
to: Liangjun Su, School of Economics and Management, Tsinghua University,
Beijing, 100084, China. E-mail: sulj@sem.tsinghua.edu.cn. }}
\author{Yiren Wang$^{a}$, Peter C B Phillips$^{b}$ and Liangjun Su$^{c}$ \\
$^{a}$School of Economics, Singapore Management University, Singapore \\
$^{b}$Yale University, University of Auckland, 
Singapore Management University\\
$^{c}$School of Economics and Management, Tsinghua University, China}
\maketitle

\begin{abstract}
This paper considers a linear panel model with interactive fixed effects and
unobserved individual and time heterogeneities that are captured by some
latent group structures and an unknown structural break, respectively. To enhance realism the model may have different numbers of groups and/or different group
memberships before and after the break. With the preliminary
nuclear-norm-regularized estimation followed by row- and column-wise linear
regressions, we estimate the break point based on the idea of binary
segmentation and the latent group structures together with the number of
groups before and after the break by sequential testing K-means algorithm
simultaneously. It is shown that the break point, the number of groups and the
group memberships can each be estimated correctly with probability approaching
one. Asymptotic distributions of the estimators of the slope coefficients
are established. Monte Carlo simulations demonstrate excellent finite sample
performance for the proposed estimation algorithm. An empirical application to real house price data across 377 Metropolitan Statistical Areas in the US from 1975 to 2014 suggests the presence both of structural breaks and of changes in group membership. \medskip 

\noindent \textbf{Key words:} Interactive fixed effects, latent group
structure, structural break, nuclear norm regularization, sequential testing
K-means algorithm. \medskip

\noindent \textbf{JEL Classification:} C23, C33, C38, C51\medskip\ 
\end{abstract}

\ \ \ \ \pagebreak

\section{Introduction}

\label{sec:intro} Heterogeneous panel data models have been widely used in
empirical research in economics because they can capture a rich degree of
unobserved heterogeneity. But models with complete heterogeneity along
either the cross section or time dimension tend to possess too many
parameters to be identified, which results in slow convergence and
inefficient estimates. For this reason, researchers now frequently advocate
the use of panel data models with certain structures imposed along either
the cross section or time dimension. On the one hand, the recent burgeoning
of panels with latent group structures can be motivated from the observation
that different groups of individuals respond differently to exogenous
shocks. For instance, \cite{durlauf1995multiple}, \cite%
{berthelemy1996economic}, and \cite{ben1998convergence} show economies in
different groups of income per capita and/or education level may converge to
different steady state equilibria. \cite{klapper2011impact}, \cite%
{chu2012logistics}, and \cite{zhang2019threshold} show an exogenous shock
like policy implementation has different impacts on different individuals,
and \cite{long2012impact} argue that the influence of the 2008 financial
crisis on economic growth differs for emergent and developed
economies. On the other hand, the recent popularity of panels that evidence 
structural change can be motivated by events such as financial crises,
the economic impact of technological progress, and more general economic transitions that occur during the time
periods covered by the data. See \cite{qian2016shrinkage} for a survey of
panel data models and research that consider estimation and inference concerning structural change.

In spite of the now large literature that studies separately individual heterogeneity or
time heterogeneity in the slope coefficients of panel models, few works consider both types of heterogeneity simultaneously. Exceptions include \cite{keane2020climate} and \cite{lu2023uniform} who consider linear
panel data models with two-dimensional unobserved heterogeneity in the slope
coefficients that are modelled via the usual additive structure, and \cite%
{chernozhukov2019inference} and \cite{wang2022low-rank} who model the slope
coefficients via the use of low-rank matrices for conditional mean and
quantile regressions, respectively. In addition, \cite{okui2021heterogeneous}
and \cite{lumsdaine2023estimation} consider both individual heterogeneity
and time heterogeneity by modeling them as a grouped pattern and as structural
breaks, respectively. Specifically, \cite{okui2021heterogeneous} develop a
new panel data model with latent groups where the number of groups and the
group memberships do not change over time but the coefficients within each
group can change over time and they may have different break-dates; \cite%
{lumsdaine2023estimation} consider the panels with a grouped pattern of
heterogeneity when the latent group membership structure and/or the values
of slope coefficients change at a break point. Both papers provide
algorithms to recover the latent group structure based on linear panel
models with or without individual fixed effects, but cannot allow for the
presence of more complicated fixed effects such as interactive fixed
effects (IFEs) to capture strong cross-sectional dependence in the data.

This paper proposes a linear panel data model with IFEs that enable the
slope coefficients to exhibit two-way heterogeneity. Following the lead of \cite%
{okui2021heterogeneous} and \cite{lumsdaine2023estimation} and to encourage
the parameter parsimony, we use a latent group structure to capture 
individual heterogeneity and an unknown structural break to capture time
heterogeneity. The latent group structure of the model accommodates different 
group numbers and different group memberships before and after the
break. Given this complicated structure, the approach proposed is to estimate
the break point, the number of groups before and after the break, the group
membership before and after the break, and the group-specific parameters in
multiple steps. The key insight that permits this degree of complication is that 
the slope coefficients of each of the $p$ regressors in the model are permitted
to vary across both cross section and time dimensions by means of a factor structure with a fixed
number of factors so that they may be conveniently stacked into a low-rank matrix.

In the first step, the low-rank nature of the slope matrices is explored and
initial estimates are obtained by nuclear norm regularization
(NNR), a machine learning technique popular in computer science that is increasingly used in econometrics. Such
initial matrix estimates are consistent in terms of the Frobenius norm but do
not have pointwise or uniform convergence for their elements. Despite
this, by applying singular value decomposition (SVD) to these estimates, we
can obtain estimates of the associated factors and the factor loadings that are
also consistent in terms of the Frobenius norm. In the second step, we use the
first-step initial estimates of the factors and factor loadings to run the
row- and column-wise linear regressions to update the estimates of the
factors and factor loadings which now possess pointwise and uniform
consistency and can be used for subsequent analyses. In the third step, we
estimate the break point by using the celebrated idea of binary segmentation,
as commonly used for break point estimation in the time series literature.
Once the break point is estimated, the full-sample is naturally split into
two subsamples. In the fourth step, we follow the lead of \cite%
{lin2012estimation} and \cite{Jin2022Optimal} to focus on each subsample
before and after the estimated break point and propose a sequential testing
K-means algorithm to recover the latent group structure and obtain the
number of groups simultaneously. In the last step, we use the estimated
group structure to estimate the group-specific parameters. Asymptotic
analyses show that the break point, the number of groups and the group
memberships can be consistently estimated in Steps 3-4, so that the final
step estimates for the group-specific coefficients can enjoy the oracle
property. This means they have the same asymptotic distributions as the ones
obtained by knowing the break point and the latent group structures before
and after the break points.


The present paper relates to two branches of literature. First, it contributes to the panel data literature on one-way heterogeneity, especially with either latent group structures or structural breaks. With respect to latent group structures, there are several popular ways to recover the
latent groups. The first approach is the K-means algorithm. \cite%
{lin2012estimation} apply the K-means algorithm to linear panel data models
with grouped slope coefficients and propose an information criterion and a
sequential testing approach to estimate the true number of groups. \cite%
{sarafidis2015partially} analyze the unknown grouped slopes in the large $N$
and fixed $T$ framework, and \cite{zhang2019quantile} provide an iterative
algorithm based on K-means clustering for a panel quantile regression model. 
\cite{bonhomme2015grouped} and \cite{ando2016panel} consider panels with
grouped fixed effects. The second approach is the Classifier-Lasso (C-Lasso)
that has become a popular clustering method since \cite{su2016identifying}.
This method is extended by \cite{lu2017determining}, \cite{su2018identifying}%
, \cite{su2019sieve}, \cite{wang2019heterogeneous}, and \cite%
{huang2020identifying} to various contexts. In addition, both the clustering
algorithm in regression via a data-driven segmentation (CARDS) approach and
binary segmentation are also considered in \cite{ke2015homogeneity}, \cite%
{wang2018homogeneity}, \cite{ke2016structure} and \cite{wang2021identifying}%
, among others. As for panel models with structural breaks, binary
segmentation has become a common approach to estimate the break point. See 
\cite{bai2010common}, \cite{hsu2011change}, \cite{kim2011estimating}, \cite%
{kim2014common} and \cite{baltagi2017estimation}, among others. 
These papers focus on the case of a single break point in the model. In contrast, 
\cite{qian2016shrinkage} and \cite{li2016panel} allow for multiple breaks in
linear panel models with either classical fixed effects or IFEs, and propose an adaptive grouped fused lasso (AGFL) approach to estimate the break points. Compared to the existing panel literature on one-way
heterogeneity, our model allows for two-way heterogeneity. In
particular, not only are different membership structures in different
time blocks permitted but also changes in the number of groups over time. As a result,
our model is more flexible than all existing models that allow only for 
latent group structures or structural breaks, but not both.

Second, this paper contributes to the recent burgeoning literature that
models two-way heterogeneity in the slope coefficients of a panel model. 
As mentioned above, there are two approaches to model two-way
heterogeneity in the slope coefficients. One approach models them in an additive
structure so that both individual and time effects enter the slope
coefficients additively, as in \cite{keane2020climate} and \cite%
{lu2023uniform}. The other approach imposes certain low-rank structures on the
slope coefficient matrices in which case one models each slope coefficient
via the use of IFEs to capture strong cross sectional dependence
in the panel. In view of the low-rank structures, we can resort to NNR estimation 
which has attracted increasing attention recently in panel data analyses.
NNR has been used in recent econometric research -- see \cite%
{bai2019rank}, \cite{moon2018nuclear}, \cite{chernozhukov2019inference}, 
\cite{belloni2023high}, \cite{miao2023high}, \cite{feng2023regularized}, and 
\cite{hong2023profile}, among others. But none of these papers imposes any
latent group structures on the slope coefficients. With latent group
structures and structural breaks imposed, \cite{okui2021heterogeneous} allows
the slope coefficients within each group to have common breaks and the break
points to vary across different groups, and they propose to estimate the
latent group structures, the structural breaks, and the group-specific
regression parameters by the grouped adaptive group fused lasso (GAGFL).
But neither the number of groups nor the group memberships is allowed
to change over time in \cite{okui2021heterogeneous}. In a companion paper, 
\cite{lumsdaine2023estimation} allows the latent group membership structure
and/or the values of slope coefficients to change at a break point, and
proposes an estimation algorithm similar to the K-means of \cite%
{bonhomme2015grouped}. Both \cite{okui2021heterogeneous} and \cite%
{lumsdaine2023estimation} allow for at most one-way heterogeneity
(individual fixed effects) in the intercept and neither allows for IFEs to
capture strong cross section dependence. In contrast, this paper proposes an
algorithm to detect the unknown break point and to recover the group
structure based on linear panel model with IFEs, which involves a more
general model. In addition, \cite{lumsdaine2023estimation} first assume the
number of groups is known in the estimation algorithm and then estimate the
number of groups via an information criterion but they do not establish consistency 
for such an estimate. Instead, we estimate the number of
groups and group membership simultaneously by the sequential testing K-means
algorithm and establish the consistency of the number of groups estimator.

The rest of the paper is organized as follows. We first introduce the linear
panel model with time-varying latent group structures in Section \ref%
{sec:model} and provide the estimation algorithm in Section \ref%
{sec:estimation}. The asymptotic properties are given in Section \ref%
{sec:theory}. In Section \ref{sec:extension}, we propose an alternative
approach to detect the break point, provide the test statistics for the null
that the slope coefficients exhibit no structure change against the
alternative with one break point, and discuss the estimation for the model
with multiple breaks. In Sections \ref{sec:simul} and \ref{sec:emp}, we show
the finite sample performance of our method by Monte Carlo simulations and
an empirical application, respectively. Section \ref{sec:concl} concludes.
All proofs are provided in the online supplement.

\textit{Notation.} Let $\left\Vert \cdot \right\Vert _{\max },$ $\left\Vert
\cdot \right\Vert _{op},$ $\left\Vert \cdot \right\Vert $, and $\left\Vert
\cdot \right\Vert _{\ast }$ denote the (elementwise) maximum norm, operator
norm, Frobenius norm, and nuclear norm, respectively. Let $\odot $ denote
the element-wise Hadamard product. $\lfloor \cdot \rfloor $ and $\lceil
\cdot \rceil $ denote the floor and ceiling functions, respectively. Let $%
a\vee b=\max \left( a,b\right) $ and $a\wedge b=\min (a,b)$. $a_{n}\lesssim
b_{n}$ means $a_{n}/b_{n}=O_{p}\left( 1\right) $ and $a_{n}\gg b_{n}$ means $%
b_{n}a_{n}^{-1}=o(1)$. Let $A=\{A_{it}\}$ be a matrix with its $(i,t)$-th
entry denoted as $A_{it}$. Let denote $\{A_{j}\}_{j\in \lbrack p]\cup \{0\}}$
be the collection of matrices $A_{j},$ $j\in \{0,1,\cdots ,p\}$. For a
specific $A\in \mathbb{R}^{m\times n}$ with rank $n$, let $P_{A}=A(A^{\prime
}A)^{-1}A^{\prime }$ and $M_{A}=I_{m}-P_{A}$. When $A$ is symmetric, $%
\lambda _{\max }(A)$, $\lambda _{\min }(A)$ and $\lambda _{n}(A)$ denote its
largest, smallest and $n$-th largest eigenvalues, respectively. The
operators $\rightsquigarrow $ and $\overset{p}{\longrightarrow }$ denote
convergence in distribution and in probability, respectively. Let $[n]=$ $%
\{1,\cdots ,n\}$ for any positive integer $n$, let $\mathbf{1}\{\cdot \}$ be
the usual indicator function, and w.p.a.1 and a.s. abbreviate
\textquotedblleft with probability approaching 1\textquotedblright\ and
\textquotedblleft almost surely\textquotedblright , respectively.

\section{Model Setup}

\label{sec:model} In this paper we consider the following linear panel
model with IFEs: 
\begin{equation}
Y_{it}=\Theta _{0,it}^{0}+X_{it}^{\prime }\Theta _{it}^{0}+e_{it},
\label{eq2.1}
\end{equation}%
where $i\in \lbrack N]$, $t\in \lbrack T]$, $Y_{it}$ is the dependent
variable, $X_{it}=(X_{1,it},\cdots ,X_{p,it})^{\prime }$ is a $p\times 1$
vector of regressors, $\Theta _{it}^{0}=(\Theta _{1,it}^{0},\cdots ,\Theta
_{p,it}^{0})^{\prime }$ is a $p\times 1$ vector of slope coefficients, $%
\Theta _{0,it}^{0}=\lambda _{i}^{0\prime }f_{t}^{0}$ is an intercept term
that exhibits a factor structure with $r_{0}$ factors, and $e_{it}$ is the
error term. Here, we assume $r_{0}$ is a fixed integer that does not change
as $\left( N,T\right) \rightarrow \infty .$ Let $\Lambda ^{0}=\left( \lambda
_{1}^{0},\cdots ,\lambda _{N}^{0}\right) ^{\prime }$ and $F^{0}=\left(
f_{1}^{0},\cdots ,f_{T}^{0}\right) ^{\prime }$. Let $Y=\left\{
Y_{it}\right\} ,$ $X_{j}=\left\{ X_{j,it}\right\} ,$ $\Theta
_{j}^{0}=\{\Theta _{j,it}^{0}\}$ and $E=\left\{ e_{it}\right\} ,$ all of
which are $N\times T$ matrices. Then we can rewrite (\ref{eq2.1}) in matrix form as 
\begin{equation}
Y=\Theta _{0}^{0}+\sum_{j=1}^{p}X_{j}\odot \Theta _{j}^{0}+E.  \label{eq2.2}
\end{equation}%
We assume that the slope coefficients follow time-varying latent group
structures, viz., 
\begin{equation*}
\Theta _{it}^{0}=\sum_{k\in \lbrack K_{t}]}\alpha _{kt}\mathbf{1}\{i\in
G_{kt}\},
\end{equation*}%
where $\left\{ G_{kt}\right\} _{k\in K_{t}}$ forms a partition of $[N]$ for
each specific time $t$ with $K_{t}$ being the number of groups at time $t$.
Moreover, we assume that the group-specific slope coefficients $\alpha _{kt}$
or the memberships change at an unknown time point $T_{1}$, i.e., 
\begin{eqnarray*}
\alpha _{kt} &=&\left\{ \begin{aligned}
&\alpha_{k}^{(1)},\quad\text{for}\quad t=1,\dots,T_{1},\\ &
\alpha_{k}^{(2)},\quad\text{for}\quad t=T_{1}+1,\dots,T,\\ \end{aligned}%
\right.  \\
G_{kt} &=&\left\{ \begin{aligned} &G_{k}^{(1)},\quad\text{for}\quad
t=1,\dots,T_{1},~ k=1,\dots,K^{(1)},\\ & G_{k}^{(2)},\quad \text{for}\quad
t=T_{1}+1,\dots,T,~ k=1,\dots,K^{(2)},\\ \end{aligned}\right. 
\end{eqnarray*}%
with $K^{(1)}$ and $K^{(2)}$ being the number of latent groups before and
after the break point, respectively. Let $g_{i}^{(1)}$ and $g_{i}^{(2)}$
respectively denote the individual group indices before and after the break: 
\begin{equation*}
g_{i}^{(1)}=\sum_{k\in K^{(1)}}k\mathbf{1}\{i\in G_{k}^{(1)}\}\quad \text{and%
}\quad g_{i}^{(2)}=\sum_{k\in K^{(2)}}k\mathbf{1}\{i\in G_{k}^{(2)}\}.
\end{equation*}

Let $r_{j}$ be the rank of $\Theta _{j}^{0}$ for $j\in \lbrack p]\cup \{0\}$%
. It is easy to see that $\Theta _{j}^{0}$ exhibits a low-rank structure for
all $j.$ By the SVD, we have 
\begin{equation*}
\Theta _{j}^{0}=\sqrt{NT}\mathcal{U}_{j}^{0}\Sigma _{j}^{0}\mathcal{V}%
_{j}^{0\prime }:=U_{j}^{0}V_{j}^{0\prime },\quad j\in \lbrack p]\cup \{0\},
\end{equation*}%
where $\mathcal{U}_{j}^{0}\in \mathbb{R}^{N\times r_{j}}$, $\mathcal{V}%
_{j}^{0}\in \mathbb{R}^{T\times r_{j}}$, $\Sigma _{j}^{0}=\text{diag}(\sigma
_{1,j},\cdots ,\sigma _{r_{j},j})$, $U_{j}^{0}=\sqrt{N}\mathcal{U}%
_{j}^{0}\Sigma _{j}^{0}$ with each row being $u_{i,j}^{0\prime }$, and $%
V_{j}^{0}=\sqrt{T}\mathcal{V}_{j}^{0}$ with each row being $v_{t,j}^{0\prime
}$.

Note that we allow $\left\{ \Theta _{it}^{0}\right\} _{i=1}^{N}$ to exhibit
latent group structures before and after the break. For a particular $j\in
\lbrack p]$, the $N\times T$ matrix $\Theta _{j}^{0}$ may have no group
structure before or after the break, or no break, or more or fewer groups
after the break. Let $K_{j}^{(1)}$ and $K_{j}^{(2)}$ denote the number of
groups before and after the break, respectively, for $\{\Theta
_{j,it}^{0}\}_{i=1}^{N}$. Let $\mathcal{G}_{j}^{(\ell )}=\{G_{1,j}^{(\ell
)},\cdots ,G_{K_{j}^{(\ell )},j}^{(\ell )}\},$ $\ell =1,2,$ denote the
associated latent group structures. Define $N_{k,j}^{(\ell
)}=|G_{k,j}^{(\ell )}|$ and $\pi _{k,j}^{(\ell )}=\frac{N_{k,j}^{(\ell )}}{N}
$ for $\ell =1,2,$ where $\left\vert A\right\vert $ denotes the cardinality
of set $A.$ Further define $\tau _{T}:=\frac{T_{1}}{T}$. We show that $%
\Theta _{j}^{0}$ has a low-rank structure in all of the following cases:

\begin{description}
\item[Case 1:] $\Theta _{j}^{0}$ exhibits neither structural break nor group
structure.\newline
In this case, $K_{j}^{(1)}=K_{j}^{(2)}=1$, and $\Theta _{j,it}^{0}=\alpha
_{j}$ $\forall \left( i,t\right) \in \lbrack N]\times \lbrack T]$. Without
loss of generality, assume that $\alpha _{j}>0.$ Then by the SVD, we have 
\begin{align*}
& \mathcal{U}_{j}=\frac{1}{\sqrt{N}}\iota _{N}\in \mathbb{R}^{N\times
1},\quad \Sigma _{j}=\alpha _{j},\quad \mathcal{V}_{j}=\frac{1}{\sqrt{T}}%
\iota _{T}\in \mathbb{R}^{T\times 1}, \\
& U_{j}=\alpha _{j}\iota _{N}\in \mathbb{R}^{N\times 1},\quad V_{j}=\iota
_{T}\in \mathbb{R}^{T\times 1},
\end{align*}%
where $\iota _{d}=\left( 1,\cdots ,1\right) ^{\prime }\in \mathbb{R}%
^{d\times 1}$ for any natural number $d$. Obviously, $r_{j}=1$ in Case 1.

\item[Case 2:] $\Theta _{j}^{0}$ exhibits no structural break but a group
structure.\newline
In this case, $K_{j}^{(1)}=K_{j}^{(2)}=K_{j}$, $%
G_{k,j}^{(1)}=G_{k,j}^{(2)}=G_{k,j}$, $N_{k,j}^{(1)}=N_{k,j}^{(2)}=N_{k,j}$, 
$\pi _{k,j}^{(1)}=\pi _{k,j}^{(2)}=\pi _{k,j}~\forall k\in \left[ K_{j}%
\right] $, and $\Theta _{j,it}^{0}=\operatornamewithlimits{\sum}\limits%
_{k\in \lbrack K_{j}]}\alpha _{k,j}\mathbf{1}\left\{ i\in G_{k,j}\right\} $
for $t\in \lbrack T]$. Therefore, we have 
\begin{align*}
& \mathcal{U}_{j,i}=\frac{\sum_{k\in \lbrack K_{j}]}\alpha _{k,j}\mathbf{1}%
\left\{ i\in G_{k,j}\right\} }{\sqrt{\sum_{k\in \lbrack K_{j}]}N_{k,j}\left(
\alpha _{k,j}\right) ^{2}}},\quad \Sigma _{j}=\sqrt{\sum_{k\in \lbrack
K_{j}]}\pi _{k,j}\left( \alpha _{k,j}\right) ^{2}},\quad \mathcal{V}_{j}=%
\frac{1}{\sqrt{T}}\iota _{T}, \\
& u_{i,j}=\sum_{k\in \lbrack K_{j}]}\alpha _{k,j}\mathbf{1}\left\{ i\in
G_{k,j}\right\} ,\quad V_{j}=\iota _{T},
\end{align*}%
where $\mathcal{U}_{j,i}$ is the $i$-th element in $\mathcal{U}_{j}$.
Obviously, $r_{j}=1$ in this case.

\item[Case 3:] $\Theta _{j}^{0}$ exhibits both a structural break and a
group structure. \ 
\end{description}

\begin{itemize}
\item[\textbf{(i)}] $K_{j}^{(1)}\neq K_{j}^{(2)}$, where we have different
numbers of groups before and after the break;

\item[\textbf{(ii)}] $K_{j}^{(1)}=K_{j}^{(2)}=K_{j}$ and $G_{k,j}^{(1)}\neq
G_{k,j}^{(2)}$, where we have the same number of groups before and after the
break, but the group memberships change after the break point;

\item[\textbf{(iii)}] $K_{j}^{(1)}=K_{j}^{(2)}=K_{j}$, $%
G_{k,j}^{(1)}=G_{k,j}^{(2)}=G_{k,j}$ for $\forall k\in \lbrack K_{j}]$, and $%
\alpha _{k,j}^{(1)}\neq \alpha _{k,j}^{(2)}$ for at least one $k\in \lbrack
K_{j}]$, where even though neither the number of groups nor group membership
changes after the break, there exists at least one group whose slope
coefficients change.
\end{itemize}

For any positive integer $d$, we use $\mathbf{0}_{d}$ to denote a $d\times 1$
vector of zeros. The following lemma lays down the foundation for break point 
detection in our model.

\begin{lemma}
{\small \label{Lem:idfct} } For any $j\in [p]$ such that $\Theta_{j}^{0}$
lies in Case 3 above, we have $rank(\Theta_{j}^{0})\leq 2$. When $%
rank(\Theta_{j}^{0})=2$, we have

\begin{itemize}
\item[(i)] $\Theta_{j}^{0}=\mathcal{U}_{j}\Sigma_{j}\mathcal{V}%
_{j}^{\prime}=U_{j}V_{j}^{\prime }$ where $U_{j}=\mathcal{U}_{j}\Sigma_{j}/%
\sqrt{T},$ $V_{j}=\sqrt{T}\mathcal{V}_{j}=D_{j}R_{j}$$,$ $D_{j}=%
\begin{bmatrix}
\frac{1}{\sqrt{\tau_{T}}}\iota_{T_{1}} & \mathbf{0}_{T_{1}} \\ 
\mathbf{0}_{T-T_{1}} & \frac{1}{\sqrt{1-\tau_{T}}}\iota_{T-T_{1}}%
\end{bmatrix}%
$ and $R_{j}^{\prime }R_{j}=I_{2};$

\item[(ii)] $\left\Vert \frac{v_{t,j}^{0}}{\left\Vert v_{t,j}^{0}\right\Vert 
}-\frac{v_{t^{\ast },j}^{0}}{\left\Vert v_{t^{\ast },j}^{0}\right\Vert }%
\right\Vert =\sqrt{2}$ for any $t\leq T_{1}$ and $t^{\ast }>T_{1}$.
\end{itemize}
\end{lemma}

By Lemma \ref{Lem:idfct} for Case 3 and the above analyses for Cases 1 and
2, we conclude that $\Theta _{j}^{0}$ is a low-rank matrix with rank equal
to or less than $2$. In view of the low-rank structure of the slope
matrices, we propose to adopt the NNR to obtain the preliminary estimates
below. Moreover, under Case 3, Lemma \ref{Lem:idfct}(ii) indicates that
singular vectors of the slope matrix with rank 2 contain the structural
break information.

\section{Estimation}

\label{sec:estimation} In this section we provide the estimation algorithm.
We first assume that the ranks $r_{j}$ for $j\in [p]\cup \{0\}$ are known,
and then propose a singular value thresholding (SVT) procedure to estimate
them. After we recover the break point and the latent group structures, we
propose consistent estimates of the group-specific parameters.

\subsection{Estimation Algorithm}

Given $r_{j},~\forall j\in \lbrack p]\cup \{0\}$, we propose the following
four-step procedure to estimate the break point and to recover the latent
group structures before and after the break.

\begin{description}
\item[\textbf{Step 1:}] \textbf{Nuclear Norm Regularization (NNR).} We run
the nuclear norm regularized regression and obtain the preliminary estimates
as follows: 
\begin{equation}
\{\tilde{\Theta}_{j}\}_{j\in \lbrack p]\cup \{0\}}=\argmin\limits_{\left\{
\Theta _{j}\right\} _{j=0}^{p}}\frac{1}{NT}\left\Vert
Y-\sum_{j=1}^{p}X_{j}\odot \Theta _{j}-\Theta _{0}\right\Vert
^{2}+\sum_{j=0}^{p}\nu _{j}\left\Vert \Theta _{j}\right\Vert _{\ast },
\label{obj}
\end{equation}%
where $\nu _{j}$ is the tuning parameter for $j\in \lbrack p]\cup \{0\}$.
For each $j$, conduct the SVD: $\frac{1}{\sqrt{NT}}\tilde{\Theta}_{j}=\hat{%
\tilde{\mathcal{U}}}_{j}\hat{\tilde{\Sigma}}_{j}\hat{\tilde{\mathcal{V}}}%
_{j}^{\prime }$, where $\hat{\tilde{\Sigma}}_{j}$ is a diagonal matrix with
the diagonal elements being the descending singular values of $\tilde{\Theta}%
_{j}$. Let $\tilde{\mathcal{V}}_{j}$ consist of the first $r_{j}$ columns of 
$\hat{\tilde{\mathcal{V}}}_{j},$ and $\tilde{V}_{j}=\sqrt{T}\tilde{\mathcal{V%
}}_{j}.$ Let $\tilde{v}_{t,j}^{\prime }$ denote the $t$-th row of $\tilde{V}%
_{j}$ for $t\in \lbrack T]$.

\item[\textbf{Step 2:}] \textbf{Row- and Column-Wise Regressions.} First run
the row-wise regressions of $Y_{it}$ on $\left( \tilde{v}_{t,0},\{\tilde{v}%
_{t,j}X_{j,it}\}_{j\in \lbrack p]}\right) $ to obtain $\{\dot{u}%
_{i,j}\}_{j\in \lbrack p]\cup \{0\}}$ for $i\in \lbrack N]$. Then run the
column-wise regressions of $Y_{it}$ on $\left( \dot{u}_{i,0},\{\dot{u}%
_{i,j}X_{j,it}\}_{j\in \lbrack p]}\right) $ to obtain $\{\dot{v}%
_{t,j}\}_{j\in \lbrack p]\cup \{0\}}$ for $t\in \lbrack T].$ Let $\dot{\Theta%
}_{j,it}=\dot{u}_{i,j}^{\prime }\dot{v}_{t,j}$ for $\left( i,t\right) \in
\lbrack N]\times \lbrack T]$ and $j\in \lbrack p]\cup \{0\}$. Specifically,
the row- and column-wise regressions are given by 
\begin{align}
\left\{ \dot{u}_{i,j}\right\} _{j\in \lbrack p]\cup \{0\}}& =\argmin%
_{\{u_{i,j}\}_{j\in \lbrack p]\cup \{0\}}}\frac{1}{T}\sum_{t\in \lbrack
T]}\left( Y_{it}-u_{i,0}^{\prime }\tilde{v}_{t,0}-\sum_{j=1}^{p}u_{i,j}^{%
\prime }\tilde{v}_{t,j}X_{j,it}\right) ^{2},\quad i\in \lbrack N],
\label{obj2} \\
\left\{ \dot{v}_{t,j}\right\} _{j\in \lbrack p]\cup \{0\}}& =\argmin%
_{\{v_{t,j}\}_{j\in \lbrack p]\cup \{0\}}}\frac{1}{N}\sum_{i\in \lbrack
N]}\left( Y_{it}-v_{t,0}^{\prime }\dot{u}_{i,0}-\sum_{j=1}^{p}v_{t,j}^{%
\prime }\dot{u}_{i,j}X_{j,it}\right) ^{2},\quad t\in \lbrack T].
\end{align}

\item[\textbf{Step 3:}] \textbf{Break Point Estimation.} We estimate the
break point as follows: 
\begin{equation}
\hat{T}_{1}=\argmin_{s\in \{2,\cdots ,T-1\}}\frac{1}{pNT}\sum_{j\in \lbrack
p]}\sum_{i\in \lbrack N]}\left\{ \sum_{t=1}^{s}\left( \dot{\Theta}_{j,it}-%
\bar{\dot{\Theta}}_{j,i}^{(1s)}\right) ^{2}+\sum_{t=s+1}^{T}\left( \dot{%
\Theta}_{j,it}-\bar{\dot{\Theta}}_{j,i}^{(2s)}\right) ^{2}\right\} ,
\label{BP_estimates}
\end{equation}%
where $\bar{\dot{\Theta}}_{j,i}^{(1s)}=\frac{1}{s}\sum_{t=1}^{s}{\dot{\Theta}%
_{j,it}}$ and $\bar{\dot{\Theta}}_{j,i}^{(2s)}=\frac{1}{T-s}\sum_{t=s+1}^{T}{%
\dot{\Theta}_{j,it}}$.

\item[\textbf{Step 4:}] \textbf{Sequential Testing K-means (STK).} In this
step, we estimate the number of groups and the group membership before and
after the break by using the STK algorithm. For each $j\in \lbrack p]$,
define $\dot{\Theta}_{j,i}^{(1)}=(\dot{\Theta}_{j,i1},\cdots ,\dot{\Theta}%
_{j,i\hat{T}_{1}})^{\prime }$, $\dot{\Theta}_{j,i}^{(2)}=(\dot{\Theta}_{j,i,%
\hat{T}_{1}+1},\cdots ,\dot{\Theta}_{j,iT})^{\prime }$, $\dot{\beta}%
_{i}^{(1)}=\frac{1}{\sqrt{\hat{T}_{1}}}(\dot{\Theta}_{1,i}^{(1)\prime
},\cdots ,\dot{\Theta}_{p,i}^{(1)\prime })^{\prime },$ and $\dot{\beta}%
_{i}^{(2)}=\frac{1}{\sqrt{\hat{T}_{2}}}(\dot{\Theta}_{1,i}^{(2)\prime
},\cdots ,\dot{\Theta}_{p,i}^{(2)\prime })^{\prime }$. 
Let $z_{\varsigma }$ be some predetermined value which will be specified in
the next subsection. Given the subsample before and after the estimated
break point, initialize $m=1$ and classify each subsample into $m$ groups
by the K-means algorithm with group membership obtained as $\hat{\mathcal{G}%
}_{m}^{(\ell )}:=\{\hat{G}_{k,m}^{(\ell )}\}_{k\in \lbrack m]}$. Next, we
construct a suitable test statistic $\hat{\Gamma}_{m}^{(\ell )}$, defined by \eqref{eq:gamma-m} in the next subsection, and compare it to its critical value $z_{\varsigma }$ at significance level $\varsigma$ under the null hypothesis of $m$ subgroups, setting $m=m+1$ and moving to the next iteration if $\hat{\Gamma}_{m}^{(\ell )}>z_{\varsigma }$ and stopping the STK algorithm otherwise. Lastly, define $\hat{K}^{(\ell )}=m$ and $\hat{\mathcal{G}}^{(\ell )}=\hat{%
\mathcal{G}}_{m}^{(\ell )}$. The next subsection provides a detailed breakdown of these steps in the STK algorithm. 
\begin{figure}[h]
\centering
\includegraphics[width=0.8\textwidth]{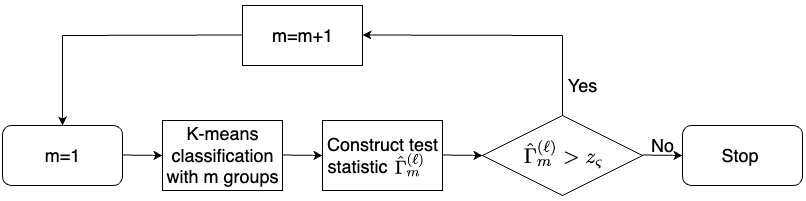}
\caption{The flow chart of STK algorithm}
\label{fig:STK}
\end{figure}

\end{description}

Several remarks are in order. First, the ranks of the intercept
and slope matrices are assumed known in Step 1 but consistent estimates 
for them are proposed in the SVT below. Second, we obtain preliminary estimates by
NNR based on the low-rank structure of the intercept and slope matrices in
the model. These estimates are consistent in terms of the Frobenius norm but pointwise or uniform convergence for their elements is not established. Nonetheless, SVD can be employed to obtain preliminary estimates of the
factors and factor loadings to be used subsequently. Third, row- and column-wise linear regressions are conducted to obtain updated estimates of the factors and factor loadings for which we can establish pointwise and
uniform convergence rates. Fourth, using the consistent estimates obtained in
the second step, we can estimate the break point in Step 3 consistently by
using a binary segmentation process. Fifth, the STK algorithm in Step
4 then yields the estimated number of groups and the group memberships together. 

In the latent group literature, it is standard and popular to assume the
number of groups in the K-means algorithm is known and then to estimate the
number of groups by using certain information criteria. In this case, one
needs to consider not only under- and just-fitting cases, but also
over-fitting cases. It is well known that the major difficulty with this
approach is showing that the over-fitting case occurs with probability
approaching zero. The STK algorithm ensures a focus on the
under-and just-fitting cases, which helps to avoid the difficulty caused
by K-means classification with a larger than true
number of groups. In addition, although this sequential algorithm approach is adopted, the
error from the previous iteration does not accumulate in the following
iterations owing to fact that the classification in each iteration is new
and is not based on the K-means outcomes in previous iterations.

\subsection{The STK algorithm}

This subsection describes the K-means algorithm and the construction
of the test statistics $\hat{\Gamma}_{m}^{(\ell )}$ that are used in the STK algorithm for $%
\ell \in \{1,2\}$.

First, we define the objective function for the K-means algorithm with $m$
clusters at each iteration. Let $a_{k,m}^{(\ell )}$ be a $p\hat{T}_{1}\times
1$ and $p(T-\hat{T}_{1})\times 1$ vector for $\ell =1,2$, respectively. We
obtain the group membership with $m$ groups by solving the following
minimization problem 
\begin{equation}
\left\{ \dot{a}_{k,m}^{(\ell )}\right\} _{k\in \lbrack m]}=\argmin_{\left\{
a_{k,m}^{(\ell )}\right\} _{k\in \lbrack m]}}\frac{1}{N}\sum_{i\in \lbrack
N]}\min_{k\in \lbrack m]}\left\Vert \dot{\beta}_{i}^{(\ell )}-a_{k,m}^{(\ell
)}\right\Vert ^{2},  \label{kmeans_obj}
\end{equation}%
which yields the membership estimates for each individual at the $m$-th
iteration as 
\begin{equation}
\hat{g}_{i,m}^{(\ell )}=\argmin_{k\in \lbrack m]}\left\Vert \dot{\beta}%
_{i}^{(\ell )}-\dot{a}_{k,m}^{(\ell )}\right\Vert \quad \forall i\in \lbrack
N].  \label{group estimates}
\end{equation}%
Let $\hat{G}_{k,m}^{(\ell )}:=\{i\in \lbrack N]:\hat{g}_{i,m}^{(\ell )}=k\}$.

Second, we discuss the construction of the test statistic based on the idea
of homogeneity test for several subsamples. At iteration $m$, we have $m$
potential subgroups $(\hat{G}_{1,m}^{(\ell )},\cdots ,\hat{G}_{m,m}^{(\ell
)})$ after the K-means classification for $\ell =1$ and 2. Let $\mathcal{%
\hat{T}}_{1}=[\hat{T}_{1}]$, $\mathcal{\hat{T}}_{2}=[T]\backslash \lbrack 
\hat{T}_{1}]$, $\mathcal{\hat{T}}_{1,-1}=\mathcal{\hat{T}}_{1}\backslash \{%
\hat{T}_{1}\}$, $\mathcal{\hat{T}}_{2,-1}=\mathcal{\hat{T}}_{2}\backslash
\{T\}$, $\mathcal{\hat{T}}_{1,j}=\{1+j,\cdots ,\hat{T}_{1}\}$, and $\mathcal{%
\hat{T}}_{2,j}=\{\hat{T}_{1}+1+j,\cdots ,T\}$ for some specific $j\in \hat{%
\mathcal{T}}_{\ell ,-1}$. Based on these estimated subgroups, we can obtain
the estimates of the coefficients, factors and factor loadings for each
subgroup in regime $\ell $ as follows: 
\begin{equation*}
\left( \left\{ \hat{\theta}_{i,k,m}^{(\ell )}\right\} _{i\in \hat{G}%
_{k,m}^{(\ell )}},\hat{\Lambda}_{k,m}^{(\ell )},\hat{F}_{k,m}^{(\ell
)}\right) =\argmin_{\left\{ \theta _{i},\lambda _{i},f_{t}\right\} _{i\in 
\hat{G}_{k,m}^{(\ell )},t\in \mathcal{\hat{T}}_{\ell }}}\sum_{i\in \hat{G}%
_{k,m}^{(\ell )}}\sum_{t\in \mathcal{\hat{T}}_{\ell }}\left(
Y_{it}-X_{it}^{\prime }\theta _{i}-\lambda _{i}^{\prime }f_{t}\right) ^{2},
\end{equation*}%
where $\hat{\Lambda}_{k,m}^{(\ell )}=\{\hat{\lambda}_{i,k,m}^{\left( \ell
\right) }\}_{i\in \hat{G}_{k,m}^{(\ell )}}$ and $\hat{F}_{k,m}^{(\ell )}=\{%
\hat{f}_{t,k,m}^{\left( \ell \right) }\}_{t\in \hat{\mathcal{T}}_{\ell }}.$
For all $i\in \lbrack N]$ and $t\in \lbrack T]$, define the residuals%
\begin{equation*}
\hat{e}_{it}=\sum_{\ell =1}^{2}\left( Y_{it}-\hat{f}_{t,k,m}^{(\ell )\prime }%
\hat{\lambda}_{i,k,m}^{(\ell )}-X_{it}^{\prime }\hat{\theta}_{i,k,m}^{(\ell
)}\right) \mathbf{1}\{t\in \mathcal{\hat{T}}_{\ell }\}.
\end{equation*}%
Let $\hat{X}_{i}^{(1)}=(X_{i1},\cdots ,X_{i\hat{T}_{1}})^{\prime }$, $\hat{X}%
_{i}^{(2)}=(X_{i,\hat{T}_{1}+1},\cdots ,X_{iT})^{\prime }$, and $\hat{T}%
_{2}=T-\hat{T}_{1}.$ Define 
\begin{align*}
& \hat{\bar{\theta}}_{k,m}^{(\ell )}=\frac{1}{|\hat{G}_{k,m}^{(\ell )}|}%
\sum_{i\in \hat{G}_{k,m}^{(\ell )}}\hat{\theta}_{i,k,m}^{(\ell )},\quad M_{%
\hat{F}_{k,m}^{(\ell )}}=I_{\hat{T}_{\ell }}-\frac{1}{\hat{T}_{\ell }}\hat{F}%
_{k,m}^{(\ell )}\hat{F}_{k,m}^{(\ell )\prime }, \\
& \hat{S}_{ii,k,m}^{(\ell )}=\frac{1}{\hat{T}_{\ell }}\hat{X}_{i}^{(\ell
)\prime }M_{\hat{F}_{k,m}^{(\ell )}}\hat{X}_{i}^{(\ell )},\quad \hat{a}%
_{ii,k}^{(\ell )}=\hat{\lambda}_{i,k,m}^{(\ell )\prime }\left( |\hat{G}%
_{k,m}^{(\ell )}|^{-1}\hat{\Lambda}_{k,m}^{(\ell )\prime }\hat{\Lambda}%
_{k,m}^{(\ell )}\right) ^{-1}\hat{\lambda}_{i,k,m}^{(\ell )}.
\end{align*}%
Let $\hat{z}_{it}^{(\ell )\prime }$ be the $t$-th row of $M_{\hat{F}%
_{k,m}^{(\ell )}}\hat{X}_{i}^{(\ell )}.$ For each subgroup $\hat{G}%
_{k,m}^{(\ell )}$ with $k\in \lbrack m]$, we follow the lead of \cite%
{pesaran2008testing} and \cite{ando2015simple} and define the following test statistic components
\begin{equation} \label{eq:gamma-km}
\hat{\Gamma}_{k,m}^{(\ell )}=\sqrt{|\hat{G}_{k,m}^{(\ell )}|}\cdot \frac{%
\frac{1}{|\hat{G}_{k,m}^{(\ell )}|}\sum_{i\in \hat{G}_{k,m}^{(\ell )}}\hat{%
\mathbb{S}}_{i,k,m}^{(\ell )}-p}{\sqrt{2p}},
\end{equation}%
where 
\begin{align*}
& \hat{\mathbb{S}}_{i,k,m}^{(\ell )}=\hat{T}_{\ell }(\hat{\theta}%
_{i,k,m}^{(\ell )}-\hat{\bar{\theta}}_{k,m}^{(\ell )})^{\prime }\hat{S}%
_{ii,k,m}^{(\ell )}(\hat{\Omega}_{i,k,m}^{(\ell )})^{-1}\hat{S}%
_{ii,k,m}^{(\ell )}(\hat{\theta}_{i,k}^{(\ell )}-\hat{\bar{\theta}}%
_{k}^{(\ell )})\left( 1-\frac{\hat{a}_{ii,k}^{(\ell )}}{|\hat{G}%
_{k,m}^{(\ell )}|}\right) ^{2}, \\
& \hat{\Omega}_{i,k,m}^{(\ell )}=\frac{1}{\hat{T}_{\ell }}\sum_{t\in \hat{%
\mathcal{T}}_{\ell }}\hat{z}_{it}^{(\ell )}\hat{z}_{it}^{(\ell )\prime }\hat{%
e}_{it}^{2}+\frac{1}{\hat{T}_{\ell }}\sum_{j\in \hat{\mathcal{T}}_{\ell
,-1}}k(j/S_{T})\sum_{t\in \mathcal{\hat{T}}_{\ell ,j}}[\hat{z}_{it}^{(\ell )}%
\hat{z}_{i,t+j}^{(\ell )\prime }\hat{e}_{it}\hat{e}_{i,t+j}+\hat{z}%
_{i,t-j}^{(\ell )}\hat{z}_{it}^{(\ell )\prime }\hat{e}_{i,t-j}\hat{e}_{i,t}],
\end{align*}%
$k(\cdot )$ is a kernel function, $S_{T}$ is a bandwidth/truncation
parameter, and $\hat{\Omega}_{i,k,m}^{(\ell )}$ is a traditional HAC
estimator. Using the components \eqref{eq:gamma-km} we now define the test statistic
\begin{align} \label{eq:gamma-m}
\hat{\Gamma}_{m}^{(\ell )}=\max_{k\in \lbrack m]}(\hat{\Gamma%
}_{k,m}^{(\ell )})^{2}.
\end{align}

We will show that $\hat{\Gamma}_{m}^{(\ell )}$ is asymptotically distributed
as the maximum of $m$ independent $\chi ^{2}(1)$ random variables under the
null hypothesis that the slope coefficients in each of the $m$ subsamples are
homogeneous, whereas it diverges to infinity under the alternative. Let $%
z_{\varsigma }$ denote the critical value at significance level $\varsigma ,$
which is calculated from the maximum of $m$ independent $\chi ^{2}(1)$
random variables. We reject the null of $m$ subgroups in favor of more
groups at level $\varsigma $ if $\hat{\Gamma}_{m}^{(\ell )}>z_{\varsigma }.$

\subsection{Rank Estimation}

To obtain the rank estimator, we use the low-rank estimators from (\ref{obj}%
) and estimate $r_{j}$ by the singular value thresholding (SVT) criterion 
\begin{equation*}
\hat{r}_{j}=\sum_{i=1}^{N\wedge T}\mathbf{1}\left\{ \sigma _{i}\left( \tilde{%
\Theta}_{j}\right) \geq 0.5\left( \nu _{j}\left\Vert \tilde{\Theta}%
_{j}\right\Vert _{op}\right) ^{1/2}\right\} \quad \forall j\in \{0\}\cup
\lbrack p],
\end{equation*}%
where $\sigma _{i}\left( A\right) $ denotes the $i$-th largest singular
value of $A$ and $N\wedge T=\min (N,T).$ By arguments as used in the proof
of Proposition D.1 in \cite{chernozhukov2019inference} and that of Theorem
3.2 in \cite{hong2023profile}, we can show that $\mathbb{P}(\hat{r}%
_{j}=r_{j})\rightarrow 1$ for each $j$ as $\left( N,T\right) \rightarrow
\infty .$

\subsection{Parameter Estimation}

Once we obtain the estimated break point, the number of groups and the group
membership before and after the estimated break point, we can estimate the
group-specific slope coefficients $\{\alpha _{k}^{(\ell )}\}_{k\in \lbrack 
\hat{K}^{(\ell )}]}$ along with the factors and factor loadings as follows 
\begin{equation}
\left( \hat{\Lambda}^{\left( \ell \right) },\hat{F}^{\left( \ell \right)
},\left\{ \hat{\alpha}_{k}^{(\ell )}\right\} _{k\in \lbrack \hat{K}^{(\ell
)}]}\right) =\argmin\mathbb{L}\left( \Lambda ,F,\left\{ a_{k}^{(\ell
)}\right\} _{k\in \lbrack \hat{K}^{(\ell )}]}\right) ,  \label{Post_estimate}
\end{equation}%
where $\mathbb{L}\left( \Lambda ,F,\left\{ a_{k}^{(\ell )}\right\} _{k\in
\lbrack \hat{K}^{(\ell )}]}\right) =\frac{1}{N\hat{T}_{\ell }}\sum_{k=1}^{%
\hat{K}^{(\ell )}}\sum_{i\in \hat{G}_{k}^{(\ell )}}\sum_{t\in \mathcal{\hat{T%
}}_{\ell }}\left( Y_{it}-\lambda _{i}^{\prime }f_{t}-X_{it}^{\prime
}a_{k}^{(\ell )}\right) ^{2}.$ Here, we ignore the fact that the prior- and
post-break regimes share the same set of factor loadings and estimate the
group-specific parameters separately for the two regimes at the cost of
sacrificing some efficiency for the factor loading estimates. Alternatively,
we can pool the observations before and after the break to estimate the
parameters as 
\begin{equation*}
\left( \hat{\Lambda},\hat{F},\left\{ \hat{\alpha}_{k}^{(1)}\right\} _{k\in
\lbrack \hat{K}^{(1)}]},\left\{ \hat{\alpha}_{k}^{(2)}\right\} _{k\in
\lbrack \hat{K}^{(2)}]}\right) =\argmin\mathbb{L}\left( \Lambda ,F,\left\{
a_{k}^{(1)}\right\} _{k\in \lbrack \hat{K}^{(1)}]},\left\{
a_{k}^{(2)}\right\} _{k\in \lbrack \hat{K}^{(2)}]}\right) 
\end{equation*}%
where 
\begin{equation}
\mathbb{L}\left( \Lambda ,F,\left\{ a_{k}^{(1)}\right\} _{k\in \lbrack \hat{K%
}^{(1)}]},\left\{ a_{k}^{(2)}\right\} _{k\in \lbrack \hat{K}^{(2)}]}\right) =%
\mathbb{L}\left( \Lambda ,F,\left\{ a_{k}^{(1)}\right\} _{k\in \lbrack \hat{K%
}^{(1)}]}\right) +\mathbb{L}\left( \Lambda ,F,\left\{ a_{k}^{(2)}\right\}
_{k\in \lbrack \hat{K}^{(2)}]}\right) .
\end{equation}

In either case, as should be clear because of the presence of the group structures,
establishment of the asymptotic properties of the post-classification
estimators of the group-specific slope coefficients becomes much more
involved than in \cite{bai2009panel} and \cite{moon2017dynamic}. For
this reason, we will focus on the estimates defined in (\ref{Post_estimate}).

\section{Asymptotic Theory}

\label{sec:theory} This section develops the asymptotic properties of
the estimators introduced above.

\subsection{Basic Assumptions}

Define $e_{i}=\left( e_{i1},\cdots ,e_{iT}\right) ^{\prime }$ and $%
X_{j,i}=\left( X_{j,i1},\cdots ,X_{j,iT}\right) ^{\prime }.$ Let $V_{j}^{0}$
be a $T\times r_{j}$ matrix with its $t$-th row being $v_{t,j}^{0\prime }$,
and $U_{j}^{0}$ be the $N\times r_{j}$ matrix with its $i$-th row being $%
u_{i,j}^{0\prime }$. Throughout the paper, we treat the factors $%
\{V_{j}^{0}\}_{j\in \lbrack p]\cup \{0\}}$ as random and their loadings $%
\{U_{j}^{0}\}_{j\in \lbrack p]\cup \{0\}}$ as deterministic. Let $\mathscr{D}%
:=\sigma (\{V_{j}^{0}\}_{j\in \lbrack p]\cup \{0\}})$ denote the minimum $%
\sigma $-field generated by $\{V_{j}^{0}\}_{j\in \lbrack p]\cup \{0\}}.$
Similarly, let $\mathscr{G}_{t}:=\sigma (\mathscr{D},\left\{ X_{is}\right\}
_{i\in \lbrack N],s\leq t+1},\left\{ e_{is}\right\} _{i\in \lbrack N],s\leq
t})$. Let $\max_{i}=\max_{i\in \lbrack N]},$ $\max_{t}=\max_{t\in \lbrack
T]}\ $and $\max_{i,t}=\max_{i\in \lbrack N],t\in \lbrack T]}.$ Let $M$ and $%
C $ be generic bounded positive constants which may vary across lines. \ 

\begin{ass}
\label{ass:1}

\begin{enumerate}
\item[(i)] $\left\{ e_{it},X_{it}\right\}_{t\in [T]}$ are conditionally
independent across $i$ given $\mathscr{D}$.

\item[(ii)] $\mathbb{E}(e_{it}|X_{it},\mathscr{D})=0$.

\item[(iii)] For each $i$, $\left\{ \left( e_{it},X_{it}\right) ,t\geq
1\right\} $ is strong mixing conditional on $\mathscr{D}$ with the mixing
coefficient $\alpha _{i}(\cdot )$ satisfying $\max_{i}\alpha _{i}(z)\leq
M\vartheta ^{z}$ for some constant $\vartheta \in \left( 0,1\right) $.

\item[(iv)] There exists a constant $C>0$ such that $\max_{i}\frac{1}{T}%
\sum_{t\in \lbrack T]}\left\Vert \xi _{it}\right\Vert ^{2}\leq C~a.s.\quad $%
and$\quad \max_{t}\frac{1}{N}\sum_{i\in \lbrack N]}\left\Vert \xi
_{it}\right\Vert ^{2}$ $\leq C~a.s.$ for $\xi _{it}=e_{it},$ $X_{it}$ and $%
X_{it}e_{it}$.

\item[(v)] $\max_{i,t}\mathbb{E}[\left\Vert \xi _{it}\right\Vert ^{q}\big|%
\mathscr{D}]\leq M~a.s.$ and $\max_{i,i^{\ast },t}\mathbb{E}[\left\Vert
X_{it}e_{i^{\ast }t}\right\Vert ^{q}\big|\mathscr{D}]\leq M~a.s.$ for some $%
q>8$ and $\xi _{it}=e_{it},$ $X_{it}$ and $X_{it}e_{it}.$

\item[(vi)] As $(N,T)\rightarrow \infty $, $\sqrt{N}(\log
N)^{2}T^{-1}\rightarrow 0$ and $T(\log N)^{2}N^{-3/2}\rightarrow 0$.
\end{enumerate}
\end{ass}

\begin{assp}{\ref*{ass:1}{$^*$} }
\label{ass:1*}
(i), (iv) and (v) are same as Assumption \ref{ass:1}(i), (iv) and (v). In addition:

\begin{itemize}
\item[(ii)] $\mathbb{E}(e_{it}|\mathscr{G}_{t-1})=0$ $\forall (i,t)\in
\lbrack N]\times \lbrack T]$, and $\max_{i,t}\mathbb{E(}e_{it}^{2}\big|%
\mathscr{G}_{t-1})\leq M~a.s.$.

\item[(iii)] $\left\{ e_{it}\right\} _{i\in \lbrack N]}$ is conditionally
independent across $t$ given $\mathscr{D}$.
\end{itemize}
\end{assp}


Assumption \ref{ass:1}(i) imposes conditional independence on $\left\{
e_{it},X_{it}\right\} _{t\in \lbrack T]}$ across the cross sectional units.
Assumption \ref{ass:1}(ii) is the conditional moment condition. Assumption \ref%
{ass:1}(iii) imposes conditional strong mixing conditions along the time
dimension. See \cite{Prakasa_Rao2009} for the definition of conditional
strong mixing and \cite{su2013testing} for an application in the panel
setup. Assumptions \ref{ass:1}(iv)-(v) impose conditions that 
restrict the tail behavior of $\xi _{it}$. Note that neither the regressors 
nor the errors are constrained to be bounded. Assumption \ref{ass:1}%
(vi) imposes restrictions on $N$ and $T$ but does not require $N$ and $
T$ to diverge at the the same rate. It is possible to allow $N$ to
diverge to infinity faster but not too much faster than $T$, and vice versa.

Assumption \ref{ass:1*} is used for the study of dynamic panel data models. To be
specific, Assumption \ref{ass:1*}(ii) requires that the error sequence $%
\left\{ e_{it},t\geq 1\right\} $ be a martingale difference sequence
(m.d.s.) with respect to the filtration $\mathscr{G}_{t},$ which allows for
lagged dependent variables in $X_{it}$. Assumption \ref{ass:1*}(iii) imposes
conditional independence of the errors over $t$. The presence of
serially correlated errors in dynamic panels typically induces endogeneity,
which invalidates least-squares-based PCA estimation.

\begin{ass}
\label{ass:2} rank$(\Theta _{j}^{0})=r_{j}\leq \bar{r}$ for $j\in \lbrack
p]\cup \{0\}$ and some fixed $\bar{r},$ and $\max_{j\in \lbrack p]\cup
\{0\}}||\Theta _{j}^{0}||_{\max }\leq M$.
\end{ass}

\noindent Assumption \ref{ass:2} imposes low-rank conditions on the coefficient
matrices, which facilitate the use of NNR in obtaining preliminary
estimates in the first step. As discussed in the previous section, we see
that the low-rank assumption for the slope matrices is satisfied for the
model in Section \ref{sec:model}. Moreover, we follow \cite{ma2020detecting}
and assume the elements of the coefficient matrices are uniformly bounded to
simplify the proofs. The boundedness of the slope coefficients is reasonable
given that their cardinality does not grow with the sample size. The
boundedness assumption for the intercept coefficient can be relaxed at the
cost of more lengthy arguments.

\begin{ass}
\label{ass:3} Let $\sigma _{l,j}$ denote the $l$-th largest singular values
of $\Theta _{j}^{0}$ for $j\in \lbrack p]\cup \{0\}.$ There exist some
constants $C_{\sigma }$ and $c_{\sigma }$ such that 
\begin{equation*}
\infty >C_{\sigma }\geq \lim \sup_{(N,T)\rightarrow \infty }\max_{j\in
\lbrack p]}\sigma _{1,j}\geq \lim \inf_{(N,T)\rightarrow \infty }\min_{j\in
\lbrack p]}\sigma _{r_{j},j}\geq c_{\sigma }>0.
\end{equation*}
\end{ass}

\noindent Assumption \ref{ass:3} imposes some conditions on the singular values of the
coefficient matrices. These ensure that only pervasive factors are allowed when
the matrices are written as a factor structure. The assumption can be
readily verified given the latent group structures of the slope coefficients.

Consider the SVD for $\Theta _{j}^{0}$: $\Theta _{j}^{0}=\mathcal{U}%
_{j}\Sigma _{j}\mathcal{V}_{j}^{\prime }$ $\forall j\in \lbrack p]\cup \{0\}$%
. Decompose $\mathcal{U}_{j}=\left( \mathcal{U}_{j,r},\mathcal{U}%
_{j,0}\right) $ and $\mathcal{V}_{j}=\left( \mathcal{V}_{j,r},\mathcal{V}%
_{j,0}\right) $ with $\left( \mathcal{U}_{j,r},\mathcal{V}_{j,r}\right) $
being the singular vectors corresponding to nonzero singular values and $%
\left( \mathcal{U}_{j,0},\mathcal{V}_{j,0}\right) $ being the singular
vectors corresponding to zero singular values. Hence, for any matrix $W\in 
\mathbb{R}^{N\times T}$, we define 
\begin{equation*}
\mathcal{P}_{j}^{\bot }\left( W\right) =\mathcal{U}_{j,0}\mathcal{U}%
_{j,0}^{\prime }W\mathcal{V}_{j,0}\mathcal{V}_{j,0}^{\prime },\quad \mathcal{%
P}_{j}\left( W\right) =W-\mathcal{P}_{j}^{\bot }\left( W\right) ,
\end{equation*}%
where $\mathcal{P}_{j}\left( W\right) $ can be seen as the linear projection
of matrix $W$ into the low-rank space with $\mathcal{P}_{j}^{\bot }\left(
W\right) $ being its orthogonal space. Let $\Delta _{\Theta _{j}}=\Theta
_{j}-\Theta _{j}^{0}$ for any $\Theta _{j}$. Based on the spaces constructed
above, with some positive constants $C_{1}$ and $C_{2}$, we define the
restricted set for full-sample parameters as follows: 
\begin{align}
\mathcal{R}(C_{1},C_{2})& :=\bigg\{(\{\Delta _{\Theta _{j}}\}_{j\in \lbrack
p]\cup \{0\}}):\sum_{j\in \lbrack p]\cup \{0\}}\left\Vert \mathcal{P}%
_{j}^{\bot }(\Delta _{\Theta _{j}})\right\Vert _{\ast }\leq C_{1}\sum_{j\in
\lbrack p]\cup \{0\}}\left\Vert \mathcal{P}_{j}(\Delta _{\Theta
_{j}})\right\Vert _{\ast },  \notag  \label{RS} \\
& \sum_{j\in \lbrack p]\cup \{0\}}\left\Vert \Theta _{j}\right\Vert ^{2}\geq
C_{2}\sqrt{NT}\bigg\}.
\end{align}

Lemma \ref{Lem:RS} in the online appendix shows that our nuclear norm
estimators are in a restricted set larger than (\ref{RS}), which derives from
the restriction on the Frobenius norm in the definition of $\mathcal{R}%
\left( C_{1},C_{2}\right) $. Intuitively, the first restriction in (\ref{RS}%
)\ means the projection onto the orthogonal low-rank space of the estimator
error can be controlled by its projection onto the low-rank space. Theorem %
\ref{Thm1} below largely hinges on this property.

\begin{ass}
\label{ass:4} For any $C_{2}>0$, there are constants $C_{3}$ and $C_{4}$
such that for any $(\{\Delta _{\Theta _{j}}\}_{j\in \lbrack p]\cup \{0\}})$ $%
\in \mathcal{R}(3,C_{2})$, we have 
\begin{equation*}
\left\Vert \Delta _{\Theta _{0}}+\sum_{j=1}^{p}\Delta _{\Theta _{j}}\odot
X_{j}\right\Vert ^{2}\geq C_{3}\sum_{j\in \lbrack p]\cup \{0\}}\left\Vert
\Delta _{\Theta _{j}}\right\Vert ^{2}-C_{4}(N+T)\quad \text{w.p.a.1.}
\end{equation*}
\end{ass}

\noindent Assumption \ref{ass:4} imposes the restricted strong convexity (RSC)
condition, which is similar to Assumption 3.1 in \cite%
{chernozhukov2019inference}. The latter authors also provide some sufficient
conditions to verify such an assumption.

Let $r=\sum_{j\in \lbrack p]\cup \{0\}}r_{j}.$ Define the following $r\times
r$ matrices: 
\begin{equation*}
\Phi _{i}=\frac{1}{T}\sum_{t=1}^{T}\phi _{it}^{0}\phi _{it}^{0\prime
}~~\forall i\in \lbrack N]\text{ and }\Psi _{t}=\frac{1}{N}\sum_{i\in
\lbrack N]}\psi _{it}^{0}\psi _{it}^{0\prime }~~\forall t\in \lbrack T],
\end{equation*}%
where $\phi _{it}^{0}=(v_{t,0}^{0\prime },v_{t,1}^{0\prime }X_{1,it},\cdots
,v_{t,p}^{0\prime }X_{p,it})^{\prime },$ and $\psi
_{it}^{0}=(u_{i,0}^{0\prime },u_{i,1}^{0\prime }X_{1,it},\cdots
,u_{i,p}^{0\prime }X_{p,it})^{\prime }.$

\begin{ass}
\label{ass:5} There exist constants $C_{\phi}$ and $c_{\phi}$ such that 
\begin{align*}
&\infty>C_{\phi}\geq \limsup_{T}\max_{t\in[T]}\lambda_{\max}(\Psi_{t})\geq%
\liminf_{T}\min_{t\in[T]}\lambda_{\min}(\Psi_{t})\geq c_{\phi}>0, \\
&\infty>C_{\phi}\geq \limsup_{N}\max_{i\in
[N]}\lambda_{\max}(\Phi_{i})\geq\liminf_{N}\min_{i\in
[N]}\lambda_{\min}(\Phi_{i})\geq c_{\phi}>0.
\end{align*}
\end{ass}
\noindent Assumption \ref{ass:5} is similar to Assumption 8 in \cite{ma2020detecting}
and it imposes some rank conditions.


\subsection{Asymptotics of NNR Estimators and Singular Vector
Estimators}

Let $\eta _{N,1}=\frac{\sqrt{\log T}}{\sqrt{N\wedge T}}$ and $\eta _{N,2}=%
\frac{\sqrt{\log (N\vee T)}}{\sqrt{N\wedge T}}(NT)^{1/q}.$ Let $\tilde{\sigma%
}_{k,j}$ denote the $k$-th largest singular value of $\tilde{\Theta}_{j}$
for $j\in \lbrack p]\cup \{0\}.$ Our first main result is about the
consistency of the first-stage NNR estimators and the second-stage singular
vector estimators.

\begin{theorem}
\label{Thm1} Suppose that Assumptions \ref{ass:1}--\ref{ass:4} hold. Then $%
\forall j\in \lbrack p]\cup \{0\}$, we have

\begin{itemize}
\item[(i)] $\frac{1}{\sqrt{NT}}||\tilde{\Theta}_{j}-\Theta
_{j}^{0}||=O_{p}(\eta _{N,1})$, $\max_{k\in \lbrack r_{j}]}|\tilde{\sigma}%
_{k,j}-\sigma _{k,j}|=O_{p}(\eta _{N,1})$, and $||V_{j}^{0}-\tilde{V}%
_{j}O_{j}||=O_{p}(\sqrt{T}\eta _{N,1})$ where $O_{j}$ is an orthogonal
matrix defined in the proof.

If in addition Assumption \ref{ass:5} is also satisfied, then we have

\item[(ii)] $\operatornamewithlimits{\max}\limits_{i\in \lbrack N]}||\dot{u}%
_{i,j}-O_{j}u_{i,j}^{0}||=O_{p}(\eta _{N,2}),$ $\operatornamewithlimits{\max}%
\limits_{t\in \lbrack T]}\left\Vert \dot{v}_{t,j}-O_{j}v_{t,j}^{0}\right%
\Vert _{2}=O_{p}(\eta _{N,2})$,

\item[(iii)] $\operatornamewithlimits{\max}\limits_{i\in \lbrack N],t\in
\lbrack T]}|\dot{\Theta}_{j,it}-\Theta _{j,it}^{0}|=O_{p}(\eta _{N,2})$.
\end{itemize}
\end{theorem}

Theorem \ref{Thm1}(i) reports the error bounds for $\tilde{\Theta}_{j},$ $%
\tilde{\sigma}_{k,j},$ and $\tilde{V}_{j}$. The $\log T$ term in the
numerator of $\eta _{N,1}$ is due to the use of some exponential inequality
for (conditional)\ strong mixing processes. Theorem \ref{Thm1}(ii)--(iii)
report the uniform convergence rate of the factor and factor loading
estimators. The extra $(NT)^{1/q}$ term in the definition of $\eta _{N,2}$
is due to the nonboundedness of $X_{j,it}$ in Assumption \ref{ass:1}(v), and
it disappears when $X_{j,it}$ is assumed to be uniformly bounded.

\subsection{Consistency of the Break Point Estimate}

Recall that $g_{i}^{(1)}$ and $g_{i}^{(2)}$ denote the true group individual 
$i$ belongs to before and after the break, respectively. To estimate the
break point consistently, we add the following condition.

\begin{ass}
\label{ass:6}

\begin{enumerate}
\item[(i)] $\sqrt{\frac{1}{N}\sum_{i\in \lbrack N]}||\alpha
_{g_{i}^{(1)}}-\alpha _{g_{i}^{(2)}}||^{2}}=C_{5}\zeta _{NT}$, where $C_{5}$
is a positive constant and $\zeta _{NT}\gg \eta _{N,2}$.

\item[(ii)] $\tau_{T}:=\frac{T_{1}}{T}\rightarrow \tau \in (0,1)$ as $%
T\rightarrow \infty.$
\end{enumerate}
\end{ass}

Assumption \ref{ass:6}(i) imposes conditions on the break size in order to
identify the break point. Note that we allow the average break size to
shrink to zero at a rate slower than $\sqrt{\frac{\log (N\vee T)}{N\wedge T%
}}(NT)^{1/q}$. This rate is of much bigger magnitude than the optimal $%
\left( NT\right) ^{-1/2}$-rate that can be detected in the panel threshold
regressions (PTRs) for several reasons. First, in PTRs, the slope
coefficients are usually assumed to be homogeneous so that each individual
is subject to the same change in the slope coefficients and one can use the
cross-sectional information effectively. In contrast, we allow for
heterogeneous slope coefficients here and the change can occur only for a
subset of cross section units but not all. In addition, in the presence of
latent group structures, we not only allow the slope coefficients of some
specific groups to change with group membership fixed, but also allow the
slope coefficient to remain the same for some groups while the group
memberships change after the break. Second, our break point estimation
relies on the binary segmentation idea borrowed from the time series
literature where one can allow break sizes of bigger magnitude than $%
T^{-1/2} $ in order to identify the break ratio consistently but not the
break point consistently. As is apparent, even though we require bigger break
sizes, we can estimate the break date consistently by using information from
both the cross section and time dimensions. Third, as mentioned above, the
additional term $\log (N\vee T)$ in the above rate is mainly due to the use
of an exponential inequality and the factor $(NT)^{1/q}$ is due to the fact
that we only assume the existence of $q$-th order moments for some random
variables.

The following theorem indicates that we can estimate the break date $T_{1}$
consistently.

\begin{theorem}
\label{Thm2} Suppose Assumptions \ref{ass:1}--\ref{ass:6} hold, with the
true break point being $T_{1}$ and the estimator defined in (\ref%
{BP_estimates}). Then $\mathbb{P}(\hat{T}_{1}=T_{1})\rightarrow 1$ as $%
\left( N,T\right) \rightarrow \infty $.
\end{theorem}

Theorem \ref{Thm2} shows that we can estimate the true break date
consistently w.p.a.1 despite the fact that we allow the break size to shrink
to zero at a certain rate.

\subsection{Consistency of the Estimates of the Number of Groups and the
Latent Group Structures}

To study the asymptotic properties of the estimates of the number of groups
and the recovery of the latent group structures, we first add the following
definition.

\begin{definition}
\label{def:1} Fix $K^{(\ell)}>1$ and $m\leq K^{(\ell)}$. The estimated group
structure $\hat{\mathcal{G}}_{m}^{(\ell)}$ satisfies the non-splitting
property (NSP) if for any pair of individuals in the same true group, the
estimated group labels are the same.
\end{definition}

\noindent Definition \ref{def:1} describes the non-splitting property introduced by 
\cite{Jin2022Optimal}. The latter authors show that the STK algorithm yields
the estimated group structures enjoying the NSP.

To proceed, we add the following assumptions.

\begin{ass}
\label{ass:7}

\begin{enumerate}
\item[(i)] Let $k_{s}$ and $k_{s^{\ast }}$ be different group indices.
Assume that $\min_{1\leq k_{s}<k_{s^{\ast }}\leq K^{(\ell )}}\left\Vert
\alpha _{k_{s}}^{(\ell )}-\alpha _{k_{s^{\ast }}}^{(\ell )}\right\Vert _{2}$ 
$\geq C_{5}$ for$~\ell \in \{1,2\}.$

\item[(ii)] Let $N_{k}^{(\ell )}$ be the number of individuals in group $k\ $%
for $k\in \lbrack K^{(\ell )}]$. Define $\pi _{k}^{(\ell )}=\frac{%
N_{k}^{(\ell )}}{N}$ for $\ell =1,2.$ Assume $K^{(\ell )}$ is fixed and 
\begin{equation*}
\infty >\bar{C}\geq \lim \sup_{N}\sup_{k\in \lbrack K^{(\ell )}]}\pi
_{k}^{(\ell )}\geq \lim \inf_{N}\inf_{k\in \lbrack K^{(\ell )}]}\pi
_{k}^{(\ell )}\geq \underline{c}>0,~\ell =1,2.
\end{equation*}

\item[(iii)] For any combination of the collection of $n$ true groups with $%
n\geq 2$, we have 
\begin{equation*}
\frac{T_{\ell }}{\sqrt{N}}\sum_{s=1}^{n}N_{k_{s}}^{(\ell )}\left\Vert
\sum_{s^{\ast }\in \left[ n\right] ,s^{\ast }\neq s}(\alpha _{k_{s^{\ast
}}}^{(\ell )}-\alpha _{k_{s}}^{(\ell )})\right\Vert ^{2}/(\log
N)^{1/2}\rightarrow \infty ,~\ell =1,2.
\end{equation*}
\end{enumerate}
\end{ass}

\noindent \textbf{Remark 3.} Assumption \ref{ass:7}(i)--(ii) is the standard
assumption for K-means algorithm, which is comparable to Assumption 4 in 
\cite{su2020strong} and greatly facilitates the subsequent analyses.
Assumption \ref{ass:7}(i) assumes that the minimum distance of two distinct
groups is bounded away from 0, and Assumption \ref{ass:7}(ii) imposes that
each group has asymptotically non-negligible number of units. For Assumption %
\ref{ass:7}(iii), it can be shown to hold under mild conditions. Below we
explain this assumption in detail. When $n=2$, it is clear that 
\begin{align*}
\frac{T_{\ell }}{\sqrt{N}}\sum_{s=1}^{n}N_{k_{s}}^{(\ell )}\left\Vert
\sum_{s^{\ast }\in \left[ n\right] ,s^{\ast }\neq s}(\alpha _{k_{s^{\ast
}}}^{(\ell )}-\alpha _{k_{s}}^{(\ell )})\right\Vert ^{2}& =\frac{T_{\ell }}{%
\sqrt{N}}\left( N_{k_{1}}^{(\ell )}\left\Vert \alpha _{k_{2}}^{(\ell
)}-\alpha _{k_{1}}^{(\ell )}\right\Vert ^{2}+N_{k_{2}}^{(\ell )}\left\Vert
\alpha _{k_{1}}^{(\ell )}-\alpha _{k_{2}}^{(\ell )}\right\Vert ^{2}\right) \\
& \geq \frac{C_{5}^{2}T_{\ell }(N_{k_{1}}^{(\ell )}+N_{k_{2}}^{(\ell )})}{%
\sqrt{N}}=O(T\sqrt{N})
\end{align*}%
by Assumptions \ref{ass:6}(ii) and \ref{ass:7}(i)--(ii). When $n>2$, we
consider a special case such that $S_{s}=:||\sum_{s^{\ast }\in \left[ n%
\right] ,s^{\ast }\neq s}(\alpha _{k_{s^{\ast }}}^{(\ell )}-\alpha
_{k_{s}}^{(\ell )})||=0$ for some specific $s=s_{0}\in \lbrack n]$. Then it
is easy to see $S_{s}$ is non-zero for all $s\in \lbrack n]\backslash
\{s_{0}\}$ under Assumption \ref{ass:7}(i). Hence, if we assume $S_{s}$ is
lower bounded by a constant $c$ for any $s\in \lbrack n]\backslash \{s_{0}\}$%
, Assumption \ref{ass:7}(iii) will holds naturally. Similar arguments hold
for the other general cases.

\begin{ass}
\label{ass:8} Let $\mathcal{T}_{1}=\left[ T_{1}\right] \ $and $\mathcal{T}%
_{2}=[T]\backslash \left[ T_{1}\right] $.

\begin{itemize}
\item[(i)] $\frac{1}{T_{\ell }}\sum_{t\in \mathcal{T}_{\ell
}}f_{t}^{0}f_{t}^{0\prime }\overset{p}{\longrightarrow }\Sigma _{F}^{(\ell
)}>0$ as $T\rightarrow \infty $. $\frac{1}{N_{k}^{(\ell )}}\Lambda
_{k}^{0,\left( \ell \right) \prime }\Lambda _{k}^{0,\left( \ell \right) }%
\overset{p}{\longrightarrow }\Sigma _{\Lambda ,k}^{\left( \ell \right) }>0$
as $N\rightarrow \infty $, where $\Lambda _{k}^{0,\left( \ell \right) }$ is
a stack of $\lambda _{i}^{0}$ for all individuals in group $k\ $and $k\in
\lbrack K^{(\ell )}]$.

\item[(ii)] There exists a constant $C>0$ such that $\max_{i}\frac{1}{%
T_{\ell }}\sum_{t\in \mathcal{T}_{\ell }}\left\Vert \xi _{it}\right\Vert
^{2}\leq C~a.s.$ for $\xi _{it}=e_{it},$ $X_{it}$ and $X_{it}e_{it}.$
\end{itemize}
\end{ass}



\noindent Assumption \ref{ass:8}(i) imposes some standard assumptions on the factors
and factor loadings. Assumption \ref{ass:8}(ii) is similar as Assumption \ref%
{ass:1}(iv), which strengthens Assumption \ref{ass:1}(iv) to hold for two
time regimes. 

The next theorem reports the asymptotic properties of the STK estimators.

\begin{theorem}
\label{Thm3} Let $\varsigma =\varsigma _{N}\rightarrow 0$ at rate $N^{-c}$
for some $c>0$ as $N\rightarrow \infty .$ Suppose that Assumption \ref{ass:1*}
and Assumptions \ref{ass:2}--\ref{ass:8} hold. Then for $\ell \in
\{1,2\}$, we have

\begin{itemize}
\item[(i)] if $m=K^{(\ell )}$,

\begin{itemize}
\item[(a)] $\operatornamewithlimits{\max}\limits_{i\in \lbrack N]}\mathbf{1}%
\{\hat{g}_{i,K^{(\ell )}}^{(\ell )}\neq g_{i}^{(\ell )}\}=0$ w.p.a.1,

\item[(b)] $\hat{\Gamma}_{K^{(\ell )}}^{(\ell )}$ is asymptotically
distributed as the maximum of $K^{(\ell )}$ independent $\chi ^{2}(1)$
random variables,

\item[(c)] $\mathbb{P}(\hat{K}^{(\ell )}\leq K^{(\ell )})\geq 1-\alpha +o(1)$%
,
\end{itemize}

\item[(ii)] if $m<K^{(\ell )}$, $\hat{\Gamma}_{m}^{(\ell )}/\log
N\rightarrow \infty $ w.p.a.1. Thus $\mathbb{P}(\hat{K}^{(\ell )}\neq
K^{(\ell )})\leq \varsigma +o(1)$.
\end{itemize}
\end{theorem}

Theorem \ref{Thm3} studies the asymptotic properties of the STK algorithm.
Since we allow $\varsigma =\varsigma _{N}$ to shrink to zero at rate $N^{-c},
$ the critical value $z_{\varsigma }$ diverges to infinity at rate $\log N$
as $N\rightarrow \infty $ by virtue of the tail properties of $\chi ^{2}(1)$ random
variables. At iteration $m$ such that $m<K^{(\ell )}$, w.p.a.1, the test
statistics $\hat{\Gamma}_{m}^{(\ell )}$ diverges to infinity at a rate
faster than $\log N$, which means the iteration will continue at the $(m+1)$%
-th iteration. At iteration $m$ with $m=K^{(\ell )}$, however, we can easily
find that $z_{\varsigma }\rightarrow \infty $ while the test statistic $%
\hat{\Gamma}_{m}^{(\ell )}$ is stochastically bounded. As a result, the
iteration stops w.p.a.1 and we have $\mathbb{P}(\hat{K}^{(\ell )}=K^{(\ell
)})\rightarrow 1$. As aforementioned, Theorem \ref{Thm3} ensures the
application of K-means algorithm only for the under-fitting and just-fitting
cases and it avoids the theoretical challenge of handling the over-fitting
case in the classification.


For dynamic panels, we can focus on Assumption \ref{ass:1*}, where the error
term is an m.d.s. Under this assumption, the HAC estimator $\hat{\Omega}%
_{i,k,m}^{(\ell )}$ degenerates to $\frac{1}{\hat{T}_{\ell }}\sum_{t\in \hat{%
\mathcal{T}}_{\ell }}\hat{z}_{it}^{(\ell )}\hat{z}_{it}^{(\ell )\prime }\hat{%
e}_{it}^{2}$. For static panels, we typically allow for serially correlated
errors and employ the HAC estimator, and the results in Theorem \ref{Thm3}
continue to hold with Assumption \ref{ass:1*} replaced by Assumption \ref%
{ass:1}. For the kernel function and bandwidth, we can follow \cite%
{andrews1991heteroskedasticity} and let $k\left( \cdot \right) $ belong to
the following class of kernels 
\begin{align*}
\mathcal{K}=\bigg\{& k(\cdot ):\mathbb{R}\mapsto \lbrack
-1,1]~|~k(0)=1,~k(u)=k(-u),~\int \left\vert k\left( u\right) \right\vert
du<\infty , \\
& k(\cdot )\text{ is continuous at 0 and at all but a finite number of other
points}\bigg\}.
\end{align*}%
See, e.g., \cite{andrews1991heteroskedasticity} and \cite%
{white2014asymptotic} for details.

\subsection{Distribution Theory for the Group-specific Slope Estimators}

For $\ell \in \{1,2\}$, let $\{\hat{\alpha}_{k}^{\ast (\ell )}\}_{k\in
K^{(\ell )}}$ be the oracle estimators of the group-specific slope
coefficients before and after the break point by using the true break and
membership information for all individuals. To study the asymptotic
distribution theory for $\{\hat{\alpha}_{k}^{(\ell )}\}_{k\in K^{(\ell
)}},~\ell \in \{1,2\}$, we only need to show that for the oracle estimators $%
\{\hat{\alpha}_{k}^{\ast (\ell )}\}_{k\in K^{(\ell )}}$ based on Theorems %
\ref{Thm2} and \ref{Thm3} by extending the result of \cite{bai2009panel} and 
\cite{moon2017dynamic}.

To proceed, we add some notation. For $\ell \in \{1,2\}$, we first define
the matrix notation for each subgroup. For $j\in \lbrack p]$, let $%
X_{j,i}^{(1)}=\left( X_{j,i1},\cdots ,X_{j,iT_{1}}\right) ^{\prime }$, $%
X_{j,i}^{(2)}=\left( X_{j,i(T_{1}+1)},\cdots ,X_{j,iT}\right) ^{\prime }$, $%
e_{i}^{(1)}=\left( e_{i1},\cdots ,e_{iT_{1}}\right) ^{\prime }$ and $%
e_{i}^{(2)}=\left( e_{i(T_{1}+1)},\cdots ,e_{iT}\right) ^{\prime }$. Then we
use $\mathbb{X}_{j,k}^{(\ell )}\in \mathbb{R}^{N_{k}^{(\ell )}\times T_{\ell
}}$ and $E_{k}^{(\ell )}\in \mathbb{R}^{N_{k}^{(\ell )}\times T_{\ell }}$ to
denote the regressor and error matrix for subgroup $k\in \lbrack K^{(\ell )}]
$ with each row being $X_{j,i}^{(\ell )}$ and $e_{i}^{(\ell )}$,
respectively. Let $\mathcal{X}_{j,k}^{(\ell )}=M_{\Lambda _{k}^{0,(\ell )}}%
\mathbb{X}_{j,k}^{(\ell )}M_{F^{0,(\ell )}}\in \mathbb{R}^{N_{k}^{(\ell
)}\times T_{\ell }}$ with the $\left( i,t\right) $-th entry given by $%
\mathcal{X}_{j,k,it}^{(\ell )}.$ Let $\mathcal{X}_{k,it}^{(\ell )}=(\mathcal{%
X}_{1,k,it}^{(\ell )},\cdots ,\mathcal{X}_{p,k,it}^{(\ell )})^{\prime }$.
Further define 
\begin{align*}
& \mathbb{B}_{NT,1,j,k}^{(\ell )}=\frac{1}{N_{k}^{(\ell )}}tr\left[
P_{F^{0,(\ell )}}\mathbb{E}\left( E_{k}^{(\ell )\prime }\mathbb{X}%
_{j,k}^{(\ell )}\big|\mathscr{D}\right) \right] , \\
& \mathbb{B}_{NT,2,j,k}^{(\ell )}=\frac{1}{T_{\ell }}tr\left[ \mathbb{E}%
\left( E_{k}^{(\ell )}E_{k}^{(\ell )\prime }\big|\mathscr{D}\right)
M_{\Lambda _{k}^{0,(\ell )}}\mathbb{X}_{j,k}^{(\ell )}F^{0,(\ell )}\left(
F^{0,(\ell )\prime }F^{0,(\ell )}\right) ^{-1}\left( \Lambda _{k}^{0,(\ell
)\prime }\Lambda _{k}^{0,(\ell )}\right) ^{-1}\Lambda _{k}^{0,(\ell )\prime }%
\right] , \\
& \mathbb{B}_{NT,3,j,k}^{(\ell )}=\frac{1}{N_{k}^{(\ell )}}tr\left[ \mathbb{E%
}\left( E_{k}^{(\ell )}E_{k}^{(\ell )\prime }\big|\mathscr{D}\right)
M_{F^{0,(\ell )}}\mathbb{X}_{j,k}^{(\ell )}\Lambda _{k}^{0,(\ell )}\left(
\Lambda _{k}^{0,(\ell )\prime }\Lambda _{k}^{0,(\ell )}\right) ^{-1}\left(
F^{0,(\ell )\prime }F^{0,(\ell )}\right) ^{-1}F^{0,(\ell )\prime }\right] ,
\\
& \mathbb{B}_{NT,m,k}^{(\ell )}=\left( \mathbb{B}_{NT,m,1,k}^{(\ell
)},\cdots ,\mathbb{B}_{NT,m,p,k}^{(\ell )}\right) ^{\prime },\quad \forall
m\in \{1,2,3\}, \\
& \Omega _{k}^{(\ell )}=\frac{1}{N_{k}^{(\ell )}T_{\ell }}\sum_{i\in
G_{k}^{(\ell )}}\sum_{t\in \mathcal{T}_{\ell }}e_{it}^{2}\mathcal{X}%
_{k,it}^{(\ell )}\mathcal{X}_{k,it}^{(\ell )\prime }.
\end{align*}%
Let $\mathbb{W}_{NT,k}^{(\ell )}$ be a $p\times p$ matrix with $(j_{1},j_{2})
$-th entry $\frac{1}{N_{k}^{(\ell )}T_{\ell }}tr\left( M_{F^{0,(\ell
)}}\mathbb{X}_{j_{1},k}^{(\ell )\prime }M_{\Lambda _{k}^{0,(\ell )}}\mathbb{X%
}_{j_{2},k}^{(\ell )}\right) $. Then we define the overall bias term for
each subgroup as 
\begin{equation*}
\mathbb{B}_{NT,k}^{(\ell )}=-\rho _{k}^{(\ell )}\mathbb{B}_{NT,1,k}^{(\ell
)}-(\rho _{k}^{(\ell )})^{-1}\mathbb{B}_{NT,2,k}^{(\ell )}-\rho _{k}^{(\ell
)}\mathbb{B}_{NT,3,k}^{(\ell )},
\end{equation*}%
where $\rho _{k}^{(\ell )}=\sqrt{\frac{N_{k}^{(\ell )}}{T_{\ell }}}$. To
state the main result in this subsection, we add the following assumption.

\begin{ass}
\label{ass:10}

\begin{itemize}
\item[(i)] As $(N,T)\to\infty$, $T(\log T)N^{-4/3}\to 0$.

\item[(ii)] $\plim_{(N,T)\rightarrow \infty }\frac{1}{N_{k}^{(\ell )}T_{\ell
}}\sum_{i\in G_{k}^{(\ell )}}\sum_{t\in \mathcal{T}_{\ell
}}X_{it}X_{it}^{\prime }>0$\ for $\ell \in \{1,2\}$ and $k\in \lbrack
K^{(\ell )}]$.

\item[(iii)] For $\ell\in\{1,2\}$ and $k\in[K^{(\ell)}]$, separate the $p$
regressors of each subgroups into $p_{1}$ ``low-rank regressors" $\mathbb{X}%
_{j,k}^{(\ell)}$ such that $rank(\mathbb{X}_{j,k}^{(\ell)})=1,~\forall
j\in\{1,\cdots,p_{1}\},$ and ``high-rank regressors" $\mathbb{X}%
_{j,k}^{(\ell)}$ such that $rank(\mathbb{X}_{j,k}^{(\ell)})>1,~\forall
j\in\{p_{1}+1,\cdots,p\}.$ Let $p_{2}:=p-p_{1}$. These two types of
regressors satisfy:

\begin{itemize}
\item[(iii.a)] Consider the linear combinations $b\cdot \mathbb{X}%
_{high,k}^{(\ell )}:=\sum_{j=p_{1}+1}^{p}b_{j}\mathbb{X}_{j,k}^{(\ell )}$
for high-rank regressors with $p_{2}$-vectors $b$ such that $\left\Vert
b\right\Vert _{2}=1$ and $b=\left( b_{p_{1}+1},\cdots ,b_{p}\right) ^{\prime
}$. There exists a positive constant $C_{b}$ such that 
\begin{equation*}
\operatornamewithlimits{\min}_{\left\{ \left\Vert b\right\Vert
_{2}=1\right\} }\sum_{n=2r_{0}+p_{1}+1}^{N}\lambda _{n}\left[ \frac{1}{%
NT_{\ell }}\left( b\cdot \mathbb{X}_{high,k}^{(\ell )}\right) \left( b\cdot 
\mathbb{X}_{high,k}^{(\ell )}\right) ^{\prime }\right] \geq C_{b}\quad
w.p.a.1.
\end{equation*}

\item[(iii.b)] For $j\in \lbrack p_{1}]$, write $\mathbb{X}_{j,k}^{(\ell
)}=w_{j,k}^{(\ell )}v_{j,k}^{(\ell )\prime }$ with $N_{k}^{(\ell )}$-vectors 
$w_{j}^{(\ell )}$ and $T_{\ell }$-vectors $v_{j}^{(\ell )}$. Let $%
w_{k}^{(\ell )}=(w_{1,k}^{(\ell )},\cdots ,w_{p_{1},k}^{(\ell )})\in \mathbb{%
R}^{N\times p_{1}}$, $v^{(\ell )}=(v_{1}^{(\ell )},\cdots ,v_{p_{1}}^{(\ell
)})\in \mathbb{R}^{T_{\ell }\times p_{1}}$, $M_{w_{k}^{(\ell
)}}=I_{N_{k}^{(\ell )}}-w_{k}^{(\ell )}(w_{k}^{(\ell )\prime }w_{k}^{(\ell
)})^{-1}w_{k}^{(\ell )\prime }$ and $M_{v^{(\ell )}}=I_{T_{\ell }}-v^{(\ell
)}\left( v^{(\ell )\prime }v^{(\ell )}\right) ^{-1}v^{(\ell )\prime }$.
There exists a positive constant $C_{B}$ such that $(N_{k}^{(\ell
)})^{-1}\Lambda _{k}^{0,(\ell )\prime }M_{w_{k}^{(\ell )}}\Lambda
_{k}^{0,(\ell )}>C_{B}I_{r_{0}}$ and $T_{\ell }^{-1}F^{0,(\ell )\prime
}M_{w^{(\ell )}}F^{0,(\ell )}>C_{B}I_{r_{0}}$ w.p.a.1.
\end{itemize}

\item[(iv)] For $\forall j\in \lbrack p]$, $\ell \in \{1,2\}$, and $k\in
K^{(\ell )}$, 
\begin{equation*}
\frac{1}{N_{k}^{(\ell )}\left( T_{\ell }\right) ^{2}}\sum_{i\in G_{k}^{(\ell
)}}\sum_{t_{1}\in \mathcal{T}_{\ell }}\sum_{t_{2}\in \mathcal{T}_{\ell
}}\sum_{s_{1}\in \mathcal{T}_{\ell }}\sum_{s_{2}\in \mathcal{T}_{\ell
}}\left\vert Cov\left( e_{it_{1}}\tilde{X}_{j,it_{2}},e_{is_{1}}\tilde{X}%
_{j,is_{2}}\right) \right\vert =O_{p}(1),
\end{equation*}%
where $\tilde{X}_{j,it}=X_{j,it}-\mathbb{E}\left( X_{j,it}|\mathscr{D}%
\right) $.
\end{itemize}
\end{ass}

Assumption \ref{ass:10} imposes some conditions to help derive the asymptotic
distribution theory for the panel model with IFEs which allows for dynamics.
Assumption \ref{ass:10}(i) slightly strengthens Assumption \ref{ass:1}(vi).
Assumption \ref{ass:10}(ii) is the standard non-collinearity condition for
regressors, which is analogous to Assumption 4(i) in \cite{moon2017dynamic}.
Assumption \ref{ass:10}(iii) is the identification assumption which is
comparable to Assumption 4 in \cite{moon2017dynamic}. With the conditional
strong mixing condition in Assumption \ref{ass:1}(iii), we can verify
Assumption \ref{ass:10}(iv).

The following theorem establishes the asymptotic distribution of $\{\hat{%
\alpha}_{k}^{(\ell )}\}_{k\in K^{(\ell )}}.$

\begin{theorem}
\label{Thm4} Suppose that Assumption \ref{ass:1} or \ref{ass:1*} and
Assumptions \ref{ass:2}--\ref{ass:10} hold. For $\ell \in \{1,2\}$, the
estimators $\{\hat{\alpha}_{k}^{(\ell )}\}_{k\in K^{(\ell )}}$ are
asymptotically equivalent to the oracle estimators $\{\hat{\alpha}_{k}^{\ast
(\ell )}\}_{k\in K^{(\ell )}}$, and we have 
\begin{equation*}
\mathbb{W}_{NT}^{(\ell )}\mathbb{D}_{NT}^{(\ell )}%
\begin{pmatrix}
\hat{\alpha}_{1}^{(\ell )}-\alpha _{1}^{(\ell )} \\ 
\vdots \\ 
\hat{\alpha}_{K^{(\ell )}}^{(\ell )}-\alpha _{K^{(\ell )}}^{(\ell )}%
\end{pmatrix}%
-\mathbb{B}_{NT}^{(\ell )}\rightsquigarrow \mathbb{N}\left( 0,\Omega ^{(\ell
)}\right) ,
\end{equation*}%
such that $\mathbb{D}_{NT}^{(\ell )}=\text{diag}\left( \sqrt{N_{1}^{(\ell
)}T_{\ell }},\cdots ,\sqrt{N_{K^{(\ell )}}^{(\ell )}T_{\ell }}\right) $, $%
\mathbb{W}_{NT}^{(\ell )}=\text{diag}\left( \mathbb{W}_{NT,1}^{(\ell
)},\cdots ,\mathbb{W}_{NT,K^{(\ell )}}^{(\ell )}\right) $, $\mathbb{B}%
_{NT}^{(\ell )}=\text{diag}\left( \mathbb{B}_{NT,1}^{(\ell )},\cdots ,%
\mathbb{B}_{NT,K^{(\ell )}}^{(\ell )}\right) $ and $\Omega ^{(\ell )}=\text{%
diag}\left( \Omega _{1}^{(\ell )},\cdots ,\Omega _{K^{(\ell )}}^{(\ell
)}\right) $.
\end{theorem}

\noindent Theorem \ref{Thm4} establishes the asymptotic distribution for the
estimators of the group-specific slope coefficients before and after the
break. It shows that the parameter estimators from our algorithm enjoy the
oracle property given the results in Theorems \ref{Thm2} and \ref{Thm3}. The
proof of the above theorem can be done by following \cite{moon2017dynamic}
and \cite{lu2016shrinkage}.

\section{Alternatives and Extensions}

\label{sec:extension} This section first considers an alternative
method to estimate the break point and then discusses several possible
extensions.

\subsection{Alternative for Break Point Detection}

The algorithm proposed in Section \ref{sec:estimation} uses low-rank
estimates of $\Theta _{j}^{0}$ to find the break point estimates. However,
by Lemma \ref{Lem:idfct}(ii), we observe that the right singular vector
matrix of $\Theta _{j}^{0}$, i.e., $V_{j}^{0}$, contains the structural
break information when $r_{j}=2$. For this reason, we can propose an
alternative way to estimate the break point under the case where the maximum
rank of the slope matrix in the model is 2. Let $\dot{v}_{t,j}^{\ast }:=%
\frac{\dot{v}_{t,j}}{\left\Vert \dot{v}_{t,j}\right\Vert }$ and $\dot{v}%
_{t}^{\ast }:=\left( \dot{v}_{t,1}^{\ast \prime },\cdots ,\dot{v}%
_{t,p}^{\ast \prime }\right) ^{\prime },$ with the true values being $%
v_{t,j}^{\ast }:=\frac{O_{j}v_{t,j}^{0}}{\left\Vert
O_{j}v_{t,j}^{0}\right\Vert }$ and $v_{t}^{\ast }:=\left( v_{t,1}^{\ast
\prime },\cdots ,v_{t,p}^{\ast \prime }\right) ^{\prime }$, respectively.
Then Step 3 can be replaced by Step 3* below:

\begin{description}
\item[\textbf{Step 3*:}] \textbf{Break Point Estimation by Singular Vectors.}
We estimate the break point as follows: 
\begin{equation}
\tilde{T}_{1}=\argmin_{s\in \{2,\cdots ,T-1\}}\frac{1}{T}\left\{
\sum_{t=1}^{s}\left\Vert \dot{v}_{t}^{\ast }-\bar{\dot{v}}^{\ast
(1)s}\right\Vert ^{2}+\sum_{t=s+1}^{T}\left\Vert \dot{v}_{t}^{\ast }-\bar{%
\dot{v}}^{\ast (2)s}\right\Vert ^{2}\right\} ,  \label{BP alternative}
\end{equation}%
where $\bar{\dot{v}}^{\ast (1)s}=\frac{1}{s}\sum_{t=1}^{s}{\dot{v}_{t}}%
^{\ast }$ and $\bar{\dot{v}}^{\ast (2)s}=\frac{1}{T-s}\sum_{t=s+1}^{T}{\dot{v%
}}_{t}^{\ast }$.
\end{description}

The following two theorems establish the consistency of $\dot{v}_{t}^{\ast }$
and $\tilde{T}_{1},$ respectively.

\begin{theorem}
\label{Thm5} Suppose that Assumptions \ref{ass:1}--\ref{ass:5} hold. Then $%
\max_{t}\left\Vert \dot{v}_{t}^{\ast }-v_{t}^{\ast }\right\Vert =O_{p}(\eta
_{N,2}).$
\end{theorem}

\begin{theorem}
\label{Thm6} Suppose that Assumptions \ref{ass:1}--\ref{ass:6} hold. Then $%
\mathbb{P(}\tilde{T}_{1}=T_{1})\rightarrow 1$ as $\left( N,T\right)
\rightarrow \infty .$
\end{theorem}

\noindent Since the singular vectors of the slope matrices contain the structural
change information, Theorem \ref{Thm5} indicates that we can consistently
estimate the break point by using the factor estimates instead of the slope
matrix estimates in (\ref{BP_estimates}). Given Theorem \ref{Thm5} and Lemma %
\ref{Lem:idfct}(iii), we can prove Theorem \ref{Thm6} with arguments
analogous to those used in the proof of Theorem \ref{Thm2}.

\subsection{Test for the Presence of a Structural Break}

\label{sec:extension_test} Section \ref{sec:model} considers
time-varying latent group structures with one break point. In this
subsection, we propose a test for the null that the slope coefficients are
time-invariant against the alternative that there's one structural break as
assumed in Section \ref{sec:model}.

Since various scenarios can occur once we allow for the presence of a
structural break in the latent group structures, the number of groups may
or may not change under the alternative and so may some of the
group-specific coefficients. First, information on the latent group structures may
be ignored and  the slope coefficients can be tested for possible time-variation. In this case, we can rewrite $\Theta
_{it}^{0}$ as  
\begin{equation*}
\Theta _{it}^{0}=\Theta _{i}^{0}+c_{it},
\end{equation*}%
where $\Theta _{i}^{0}:=\frac{1}{T}\sum_{t\in \lbrack T]}\Theta _{it}^{0}$.
The null and alternative hypotheses can then be specified as 
\begin{eqnarray}
H_{0} &:&c_{it}=0~~\text{for all }i\in \lbrack N],\quad \text{and}  \notag
\label{BP_test} \\
H_{1} &:&c_{it}\neq 0~~\text{for some }i\in \lbrack N].
\end{eqnarray}%
To construct the test statistics we follow \cite%
{bai1998estimating} and consider a sup-$F$ test. Let $\mathcal{T}_{\epsilon
}:=\left\{ T_{1}:\epsilon T\leq T_{1}\leq (1-\epsilon )T\right\} ,$ where $%
\epsilon >0$ is a tuning parameter that avoids breaks at the end of the
sample. Define 
\begin{equation*}
F_{NT}(1|0):=\operatornamewithlimits{\max}\limits_{i\in \lbrack N]}%
\operatornamewithlimits{\sup}\limits_{T_{1}\in \mathcal{T}_{\epsilon
}}F_{i}(T_{1}),
\end{equation*}%
where 
\begin{equation*}
F_{i}(T_{1})=\frac{T-2p}{p}\left[ \tilde{\beta}_{i}^{\left( 1\right)
}(T_{1})-\tilde{\beta}_{i}^{\left( 2\right) }\left( T_{1}\right) \right]
^{\prime }\left[ \hat{\Sigma}_{i}(T_{1})\right] ^{-1}\left[ \tilde{\beta}%
_{i}^{\left( 1\right) }(T_{1})-\tilde{\beta}_{i}^{\left( 2\right) }\left(
T_{1}\right) \right] ,
\end{equation*}%
$\tilde{\beta}_{i}^{\left( 1\right) }\left( T_{1}\right) $ and $\tilde{\beta}%
_{i}^{\left( 2\right) }(T_{1})$ are the PCA slope estimators of $\Theta
_{i}^{0}$ in the linear panels with IFEs for each individual $i$ with the
prior-break observations $\left\{ (i,t):i\in \lbrack N],t\in \lbrack
T_{+}]\right\} $ and post-break observations $\left\{ (i,t):i\in \lbrack
N],t\in \lbrack T]\backslash \lbrack T_{+}]\right\} $, respectively,%
\footnote{%
See Section \ref{sec:panel_IFE} in the appendix for the detail of the PCA
estimation in linear panels with IFEs.} and $\hat{\Sigma}_{i}(T_{+})$ is the
consistent estimator for the asymptotic variance of $\tilde{\beta}%
_{i}^{\left( 1\right) }(T_{1})-\tilde{\beta}_{i}^{\left( 2\right) }\left(
T_{1}\right) $. Following \cite{bai1998estimating}, we conjecture that the
asymptotic distribution of $\operatornamewithlimits{\sup}\limits_{T_{1}\in 
\mathcal{T}_{\epsilon }}F_{i}(T_{1})$ depends on a $p$-vector of
Wiener processes on $\left[ 0,1\right] ,$ based on which the
corresponding distribution of $F_{NT}(1|0)$ can be obtained.

Alternatively, we can estimate the model with latent group structures by
assuming the presence of a break point at $T_{1}.$ Then we obtain the
estimates of the group-specific parameters $\{\alpha _{j}^{\left( 1\right)
}\left( T_{1}\right) \}_{j\in K^{\left( 1\right) }}$ prior to the potential
break point $T_{1}$ and those of the group-specific parameters $\{\alpha
_{j}^{\left( 2\right) }\left( T_{1}\right) \}_{j\in K^{\left( 2\right) }}$
after the potential break point $T_{1}.$ It is possible to construct a test
statistic based on the contrast of these two sets of estimates or the
corresponding residual sum of squares (RSS) and then take the supremum over $%
T_{1}\in \mathcal{T}_{\epsilon }.$ Evidently, this approach is also
quite involved as it is necessary to determine the number of groups before and after
the break, $K^{\left( 1\right) }$ and $K^{\left( 1\right) },$ at each $%
T_{1}. $ It is not clear how the estimation errors from these estimates and
those of the factors and factor loadings with slow convergence rates affect
the asymptotic properties of the estimators of the group-specific parameters.

Last, it is also possible to estimate the model with latent group structures
under the case of no structural changes to obtain the restricted residuals.
If there exists a structural change in the latent group structure, it should
be reflected in properties of the restricted residuals obtained under the null. We can then
consider the regression of the restricted residuals on the regressors
and construct an LM-type test statistic to check the goodness of fit of 
such an auxiliary regression model as in \cite{su2013testing} and \cite%
{su2020testing}. This analysis is left for future research.

\subsection{The Case of Multiple Breaks}

Section \ref{sec:model} considers only a one-time structural break in
the latent group structures. In practice it is possible to have multiple
breaks especially if $T$ is large. Here we generalize the model in Section %
\ref{sec:model} to allow for multiple breaks. In this case, we have 
\begin{equation*}
\alpha_{kt}=\left\{ \begin{aligned} &\alpha_{k}^{(1)}, \quad\text{for}\quad
t=1,\dots,T_{1},\\ & \alpha_{k}^{(2)}, \quad\text{for}\quad
t=T_{1}+1,\dots,T_{2},\\ &\vdots\\ & \alpha_{k}^{(b+1)},\quad\text{for}\quad
t=T_{b}+1,\dots,T,\\ \end{aligned}\right.
\end{equation*}%
where $b\geq 1$ denotes the number of breaks.

To estimate the number of breaks and the break points $T_{1},\cdots ,T_{b}$,
in principle we can follow the sequential method proposed by \cite%
{bai1998estimating}. First, using the full-sample data, we can construct $%
F_{NT}(1|0)$ defined in the previous subsection and estimate the break point
as in \eqref{BP_estimates}. Second, for each regime before and after the
estimated break point, we test the hypothesis in \eqref{BP_test} and
estimate the break point for each regime separately. This sequential method is repeated
until the null can not be rejected for all subsamples, leading to the break point
estimates $\{\hat{T}_{a}\}_{a\in\lbrack \hat{b}]}$ where $\hat{b}$ is the estimated number of breaks. We
conjecture that we can establish the consistency of $\hat{b}$ and $\{\hat{T}%
_{a}\}.$

After the estimated number of breaks and break points are obtained for each
subsample 
\begin{equation*}
\left\{ (i,t):i\in [N],t\in \{\hat{T}_{a-1}+1,\cdots ,\hat{T}_{a}\}\right\},
\end{equation*}
where $a\in [\hat{b}+1]$ with $\hat{T}_{0}:=0$ and $\hat{T}_{\hat{b}+1}:=T$, we
can continue Step 4 in the estimation algorithm in Section \ref%
{sec:estimation} to obtain the estimated group structure for each subsample.

\section{Monte Carlo Simulations}

\label{sec:simul} In this section we report simulation results for the
low-rank estimates, break point estimates, group membership estimates and
the group number estimates based on 1,000 replications, and tuning
parameter $\nu _{j}$ chosen by a procedure similar to that described in \cite%
{chernozhukov2019inference}. We focus on the linear panel model $%
Y_{it}=\lambda _{i}^{\prime }f_{t}+X_{it}^{\prime }\Theta _{it}+e_{it}$,
where $X_{it}=(X_{1,it},X_{2,it})^{\prime }$ and $\Theta _{it}=(\Theta
_{1,it},\Theta _{2,it})^{\prime }$.

\subsection{Data Generating Processes (DGPs)}

The following four main DGPs are employed.

\begin{description}
\item[\textbf{DGP 1:}] [\textbf{Static panel with homoskedasticity}] $%
X_{1,it}\sim $ $i.i.d.$ $U(-2,2)$, $X_{2,it}\sim i.i.d.$ $U(-2,2)$, and errors 
$e_{it}\sim i.i.d.$ $\mathbb{N}(0,1)$. For $\Theta _{1}$, we
randomly choose the break point $T_{1}$ from $0.4T$ to $0.6T$.

\item[\textbf{DGP 2:}] [\textbf{Static panel with heteroscedasticity}]
Compared to DGP 1, the errors $e_{it}\sim i.i.d.$ $\mathbb{N}%
(0,\sigma _{it}^{2})$ with $\sigma _{it}^{2}\sim i.i.d.$ $U(0.5,1)$.
The settings for the regressors and break point are the same as those in DGP
1.

\item[\textbf{DGP 3:}] [\textbf{Serially correlated error}] Compared to 
DGP 2, the errors $e_{it}=0.2e_{i,t-1}+\eta _{it},$ where $\eta
_{it}\sim i.i.d.$ $\mathbb{N}(0,1)$ and all other settings are the same as
in DGP 2.

\item[\textbf{DGP 4:}] [\textbf{Dynamic panel}] In this case, $%
X_{1,it}=Y_{i,t-1}$ with $Y_{i,0}\sim i.i.d.~\mathbb{N}(0,1)$. $X_{2,it}\sim
i.i.d.~U(-2,2)$, and $e_{it}\sim i.i.d.~\mathbb{N}(0,0.5)$.
\end{description}

For each DGP above, we set $r_{0}=0$ and draw $\lambda _{i}$ and $f_{t}$
from $\mathbb{N}(0,1)$ independently. We define the slope coefficient based
on three subcases below.

\begin{description}
\item[\textbf{DGP X.1:}] In this case, the group membership and the number
of groups do not change after the break point and only the value of the slope
coefficient changes. We set the number of groups to be 2, the ratio of
individuals among the two groups is $N_{1}:N_{2}=0.5:0.5$, and the group
membership $G_{1}$ is obtained by a random draw from $[N]$ without
replacement. For DGPs 1.1, 2.1, and 3.1, 
\begin{equation*}
\Theta _{1,it}=\Theta _{2,it}=\left\{ \begin{aligned} & 0.1, \quad i\in
G_{1},~t\in\{1,\cdots,T_{1} \},\\ & 0.9, \quad i\in
G_{2},~t\in\{1,\cdots,T_{1} \},\\ & 0.05, \quad i\in
G_{1},~t\in\{T_{1}+1,\cdots,T \},\\ & 0.45, \quad i\in
G_{2},~t\in\{T_{1}+1,\cdots,T \}. \end{aligned}\right.
\end{equation*}%
For DGP 4.1, $\Theta _{2,it}$ is same as other DGPs X.1 for X$\in \{1,2,3\}$%
, and 
\begin{equation*}
\Theta _{1,it}=\left\{ \begin{aligned} & 0.1, \quad i\in
G_{1},~t\in\{1,\cdots,T_{1} \},\\ & 0.7, \quad i\in
G_{2},~t\in\{1,\cdots,T_{1} \},\\ & 0.05, \quad i\in
G_{1},~t\in\{T_{1}+1,\cdots,T \},\\ & 0.35, \quad i\in
G_{2},~t\in\{T_{1}+1,\cdots,T \}. \end{aligned}\right.
\end{equation*}

\item[\textbf{DGP X.2:}] Compared to DGP X.1, the values of the slope
coefficients for different groups do not change after the break point, but
the group membership changes. The number of groups is 2, the ratio of
individuals among the groups is still $N_{1}:N_{2}=0.5:0.5$.
Nevertheless, $\{G_{1}^{(1)},G_{2}^{(1)}\}$ is different from $%
\{G_{1}^{(2)},G_{2}^{(2)}\}$ so that elements in both $G_{1}^{(1)}$ and $%
G_{1}^{(2)}$ are independent draws from $[N]$ without replacement. In
addition, for DGPs 1.2, 2.2, and 3.2, 
\begin{equation*}
\Theta _{1,it}=\Theta _{2,it}=\left\{ \begin{aligned} & 0.1, \quad i\in
G_{1}^{(1)},~t\in\{1,\cdots,T_{1} \},\\ & 0.9, \quad i\in
G_{2}^{(1)},~t\in\{1,\cdots,T_{1} \},\\ & 0.1, \quad i\in
G_{1}^{(2)},~t\in\{T_{1}+1,\cdots,T \},\\ & 0.9, \quad i\in
G_{2}^{(2)},~t\in\{T_{1}+1,\cdots,T \}. \end{aligned}\right.
\end{equation*}%
For DGP 4.2, $\Theta _{2,it}$ is defined same as other DGPs X.2 for $X\in
\{1,2,3\}$, and 
\begin{equation*}
\Theta _{1,it}=\left\{ \begin{aligned} & 0.1, \quad i\in
G_{1}^{(1)},~t\in\{1,\cdots,T_{1} \},\\ & 0.7, \quad i\in
G_{2}^{(1)},~t\in\{1,\cdots,T_{1} \},\\ & 0.1, \quad i\in
G_{1}^{(2)},~t\in\{T_{1}+1,\cdots,T \},\\ & 0.7, \quad i\in
G_{2}^{(2)},~t\in\{T_{1}+1,\cdots,T \}. \end{aligned}\right.
\end{equation*}

\item[\textbf{DGP X.3:}] Under this scenario, the number of groups changes
after the break. We set $N_{1}^{(1)}:N_{2}^{(1)}=0.5:0.5$ and $%
N_{1}^{(2)}:N_{2}^{(2)}:N_{3}^{(2)}=0.4:0.3:0.3$ before and after the break,
respectively. Specifically, for DGPs 1.3, 2.3, and 3.3, we have 
\begin{equation*}
\Theta _{1,it}=\Theta _{2,it}=\left\{ \begin{aligned} & 0.1, \quad i\in
G_{1}^{(1)},~t\in\{1,\cdots,T_{1} \},\\ & 0.9, \quad i\in
G_{2}^{(1)},~t\in\{1,\cdots,T_{1} \},\\ & 0.1, \quad i\in
G_{1}^{(2)},~t\in\{T_{1}+1,\cdots,T \},\\ & 0.5, \quad i\in
G_{2}^{(2)},~t\in\{T_{1}+1,\cdots,T \},\\ & 0.9, \quad i\in
G_{3}^{(2)},~t\in\{T_{1}+1,\cdots,T \}. \end{aligned}\right.
\end{equation*}%
For DGP 4.3, $\Theta _{2,it}$ is defined as in DGP X.3 for X$\in \{1,2,3\}$,
and 
\begin{equation*}
\Theta _{1,it}=\left\{ \begin{aligned} & 0.1, \quad i\in
G_{1}^{(1)},~t\in\{1,\cdots,T_{1} \},\\ & 0.7, \quad i\in
G_{2}^{(1)},~t\in\{1,\cdots,T_{1} \},\\ & 0.1, \quad i\in
G_{1}^{(2)},~t\in\{T_{1}+1,\cdots,T \},\\ & 0.4, \quad i\in
G_{2}^{(2)},~t\in\{T_{1}+1,\cdots,T \},\\ & 0.7, \quad i\in
G_{3}^{(2)},~t\in\{T_{1}+1,\cdots,T \}. \end{aligned}\right.
\end{equation*}
\end{description}

\subsection{Results}

Table \ref{tab:rank} reports the proportion of correct rank estimation for
the intercept (IFE) and slope matrices based on the SVT in Section 3.3. Note
that $r_{0}$ denotes the true rank of the intercept matrix and $r_{1}$ and $%
r_{2}$ denote those of the two slope matrices. From Table \ref{tab:rank}, we
notice that the true ranks of both the intercept and slope matrices can be
almost perfectly estimated for the sample sizes under investigation.

\begin{table}[p]
\caption{Frequency of correct rank estimation} \label{tab:rank} \scriptsize \centering
\begin{tabular}{cccccccccccc}
\toprule\toprule\multicolumn{2}{c}{N} & \multicolumn{2}{c}{100} & \multicolumn{2}{c}{200} & 
\multicolumn{2}{c}{N} & \multicolumn{2}{c}{100} & \multicolumn{2}{c}{200} \\ 
\multicolumn{2}{c}{T} & 100 & 200 & 100 & 200 & \multicolumn{2}{c}{T} & 100
& 200 & 100 & 200 \\ 
\midrule \multirow{3}[2]{*}{DGP 1.1} & $r_{0}=1$ & 1.00 & 1.00 & 1.00 & 1.00
& \multirow{3}[2]{*}{DGP 3.1} & $r_{0}=1$ & 1.00 & 1.00 & 1.00 & 1.00 \\ 
& $r_{1}=1$ & 1.00 & 1.00 & 1.00 & 1.00 &  & $r_{1}=1$ & 1.00 & 1.00 & 1.00
& 1.00 \\ 
& $r_{2}=1$ & 1.00 & 1.00 & 1.00 & 1.00 &  & $r_{2}=1$ & 1.00 & 1.00 & 1.00
& 1.00 \\ 
\midrule \multirow{3}[2]{*}{DGP 1.2} & $r_{0}=1$ & 1.00 & 1.00 & 1.00 & 1.00
& \multirow{3}[2]{*}{DGP 3.2} & $r_{0}=1$ & 1.00 & 1.00 & 1.00 & 1.00 \\ 
& $r_{1}=2$ & 1.00 & 1.00 & 1.00 & 1.00 &  & $r_{1}=2$ & 1.00 & 1.00 & 1.00
& 1.00 \\ 
& $r_{2}=2$ & 1.00 & 1.00 & 1.00 & 1.00 &  & $r_{2}=2$ & 1.00 & 1.00 & 1.00
& 1.00 \\ 
\midrule \multirow{3}[2]{*}{DGP 1.3} & $r_{0}=1$ & 1.00 & 1.00 & 1.00 & 1.00
& \multirow{3}[2]{*}{DGP 3.3} & $r_{0}=1$ & 1.00 & 1.00 & 1.00 & 1.00 \\ 
& $r_{1}=2$ & 1.00 & 1.00 & 1.00 & 1.00 &  & $r_{1}=2$ & 1.00 & 1.00 & 1.00
& 1.00 \\ 
& $r_{2}=2$ & 1.00 & 1.00 & 1.00 & 1.00 &  & $r_{2}=2$ & 0.998 & 1.00 & 1.00
& 1.00 \\ 
\midrule \multirow{3}[2]{*}{DGP 2.1} & $r_{0}=1$ & 1.00 & 1.00 & 1.00 & 1.00
& \multirow{3}[2]{*}{DGP 4.1} & $r_{0}=1$ & 1.00 & 1.00 & 1.00 & 1.00 \\ 
& $r_{1}=1$ & 1.00 & 1.00 & 1.00 & 1.00 &  & $r_{1}=1$ & 1.00 & 1.00 & 1.00
& 1.00 \\ 
& $r_{2}=1$ & 1.00 & 1.00 & 1.00 & 1.00 &  & $r_{2}=1$ & 1.00 & 1.00 & 1.00
& 1.00 \\ 
\midrule \multirow{3}[2]{*}{DGP 2.2} & $r_{0}=1$ & 1.00 & 1.00 & 1.00 & 1.00
& \multirow{3}[2]{*}{DGP 4.2} & $r_{0}=1$ & 1.00 & 1.00 & 1.00 & 1.00 \\ 
& $r_{1}=2$ & 1.00 & 1.00 & 1.00 & 1.00 &  & $r_{1}=1$ & 1.00 & 1.00 & 1.00
& 1.00 \\ 
& $r_{2}=2$ & 1.00 & 1.00 & 1.00 & 1.00 &  & $r_{2}=2$ & 1.00 & 1.00 & 1.00
& 1.00 \\ 
\midrule \multirow{3}[2]{*}{DGP 2.3} & $r_{0}=1$ & 1.00 & 1.00 & 1.00 & 1.00
& \multirow{3}[2]{*}{DGP 4.3} & $r_{0}=1$ & 1.00 & 1.00 & 1.00 & 1.00 \\ 
& $r_{1}=2$ & 1.00 & 1.00 & 1.00 & 1.00 &  & $r_{1}=1$ & 1.00 & 1.00 & 1.00
& 1.00 \\ 
& $r_{2}=2$ & 1.00 & 1.00 & 1.00 & 1.00 &  & $r_{2}=2$ & 1.00 & 1.00 & 1.00
& 1.00 \\ 
\bottomrule &  &  &  &  &  &  &  &  &  &  & 
\end{tabular}%
\end{table}

Table \ref{tab:BP_largeT} reports the results for the break point estimation
in Step 3 based on different $(N,T)$ combinations. We summarize some
important findings from Table \ref{tab:BP_largeT}. First, when the group
membership and the number of groups do not change as in DGP X.1 for X$\in %
\left[ 3\right] ,$ the frequency of correct break point estimation may not
be 1 especially if $N$ is not large. This suggests that the binary
segmentation does not work perfectly in such a scenario. Second, the change
of group membership or the number of groups help to identify the break point
as reflected in the simulation results for DGP X.2 and X.3 for X$\in \left[ 4%
\right] .$ In general, the binary segmentation works well in our setting.

\begin{table}[tbph]
\caption{Frequency of correct break point estimation}
\label{tab:BP_largeT}\scriptsize \centering%
\begin{tabular}{cccccccccc}
\toprule\toprule N & \multicolumn{2}{c}{100} & \multicolumn{2}{c}{200} & N & 
\multicolumn{2}{c}{100} & \multicolumn{2}{c}{200} \\ 
\midrule T & 100 & 200 & 100 & 200 & T & 100 & 200 & 100 & 200 \\ 
\midrule DGP 1.1 & 0.980 & 0.993 & 1.00 & 1.00 & DGP 3.1 & 0.985 & 0.972 & 
1.00 & 0.999 \\ 
DGP 1.2 & 0.999 & 1.00 & 1.00 & 1.00 & DGP 3.2 & 1.00 & 1.00 & 1.00 & 1.00
\\ 
DGP 1.3 & 1.00 & 1.00 & 1.00 & 1.00 & DGP 3.3 & 1.00 & 1.00 & 1.00 & 1.00 \\ 
\midrule DGP 2.1 & 0.998 & 0.999 & 1.00 & 1.00 & DGP 4.1 & 1.00 & 1.00 & 1.00
& 1.00 \\ 
DGP 2.2 & 1.00 & 1.00 & 1.00 & 1.00 & DGP 4.2 & 1.00 & 1.00 & 1.00 & 1.00 \\ 
DGP 2.3 & 1.00 & 1.00 & 1.00 & 1.00 & DGP 4.3 & 1.00 & 1.00 & 1.00 & 1.00 \\ 
\bottomrule &  &  &  &  &  &  &  &  & 
\end{tabular}
\end{table}

Table \ref{tab:group} reports the results for the group membership
estimation when the number of groups are either known (infeasible in
practice) or estimated from the data (feasible). With known number of
groups, the STK algorithm degenerates to the traditional K-means algorithm.
The \textquotedblleft Infeasible\textquotedblright\ part of Table \ref%
{tab:group} reports the frequency of correct group membership estimation
before and after the estimated break point, $G_{B}$ and $G_{A}$, based on
the known true number of groups and K-means algorithm. Evidently, the
K-means classification exhibits excellent performance in this case.
Nevertheless, without prior information on the true number of groups, the STK
algorithm is able to estimate the group membership and the number of groups
simultaneously. In this case, the frequency of correct estimation of the
group membership and that of the number of groups are shown in the
\textquotedblleft Feasible\textquotedblright\ part in Table \ref{tab:group}
and in Table \ref{tab:K_1}, respectively. To implement the STK algorithm
with unknown number of groups, we set $\varsigma _{N}=N^{-2}$ to ensure the
consistency of the group number estimators. As expected, the performance of
the STK algorithm is slightly worse than that of the K-means algorithm with knowledge of the
true number of groups. But the performance improves when both $N$ and $T$
increase. Table \ref{tab:K_1} suggests that the number of groups can be
nearly perfectly estimated in DGPs 1.1, 1.2, 1.3 and 2.1. For the more
complicated DGPs (e.g., the dynamic case in DGPs 4.1, 4.2, and 4.3 or the
static panel with serially correlated errors in DGPs 3.1, 3.2, and 3.3), the
performance is not as good as that in the simple DGPs.

\begin{table}[p]
\caption{Frequency of correct group membership estimation}
\label{tab:group}\scriptsize \centering
\begin{tabular}{cccccccccccccc}
\toprule \toprule \multirow{26}[12]{*}{Infeasible} & \multicolumn{2}{c}{N} & 
\multicolumn{2}{c}{100} & \multicolumn{2}{c}{200} & %
\multirow{26}[12]{*}{Feasible} & \multicolumn{2}{c}{N} & \multicolumn{2}{c}{
100} & \multicolumn{2}{c}{200} \\ 
\cmidrule{2-7}\cmidrule{9-14} & \multicolumn{2}{c}{T} & 100 & 200 & 100 & 200
&  & \multicolumn{2}{c}{T} & 100 & 200 & 100 & 200 \\ 
\cmidrule{2-7}\cmidrule{9-14} & \multirow{2}[1]{*}{DGP 1.1} & $G_{B}$ & 1.00
& 1.00 & 1.00 & 1.00 &  & \multirow{2}[1]{*}{DGP 1.1} & $G_{B}$ & 1.00 & 1.00
& 1.00 & 1.00 \\ 
&  & $G_{A}$ & 1.00 & 1.00 & 1.00 & 1.00 &  &  & $G_{A}$ & 1.00 & 1.00 & 1.00
& 1.00 \\ 
& \multirow{2}[0]{*}{DGP 1.2} & $G_{B}$ & 1.00 & 1.00 & 1.00 & 1.00 &  & %
\multirow{2}[0]{*}{DGP 1.2} & $G_{B}$ & 1.00 & 1.00 & 1.00 & 1.00 \\ 
&  & $G_{A}$ & 1.00 & 1.00 & 1.00 & 1.00 &  &  & $G_{A}$ & 1.00 & 1.00 & 1.00
& 1.00 \\ 
& \multirow{2}[1]{*}{DGP 1.3} & $G_{B}$ & 1.00 & 1.00 & 1.00 & 1.00 &  & %
\multirow{2}[1]{*}{DGP 1.3} & $G_{B}$ & 1.00 & 1.00 & 1.00 & 1.00 \\ 
&  & $G_{A}$ & 0.989 & 0.999 & 0.978 & 0.999 &  &  & $G_{A}$ & 0.989 & 0.999
& 0.978 & 0.999 \\ 
\cmidrule{2-7}\cmidrule{9-14} & \multirow{2}[1]{*}{DGP 2.1} & $G_{B}$ & 1.00
& 1.00 & 1.00 & 1.00 &  & \multirow{2}[1]{*}{DGP 2.1} & $G_{B}$ & 1.00 & 1.00
& 1.00 & 1.00 \\ 
&  & $G_{A}$ & 1.00 & 1.00 & 1.00 & 1.00 &  &  & $G_{A}$ & 1.00 & 1.00 & 1.00
& 1.00 \\ 
& \multirow{2}[0]{*}{DGP 2.2} & $G_{B}$ & 1.00 & 1.00 & 1.00 & 1.00 &  & %
\multirow{2}[0]{*}{DGP 2.2} & $G_{B}$ & 0.989 & 0.999 & 0.992 & 0.999 \\ 
&  & $G_{A}$ & 1.00 & 1.00 & 1.00 & 1.00 &  &  & $G_{A}$ & 0.992 & 0.999 & 
0.977 & 0.998 \\ 
& \multirow{2}[1]{*}{DGP 2.3} & $G_{B}$ & 1.00 & 1.00 & 1.00 & 1.00 &  & %
\multirow{2}[1]{*}{DGP 2.3} & $G_{B}$ & 0.992 & 0.999 & 0.961 & 0.999 \\ 
&  & $G_{A}$ & 0.998 & 1.00 & 0.999 & 1.00 &  &  & $G_{A}$ & 0.989 & 0.999 & 
0.992 & 0.999 \\ 
\cmidrule{2-7}\cmidrule{9-14} & \multirow{2}[1]{*}{DGP 3.1} & $G_{B}$ & 1.00
& 1.00 & 1.00 & 1.00 &  & \multirow{2}[1]{*}{DGP 3.1} & $G_{B}$ & 0.981 & 
0.999 & 0.949 & 0.999 \\ 
&  & $G_{A}$ & 1.00 & 1.00 & 1.00 & 1.00 &  &  & $G_{A}$ & 0.981 & 0.993 & 
0.979 & 0.996 \\ 
& \multirow{2}[0]{*}{DGP 3.2} & $G_{B}$ & 1.00 & 1.00 & 1.00 & 1.00 &  & %
\multirow{2}[0]{*}{DGP 3.2} & $G_{B}$ & 0.985 & 0.996 & 0.962 & 0.993 \\ 
&  & $G_{A}$ & 1.00 & 1.00 & 1.00 & 1.00 &  &  & $G_{A}$ & 0.985 & 0.994 & 
0.973 & 0.998 \\ 
& \multirow{2}[1]{*}{DGP 3.3} & $G_{B}$ & 1.00 & 1.00 & 1.00 & 1.00 &  & %
\multirow{2}[1]{*}{DGP 3.3} & $G_{B}$ & 0.985 & 0.998 & 0.973 & 0.995 \\ 
&  & $G_{A}$ & 0.981 & 0.997 & 0.982 & 0.999 &  &  & $G_{A}$ & 0.971 & 0.994
& 0.968 & 0.998 \\ 
\cmidrule{2-7}\cmidrule{9-14} & \multirow{2}[1]{*}{DGP 4.1} & $G_{B}$ & 1.00
& 1.00 & 1.00 & 1.00 &  & \multirow{2}[1]{*}{DGP 4.1} & $G_{B}$ & 0.975 & 
0.999 & 0.984 & 0.999 \\ 
&  & $G_{A}$ & 1.00 & 1.00 & 1.00 & 1.00 &  &  & $G_{A}$ & 0.985 & 0.998 & 
0.949 & 0.997 \\ 
& \multirow{2}[0]{*}{DGP 4.2} & $G_{B}$ & 1.00 & 1.00 & 1.00 & 1.00 &  & %
\multirow{2}[0]{*}{DGP 4.2} & $G_{B}$ & 0.994 & 0.998 & 0.952 & 0.997 \\ 
&  & $G_{A}$ & 1.00 & 1.00 & 1.00 & 1.00 &  &  & $G_{A}$ & 0.977 & 0.999 & 
0.985 & 0.999 \\ 
& \multirow{2}[1]{*}{DGP 4.3} & $G_{B}$ & 1.00 & 1.00 & 1.00 & 1.00 &  & %
\multirow{2}[1]{*}{DGP 4.3} & $G_{B}$ & 0.983 & 0.998 & 0.948 & 0.999 \\ 
&  & $G_{A}$ & 1.00 & 1.00 & 1.00 & 1.00 &  &  & $G_{A}$ & 0.982 & 0.998 & 
0.983 & 0.998 \\ 
\bottomrule &  &  &  &  &  &  &  &  &  &  &  &  & 
\end{tabular}
\end{table}

\begin{table}[p]
\caption{Frequency of correct estimation of the number of groups}
\label{tab:K_1}\scriptsize \centering
\begin{tabular}{cccccccccccc}
	\toprule \toprule 
\multicolumn{2}{c}{N} & \multicolumn{2}{c}{100} & \multicolumn{2}{c}{200} & 
\multicolumn{2}{c}{N} & \multicolumn{2}{c}{100} & \multicolumn{2}{c}{200} \\ 
\multicolumn{2}{c}{T} & 100 & 200 & 100 & 200 & \multicolumn{2}{c}{T} & 100
& 200 & 100 & 200 \\ 
\midrule \multirow{2}[1]{*}{DGP 1.1} & $K^{(1)}=2$ & 1.00 & 1.00 & 0.999 & 
1.00 & \multirow{2}[1]{*}{DGP 3.1} & $K^{(1)}=2$ & 0.880 & 0.993 & 0.675 & 
0.985 \\ 
& $K^{(2)}=2$ & 1.00 & 1.00 & 1.00 & 1.00 &  & $K^{(2)}=2$ & 0.890 & 0.960 & 
0.873 & 0.971 \\ 
\multirow{2}[0]{*}{DGP 1.2} & $K^{(1)}=2$ & 1.00 & 1.00 & 1.00 & 1.00 & %
\multirow{2}[0]{*}{DGP 3.2} & $K^{(1)}=2$ & 0.868 & 0.985 & 0.759 & 0.940 \\ 
& $K^{(2)}=2$ & 1.00 & 1.00 & 1.00 & 0.999 &  & $K^{(2)}=2$ & 0.897 & 0.971
& 0.829 & 0.987 \\ 
\multirow{2}[1]{*}{DGP 1.3} & $K^{(1)}=2$ & 0.999 & 1.00 & 1.00 & 1.00 & %
\multirow{2}[1]{*}{DGP 3.3} & $K^{(1)}=2$ & 0.889 & 0.988 & 0.802 & 0.965 \\ 
& $K^{(2)}=3$ & 1.00 & 0.999 & 1.00 & 1.00 &  & $K^{(2)}=3$ & 0.932 & 0.977
& 0.907 & 0.988 \\ 
\midrule \multirow{2}[1]{*}{DGP 2.1} & $K^{(1)}=2$ & 1.00 & 1.00 & 1.00 & 
1.00 & \multirow{2}[1]{*}{DGP 4.1} & $K^{(1)}=2$ & 0.807 & 0.981 & 0.825 & 
0.982 \\ 
& $K^{(2)}=2$ & 1.00 & 1.00 & 1.00 & 1.00 &  & $K^{(2)}=2$ & 0.919 & 0.988 & 
0.714 & 0.980 \\ 
\multirow{2}[0]{*}{DGP 2.2} & $K^{(1)}=2$ & 0.919 & 0.995 & 0.940 & 0.994 & %
\multirow{2}[0]{*}{DGP 4.2} & $K^{(1)}=2$ & 0.933 & 0.988 & 0.630 & 0.975 \\ 
& $K^{(2)}=2$ & 0.930 & 0.993 & 0.809 & 0.982 &  & $K^{(2)}=2$ & 0.758 & 
0.988 & 0.870 & 0.989 \\ 
\multirow{2}[1]{*}{DGP 2.3} & $K^{(1)}=2$ & 0.940 & 0.989 & 0.724 & 0.990 & %
\multirow{2}[1]{*}{DGP 4.3} & $K^{(1)}=2$ & 0.877 & 0.991 & 0.657 & 0.991 \\ 
& $K^{(2)}=3$ & 0.946 & 0.995 & 0.952 & 0.992 &  & $K^{(2)}=3$ & 0.900 & 
0.987 & 0.874 & 0.980 \\ 
\bottomrule &  &  &  &  &  &  &  &  &  &  & 
\end{tabular}%
\end{table}

Table \ref{tab:K_2} presents more detailed results for the estimation of the
number of groups. For DGPs 1.X and DGP 2.X where we have static panels with
independent errors, the results show that the group membership and the
number of groups can be well estimated with nearly 100\% accuracy under
different $(N,T)$ combinations. For DGPs 3.X and 4.X where we have static
panels with serially correlated errors and dynamic panels, respectively, the
frequency of correct estimation of both the group membership and the number of
groups is not great when $T$ is small, but gradually
approaches unity as $T$ increases. One reason for this is that we need to use
HAC estimates of certain long-run variance objects in the STK algorithm and
it is well known that a relatively large value of $T$ is required in order
for the HAC estimates to be reasonably well behaved in finite samples.

\begin{table}[p]
\caption{Determination of the number of groups}
\label{tab:K_2}\scriptsize \centering
\begin{tabular}{ccccccccccc}
\toprule \toprule \multirow{2}[4]{*}{DGP } & \multirow{2}[4]{*}{N} & %
\multirow{2}[4]{*}{T} & \multicolumn{4}{c}{$\hat{K}^{(1)}$} &  & $\hat{K}%
^{(2)}$ &  &  \\ 
\cmidrule{4-11} &  &  & 2 & 3 & 4 & $\geq$5 & 2 & 3 & 4 & $\geq$5 \\ 
\midrule \multirow{4}[1]{*}{DGP 1.1} & \multirow{2}[1]{*}{100} & 100 & 
\textbf{1.00} & 0 & 0 & 0 & \textbf{1.00} & 0 & 0 & 0 \\ 
&  & 200 & \textbf{1.00} & 0 & 0 & 0 & \textbf{1.00} & 0 & 0 & 0 \\ 
& \multirow{2}[0]{*}{200} & 100 & \textbf{0.999} & 0.001 & 0 & 0 & \textbf{%
1.00} & 0 & 0 & 0 \\ 
&  & 200 & \textbf{1.00} & 0 & 0 & 0 & \textbf{1.00} & 0 & 0 & 0 \\ 
\multirow{4}[0]{*}{DGP 1.2} & \multirow{2}[0]{*}{100} & 100 & \textbf{1.00}
& 0.00 & 0.00 & 0.00 & \textbf{1.00} & 0.00 & 0.00 & 0.00 \\ 
&  & 200 & \textbf{1.00} & 0.00 & 0.00 & 0.00 & \textbf{1.00} & 0.00 & 0.00
& 0.00 \\ 
& \multirow{2}[0]{*}{200} & 100 & \textbf{1.00} & 0.00 & 0.00 & 0.00 & 
\textbf{1.00} & 0.00 & 0.00 & 0.00 \\ 
&  & 200 & \textbf{1.00} & 0.00 & 0.00 & 0.00 & \textbf{1.00} & 0.00 & 0.00
& 0.00 \\ 
\multirow{4}[1]{*}{DGP 1.3} & \multirow{2}[0]{*}{100} & 100 & \textbf{0.999}
& 0.001 & 0.00 & 0.00 & 0.00 & \textbf{1.00} & 0.00 & 0.00 \\ 
&  & 200 & \textbf{1.00} & 0.00 & 0.00 & 0.00 & 0.00 & \textbf{0.999} & 0.001
& 0.00 \\ 
& \multirow{2}[1]{*}{200} & 100 & \textbf{1.00} & 0.00 & 0.00 & 0.00 & 0.00
& \textbf{1.00} & 0.00 & 0.00 \\ 
&  & 200 & \textbf{1.00} & 0.00 & 0.00 & 0.00 & 0.00 & \textbf{1.00} & 0.00
& 0.00 \\ 
\midrule \multirow{4}[1]{*}{DGP 2.1} & \multirow{2}[1]{*}{100} & 100 & 
\textbf{0.933} & 0.058 & 0.009 & 0.00 & \textbf{0.936} & 0.060 & 0.003 & 
0.001 \\ 
&  & 200 & \textbf{0.990} & 0.010 & 0.00 & 0.00 & \textbf{0.987} & 0.013 & 
0.00 & 0.00 \\ 
& \multirow{2}[0]{*}{200} & 100 & \textbf{0.864} & 0.126 & 0.010 & 0.00 & 
\textbf{0.901} & 0.090 & 0.009 & 0.00 \\ 
&  & 200 & \textbf{0.989} & 0.011 & 0.00 & 0.00 & \textbf{0.990} & 0.010 & 
0.000 & 0.00 \\ 
\multirow{4}[0]{*}{DGP 2.2} & \multirow{2}[0]{*}{100} & 100 & \textbf{0.919}
& 0.074 & 0.007 & 0.00 & \textbf{0.930} & 0.067 & 0.003 & 0.00 \\ 
&  & 200 & \textbf{0.995} & 0.003 & 0.00 & 0.002 & \textbf{0.993} & 0.006 & 
0.00 & 0.001 \\ 
& \multirow{2}[0]{*}{200} & 100 & \textbf{0.940} & 0.056 & 0.004 & 0.00 & 
\textbf{0.809} & 0.164 & 0.027 & 0.00 \\ 
&  & 200 & \textbf{0.994} & 0.006 & 0.00 & 0.00 & \textbf{0.982} & 0.018 & 
0.00 & 0.00 \\ 
\multirow{4}[1]{*}{DGP 2.3} & \multirow{2}[0]{*}{100} & 100 & \textbf{0.940}
& 0.055 & 0.005 & 0.00 & 0.00 & \textbf{0.946} & 0.039 & 0.015 \\ 
&  & 200 & \textbf{0.989} & 0.011 & 0.00 & 0.00 & 0.00 & \textbf{0.995} & 
0.002 & 0.003 \\ 
& \multirow{2}[1]{*}{200} & 100 & \textbf{0.724} & 0.230 & 0.046 & 0.00 & 
0.00 & \textbf{0.952} & 0.031 & 0.017 \\ 
&  & 200 & \textbf{0.990} & 0.010 & 0.00 & 0.00 & 0.00 & \textbf{0.992} & 
0.006 & 0.002 \\ 
\midrule \multirow{4}[1]{*}{DGP 3.1} & \multirow{2}[1]{*}{100} & 100 & 
\textbf{0.880} & 0.097 & 0.022 & 0.001 & \textbf{0.890} & 0.062 & 0.031 & 
0.017 \\ 
&  & 200 & \textbf{0.993} & 0.007 & 0 & 0 & \textbf{0.960} & 0.019 & 0.012 & 
0.009 \\ 
& \multirow{2}[0]{*}{200} & 100 & \textbf{0.675} & 0.224 & 0.099 & 0.002 & 
\textbf{0.873} & 0.081 & 0.041 & 0.005 \\ 
&  & 200 & \textbf{0.985} & 0.015 & 0 & 0 & \textbf{0.971} & 0.023 & 0.005 & 
0.001 \\ 
\multirow{4}[0]{*}{DGP 3.2} & \multirow{2}[0]{*}{100} & 100 & \textbf{0.868}
& 0.109 & 0.023 & 0.00 & \textbf{0.897} & 0.099 & 0.004 & 0.00 \\ 
&  & 200 & \textbf{0.985} & 0.008 & 0.003 & 0.004 & \textbf{0.971} & 0.021 & 
0.006 & 0.002 \\ 
& \multirow{2}[0]{*}{200} & 100 & \textbf{0.759} & 0.198 & 0.042 & 0.001 & 
\textbf{0.829} & 0.147 & 0.024 & 0.00 \\ 
&  & 200 & \textbf{0.940} & 0.055 & 0.005 & 0.00 & \textbf{0.987} & 0.013 & 
0.000 & 0.00 \\ 
\multirow{4}[1]{*}{DGP 3.3} & \multirow{2}[0]{*}{100} & 100 & \textbf{0.889}
& 0.100 & 0.011 & 0.00 & 0.00 & \textbf{0.932} & 0.055 & 0.013 \\ 
&  & 200 & \textbf{0.988} & 0.009 & 0.003 & 0.00 & 0.00 & \textbf{0.977} & 
0.013 & 0.010 \\ 
& \multirow{2}[1]{*}{200} & 100 & \textbf{0.802} & 0.175 & 0.023 & 0.00 & 
0.00 & \textbf{0.907} & 0.073 & 0.020 \\ 
&  & 200 & \textbf{0.965} & 0.035 & 0.000 & 0.000 & 0.000 & \textbf{0.988} & 
0.010 & 0.002 \\ 
\midrule \multirow{4}[1]{*}{DGP 4.1} & \multirow{2}[1]{*}{100} & 100 & 
\textbf{0.807} & 0.084 & 0.089 & 0.02 & \textbf{0.919} & 0.041 & 0.019 & 
0.021 \\ 
&  & 200 & \textbf{0.981} & 0.013 & 0.004 & 0.002 & \textbf{0.988} & 0.004 & 
0.005 & 0.003 \\ 
& \multirow{2}[0]{*}{200} & 100 & \textbf{0.825} & 0.107 & 0.061 & 0.007 & 
\textbf{0.714} & 0.118 & 0.084 & 0.084 \\ 
&  & 200 & \textbf{0.982} & 0.011 & 0.006 & 0.001 & \textbf{0.98} & 0.010 & 
0.004 & 0.006 \\ 
\multirow{4}[0]{*}{DGP 4.2} & \multirow{2}[0]{*}{100} & 100 & \textbf{0.933}
& 0.051 & 0.012 & 0.004 & \textbf{0.758} & 0.141 & 0.089 & 0.012 \\ 
&  & 200 & \textbf{0.988} & 0.006 & 0.004 & 0.002 & \textbf{0.988} & 0.005 & 
0.006 & 0.001 \\ 
& \multirow{2}[0]{*}{200} & 100 & \textbf{0.630} & 0.158 & 0.196 & 0.016 & 
\textbf{0.870} & 0.080 & 0.048 & 0.002 \\ 
&  & 200 & \textbf{0.975} & 0.013 & 0.012 & 0.000 & \textbf{0.989} & 0.009 & 
0.002 & 0.000 \\ 
\multirow{4}[1]{*}{DGP 4.3} & \multirow{2}[0]{*}{100} & 100 & \textbf{0.877}
& 0.076 & 0.042 & 0.005 & 0.000 & \textbf{0.900} & 0.055 & 0.045 \\ 
&  & 200 & \textbf{0.991} & 0.006 & 0.002 & 0.001 & 0.000 & \textbf{0.987} & 
0.010 & 0.003 \\ 
& \multirow{2}[1]{*}{200} & 100 & \textbf{0.657} & 0.191 & 0.129 & 0.023 & 
0.000 & \textbf{0.874} & 0.072 & 0.054 \\ 
&  & 200 & \textbf{0.991} & 0.005 & 0.004 & 0.000 & 0.000 & \textbf{0.980} & 
0.012 & 0.008 \\ 
\bottomrule &  &  &  &  &  &  &  &  &  & 
\end{tabular}%
\end{table}

Table \ref{tab:infer} shows results for the post-classification estimator
for the first slope coefficient. We follow \cite{su2016identifying} to
define the evaluation criteria as bias and coverage. Specifically, we define
the bias to be the weighted versions of bias for slope estimator from all
estimated groups, i.e. $\text{Bias}=\sum_{k=1}^{K^{(1)}}\text{Bias}(\alpha
_{k,1}^{(\ell )})$ for $\ell \in \{1,2\}$. Similarly, we define the weighted
version of the coverage ratio of the 95\% confidence interval estimators. The
\textquotedblleft Infeasible" panel shows the result assuming the number of
groups information is known, and the \textquotedblleft Feasible" panel shows
the result without knowing the number of groups information by the STK
algorithm. From Table \ref{tab:infer}, we notice that the coverage ratio for
DGP 1 and 2 is close to 95\% under different combinations of $N$ and $T$ for
both the \textquotedblleft Infeasible" and \textquotedblleft Feasible"
panels, which is due to the higher correct classification ratio. For DGPs 3
and 4, by using the STK algorithm, the coverage ratio is a bit lower for $%
T=100$, which is due to the inaccuracy of the group number and membership
estimators; but coverage approaches 95\% quickly when $T$ doubles.

\begin{table}[p]
\caption{Point estimation of $\protect\alpha _{\cdot ,1}^{(1)}$ and $\protect%
\alpha _{\cdot ,1}^{(2)}$}
\label{tab:infer}\scriptsize \centering
\scalebox{0.9}{\begin{tabular}{ccc|cccc|cccc}
    \toprule
    \toprule
    \multirow{3}[6]{*}{DGP} & \multirow{3}[6]{*}{N} & \multicolumn{1}{c}{\multirow{3}[6]{*}{T}} & \multicolumn{4}{c}{Infeasible} & \multicolumn{4}{c}{Feasible} \\
\cmidrule{4-11}          &       & \multicolumn{1}{c}{} & \multicolumn{2}{c}{Before the break} & \multicolumn{2}{c}{After the break} & \multicolumn{2}{c}{Before the break} & \multicolumn{2}{c}{After the break} \\
\cmidrule{4-11}          &       & \multicolumn{1}{c}{} & Bias($\times 10^{-6}$) & Coverage & Bias($\times 10^{-6}$) & \multicolumn{1}{c}{Coverage} & Bias($\times 10^{-6}$) & Coverage & Bias($\times 10^{-6}$) & Coverage \\
    \midrule
    \multirow{4}[2]{*}{1.1} & \multirow{2}[1]{*}{100} & 100   & 2.585 & 0.951 & -2.869 & 0.946 & 2.585 & 0.951 & -2.869 & 0.946 \\
          &       & 200   & -1.944 & 0.944 & -8.920 & 0.945 & -1.958 & 0.944 & -8.920 & 0.945 \\
          & \multirow{2}[1]{*}{200} & 100   & -1.096 & 0.943 & 1.407 & 0.947 & -1.096 & 0.943 & 1.407 & 0.947 \\
          &       & 200   & -1.910 & 0.945 & 0.960 & 0.947 & -1.910 & 0.945 & 0.960 & 0.947 \\
    \midrule
    \multirow{4}[2]{*}{1.2} & \multirow{2}[1]{*}{100} & 100   & -1.050 & 0.949 & -27.398 & 0.941 & -1.050 & 0.949 & -27.398 & 0.941 \\
          &       & 200   & -5.449 & 0.930 & 7.616 & 0.953 & -5.449 & 0.930 & 7.655 & 0.953 \\
          & \multirow{2}[1]{*}{200} & 100   & 4.770 & 0.949 & 1.866 & 0.951 & 4.770 & 0.949 & 1.866 & 0.951 \\
          &       & 200   & 1.317 & 0.941 & 1.874 & 0.945 & 1.317 & 0.941 & 1.874 & 0.945 \\
    \midrule
    \multirow{4}[2]{*}{1.3} & \multirow{2}[1]{*}{100} & 100   & -0.961 & 0.943 & 11.417 & 0.944 & -1.050 & 0.949 & -27.398 & 0.941 \\
          &       & 200   & -4.213 & 0.951 & -5.002 & 0.941 & -5.449 & 0.930 & 7.655 & 0.953 \\
          & \multirow{2}[1]{*}{200} & 100   & -1.571 & 0.938 & -3.756 & 0.938 & 4.770 & 0.949 & 1.866 & 0.951 \\
          &       & 200   & 0.403 & 0.941 & -4.159 & 0.945 & 1.317 & 0.941 & 1.874 & 0.945 \\
    \midrule
    \multirow{4}[2]{*}{2.1} & \multirow{2}[1]{*}{100} & 100   & 14.840 & 0.944 & 9.410 & 0.950 & 14.816 & 0.943 & 9.406 & 0.950 \\
          &       & 200   & -7.222 & 0.951 & 1.795 & 0.951 & -7.222 & 0.951 & 1.795 & 0.951 \\
          & \multirow{2}[1]{*}{200} & 100   & 0.916 & 0.940 & 3.575 & 0.948 & 0.916 & 0.940 & 3.575 & 0.948 \\
          &       & 200   & 0.452 & 0.948 & -0.797 & 0.947 & 0.452 & 0.948 & -0.797 & 0.947 \\
    \midrule
    \multirow{4}[2]{*}{2.2} & \multirow{2}[1]{*}{100} & 100   & -21.379 & 0.946 & 0.234 & 0.937 & -21.379 & 0.946 & 0.234 & 0.937 \\
          &       & 200   & 0.264 & 0.942 & -15.542 & 0.953 & 0.264 & 0.942 & -15.542 & 0.953 \\
          & \multirow{2}[1]{*}{200} & 100   & -1.379 & 0.945 & -1.489 & 0.951 & -1.379 & 0.944 & -1.489 & 0.951 \\
          &       & 200   & -1.101 & 0.950 & 1.127 & 0.949 & -1.101 & 0.950 & 1.127 & 0.949 \\
    \midrule
    \multirow{4}[2]{*}{2.3} & \multirow{2}[1]{*}{100} & 100   & -8.610 & 0.945 & 5.254 & 0.952 & -8.610 & 0.945 & 5.261 & 0.952 \\
          &       & 200   & 0.927 & 0.949 & 5.840 & 0.949 & 0.927 & 0.949 & 5.840 & 0.949 \\
          & \multirow{2}[1]{*}{200} & 100   & -1.560 & 0.943 & -2.569 & 0.941 & -1.560 & 0.943 & -2.569 & 0.941 \\
          &       & 200   & -0.775 & 0.947 & 4.408 & 0.947 & -0.775 & 0.947 & 4.386 & 0.947 \\
    \midrule
    \multirow{4}[2]{*}{3.1} & \multirow{2}[1]{*}{100} & 100   & -20.928 & 0.955 & -73.947 & 0.945 & -26.250 & 0.927 & -77.613 & 0.920 \\
          &       & 200   & 3.066 & 0.949 & -12.443 & 0.937 & 2.884 & 0.940 & -13.116 & 0.934 \\
          & \multirow{2}[1]{*}{200} & 100   & -2.663 & 0.951 & -8.742 & 0.944 & -3.517 & 0.857 & -7.730 & 0.888 \\
          &       & 200   & -3.747 & 0.949 & -2.107 & 0.945 & -3.642 & 0.939 & -1.971 & 0.938 \\
    \midrule
    \multirow{4}[2]{*}{3.2} & \multirow{2}[1]{*}{100} & 100   & -55.980 & 0.952 & -10.846 & 0.943 & -58.714 & 0.926 & -15.109 & 0.863 \\
          &       & 200   & -2.774 & 0.950 & 4.690 & 0.946 & -3.218 & 0.945 & 4.913 & 0.942 \\
          & \multirow{2}[1]{*}{200} & 100   & 6.979 & 0.951 & 8.879 & 0.945 & 6.287 & 0.858 & 6.894 & 0.848 \\
          &       & 200   & -1.704 & 0.947 & 0.438 & 0.945 & -2.122 & 0.928 & 0.381 & 0.940 \\
    \midrule
    \multirow{4}[2]{*}{3.3} & \multirow{2}[1]{*}{100} & 100   & -25.340 & 0.950 & 37.639 & 0.907 & -29.905 & 0.924 & 37.016 & 0.890 \\
          &       & 200   & 2.042 & 0.947 & -6.431 & 0.960 & 1.667 & 0.940 & -6.245 & 0.960 \\
          & \multirow{2}[1]{*}{200} & 100   & -2.391 & 0.946 & 14.364 & 0.892 & -2.735 & 0.891 & 13.680 & 0.840 \\
          &       & 200   & 4.339 & 0.943 & 4.890 & 0.942 & 4.493 & 0.932 & 5.113 & 0.938 \\
    \midrule
    \multirow{4}[2]{*}{4.1} & \multirow{2}[1]{*}{100} & 100   & 800.620 & 0.930 & -466.590 & 0.929 & 777.650 & 0.928 & -454.980 & 0.924 \\
          &       & 200   & 126.760 & 0.931 & 550.210 & 0.942 & 126.220 & 0.942 & 548.160 & 0.943 \\
          & \multirow{2}[1]{*}{200} & 100   & -224.960 & 0.931 & -339.020 & 0.939 & -214.900 & 0.904 & -313.430 & 0.876 \\
          &       & 200   & 417.320 & 0.938 & 412.500 & 0.947 & 415.110 & 0.941 & 410.700 & 0.944 \\
    \midrule
    \multirow{4}[2]{*}{4.2} & \multirow{2}[1]{*}{100} & 100   & 1246.000 & 0.921 & 726.940 & 0.943 & 1205.600 & 0.918 & 709.260 & 0.903 \\
          &       & 200   & -440.880 & 0.943 & 83.433 & 0.944 & -436.670 & 0.943 & 81.585 & 0.951 \\
          & \multirow{2}[1]{*}{200} & 100   & -1600.500 & 0.930 & 1937.500 & 0.927 & -1538.300 & 0.901 & 1781.600 & 0.819 \\
          &       & 200   & -1513.300 & 0.950 & -272.500 & 0.946 & -1502.000 & 0.935 & -271.420 & 0.953 \\
    \midrule
    \multirow{4}[2]{*}{4.3} & \multirow{2}[1]{*}{100} & 100   & -2067.900 & 0.931 & -505.100 & 0.940 & -1951.700 & 0.866 & -491.200 & 0.929 \\
          &       & 200   & 317.360 & 0.946 & 411.770 & 0.945 & 316.550 & 0.951 & 407.820 & 0.935 \\
          & \multirow{2}[1]{*}{200} & 100   & 1279.900 & 0.930 & 3660.100 & 0.888 & 1246.900 & 0.906 & 3355.100 & 0.874 \\
          &       & 200   & -772.380 & 0.940 & -335.250 & 0.948 & -768.870 & 0.932 & -334.190 & 0.940 \\
    \bottomrule
    \end{tabular}} 
\end{table}

\section{Empirical Study}

\label{sec:emp} The estimation methods were applied to
analyze the time-varying latent group structure of real house price changes
in Metropolitan Statistical Areas (MSAs) in the United States. Studies
of U.S. house price changes are plentiful in the literature. \cite%
{malpezzi1999simple}, \cite{capozza2002determinants}, \cite{gallin2006long},
and \cite{ortalo2006housing} all show that the house price changes are closely
correlated with real income in the long run. \cite{su2023identifying}
consider a heterogenous spatial panel and show that real income growth
affects the U.S. house prices in different ways for different MSAs. In this
application, we examine whether there exist latent group structures for the
real income growth elasticity of house price changes and whether these structures change
over the time dimension.

\subsection{Model}

We consider the following panel data model with IFEs and two-way slope heterogeneity
\begin{equation}
\pi _{it}=\lambda _{i}^{\prime }f_{t}+\Theta _{1,it}ginc_{it}+\Theta
_{2,it}ginc_{i,t-1}+e_{it},  \label{emp:1}
\end{equation}%
where the dependent variable $\pi _{it}$ measures the percentage of real
house price growth for MSA $i$ at time period $t$. The $\lambda _{i}$ and $f_{t}$
are the individual fixed effects and time fixed effects, the
covariate $ginc_{it}$ denotes the percentage of income growth for MSA $i$ at
time period $t$, and $ginc_{i,t-1}$ is the lagged value of $ginc_{it}$.
Unlike \cite{aquaro2021estimation} and \cite{su2023identifying} who consider
individual fixed effects and additive two-way fixed effects, respectively,
we allow the model to have IFEs. In the above model, we allow the slope
parameters $\left( \Theta _{1,it},\Theta _{2,it}\right) $ to exhibit
latent group structures along the cross-sectional dimension and an unknown
break along the time dimension.

\subsection{Data}

The dataset we use is obtained from \cite{aquaro2021estimation}, which is the
quarterly data for 377 MSAs over 1975 to 2014. To construct the growth rate
and the lagged term, we lose two periods of observations, which yields $%
T=158 $. Similar to \cite{su2023identifying}, we deseasonalize the growth
rate of real house price and real income. We do not de-factor the variables
since our model contains IFEs to control the common shocks.

\subsection{Empirical Results}

We first apply the singular value thresholding to estimate the ranks of $%
\Theta _{0}=\{\lambda _{i}^{\prime }f_{t}\},$ $\Theta _{1}=\{\Theta _{1,it}\}
$ and $\Theta _{2}=\{\Theta _{2,it}\}.$ The estimates are: $\hat{r}_{0}=1,$ $%
\hat{r}_{1}=2,$ and $\hat{r}_{2}=2$. Before applying the proposed estimation
algorithm in Section \ref{sec:estimation}, we first test the presence of a
structural break as in Section \ref{sec:extension_test}. As given in 
\cite{bai2003critical}, for each individual $i$, the critical value of test
statistic $\sup_{T_{1}\in \mathcal{T}_{\epsilon }}F_{i}(T_{1})$ is 15.37. We
then construct the sup-F test statistic for each MSA. Results show that $%
\operatornamewithlimits{\text{min}}\limits_{i\in \lbrack N]}\sup_{T_{1}\in 
\mathcal{T}_{\epsilon }}F_{i}(T_{1})=0.0195$ and the final test statistic is 
$F_{NT}(1|0)=\operatornamewithlimits{\text{max}}\limits_{i\in \lbrack
N]}\sup_{T_{1}\in \mathcal{T}_{\epsilon }}F_{i}(T_{1})=2161.65$. Based on
this outcome, we reject the null that there is no structural break for slope
coefficients $\Theta _{1,it}$ and $\Theta _{2,it}$ at the 1\% significance level.

With the presence of a structural break, we apply the proposed multi-stage
estimation result in Section \ref{sec:estimation} to estimate the break date
and numbers of groups before and after the break. The estimated break date
is given by $\hat{T}_{1}=51,$ which suggests that the structural break
happens at the first quarter in 1988. We conjecture that this break may be related to
the catastrophic stock market crash that occurred on October 1987, which is
considered to be the first contemporary global financial crisis event. 

\begin{figure}[p]
\caption{Group classification result 1975Q3-1987Q4}
\label{fig:before}\centering
\includegraphics[width=0.9\textwidth]{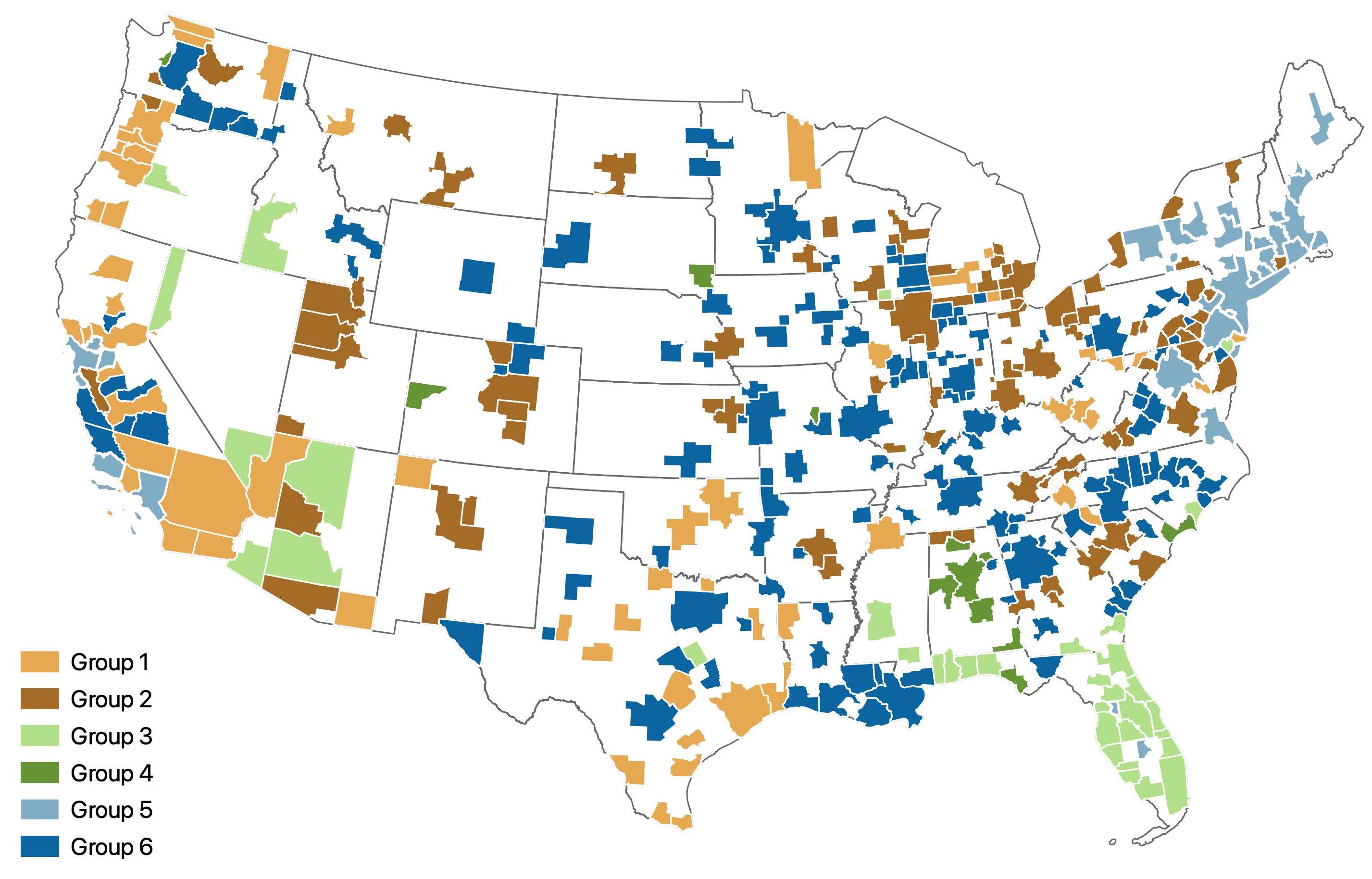} %
\includegraphics[width=0.9\textwidth]{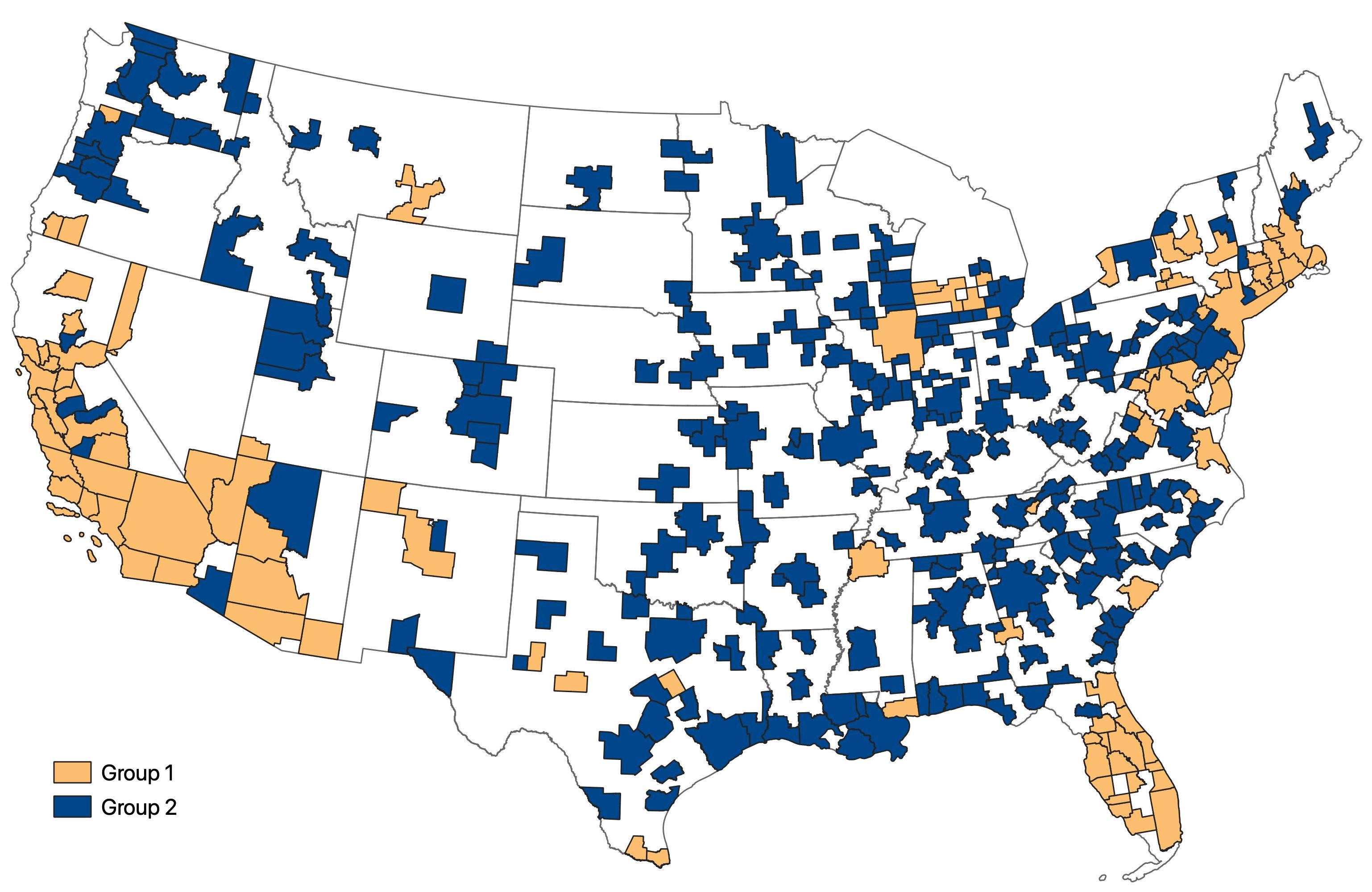}
\caption{Group classification result 1988Q1-2014Q4}
\label{fig:after}
\end{figure}

By setting $\varsigma _{N}=N^{-2}$ for the STK algorithm as in the
simulations, we obtain the estimated prior- and post-break numbers of groups
given by $\hat{K}^{(1)}=6$ and $\hat{K}^{(2)}=2,$ respectively. As for the
group structure, Figures \ref{fig:before} and \ref{fig:after} use six and
two colors to show the classification results for the 377 MSAs during 1975Q3
to 1987Q4 and 1988Q1 to 2014Q4, respectively. Table \ref{tab:emp_pool}
reports the pooled regression results for the full sample in column (1), the
subsample before the estimated break point in column (2), and the subsample
after the estimated break point in column (3). All the slope estimators are
bias-corrected. The pooled regression results in Table \ref{tab:emp_pool}
show that real income growth has positive and significant effect on 
house prices. Comparing the two subsamples before and after the estimated
break, we observe that, with a 1 percentage increase in real income
growth, the real house price growth rate will increase 0.09 percentage
before 1988, which is 0.02 percentage points higher than that after 1988.
The slope estimates for the lagged term are similar for the two subsamples.

\begin{table}[h]
\caption{Results for the pooled regressions}
\label{tab:emp_pool}\centering
\begin{tabular}{cccc}
\toprule \toprule & Pooled (full sample) & Pooled ($1975Q3-1987Q4$) & Pooled
($1988Q1-2014Q4$) \\ 
& $(1)$ & $(2)$ & $(3)$ \\ 
\midrule $ginc_{it}$ & $0.1021^{***}$ & $0.0904^{***}$ & $0.0702^{***}$ \\ 
& (0.0067) & (0.0119) & (0.0065) \\ 
$ginc_{i,t-1}$ & $0.0590^{***}$ & $0.0392^{***}$ & $0.0401^{***}$ \\ 
& (0.0067) & (0.0117) & (0.0066) \\ 
\#individuals & 377 & 377 & 377 \\ 
 \bottomrule
\multicolumn{4}{p{16cm}}{$Note$: Column (1) reports the pooled regression
results for the full sample. Columns (2) and (3) report the pooled
regression results for the subsamples before and after the estimated break
point, respectively. Slope estimators are all bias-corrected. Values in
parentheses are standard errors and $^{***}$ indicates significance at $1\%$
level.}%
\end{tabular}%
\end{table}

To examine the difference for each of the 6 estimated groups before the
break, Table \ref{tab:emp_post_before} reports the post-classification
regression results for each estimated group before the estimated break. Even
though the effects of real income growth are positive for all estimated
groups, they differ vastly across groups. The effect of real income
growth for Group 2 is the highest, followed by Groups 5 and 3, and the
effects of real income growth on real house prices in all of these three
groups are higher than 0.15 percent. In contrast, the effects of real
income growth for the remaining three groups, viz., Groups 1, 4, and 6, are
less than 0.07 percent. Similarly, Table \ref{tab:emp_post_after} reports
the post-classification regression results for each estimated group after
the estimated break. The estimated slope coefficients for both groups are
statistically significant. Especially, during 1988Q1-2014Q1, the slope
estimator for the lagged term in the first group is much higher than that
for the second group.

We also applied the C-Lasso algorithm of \cite{su2016identifying} to estimate
the group structure before and after the estimated break point. The C-Lasso
approach in conjuction with IC detects 2 groups both before and after the
break. In view of the six groups detected by our present algorithm, we conjecture that the difference may
due to the smaller time periods before the break.

\begin{table}[h]
\caption{Results for the post-classification regressions before the break}
\label{tab:emp_post_before}\centering
\begin{tabular}{ccccccc}
\toprule \toprule & Group 1 & Group 2 & Group 3 & Group 4 & Group 5 & Group 6
\\ 
& $(1)$ & $(2)$ & $(3)$ & $(4)$ & $(5)$ & $(6)$ \\ 
\midrule $ginc_{it}$ & $0.0301$ & $0.3169^{***}$ & $0.1522^{***}$ & $0.0168$
& $0.1877^{**}$ & $0.0661^{***}$ \\ 
& (0.0345) & (0.0292) & (0.0408) & (0.0561) & (0.0775) & (0.0153) \\ 
$ginc_{i,t-1}$ & $0.1217^{***}$ & $-0.0191$ & $-0.0089$ & $-0.0298$ & $%
-0.1269^{*}$ & $0.0331^{***}$ \\ 
& (0.0348) & (0.0288) & (0.0407) & (0.0560) & (0.0754) & (0.0151) \\ 
\#individuals & 60 & 92 & 35 & 12 & 36 & 142 \\ 
 \bottomrule
\multicolumn{7}{p{14cm}}{$Note$: Each column reports the regression results
for each estimated group during 1977Q3-1987Q4. Slope estimators are all
bias-corrected. Values in parentheses are standard errors. $^{***}$, $^{**}$%
, and $^{*}$ indicate significance at $1\%$ level, $5\%$ level, $10\%$
level, respectively.}%
\end{tabular}%
\end{table}

\begin{table}[h]
\caption{Results for the post-classification regressions after the break}
\label{tab:emp_post_after}\centering
\begin{tabular}{ccc}
\toprule \toprule & Group 1 & Group 2 \\ 
& $(1)$ & $(2)$ \\ 
\midrule $ginc_{it}$ & $0.0670^{***}$ & $0.0714^{***}$ \\ 
& (0.0130) & (0.0079) \\ 
$ginc_{i,t-1}$ & $0.0870^{***}$ & $0.0275^{***}$ \\ 
& (0.0135) & (0.0079) \\ 
\#individuals & 103 & 274 \\ 
\bottomrule
\multicolumn{3}{p{8cm}}{$Note$: Each column reports the regression results
for each estimated group during 1988Q1-2014Q4. Slope estimators are all
bias-corrected. Values in parentheses are standard errors. $^{***}$
indicates significance at $1\%$ level.}%
\end{tabular}%
\end{table}

\section{Conclusion}

\label{sec:concl} This paper considers a linear panel data model with interactive fixed effects
and two-way heterogeneity such that the heterogeneity across individuals is
captured by latent group structures and the heterogeneity across time is
captured by an unknown structural break. We allow the model to have
different group numbers, or different group memberships, or just changes in
the slope coefficients for some specific groups before and after the break.
To estimate the unknown structural break, the number of groups and group
memberships before and after the break point, we propose an estimation
algorithm with initial nuclear-norm-regularized estimates, followed by row-
and column-wise linear regressions. Then, the break point estimator is
obtained by binary segmentation and the group structure together with the
number of groups are estimated simultaneously using a sequential testing K-means
algorithm. We show that the structural break estimator, the group number
estimators, and the group membership estimators before and after the break
point are all consistent, and the final post-classification slope
coefficient estimators enjoy the oracle property.

There are several interesting topics for further research. First, even
though we discuss a possible test for the existence of a break in the panel
data models with latent group structures, we have not fully worked out the
asymptotic theory, a challenge that deserves separate treatment. Second, 
we assume the presence of a single break in the
data and it is interesting to extend our theory to allow for multiple
breaks. Third, the present treatment rules out both unit-root-type nonstationarity and
nonstochastic trending nonstationarity and it is interesting to extend our
theory to allow for nonstationarity. We will pursue these topics in future
research.

{\small {\ 
\bibliographystyle{apalike}
\bibliography{chapter2}
}} \newpage

{\small 
\appendix%
\linespread{1.2}%
\small%
}

{\small \setcounter{footnote}{0}\setcounter{page}{1}\setlength{%
\baselineskip}{16pt}{} }

\begin{center}
{\small {\Large Online Supplement for} }

{\small {\Large \textquotedblleft Panel Data Models with Time-Varying Latent
Group Structures\textquotedblright } }

{\small $\medskip $ }

{\small Yiren Wang$^{a}$, Peter C.B. Phillips$^{b}$ and Liangjun Su$^{c}$ }

{\small $^{a}$School of Economics, Singapore Management University,
Singapore }

{\small $^{b}$Yale University, University of Auckland,  Singapore Management University}

{\small $^{c}$School of Economics and Management, Tsinghua University, China 
}

{\small \ \ \ \ \ \ \ \ \ }
\end{center}

{\small \noindent This online supplement has five sections.
Section \ref{sec:proof_main} contains the proofs of the main results by
calling upon some technical lemmas in Sections \ref{sec:lem} and \ref%
{sec:proof_IFE}. Section \ref{sec:lem} states and proves the technical
lemmas used in Section \ref{sec:proof_main}. Section \ref{sec:panel_IFE}
contains the estimation procedure for the panel model with interactive fixed
effects (IFEs) and slope heterogeneity, and proposes the test statistics for
the slope homogeneity. Section \ref{sec:proof_IFE} shows the uniform
asymptotic theories for the slope estimators and test statistics proposed in
Section \ref{sec:panel_IFE}. Section \ref{sec:algorihm} provides details
of the algorithm for the nuclear norm regularized regressions.}

\section{\protect\small Proofs of the Main Results}

{\small \label{sec:proof_main} }

\subsection{{\protect\small Proof of Lemma \protect\ref{Lem:idfct}}}

{\small (i) Recall that $\mathcal{G}_{j}^{(\ell )}=\{G_{1,j}^{(\ell
)},\cdots ,G_{K_{\ell },j}^{(\ell )}\}$. For the special case when $\mathcal{%
G}_{j}^{(1)}=\mathcal{G}_{j}^{(2)}$ and $\alpha _{k,j}^{(1)}=\mu \alpha
_{k,j}^{(2)}$ such that $\mu $ is a constant, the group structure does not
change, the relative break size is the same for all groups, and $r_{j}=1$.
Except for this case, below we will show that $r_{j}=2$. }

{\small Let $A_{j,i}^{\left( \ell \right) }=\sum_{k=1}^{K_{\ell }}\alpha
_{k,j}^{(\ell )}\mathbf{1}\{i\in G_{k,j}^{(\ell )}\},$ $%
A_{j,i}=(A_{j,i}^{(1)},A_{j,i}^{(2)})^{\prime }$ and $A_{j}=\left(
A_{j,1},\cdots ,A_{j,N}\right) ^{\prime }\in \mathbb{R}^{N\times 2}$. Define
the $2\times 2$ symmetric matrix $B_{j}=A_{j}^{\prime }A_{j}.$ Let $B_{j}^{%
\frac{1}{2}}$ be the symmetric square root of $B_{j}.$ By the singular value
decomposition (SVD), $B_{j}^{\frac{1}{2}}%
\begin{bmatrix}
\sqrt{\tau _{T}} & 0 \\ 
0 & \sqrt{1-\tau _{T}}%
\end{bmatrix}%
$ $=L_{j}S_{j}R_{j}^{\prime },$ where $L_{j}^{\prime }L_{j}=R_{j}^{\prime
}R_{j}=I_{2}$ and $S_{j}$ is diagonal. Then 
\begin{align*}
\Theta _{j}^{0}& =\left[ 
\begin{array}{cc}
A_{j,1}^{\left( 1\right) }\iota _{T_{1}}^{\prime } & A_{j,1}^{\left(
2\right) }\iota _{T-T_{1}}^{\prime } \\ 
\vdots  & \vdots  \\ 
A_{j,N}^{\left( 1\right) }\iota _{T_{1}}^{\prime } & A_{j,N}^{\left(
2\right) }\iota _{T-T_{1}}^{\prime }%
\end{array}%
\right] =A_{j}%
\begin{bmatrix}
\iota _{T_{1}} & \mathbf{0}_{T_{1}} \\ 
\mathbf{0}_{T-T_{1}} & \iota _{T-T_{1}}%
\end{bmatrix}%
^{\prime } \\
& =A_{j}B_{j}^{-1/2}L_{j}S_{j}R_{j}^{\prime }%
\begin{bmatrix}
\frac{1}{\sqrt{\tau _{T}}} & 0 \\ 
0 & \frac{1}{\sqrt{1-\tau _{T}}}%
\end{bmatrix}%
\begin{bmatrix}
\iota _{T_{1}} & \mathbf{0}_{T_{1}} \\ 
\mathbf{0}_{T-T_{1}} & \iota _{T-T_{1}}%
\end{bmatrix}%
^{\prime } \\
& =A_{j}B_{j}^{-\frac{1}{2}}L_{j}S_{j}R_{j}^{\prime }%
\begin{bmatrix}
\frac{1}{\sqrt{\tau _{T}}}\iota _{T_{1}} & \mathbf{0}_{T_{1}} \\ 
\mathbf{0}_{T-T_{1}} & \frac{1}{\sqrt{1-\tau _{T}}}\iota _{T-T_{1}}%
\end{bmatrix}%
^{\prime } \\
& =\left( A_{j}B_{j}^{-\frac{1}{2}}L_{j}\right) \left( \sqrt{T}S_{j}\right)
\left\{ \frac{1}{\sqrt{T}}R_{j}^{\prime }%
\begin{bmatrix}
\frac{1}{\sqrt{\tau _{T}}}\iota _{T_{1}} & \mathbf{0}_{T_{1}} \\ 
\mathbf{0}_{T-T_{1}} & \frac{1}{\sqrt{1-\tau _{T}}}\iota _{T-T_{1}}%
\end{bmatrix}%
^{\prime }\right\} :=\mathcal{U}_{j}\Sigma _{j}\mathcal{V}_{j}^{\prime },
\end{align*}%
where $\mathcal{U}_{j}=A_{j}B_{j}^{-\frac{1}{2}}L_{j}\in \mathbb{R}^{N\times
2}$, $\mathcal{V}_{j}=\frac{1}{\sqrt{T}}%
\begin{bmatrix}
\frac{1}{\sqrt{\tau _{T}}}\iota _{T_{1}} & \mathbf{0}_{T_{1}} \\ 
\mathbf{0}_{T-T_{1}} & \frac{1}{\sqrt{1-\tau _{T}}}\iota _{T-T_{1}}%
\end{bmatrix}%
R_{j}\in \mathbb{R}^{T\times 2}$, and $\Sigma _{j}=\sqrt{T}S_{j}\in \mathbb{R%
}^{2\times 2}.$ It is easy to verify that 
\begin{eqnarray*}
\mathcal{U}_{j}^{\prime }\mathcal{U}_{j} &=&L_{j}^{\prime }B_{j}^{-\frac{1}{2%
}}A_{j}^{\prime }A_{j}B_{j}^{-\frac{1}{2}}L_{j}=L_{j}^{\prime }B_{j}^{-\frac{%
1}{2}}B_{j}B_{j}^{-\frac{1}{2}}L_{j}=L_{j}^{\prime }L_{j}=I_{2}\quad \text{%
and} \\
\mathcal{V}_{j}^{\prime }\mathcal{V}_{j} &=&R_{j}^{\prime }R_{j}=I_{2}.
\end{eqnarray*}%
Now, let $U_{j}=\mathcal{U}_{j}\Sigma _{j}/\sqrt{T}$ and $V_{j}=\sqrt{T}%
\mathcal{V}_{j}.$ We have $\Theta _{j}^{0}=U_{j}V_{j}^{\top }$ and $V_{j}=%
\begin{bmatrix}
\frac{1}{\sqrt{\tau _{T}}}\iota _{T_{1}} & \mathbf{0}_{T_{1}} \\ 
\mathbf{0}_{T-T_{1}} & \frac{1}{\sqrt{1-\tau _{T}}}\iota _{T-T_{1}}%
\end{bmatrix}%
R_{j}=D_{j}R_{j}.$ This proves (i). }

{\small (ii) Given $R_{j}$ is an orthonormal matrix, this follows from (i)
automatically. $\blacksquare $ }

\subsection{{\protect\small Proof of Theorem \protect\ref{Thm1}}}

\subsubsection{\protect\small Proof of Statement (i).}

{\small Let $\mathcal{R}(C_{1}):=\left\{ \{\Delta _{\Theta _{j}}\}_{j\in
\lbrack p]\cup \{0\}}\in \mathbb{R}^{N\times T}:\sum_{j\in \lbrack p]\cup
\{0\}}\left\Vert \mathcal{P}_{j}^{\bot }(\Delta _{\Theta _{j}})\right\Vert
_{\ast }\leq C_{1}\sum_{j\in \lbrack p]\cup \{0\}}\left\Vert \mathcal{P}%
_{j}(\Delta _{\Theta _{j}})\right\Vert _{\ast }\right\} $. By Lemma \ref%
{Lem:RS}, we notice that 
\begin{equation*}
\mathbb{P}\left\{ \left\{ \tilde{\Delta}_{\Theta _{j}}\right\} _{j\in
\lbrack p]\cup \{0\}}\in \mathcal{R}(3)\right\} \rightarrow 1.
\end{equation*}%
Recall from (\ref{RS}) that 
\begin{equation*}
\mathcal{R}(C_{1},C_{2}):=\bigg\{(\{\Delta _{\Theta _{j}}\}_{j\in \lbrack
p]\cup \{0\}}):\sum_{j\in \lbrack p]\cup \{0\}}\left\Vert \mathcal{P}%
_{j}^{\bot }(\Delta _{\Theta _{j}})\right\Vert _{\ast }\leq C_{1}\sum_{j\in
\lbrack p]\cup \{0\}}\left\Vert \mathcal{P}_{j}(\Delta _{\Theta
_{j}})\right\Vert _{\ast }, \sum_{j\in \lbrack p]\cup \{0\}}\left\Vert
\Theta _{j}\right\Vert _{{}}^{2}\geq C_{2}\sqrt{NT}\bigg\}.
\end{equation*}%
When $\{\tilde{\Delta}_{\Theta _{j}}\}_{j\in \lbrack p]\cup \{0\}}\in 
\mathcal{R}(3)$ and $\{\tilde{\Delta}_{\Theta _{j}}\}_{j\in \lbrack p]\cup
\{0\}}\notin \mathcal{R}(3,C_{2})$, we have $\sum_{j\in \lbrack p]\cup
\{0\}}\left\Vert \tilde{\Delta}_{\Theta _{j}}\right\Vert ^{2}<C_{2}\sqrt{NT}%
, $ which gives 
\begin{equation*}
\frac{1}{\sqrt{NT}}\left\Vert \tilde{\Delta}_{\Theta _{j}}\right\Vert <\frac{%
C_{2}}{\sqrt{N\wedge T}},\quad \forall j\in \lbrack p]\cup \{0\}.
\end{equation*}%
So it suffices to focus on the case that $\{\tilde{\Delta}_{\Theta
_{j}}\}_{j\in \lbrack p]\cup \{0\}}\in \mathcal{R}\left( 3,C_{2}\right) $. }

{\small Define the event 
\begin{equation*}
\mathscr{A}_{1,N}(c_{3})=\left\{ \left\Vert E\right\Vert _{op}\leq c_{3}(%
\sqrt{N}\vee \sqrt{T\log T}),\left\Vert X_{j}\odot E\right\Vert _{op}\leq
c_{3}(\sqrt{N}\vee \sqrt{T\log T}),\forall j\in \lbrack p]\right\} 
\end{equation*}%
for some positive constant $c_{3}$. By Lemma \ref{Lem:score op}, we have $%
\mathbb{P}(\mathscr{A}_{1,N}^{c}(c_{3}))\leq \epsilon $ for any $\epsilon >0$%
. By the definition of $\{\tilde{\Theta}_{j}\}_{j\in \lbrack p]\cup \{0\}}$
in (\ref{obj}), we have 
\begin{align*}
\sum_{j\in \lbrack p]\cup \{0\}}\nu _{j}\left( \left\Vert \Theta
_{j}^{0}\right\Vert _{\ast }-\left\Vert \tilde{\Theta}_{j}\right\Vert _{\ast
}\right) & \geq \frac{1}{NT}\left\Vert Y-\tilde{\Theta}_{0}-\sum_{j\in
\lbrack p]}X_{j}\odot \tilde{\Theta}_{j}\right\Vert ^{2}-\frac{1}{NT}%
\left\Vert Y-\Theta _{0}^{0}-\sum_{j\in \lbrack p]}X_{j}\odot \Theta
_{j}^{0}\right\Vert ^{2} \\
& =\frac{1}{NT}\left\Vert \tilde{\Delta}_{\Theta _{0}}+\sum_{j\in \lbrack p]}%
\tilde{\Delta}_{\Theta _{j}}\odot X_{j}\right\Vert ^{2}-\frac{2}{NT}tr(%
\tilde{\Delta}_{\Theta _{0}}^{\prime }E)-\frac{2}{NT}\sum_{j\in \lbrack p]}tr%
\left[ \tilde{\Delta}_{\Theta _{0}}^{\prime }\left( E\odot X_{j}\right) %
\right] .
\end{align*}%
Then conditioning on the event $\mathscr{A}_{1,N}(c_{3})$, we have 
\begin{align}
& \frac{1}{NT}\left\Vert \tilde{\Delta}_{\Theta _{0}}+\sum_{j\in \lbrack p]}%
\tilde{\Delta}_{\Theta _{j}}\odot X_{j}\right\Vert ^{2}  \notag  \label{A.1}
\\
& \leq \frac{2}{NT}tr\left( \tilde{\Delta}_{\Theta _{0}}^{\prime }E\right) +%
\frac{2}{NT}\sum_{j\in \lbrack p]}tr\left[ \tilde{\Delta}_{\Theta
_{0}}^{\prime }\left( E\odot X_{j}\right) \right] +\sum_{j\in \lbrack p]\cup
\{0\}}\nu _{j}\left( \left\Vert \Theta _{j}^{0}\right\Vert _{\ast
}-\left\Vert \tilde{\Theta}_{j}\right\Vert _{\ast }\right)   \notag \\
& \leq 2c_{3}\frac{\sqrt{N}\vee \sqrt{T\log T}}{NT}\sum_{j\in \lbrack p]\cup
\{0\}}\left\Vert \tilde{\Delta}_{\Theta _{j}}\right\Vert _{\ast }+\sum_{j\in
\lbrack p]\cup \{0\}}\nu _{j}\left( \left\Vert \mathcal{P}_{j}(\tilde{\Delta}%
_{\Theta _{j}})\right\Vert _{\ast }-\left\Vert \mathcal{P}_{j}^{\bot }(%
\tilde{\Delta}_{\Theta _{j}})\right\Vert _{\ast }\right)   \notag \\
& =2c_{3}\frac{\sqrt{N}\vee \sqrt{T\log T}}{NT}\sum_{j\in \lbrack p]\cup
\{0\}}\left( \left\Vert \mathcal{P}_{j}(\tilde{\Delta}_{\Theta
_{j}})\right\Vert _{\ast }+\left\Vert \mathcal{P}_{j}^{\bot }(\tilde{\Delta}%
_{\Theta _{j}})\right\Vert _{\ast }\right)   \notag \\
& +4c_{3}\frac{\sqrt{N}\vee \sqrt{T\log T}}{NT}\sum_{j\in \lbrack p]\cup
\{0\}}\left( \left\Vert \mathcal{P}_{j}(\tilde{\Delta}_{\Theta
_{j}})\right\Vert _{\ast }-\left\Vert \mathcal{P}_{j}^{\bot }(\tilde{\Delta}%
_{\Theta _{j}})\right\Vert _{\ast }\right)   \notag \\
& =6c_{3}\frac{\sqrt{N}\vee \sqrt{T\log T}}{NT}\sum_{j\in \lbrack p]\cup
\{0\}}\left\Vert \mathcal{P}_{j}(\tilde{\Delta}_{\Theta _{j}})\right\Vert
_{\ast }-2c_{3}\frac{\sqrt{N}\vee \sqrt{T\log T}}{NT}\sum_{j\in \lbrack
p]\cup \{0\}}\left\Vert \mathcal{P}_{j}^{\bot }(\tilde{\Delta}_{\Theta
_{j}})\right\Vert _{\ast }  \notag \\
& \leq 6c_{3}\frac{\sqrt{N}\vee \sqrt{T\log T}}{NT}\sum_{j\in \lbrack p]\cup
\{0\}}\left\Vert \mathcal{P}_{j}(\tilde{\Delta}_{\Theta _{j}})\right\Vert
_{\ast },
\end{align}%
where the second inequality holds by the definition of event $\mathscr{A}%
_{1}(c_{3}),$ the fact that $\left\vert \text{tr}\left( AB\right)
\right\vert \leq \left\Vert A\right\Vert _{\ast }\left\Vert B\right\Vert
_{op},$ and (\ref{L1.4}), the first equality holds by the fact that $%
\left\Vert \tilde{\Delta}_{\Theta _{j}}\right\Vert _{\ast }=\left\Vert 
\mathcal{P}_{j}(\tilde{\Delta}_{\Theta _{j}})\right\Vert _{\ast }+\left\Vert 
\mathcal{P}_{j}^{\bot }(\tilde{\Delta}_{\Theta _{j}})\right\Vert _{\ast }$
(see, e.g., Lemma D.2(i) in \cite{chernozhukov2019inference}) and that $\nu
_{j}=\frac{4c_{3}\left( \sqrt{N}\vee \sqrt{T\log T}\right) }{NT}$. It
follows that 
\begin{align*}
C_{3}\sum_{j\in \lbrack p]\cup \{0\}}\left\Vert \tilde{\Delta}_{\Theta
_{j}}\right\Vert _{{}}^{2}& \leq \left\Vert \tilde{\Delta}_{\Theta
_{0}}+\sum_{j\in \lbrack p]}\tilde{\Delta}_{\Theta _{j}}\odot
X_{j}\right\Vert ^{2}+C_{4}\left( N+T\right)  \\
& \leq 6c_{3}\left( \sqrt{N}\vee \sqrt{T\log T}\right) \sum_{j\in \lbrack
p]\cup \{0\}}\left\Vert \mathcal{P}_{j}(\tilde{\Delta}_{\Theta
_{j}})\right\Vert _{\ast }+C_{4}\left( N+T\right)  \\
& \leq 12\bar{r}c_{3}\left( \sqrt{N}\vee \sqrt{T\log T}\right) \sum_{j\in
\lbrack p]\cup \{0\}}\left\Vert \mathcal{P}_{j}(\tilde{\Delta}_{\Theta
_{j}})\right\Vert +C_{4}\left( N+T\right)  \\
& \leq 12\bar{r}c_{3}\left( \sqrt{N}\vee \sqrt{T\log T}\right) \sum_{j\in
\lbrack p]\cup \{0\}}\left\Vert \tilde{\Delta}_{\Theta _{j}}\right\Vert
+C_{4}\left( N+T\right)  \\
& \leq 12\bar{r}c_{3}\left( \sqrt{N}\vee \sqrt{T\log T}\right) \sqrt{p+1}%
\sqrt{\sum_{j\in \lbrack p]\cup \{0\}}\left\Vert \tilde{\Delta}_{\Theta
_{j}}\right\Vert ^{2}}+C_{4}\left( N+T\right) ,
\end{align*}%
where the first inequality holds by Assumption \ref{ass:4}, the second
inequality follows by (\ref{A.1}), the third inequality is by the fact that $%
\left\Vert \mathcal{P}_{j}(\tilde{\Delta}_{\Theta _{j}})\right\Vert _{\ast
}\leq rank(\mathcal{P}_{j}(\tilde{\Delta}_{\Theta _{j}}))\left\Vert \mathcal{%
P}_{j}(\tilde{\Delta}_{\Theta _{j}})\right\Vert $ with $rank(\mathcal{P}_{j}(%
\tilde{\Delta}_{\Theta _{j}}))\leq 2\bar{r}$ by Lemma D.2(ii) in \cite%
{chernozhukov2019inference}, the fourth inequality is by the fact that $%
\left\Vert \tilde{\Delta}_{\Theta _{j}}\right\Vert =\left\Vert \mathcal{P}%
_{j}(\tilde{\Delta}_{\Theta _{j}})\right\Vert +\left\Vert \mathcal{P}%
_{j}^{\bot }(\tilde{\Delta}_{\Theta _{j}})\right\Vert $ (see, e.g., Lemma
D.2(ii) in \cite{chernozhukov2019inference}), and the last inequality holds
by Jensen inequality. Consequently, we can conclude that 
\begin{equation*}
\sum_{j\in \lbrack p]\cup \{0\}}\left\Vert \tilde{\Delta}_{\Theta
_{j}}\right\Vert ^{2}=O_{p}\left( N\vee \left( T\log T\right) \right) .
\end{equation*}%
In addition, $\max_{k\in \lbrack r_{j}]}\left\vert \tilde{\sigma}%
_{k,j}-\sigma _{k,j}\right\vert =O_{p}(\eta _{N,1})$ by the Weyl's
inequality with $\eta _{N,1}=\frac{\sqrt{\log T}}{\sqrt{N\wedge T}}$. }

{\small Now, we show the convergence rate for the singular vector estimates.
For $\forall j\in \lbrack p]\cup \{0\}$, let $\tilde{D}_{j}=\frac{1}{NT}%
\tilde{\Theta}_{j}^{\prime }\tilde{\Theta}_{j}=\hat{\tilde{\mathcal{V}}}_{j}%
\hat{\tilde{\Sigma}}_{j}\hat{\tilde{\mathcal{V}}}_{j}^{\prime }$, and $%
D_{j}^{0}=\frac{1}{NT}\Theta _{j}^{0\prime }\Theta _{j}^{0}=\mathcal{V}%
_{j}^{0}\Sigma _{j}^{0}\mathcal{V}_{j}^{0\prime }$. Define the event 
\begin{equation*}
\mathscr{A}_{2,N}(M)=\left\{ \frac{1}{\sqrt{NT}}\left\Vert \tilde{\Theta}%
_{j}-\Theta _{j}^{0}\right\Vert _{{}}\leq M\eta _{N,1},\quad \forall j\in
\left\{ 0,\cdots ,p\right\} \right\}
\end{equation*}%
for a large enough constant $M$. By the above analyses, we have $\mathbb{P}%
\left( \mathscr{A}_{2,N}^{c}(M)\right) \leq \epsilon $ for any $\epsilon >0$%
. On the event $\mathscr{A}_{2,N}(M)$, we observe that 
\begin{equation*}
\left\Vert \tilde{D}_{j}-D_{j}^{0}\right\Vert \leq \frac{1}{NT}\left(
\left\Vert \tilde{\Theta}_{j}\right\Vert +\left\Vert \Theta
_{j}^{0}\right\Vert \right) \left\Vert \tilde{\Theta}_{j}-\Theta
_{j}^{0}\right\Vert \leq 2M^{2}\eta _{N,1}.
\end{equation*}%
With Lemma C.1 of \cite{su2020strong} and Davis-Kahan sin$\Theta $ theorem
in \cite{yu2015useful}, we are ready to show that for some orthogonal matrix 
$O_{j},$ 
\begin{align}
\left\Vert \mathcal{V}_{j}^{0}-\hat{\tilde{\mathcal{V}}}_{j}O_{j}\right\Vert
& \leq \sqrt{r_{j}}\left\Vert \mathcal{V}_{j}^{0}-\hat{\tilde{\mathcal{V}}}%
_{j}O_{j}\right\Vert _{op}\leq \sqrt{r_{j}}\frac{2\sqrt{2}M^{2}\eta _{N,1}}{%
\sigma _{K_{j},1}^{2}-2M^{2}\eta _{N,1}}  \notag  \label{A.2} \\
& \leq \sqrt{r_{j}}\frac{2\sqrt{2}M^{2}\eta _{N,1}}{c_{\sigma
}^{2}-2M^{2}\eta _{N,1}}\leq \sqrt{r_{j}}\frac{2\sqrt{2}M^{2}\eta _{N,1}}{%
C_{6}c_{\sigma }^{2}}\leq C_{7}\eta _{N,1}
\end{align}%
for $C_{7}=\frac{2\sqrt{2}M^{2}\sqrt{\bar{r}}}{C_{6}c_{\sigma }^{2}}$, where
the second inequality in line two is due to the fact that $\eta _{N,1}$ is
sufficiently small and $C_{6}$ is some positive constant. }

{\small Then $\left\Vert V_{j}^{0}-\tilde{V}_{j}O_{j}\right\Vert \leq C_{7}%
\sqrt{T}\eta _{N,1}$ by the definition of $\tilde{V}_{j}$ and $V_{j}$.
Together with the fact that $\mathbb{P}(\mathscr{A}_{2,N}^{c}(M))\rightarrow
0$ by Theorem \ref{Thm1}(i), it implies $\frac{1}{\sqrt{T}}\left\Vert
V_{j}^{0}-\tilde{V}_{j}O_{j}\right\Vert =O_{p}\left( T\eta _{N,1}\right) .$ }

\subsubsection{\protect\small Proof of Statement (ii).}

{\small Define 
\begin{align*}
& u_{i}^{0}=[u_{i,0}^{0\prime },\cdots ,u_{i,p}^{0\prime }]^{\prime },\quad 
\dot{\Delta}_{i,j}=O_{j}^{\prime }\dot{u}_{i,j}-u_{i,j}^{0},\quad \dot{\Delta%
}_{i,u}=[\dot{\Delta}_{i,0}^{\prime },\cdots ,\dot{\Delta}_{i,p}^{\prime
}]^{\prime }, \\
& \tilde{\phi}_{it}=\left[ \left( O_{0}^{\prime }\tilde{v}_{t,0}\right)
^{\prime },\left( O_{1}^{\prime }\tilde{v}_{t,1}X_{1,it}\right) ^{\prime
},\cdots ,\left( O_{p}^{\prime }\tilde{v}_{t,p}X_{p,it}\right) ^{\prime }%
\right] ^{\prime },\quad \text{and}\quad \tilde{\Phi}_{i}=\frac{1}{T}%
\sum_{t=1}^{T}\tilde{\phi}_{it}\tilde{\phi}_{it}^{\prime }.
\end{align*}%
Let $\tilde{Y}_{it}:=Y_{it}-\left( O_{0}u_{i,0}^{0}\right) ^{\prime }\tilde{v%
}_{t,0}-\sum_{j=1}^{p}\left( O_{j}u_{i,j}^{0}\right) ^{\prime }\tilde{v}%
_{t,j}X_{j,it}.$ By the definition of $\left\{ \dot{u}_{i,j}\right\} $ in (%
\ref{obj2}), we have 
\begin{align*}
0& \geq \frac{1}{T}\sum_{t\in \lbrack T]}\left( Y_{it}-\dot{u}_{i,0}^{\prime
}\tilde{v}_{t,0}-\sum_{j=1}^{p}\dot{u}_{i,j}^{\prime }\tilde{v}%
_{t,j}X_{j,it}\right) ^{2}-\frac{1}{T}\sum_{t\in \lbrack T]}\tilde{Y}%
_{it}^{2} \\
& =\frac{1}{T}\sum_{t\in \lbrack T]}\bigg(\tilde{Y}_{it}-\left( \dot{u}%
_{i,0}-O_{0}u_{i,0}^{0}\right) ^{\prime }\tilde{v}_{t,0}-\sum_{j\in \lbrack
p]}\left( \dot{u}_{i,j}-O_{j}u_{i,j}^{0}\right) ^{\prime }\tilde{v}%
_{t,j}X_{j,it}\bigg)^{2}-\frac{1}{T}\sum_{t\in \lbrack T]}\tilde{Y}_{it}^{2}
\\
& =\frac{1}{T}\sum_{t\in \lbrack T]}\left[ \left( \dot{\Delta}_{i,u}^{\prime
}\tilde{\phi}_{it}\right) ^{2}-2\left( \dot{\Delta}_{i,u}^{\prime }\tilde{%
\phi}_{it}\right) \left( Y_{it}-u_{i}^{0\prime }\tilde{\phi}_{it}\right) %
\right] ,
\end{align*}%
which implies 
\begin{align}
& \left\Vert \dot{\Delta}_{i,u}\right\Vert _{2}^{2}\lambda _{\min }\left( 
\frac{1}{T}\sum_{t\in \lbrack T]}\tilde{\phi}_{it}\tilde{\phi}_{it}^{\prime
}\right) \leq \frac{1}{T}\sum_{t\in \lbrack T]}\left( \dot{\Delta}%
_{i,u}^{\prime }\tilde{\phi}_{it}\right) ^{2}\leq \frac{2}{T}\sum_{t\in
\lbrack T]}\dot{\Delta}_{i,u}^{\prime }\tilde{\phi}_{it}\left(
Y_{it}-u_{i}^{0\prime }\tilde{\phi}_{it}\right)  \notag  \label{A.3} \\
& =\frac{2}{T}\sum_{t\in \lbrack T]}\dot{\Delta}_{i,u}^{\prime }\tilde{\phi}%
_{it}\left[ e_{it}-u_{i}^{0\prime }\left( \tilde{\phi}_{it}-\phi
_{it}^{0}\right) \right]  \notag \\
& =2\left\{ \frac{1}{T}\sum_{t\in \lbrack T]}\phi _{it}^{0}e_{it}\right\}
^{\prime }\dot{\Delta}_{i,u}+2\left\{ \frac{1}{T}\sum_{t\in \lbrack
T]}\left( \tilde{\phi}_{it}-\phi _{it}^{0}\right) e_{it}\right\} ^{\prime }%
\dot{\Delta}_{i,u}-\frac{2}{T}\sum_{t\in \lbrack T]}\tilde{\phi}%
_{it}^{\prime }\dot{\Delta}_{i,u}\left[ u_{i}^{0\prime }\left( \tilde{\phi}%
_{it}-\phi _{it}^{0}\right) \right]  \notag \\
& =:2G_{i,1}^{\prime }\dot{\Delta}_{i,u}+2G_{i,2}-2G_{i,3}.
\end{align}%
}

{\small We first deal with $G_{1,i}$. Conditional on $\mathscr{D}$, the
randomness in $G_{1,i}$ comes from $\left\{ e_{it},X_{it}\right\} _{t\in
\lbrack T]}$, which is the\ (conditional) strong mixing sequence by
Assumption \ref{ass:1}(iii). Besides, we observe that 
\begin{equation*}
\max_{i],t\in \lbrack T]}\left\Vert Var\left( \phi _{it}^{0}e_{it}\big|%
\mathscr{D}\right) \right\Vert _{{}}\lesssim \max_{i,t}\left[ \mathbb{E}%
\left( e_{it}^{2}\bigg|\mathscr{D}\right) +\sum_{j\in \lbrack p]}\mathbb{E}%
\left( X_{j,it}^{2}e_{it}^{2}\bigg|\mathscr{D}\right) \right] =O_{a.s.}(1)~
\end{equation*}%
by Lemma \ref{Lem:bounded u&v}(ii) and Assumption \ref{ass:1}(v). Following
similar arguments, we have 
\begin{align*}
& \max_{i,t}\sum_{s=t+1}^{T}\left\Vert Cov\left( \phi _{it}^{0}e_{it},\phi
_{is}^{0}e_{is}\big|\mathscr{D}\right) \right\Vert  \\
& \leq 8\max_{t}\sum_{s=t+1}^{T}\left[ \mathbb{E}\left( \left\Vert \phi
_{it}^{0}e_{it}\right\Vert _{2}^{q}\big|\mathscr{D}\right) \right] ^{1/q}%
\left[ \left\Vert \mathbb{E}\left( \phi _{is}^{0}e_{is}\right\Vert _{2}^{q}%
\big|\mathscr{D}\right) \right] ^{1/q}\left( \alpha (t-s)\right)
^{1-2/q}=O_{a.s.}(1),
\end{align*}%
where the first inequality is by the conditional Davydov's inequality that
says 
\begin{equation*}
\left\vert Cov\left[ a\left( x_{t}\right) ,a\left( x_{s}\right) |\mathscr{D}%
\right] \right\vert \leq 8\left[ \mathbb{E}[\left\Vert a\left( x_{t}\right)
\right\Vert ^{p}|\mathscr{D}]\right] ^{\frac{1}{p}}\left[ \mathbb{E}%
[\left\Vert a\left( x_{s}\right) \right\Vert ^{q}|\mathscr{D}]\right] ^{%
\frac{1}{q}}\alpha (t-s)^{\frac{1}{r}}
\end{equation*}%
for any conditional strong mixing sequence $\left( x_{t},t\in \lbrack
T]\right) $ with mixing coefficient $\alpha (\cdot )$ and $\frac{1}{p}+\frac{%
1}{q}+\frac{1}{r}=1$. See Lemma A.4 in \cite{su2013testing}. }

{\small Following this, for some constant $C_{8}$, we have 
\begin{equation*}
\max_{i,t}\left[ \left\Vert Var\left( \phi _{it}^{0}e_{it}\big|\mathscr{D}%
\right) \right\Vert +2\sum_{s=t+1}^{T}\left\Vert Cov\left( \phi
_{it}^{0}e_{it},\phi _{is}^{0}e_{is}\big|\mathscr{D}\right) \right\Vert %
\right] \leq C_{8}\text{ a.s.},
\end{equation*}%
and $\max_{i,t}\left\Vert \phi _{it}^{0}e_{it}\right\Vert _{\max }\leq
C_{8}(NT)^{1/q}$ by Lemma \ref{Lem:bounded u&v}(i) and Assumption \ref{ass:1}%
(iv). Define $\mathscr{A}_{3,N}(M)=\{\max_{i,t}\left\Vert \phi
_{it}^{0}e_{it}\right\Vert \leq M(NT)^{1/q}\}$ and $\mathscr{A}%
_{3,N,i}(M)=\{\max_{t}\left\Vert \phi _{it}^{0}e_{it}\right\Vert \leq
M(NT)^{1/q}\}$ for a large enough constant $M$. For a positive constant $%
C_{9}$, it yields that 
\begin{align*}
& \mathbb{P}\left( \max_{i}\frac{1}{T}\left\Vert \sum_{t\in \lbrack T]}\phi
_{it}^{0}e_{it}\right\Vert >C_{9}\sqrt{\frac{\log N}{T}}(NT)^{\frac{1}{q}%
}\right)  \\
& \leq \mathbb{P}\left( \max_{i}\frac{1}{T}\left\Vert \sum_{t\in \lbrack
T]}\phi _{it}^{0}e_{it}\right\Vert >C_{9}\sqrt{\frac{\log N}{T}}(NT)^{\frac{1%
}{q}},\mathscr{A}_{3,N}(M)\right) +\mathbb{P}\left( \mathscr{A}%
_{3,N}^{c}(M)\right)  \\
& \leq \sum_{i\in \lbrack N]}\mathbb{P}\left( \frac{1}{T}\left\Vert
\sum_{t\in \lbrack T]}\phi _{it}^{0}e_{it}\right\Vert >C_{9}\sqrt{\frac{\log
N}{T}}(NT)^{\frac{1}{q}},\mathscr{A}_{3,N}(M)\right) +\mathbb{P}\left( %
\mathscr{A}_{3,N}^{c}(M)\right)  \\
& \leq \sum_{i\in \lbrack N]}\mathbb{P}\left( \frac{1}{T}\left\Vert
\sum_{t\in \lbrack T]}\phi _{it}^{0}e_{it}\right\Vert >C_{9}\sqrt{\frac{\log
N}{T}}(NT)^{\frac{1}{q}},\mathscr{A}_{3,N,i}(M)\right) +\mathbb{P}\left( %
\mathscr{A}_{3,N}^{c}(M)\right)  \\
& \leq \sum_{i\in \lbrack N]}\exp \left( -\frac{c_{4}C_{9}^{2}T\log
N(NT)^{2/q}}{TC_{8}+C_{8}^{2}(NT)^{2/q}+C_{8}C_{9}(NT)^{2/q}\sqrt{T\log N}%
(\log T)^{2}}\right) +o(1)=o(1)
\end{align*}%
where the last inequality holds by Bernstein's inequality in Lemma \ref%
{Lem:Bern_mixing}(ii) and the fact that $\mathbb{P}\left( \mathscr{A}%
_{3,N}^{c}(M)\right) =o(1)$. It follows that 
\begin{equation}
\max_{i}\frac{\left\vert G_{i,1}^{\prime }\dot{\Delta}_{i,u}\right\vert }{%
\left\Vert \dot{\Delta}_{i,u}\right\Vert }\leq \max_{i}\left\Vert
G_{i,1}\right\Vert =O_{p}(\sqrt{(\log N)/T}(NT)^{\frac{1}{q}}).  \label{G1}
\end{equation}%
}

{\small For $G_{i,2}$, we notice that 
\begin{align}
\max_{i}\frac{\left\vert G_{i,2}\right\vert }{\left\Vert \dot{\Delta}%
_{i,u}\right\Vert }& =\max_{i}\frac{\left\vert \left\{ \frac{1}{T}\sum_{t\in
\lbrack T]}\left( \tilde{\phi}_{it}-\phi _{it}^{0}\right) e_{it}\right\}
^{\prime }\dot{\Delta}_{i,u}\right\vert }{\left\Vert \dot{\Delta}%
_{i,u}\right\Vert }\leq \max_{i}\left\Vert \frac{1}{T}\sum_{t\in \lbrack
T]}\left( \tilde{\phi}_{it}-\phi _{it}^{0}\right) e_{it}\right\Vert _{2} 
\notag \\
& \leq \max_{i}\sqrt{\frac{1}{T}\sum_{t\in \lbrack T]}\left\Vert \tilde{\phi}%
_{it}-\phi _{it}^{0}\right\Vert _{2}^{2}}\max_{i}\sqrt{\frac{1}{T}\sum_{t\in
\lbrack T]}\left\vert e_{it}\right\vert ^{2}}=O_{p}(\eta _{N,1}(NT)^{1/q}),
\label{G2}
\end{align}%
where the second inequality holds by Cauchy's inequality and the last
equality is by Lemma \ref{Lem:bounded u&v}(iv) and Assumption \ref{ass:1}%
(iv) }

{\small For $G_{i,3}$, we have 
\begin{align}
\max_{i}\frac{\left\vert G_{i,3}\right\vert }{\left\Vert \dot{\Delta}%
_{i,u}\right\Vert }& =\max_{i}\frac{\left\vert \frac{1}{T}\sum_{t\in \lbrack
T]}\tilde{\phi}_{it}^{\prime }\dot{\Delta}_{i,u}\left[ u_{i}^{0\prime
}\left( \tilde{\phi}_{it}-\phi _{it}^{0}\right) \right] \right\vert }{%
\left\Vert \dot{\Delta}_{i,u}\right\Vert }  \notag \\
& \leq \max_{i}\sqrt{\frac{1}{T}\sum_{t\in \lbrack T]}\left\Vert \tilde{\phi}%
_{it}\right\Vert ^{2}}\max_{i}\left\Vert u_{i}^{0}\right\Vert \max_{i,t}%
\sqrt{\frac{1}{T}\sum_{t\in \lbrack T]}\left\Vert \tilde{\phi}_{it}-\phi
_{it}^{0}\right\Vert ^{2}}=O_{p}(\eta _{N,1}(NT)^{1/q}),  \label{G3}
\end{align}%
where the inequality holds by Cauchy's inequality and the last line is by
Lemma \ref{Lem:bounded u&v}(i) and (iv). }

{\small Combining (\ref{A.3})-(\ref{G3}) and Lemma \ref{Lem:lambda_phi}
yields 
\begin{equation*}
\max_{i}\left\Vert \dot{u}_{i,j}-O_{i,j}^{(1)}u_{i,j}^{0}\right\Vert \leq
\max_{i}\left\Vert \dot{\Delta}_{i,u}\right\Vert =O_{p}\left( \sqrt{\frac{%
\log N\vee T}{N\wedge T}}(NT)^{1/q}\right) .
\end{equation*}%
}

{\small The union bound of $\dot{v}_{t,j}$ can be obtained in the same
manner and we sketch the proof here. Define 
\begin{align*}
& v_{t}^{0}=[v_{t,0}^{0\prime },\cdots ,v_{t,p}^{0\prime }]^{\prime },\quad 
\dot{\Delta}_{t,j}=O_{j}^{\prime }\dot{v}_{t,j}-v_{t,j}^{0},\quad \dot{\Delta%
}_{t,v}=[\dot{\Delta}_{t,0}^{\prime },\cdots ,\dot{\Delta}_{t,p}^{\prime
}]^{\prime }, \\
& \dot{\psi}_{it}=\left[ \left( O_{0}^{\prime }\dot{u}_{i,0}\right) ^{\prime
},\left( O_{1}^{\prime }\dot{u}_{i,1}X_{1,it}\right) ^{\prime },\cdots
,\left( O_{p}^{\prime }\dot{u}_{i,p}X_{p,it}\right) ^{\prime }\right]
^{\prime },\quad \text{and}\quad \dot{\Psi}_{t}=\frac{1}{N}\sum_{i\in
\lbrack N]}\dot{\psi}_{it}\dot{\psi}_{it}^{\prime }.
\end{align*}%
Following the steps to derive (\ref{A.3}), we can also obtain 
\begin{align}
& \left\Vert \dot{\Delta}_{t,v}\right\Vert ^{2}\lambda _{\min }\left( \frac{1%
}{N}\sum_{i\in \lbrack N]}\dot{\psi}_{it}\dot{\psi}_{it}^{\prime }\right) 
\notag \\
& =2\left\{ \frac{1}{N}\sum_{i\in \lbrack N]}\psi _{it}^{0}e_{it}\right\}
^{\prime }\dot{\Delta}_{t,v}+\frac{2}{N}\sum_{i\in \lbrack N]}\left( \dot{%
\psi}_{it}-\psi _{it}^{0}\right) ^{\prime }\dot{\Delta}_{t,v}e_{it}-\frac{2}{%
N}\sum_{i\in \lbrack N]}\dot{\psi}_{it}^{\prime }\dot{\Delta}_{t,v}\left[
v_{t}^{0\prime }\left( \dot{\psi}_{it}-\psi _{it}^{0}\right) \right] .
\label{A.6}
\end{align}%
By the fact that 
\begin{align*}
\max_{t}\frac{1}{N}\sum_{i\in \lbrack N]}\left\Vert \dot{\psi}%
_{it}\right\Vert ^{2}& =\max_{t}\frac{1}{N}\sum_{i\in \lbrack N]}\left(
\left\Vert \dot{u}_{i,0}\right\Vert ^{2}+\sum_{j\in \lbrack p]}\left\Vert 
\dot{u}_{i,j}\right\Vert ^{2}X_{j,it}^{2}\right) \\
& \leq \max_{i}\left\Vert \dot{u}_{i,0}\right\Vert ^{2}+\max_{i\in \lbrack
N],j\in \lbrack p]}\left\Vert \dot{u}_{i,j}\right\Vert ^{2}\sum_{j\in
\lbrack p]}\max_{t}\frac{1}{N}\sum_{i\in \lbrack N]}X_{j,it}^{2}=O_{p}(1), \\
\max_{t}\frac{1}{N}\sum_{i\in \lbrack N]}\left\Vert \dot{\psi}_{it}-\psi
_{it}^{0}\right\Vert ^{2}& \leq \max_{i}\left\Vert \dot{u}%
_{i,0}-u_{i,0}^{0}\right\Vert ^{2}+\max_{i\in \lbrack N],j\in \lbrack
p]}\left\Vert \dot{u}_{i,j}-u_{i,j}^{0}\right\Vert ^{2}\sum_{j\in \lbrack
p]}\max_{t}\frac{1}{N}\sum_{i\in \lbrack N]}X_{j,it}^{2}=O_{p}(\eta _{N,2}),
\end{align*}%
where $\eta _{N,2}=\sqrt{\frac{\log N\vee T}{N\wedge T}}(NT)^{1/q}$ and the
first inequality holds by Lemma \ref{Lem:bounded u&v}(i), we obtain that 
\begin{align*}
& \max_{t}\frac{\left\vert \left\{ \frac{1}{N}\sum_{i\in \lbrack N]}\psi
_{it}^{0}e_{it}\right\} ^{\prime }\dot{\Delta}_{t,v}\right\vert }{\left\Vert 
\dot{\Delta}_{t,v}\right\Vert }=O_{p}\left( \sqrt{\frac{\log T}{N}}(NT)^{%
\frac{1}{q}}\right) , \\
& \max_{t}\frac{\left\vert \frac{1}{N}\sum_{i\in \lbrack N]}\dot{\psi}%
_{it}^{\prime }\dot{\Delta}_{t,v}\left[ v_{t}^{0\prime }\left( \dot{\psi}%
_{it}-\psi _{it}^{0}\right) \right] \right\vert }{\left\Vert \dot{\Delta}%
_{t,v}\right\Vert }=O_{p}(\eta _{N,2}),\text{ and} \\
& \max_{t}\frac{\left\vert \frac{1}{N}\sum_{i\in \lbrack N]}\left( \dot{\psi}%
_{it}-\psi _{it}^{0}\right) ^{\prime }\dot{\Delta}_{t,v}e_{it}\right\vert }{%
\left\Vert \dot{\Delta}_{t,v}\right\Vert }=O_{p}(\eta _{N,2}),
\end{align*}%
where the first line is by conditional Bernstein's inequality for i.i.d.
sequence and the last two lines are by the analogous arguments in (\ref{G2})
and (\ref{G3}). It follows that 
\begin{equation*}
\max_{t}\left\Vert \dot{v}_{t,j}-O_{j}v_{t,j}^{0}\right\Vert =O_{p}\left( 
\sqrt{\frac{\log N\vee T}{N\wedge T}}(NT)^{1/q}\right) .\quad \blacksquare
\end{equation*}%
}

\subsubsection{\protect\small Proof of Statement (iii).}

{\small For $\forall j\in \lbrack p]\cup \{0\}$, $i\in \lbrack N]$ and $t\in
\lbrack T]$, we can show that 
\begin{align*}
& \dot{\Theta}_{j,it}-\Theta _{j,it}^{0}=\dot{u}_{i,j}^{\prime }\dot{v}%
_{t,j}-u_{i,j}^{0\prime }v_{t,j}^{0} \\
& =(\dot{u}_{i,j}-O_{j}u_{i,j}^{0})^{\prime }(\dot{v}%
_{t,j}-O_{j}v_{t,j}^{0})+O_{j}u_{i,j}^{0\prime }(\dot{v}%
_{t,j}-O_{j}v_{t,j}^{0})+O_{j}v_{t,j}^{0\prime }(\dot{u}%
_{i,j}-O_{j}u_{i,j}^{0}),
\end{align*}%
which implies 
\begin{align*}
\max_{i,t}\left\vert \dot{\Theta}_{j,it}-\Theta _{j,it}^{0}\right\vert &
\leq \max_{i}\left\Vert \dot{u}_{i,j}-O_{j}u_{i,j}^{0}\right\Vert
\max_{t}\left\Vert \dot{v}_{t,j}-O_{j}v_{t,j}^{0}\right\Vert \\
& +\max_{i}\left\Vert O_{j}u_{i,j}^{0}\right\Vert \max_{t}\left\Vert \dot{v}%
_{t,j}-O_{j}v_{t,j}^{0}\right\Vert +\max_{i}\left\Vert \dot{u}%
_{i,j}-O_{j}u_{i,j}^{0}\right\Vert \max_{t}\left\Vert
O_{j}v_{t,j}^{0}\right\Vert =O_{p}(\eta _{N,2}),
\end{align*}%
where the last equality combines results from Theorem \ref{Thm1}(ii) and
Lemma \ref{Lem:bounded u&v}(i). $\blacksquare $ }

\subsection{{\protect\small Proof of Theorem \protect\ref{Thm2}}}

{\small 
}

{\small To prove $\mathbb{P(}\hat{T}_{1}=T_{1})\rightarrow 1$, it suffices
to show: (i) $\mathbb{P(}\hat{T}_{1}<T_{1})\rightarrow 0$ and (ii) $\mathbb{%
P(}\hat{T}_{1}>T_{1})\rightarrow 0$. }

{\small First, we focus on (i). Let $\Delta _{it}(j)=\dot{\Theta}%
_{j,it}-\Theta _{j,it}^{0}$, $\bar{\Delta}_{s,i}(j)=\frac{1}{s}%
\sum_{t=1}^{s}(\dot{\Theta}_{j,it}-\Theta _{j,it}^{0})$ and $\bar{\Delta}%
_{s_{+},i}(j)=\frac{1}{T-s}\sum_{t=s+1}^{T}(\dot{\Theta}_{j,it}-\Theta
_{j,it}^{0})$. When $s<T_{1}$, we have 
\begin{align*}
\bar{\dot{\Theta}}_{j,i}^{(1s)}& =\frac{1}{s}\sum_{t=1}^{s}\dot{\Theta}%
_{j,it}=\frac{1}{s}\sum_{t=1}^{s}\left[ \Theta _{j,it}^{0}+(\dot{\Theta}%
_{j,it}-\Theta _{j,it}^{0})\right] =\alpha _{g_{i}^{(1)},j}^{(1)}+\bar{\Delta%
}_{s,i}(j), \\
\bar{\dot{\Theta}}_{j,i}^{(2s)}& =\frac{1}{T-s}\sum_{t=s+1}^{T}\dot{\Theta}%
_{j,it}=\frac{1}{T-s}\sum_{t=s+1}^{T}\left[ \Theta _{j,it}^{0}+(\dot{\Theta}%
_{j,it}-\Theta _{j,it}^{0})\right]  \\
& =\frac{T_{1}-s}{T-s}\alpha _{g_{i}^{(1)},j}^{(1)}+\frac{T-T_{1}}{T-s}%
\alpha _{g_{i}^{(2)},j}^{(2)}+\bar{\Delta}_{s_{+},i}(j),
\end{align*}%
with $\alpha _{g_{i}^{(1)},j}^{(1)}$ and $\alpha _{g_{i}^{(2)},j}^{(2)}$
being the $j$-th element of $\alpha _{g_{i}^{(1)}}^{(1)}$ and $\alpha
_{g_{i}^{(2)}}^{(2)}$, respectively. It yields 
\begin{align*}
\dot{\Theta}_{j,it}-\bar{\dot{\Theta}}_{j,i}^{(1s)}=& \dot{\Theta}%
_{j,it}-\alpha _{g_{i}^{(1)},j}^{(1)}-\bar{\Delta}_{s,i}(j)=\Delta _{it}(j)-%
\bar{\Delta}_{s,i}(j),\quad t\leq s,\text{ and} \\
\dot{\Theta}_{j,it}-\bar{\dot{\Theta}}_{j,i}^{(2s)}=& \dot{\Theta}_{j,it}-%
\frac{T_{1}-s}{T-s}\alpha _{g_{i}^{(1)},j}^{(1)}-\frac{T-T_{1}}{T-s}\alpha
_{g_{i}^{(2)},j}^{(2)}-\bar{\Delta}_{s_{+},i}(j) \\
=& \left\{ 
\begin{array}{ll}
\frac{T-T_{1}}{T-s}(\alpha _{g_{i}^{(1)},j}^{(1)}-\alpha
_{g_{i}^{(2)},j}^{(2)})+\Delta _{it}(j)-\bar{\Delta}_{s_{+},i}(j)\text{ } & 
\text{if }s+1\leq t\leq T_{1} \\ 
\frac{T_{1}-s}{T-s}(\alpha _{g_{i}^{(2)},j}^{(2)}-\alpha
_{g_{i}^{(1)},j}^{(1)})+\Delta _{it}(j)-\bar{\Delta}_{s_{+},i}(j) & \text{if 
}T_{1}+1\leq t\leq T%
\end{array}%
\right. .
\end{align*}%
Then, we have $\sum_{t=1}^{s}\left[ \dot{\Theta}_{j,it}-\bar{\dot{\Theta}}%
_{j,i}^{(1s)}\right] ^{2}=\sum_{t=1}^{s}\left[ \Delta _{it}(j)-\bar{\Delta}%
_{s,i}(j)\right] ^{2},$ and 
\begin{align}
\sum_{t=s+1}^{T}\left[ \dot{\Theta}_{j,it}-\bar{\dot{\Theta}}_{j,i}^{(2s)}%
\right] ^{2}& =\sum_{t=s+1}^{T_{1}}\left[ \frac{T-T_{1}}{T-s}(\alpha
_{g_{i}^{(1)},j}^{(1)}-\alpha _{g_{i}^{(2)},j}^{(2)})+\Delta _{it}(j)-\bar{%
\Delta}_{s_{+},i}(j)\right] ^{2}  \notag  \label{B.2} \\
& +\sum_{t=T_{1}+1}^{T}\left[ \frac{T_{1}-s}{T-s}(\alpha
_{g_{i}^{(2)},j}^{(2)}-\alpha _{g_{i}^{(1)},j}^{(1)})+\Delta _{it}(j)-\bar{%
\Delta}_{s_{+},i}(j)\right] ^{2}  \notag \\
& =\frac{\left( T_{1}-s\right) \left( T-T_{1}\right) ^{2}}{\left( T-s\right)
^{2}}(\alpha _{g_{i}^{(1)},j}^{(1)}-\alpha
_{g_{i}^{(2)},j}^{(2)})^{2}+\sum_{t=s+1}^{T_{1}}\left[ \Delta _{it}(j)-\bar{%
\Delta}_{s_{+},i}(j)\right] ^{2}  \notag \\
& +2\frac{T-T_{1}}{T-s}(\alpha _{g_{i}^{(1)},j}^{(1)}-\alpha
_{g_{i}^{(2)},j}^{(2)})\sum_{t=s+1}^{T_{1}}\left[ \Delta _{it}(j)-\bar{\Delta%
}_{s_{+},i}(j)\right]   \notag \\
& +\frac{\left( T-T_{1}\right) \left( T_{1}-s\right) ^{2}}{\left( T-s\right)
^{2}}(\alpha _{g_{i}^{(2)},j}^{(2)}-\alpha
_{g_{i}^{(1)},j}^{(1)})^{2}+\sum_{t=T_{1}+1}^{T}\left[ \Delta _{it}(j)-\bar{%
\Delta}_{s_{+},i}(j)\right] ^{2}  \notag \\
& +2\frac{T_{1}-s}{T-s}(\alpha _{g_{i}^{(2)},j}^{(2)}-\alpha
_{g_{i}^{(1)},j}^{(1)})\sum_{t=T_{1}+1}^{T}\left[ \Delta _{it}(j)-\bar{\Delta%
}_{s_{+},i}(j)\right]   \notag \\
& =\frac{\left( T_{1}-s\right) \left( T-T_{1}\right) }{T-s}(\alpha
_{g_{i}^{(1)},j}^{(1)}-\alpha _{g_{i}^{(2)},j}^{(2)})^{2}+\sum_{t=s+1}^{T}%
\left[ \Delta _{it}(j)-\bar{\Delta}_{s_{+},i}(j)\right] ^{2}  \notag \\
& +2\frac{T-T_{1}}{T-s}(\alpha _{g_{i}^{(1)},j}^{(1)}-\alpha
_{g_{i}^{(2)},j}^{(2)})\sum_{t=s+1}^{T_{1}}\left[ \Delta _{it}(j)-\bar{\Delta%
}_{s_{+},i}(j)\right]   \notag \\
& +2\frac{T_{1}-s}{T-s}(\alpha _{g_{i}^{(2)},j}^{(2)}-\alpha
_{g_{i}^{(1)},j}^{(1)})\sum_{t=T_{1}+1}^{T}\left[ \Delta _{it}(j)-\bar{\Delta%
}_{s_{+},i}(j)\right] .
\end{align}%
Define $L(s)=\frac{1}{pNT}\sum_{j\in \lbrack p]}\sum_{i\in \lbrack
N]}\{\sum_{t=1}^{s}[\dot{\Theta}_{j,it}-\bar{\dot{\Theta}}%
_{j,i}^{(1s)}]^{2}+\sum_{t=s+1}^{T}[\dot{\Theta}_{j,it}-\bar{\dot{\Theta}}%
_{j,i}^{(2s)}]^{2}\}$. Then we have 
\begin{align}
L(s)& =\frac{1}{pNT}\sum_{j\in \lbrack p]}\sum_{i\in \lbrack N]}\frac{\left(
T_{1}-s\right) \left( T-T_{1}\right) }{T-s}(\alpha
_{g_{i}^{(1)},j}^{(1)}-\alpha _{g_{i}^{(2)},j}^{(2)})^{2}  \notag \\
& +\frac{1}{pNT}\sum_{j\in \lbrack p]}\sum_{i\in \lbrack N]}\sum_{t=1}^{s}%
\left[ \Delta _{it}(j)-\bar{\Delta}_{s,i}(j)\right] ^{2}+\frac{1}{pNT}%
\sum_{j\in \lbrack p]}\sum_{i\in \lbrack N]}\sum_{t=s+1}^{T}\left[ \Delta
_{it}(j)-\bar{\Delta}_{s_{+},i}(j)\right] ^{2}  \notag \\
& +\frac{2}{pNT}\sum_{j\in \lbrack p]}\sum_{i\in \lbrack N]}\frac{T-T_{1}}{%
T-s}(\alpha _{g_{i}^{(1)},j}^{(1)}-\alpha
_{g_{i}^{(2)},j}^{(2)})\sum_{t=s+1}^{T_{1}}\left[ \Delta _{it}(j)-\bar{\Delta%
}_{s_{+},i}(j)\right]   \notag \\
& +\frac{2}{pNT}\sum_{j\in \lbrack p]}\sum_{i\in \lbrack N]}\frac{T_{1}-s}{%
T-s}(\alpha _{g_{i}^{(2)},j}^{(2)}-\alpha
_{g_{i}^{(1)},j}^{(1)})\sum_{t=T_{1}+1}^{T}\left[ \Delta _{it}(j)-\bar{\Delta%
}_{s_{+},i}(j)\right] :=\sum_{\ell =1}^{5}L_{\ell }(s).  \label{B.3}
\end{align}%
Obviously, 
\begin{equation}
L(T_{1})=L_{2}(T_{1})+L_{3}(T_{1}).  \label{B.4a}
\end{equation}%
Note that the event $\hat{T}_{1}<T$ implies that there exists an $s<T_{1}$
such that $L(s)-L(T_{1})<0$, which means we can prove (i) by showing that $%
\mathbb{P}\left( \exists s<T_{1},L(s)-L(T_{1})<0\right) \rightarrow 0$. By (%
\ref{B.3}) and (\ref{B.4a}), we observe that 
\begin{align}
L(s)-L(T_{1})& =L_{1}(s)+\left[ L_{2}(s)-L_{2}(T_{1})\right] +\left[
L_{3}(s)-L_{3}(T_{1})\right] +L_{4}(s)+L_{5}(s)  \notag \\
& :=A_{1}\left( s\right) +A_{2}\left( s\right) +A_{3}\left( s\right)
+A_{4}\left( s\right) +A_{5}\left( s\right) .  \label{B.5}
\end{align}%
Recall that $\eta _{N,2}=\sqrt{\frac{\log N\vee T}{N\wedge T}}(NT)^{1/q}$.
Let $\frac{T_{1}-s}{T}=\kappa _{s}$ and note that $0<\frac{1}{T}\leq \kappa
_{s}\leq \frac{T_{1}-2}{T}\asymp 1$. We analyze the five terms in (\ref{B.5}%
) in turn. }

{\small For $A_{1}\left( s\right) $, we have 
\begin{align}
A_{1}\left( s\right) & =\frac{1}{pNT}\sum_{j\in \lbrack p]}\sum_{i\in
\lbrack N]}\frac{\left( T_{1}-s\right) \left( T-T_{1}\right) }{T-s}(\alpha
_{g_{i}^{(1)},j}^{(1)}-\alpha _{g_{i}^{(2)},j}^{(2)})^{2}  \notag \\
& =\frac{T_{1}-s}{T-s}\left( 1-\tau _{T}\right) \frac{1}{pN}\sum_{j\in
\lbrack p]}\sum_{i\in \lbrack N]}(\alpha _{g_{i}^{(1)},j}^{(1)}-\alpha
_{g_{i}^{(2)},j}^{(2)})^{2}  \notag \\
& =\frac{T_{1}-s}{T-s}\left( 1-\tau _{T}\right) \frac{1}{pN}\sum_{i\in
\lbrack N]}\left\Vert \alpha _{g_{i}^{(1)}}^{(1)}-\alpha
_{g_{i}^{(2)}}^{(2)}\right\Vert ^{2}  \notag \\
& =\kappa _{s}\frac{\left( 1-\tau _{T}\right) }{1-\frac{s}{T}}D_{N\alpha
}=\kappa _{s}O(\zeta _{NT}^{2}),  \label{B.6}
\end{align}%
where $D_{N\alpha }:=\frac{1}{pN}\sum_{i\in \lbrack N]}\left\Vert \alpha
_{g_{i}^{(1)}}^{(1)}-\alpha _{g_{i}^{(2)}}^{(2)}\right\Vert ^{2}$ and the
last equality holds by Assumption \ref{ass:6}. }

{\small For $A_{2}\left( s\right) $, noting that 
\begin{equation*}
\bar{\Delta}_{T_{1},i}\left( j\right) -\bar{\Delta}_{s,i}(j)=\frac{1}{T_{1}}%
\sum_{t=1}^{T_{1}}\Delta_{it}(j)-\frac{1}{s}\left[ \sum_{t=1}^{T_{1}}%
\Delta_{it}(j)-\sum_{t=s+1}^{T_{1}}\Delta_{it}(j)\right] =\frac{s-T_{1}}{%
T_{1}s}\sum_{t=1}^{T_{1}}\Delta_{it}(j)+\frac{1}{s}\sum_{t=s+1}^{T_{1}}%
\Delta_{it}(j),
\end{equation*}%
we have 
\begin{align*}
T_{1}\bar{\Delta}_{T_{1},i}^{2}\left( j\right) -s\bar{\Delta}%
_{s,i}^{2}\left( j\right) & =\left( T_{1}-s\right) \bar{\Delta}%
_{T_{1},i}^{2}\left( j\right) +s\left[ \bar{\Delta}_{T_{1},i}^{2}\left(j%
\right) -\bar{\Delta}_{s,i}^{2}\left( j\right) \right] \\
& =\left( T_{1}-s\right) \bar{\Delta}_{T_{1},i}^{2}\left( j\right) +s\left[ 
\bar{\Delta}_{T_{1},i}\left( j\right) +\bar{\Delta}_{s,i}(j)\right] \left[%
\bar{\Delta}_{T_{1},i}\left( j\right) -\bar{\Delta}_{s,i}(j)\right] \\
& =\left( T_{1}-s\right) \bar{\Delta}_{T_{1},i}^{2}\left( j\right) +\left[%
\bar{\Delta}_{T_{1},i}\left( j\right) +\bar{\Delta}_{s,i}(j)\right] \left[%
\sum_{t=s+1}^{T_{1}}\Delta_{it}(j)-\frac{T_{1}-s}{T_{1}}\sum_{t=1}^{T_{1}}%
\Delta_{it}(j)\right] .
\end{align*}%
It follows that 
\begin{align}
A_{2}\left( s\right) & =\frac{1}{pNT}\sum_{j\in [p]}\sum_{i\in[N]%
}\left\{\sum_{t=1}^{s}\left[ \Delta_{it}(j)-\bar{\Delta}_{s,i}(j)\right]%
^{2}-\sum_{t=1}^{T_{1}}\left[ \Delta_{it}(j)-\bar{\Delta}_{T_{1},i}\left(j%
\right) \right] ^{2}\right\}  \notag \\
& =\frac{1}{pNT}\sum_{j\in [p]}\sum_{i\in
[N]}\left\{\sum_{t=1}^{s}\Delta_{it}^{2}\left( j\right) -s\bar{\Delta}%
_{s,i}^{2}\left(j\right)-\sum_{t=1}^{T_{1}}\Delta_{it}^{2}\left( j\right)
+T_{1}\bar{\Delta}_{T_{1},i}^{2}\left( j\right) \right\}  \notag \\
& =\frac{1}{pNT}\sum_{j\in [p]}\sum_{i\in[N]}\left\{-\sum_{t=s+1}^{T_{1}}%
\Delta_{it}^{2}\left( j\right) +T_{1}\bar{\Delta}_{T_{1},i}^{2}\left(
j\right) -s\bar{\Delta}_{s,i}^{2}\left(j\right)\right\}  \notag \\
& =-\kappa_{s}\frac{1}{pN\left( T_{1}-s\right) }\sum_{j\in [p]}\sum_{i\in[N]%
}\sum_{t=s+1}^{T_{1}}\Delta_{it}^{2}\left( j\right)+\kappa_{s}\frac{1}{pN}%
\sum_{j\in [p]}\sum_{i\in [N]}\bar{\Delta}_{T_{1},i}^{2}\left( j\right) 
\notag \\
& +\kappa_{s}\frac{1}{pN}\sum_{j\in [p]}\sum_{i\in [N]}\left[\bar{\Delta}%
_{T_{1},i}\left( j\right) +\bar{\Delta}_{s,i}(j)\right] \frac{1}{T_{1}-s}%
\sum_{t=s+1}^{T_{1}}\Delta_{it}(j)  \notag \\
& -\kappa_{s}\frac{1}{pN}\sum_{j\in [p]}\sum_{i\in [N]}\left[\bar{\Delta}%
_{T_{1},i}\left( j\right) +\bar{\Delta}_{s,i}(j)\right] \frac{1}{T_{1}}%
\sum_{t=1}^{T_{1}}\Delta_{it}(j)  \notag \\
& =\kappa_{s}O_{p}\left( \eta_{N,2}^{2}\right)
=\kappa_{s}o_{p}(\zeta_{NT}^{2}),  \label{B.7}
\end{align}%
where the second last equality holds by the fact that 
\begin{equation}
\max_{i\in [N],t\in [T],j\in [p]}\left\vert
\Delta_{it}(j)\right\vert=O_{p}\left( \eta_{N,2}\right)  \label{B.8}
\end{equation}%
from Theorem \ref{Thm1}(iii) and 
\begin{align}
& \max_{i\in [N],j\in [p]}\left\vert \bar{\Delta}_{s,i}(j)\right\vert=%
\max_{i\in [N],j\in [p]}\left\vert \frac{1}{s}\sum_{t=1}^{s}\Delta_{it}(j)%
\right\vert \leq \max_{i\in [N],t\in [T],j\in [p]}\left\vert
\Delta_{it}(j)\right\vert=O_{p}\left( \eta_{N,2}\right),  \notag \\
& \max_{i\in [N],j\in [p]}\left\vert \bar{\Delta}_{T_{1},i}\left(
j\right)\right\vert =\max_{i\in [N],j\in [p]}\left\vert \frac{1}{T_{1}}%
\sum_{t=1}^{T_{1}}\Delta_{it}(j)\right\vert=O_{p}\left( \eta_{N,2}\right).
\label{B.9}
\end{align}
}

{\small Similarly, noting that 
\begin{align*}
\bar{\Delta}_{T_{1+},i}\left( j\right) -\bar{\Delta}_{s_{+},i}(j)& =\frac{1}{%
T-T_{1}}\sum_{t=T_{1}+1}^{T}\Delta _{it}(j)-\frac{1}{T-s}\sum_{t=s+1}^{T}%
\Delta _{it}(j) \\
& =(\frac{1}{T-T_{1}}-\frac{1}{T-s})\sum_{t=T_{1}+1}^{T}\Delta _{it}(j)+%
\frac{1}{T-s}\left[ \sum_{t=T_{1}+1}^{T}\Delta
_{it}(j)-\sum_{t=s+1}^{T}\Delta _{it}(j)\right] \\
& =\frac{T_{1}-s}{\left( T-T_{1}\right) \left( T-s\right) }%
\sum_{t=T_{1}+1}^{T}\Delta _{it}(j)-\frac{1}{T-s}\sum_{t=s+1}^{T_{1}}\Delta
_{it}(j)
\end{align*}%
and 
\begin{align*}
& \left( T-T_{1}\right) \bar{\Delta}_{T_{1+},i}^{2}\left( j\right) -\left(
T-s\right) \bar{\Delta}_{s_{+},i}^{2}\left( j\right) \\
& =\left( s-T_{1}\right) \bar{\Delta}_{T_{1+},i}^{2}\left( j\right) +\left(
T-s\right) [\bar{\Delta}_{T_{1+},i}^{2}\left( j\right) -\bar{\Delta}%
_{s_{+},i}^{2}\left( j\right) ] \\
& =\left( s-T_{1}\right) \bar{\Delta}_{T_{1+},i}^{2}\left( j\right) +\left(
T-s\right) \left[ \bar{\Delta}_{T_{1+},i}\left( j\right) +\bar{\Delta}%
_{s_{+},i}(j)\right] \left[ \bar{\Delta}_{T_{1+},i}\left( j\right) -\bar{%
\Delta}_{s_{+},i}(j)\right] \\
& =\left( s-T_{1}\right) \bar{\Delta}_{T_{1+},i}^{2}\left( j\right) +\left[ 
\bar{\Delta}_{T_{1+},i}\left( j\right) +\bar{\Delta}_{s_{+},i}(j)\right] %
\left[ \frac{T_{1}-s}{T-T_{1}}\sum_{t=T_{1}+1}^{T}\Delta
_{it}(j)-\sum_{t=s+1}^{T_{1}}\Delta _{it}(j)\right] ,
\end{align*}%
we have%
\begin{align}
A_{3}\left( s\right) & =\frac{1}{pNT}\sum_{j\in \lbrack p]}\sum_{i\in
\lbrack N]}\left\{ \sum_{t=s+1}^{T}\left[ \Delta _{it}(j)-\bar{\Delta}%
_{s_{+},i}(j)\right] ^{2}-\sum_{t=T_{1}+1}^{T}\left[ \Delta _{it}(j)-\bar{%
\Delta}_{T_{1+},i}\left( j\right) \right] ^{2}\right\}  \notag \\
& =\frac{1}{pNT}\sum_{j\in \lbrack p]}\sum_{i\in \lbrack N]}\left\{
\sum_{t=s+1}^{T}\Delta _{it}^{2}\left( j\right) -\left( T-s\right) \bar{%
\Delta}_{s_{+},i}^{2}(j)-\sum_{t=T_{1}+1}^{T}\Delta _{it}^{2}\left( j\right)
+\left( T-T_{1}\right) \bar{\Delta}_{T_{1+},i}\left( j\right) \right\} 
\notag \\
& =\frac{1}{pNT}\sum_{j\in \lbrack p]}\sum_{i\in \lbrack N]}\left\{
\sum_{t=s+1}^{T_{1}}\Delta _{it}^{2}\left( j\right) +\left( T-T_{1}\right) 
\bar{\Delta}_{T_{1+},i}^{2}\left( j\right) -\left( T-s\right) \bar{\Delta}%
_{s_{+},i}(j)\right\}  \notag \\
& =\kappa _{s}\frac{1}{pN\left( T_{1}-s\right) }\sum_{j\in \lbrack
p]}\sum_{i\in \lbrack N]}\sum_{t=s+1}^{T_{1}}\Delta _{it}^{2}\left( j\right)
-\kappa _{s}\frac{1}{pN}\sum_{j\in \lbrack p]}\sum_{i\in \lbrack N]}\bar{%
\Delta}_{T_{1+},i}^{2}\left( j\right)  \notag \\
& +\kappa _{s}\frac{1}{pN}\sum_{j\in \lbrack p]}\sum_{i\in \lbrack N]}\left[ 
\bar{\Delta}_{T_{1+},i}\left( j\right) +\bar{\Delta}_{s_{+},i}(j)\right] 
\frac{1}{T-T_{1}}\sum_{t=T_{1}+1}^{T}\Delta _{it}(j)  \notag \\
& -\kappa _{s}\frac{1}{pN}\sum_{j\in \lbrack p]}\sum_{i\in \lbrack N]}\left[ 
\bar{\Delta}_{T_{1+},i}\left( j\right) +\bar{\Delta}_{s_{+},i}(j)\right] 
\frac{1}{T_{1}-s}\sum_{t=s+1}^{T_{1}}\Delta _{it}(j)  \notag \\
& =\kappa _{s}O_{p}\left( \eta _{N,2}^{2}\right) =\kappa _{s}o_{p}(\zeta
_{NT}^{2}),  \label{B.10}
\end{align}%
where the second last equality holds by (\ref{B.8}) and the fact that 
\begin{align}
& \max_{i\in \lbrack N],j\in \lbrack p]}\left\vert \bar{\Delta}%
_{s_{+},i}(j)\right\vert =\max_{i\in \lbrack N],j\in \lbrack p]}\left\vert 
\frac{1}{T-s}\sum_{t=s+1}^{T}\Delta _{it}(j)\right\vert =O_{p}\left( \eta
_{N,2}\right) ,  \notag \\
& \max_{i\in \lbrack N],j\in \lbrack p]}\left\vert \bar{\Delta}%
_{T_{1+},i}\left( j\right) \right\vert =\max_{i\in \lbrack N],j\in \lbrack
p]}\left\vert \frac{1}{T-T_{1}}\sum_{t=T_{1}+1}^{T}\Delta
_{it}(j)\right\vert =O_{p}\left( \eta _{N,2}\right) .  \label{B.11}
\end{align}%
}

{\small Finally, we observe that 
\begin{align}
& A_{4}\left( s\right) +A_{5}\left( s\right)  \notag \\
& =\frac{2}{pNT}\sum_{j\in \lbrack p]}\sum_{i\in \lbrack N]}(\alpha
_{g_{i}^{(1)},j}^{(1)}-\alpha _{g_{i}^{(2)},j}^{(2)})\left\{ \frac{T-T_{1}}{%
T-s}\sum_{t=s+1}^{T_{1}}\left[ \Delta _{it}(j)-\bar{\Delta}_{s_{+},i}(j)%
\right] -\frac{T_{1}-s}{T-s}\sum_{t=T_{1}+1}^{T}\left[ \Delta _{it}(j)-\bar{%
\Delta}_{s_{+},i}(j)\right] \right\}  \notag \\
& =\frac{2}{pNT}\sum_{j\in \lbrack p]}\sum_{i\in \lbrack N]}(\alpha
_{g_{i}^{(1)},j}^{(1)}-\alpha _{g_{i}^{(2)},j}^{(2)})\left[ \frac{T-T_{1}}{%
T-s}\sum_{t=s+1}^{T_{1}}\Delta _{it}(j)-\frac{T_{1}-s}{T-s}%
\sum_{t=T_{1}+1}^{T}\Delta _{it}(j)\right]  \notag \\
& =2\kappa _{s}\frac{1-\tau _{T}}{1-\frac{s}{T}}\frac{1}{pN}\sum_{j\in
\lbrack p]}\sum_{i\in \lbrack N]}(\alpha _{g_{i}^{(1)},j}^{(1)}-\alpha
_{g_{i}^{(2)},j}^{(2)})\frac{1}{T_{1}-s}\sum_{t=s+1}^{T_{1}}\Delta _{it}(j) 
\notag \\
& -2\kappa _{s}\frac{1-\tau _{T}}{1-\frac{s}{T}}\frac{1}{pN}\sum_{j\in
\lbrack p]}\sum_{i\in \lbrack N]}(\alpha _{g_{i}^{(1)},j}^{(1)}-\alpha
_{g_{i}^{(2)},j}^{(2)})\frac{1}{T-T_{1}}\sum_{t=T_{1}+1}^{T}\Delta _{it}(j) 
\notag \\
& \leq 2\kappa _{s}\frac{1-\tau _{T}}{1-\frac{s}{T}}\sqrt{\frac{1}{N}%
\sum_{i\in \lbrack N]}\left\Vert \alpha _{g_{i}^{(1)}}^{(1)}-\alpha
_{g_{i}^{(2)}}^{(2)}\right\Vert ^{2}}\sqrt{\frac{1}{N}\sum_{i\in \lbrack N]}%
\left[ \frac{1}{p(T_{1}-s)}\sum_{j\in \lbrack p]}\sum_{t=s+1}^{T_{1}}\Delta
_{it}(j)\right] ^{2}}  \notag \\
& +2\kappa _{s}\frac{1-\tau _{T}}{1-\frac{s}{T}}\sqrt{\frac{1}{N}\sum_{i\in
\lbrack N]}\left\Vert \alpha _{g_{i}^{(1)}}^{(1)}-\alpha
_{g_{i}^{(2)}}^{(2)}\right\Vert ^{2}}\sqrt{\frac{1}{N}\sum_{i\in \lbrack N]}%
\left[ \frac{1}{p(T-T_{1})}\sum_{j\in \lbrack p]}\sum_{t=T_{1}+1}^{T}\Delta
_{it}(j)\right] ^{2}}  \notag \\
& =\kappa _{s}\frac{1-\tau _{T}}{1-\frac{s}{T}}\zeta _{NT}O_{p}\left( \eta
_{N,2}\right) =\kappa _{s}o_{p}(\zeta _{NT}^{2}),  \label{B.12}
\end{align}%
where the first inequality holds by Cauchy-Schwarz inequality. }

{\small Combining (\ref{B.5}), (\ref{B.6}), (\ref{B.7}), (\ref{B.10}), (\ref%
{B.12}) and Assumption \ref{ass:6}(i) yields that 
\begin{equation*}
L(s)-L(T_{1})=\kappa_{s}\frac{\left( 1-\tau_{T}\right) }{1-\frac{s}{T}}%
D_{N\alpha }+\kappa_{s}o_{p}(\zeta_{NT}^{2}).
\end{equation*}%
Then for any $s<T_{1},$ 
\begin{equation*}
\text{plim}_{\left( N,T\right) \rightarrow \infty }\frac{1}{%
\kappa_{s}\zeta_{NT}^{2}}\left[ L(s)-L(T_{1})\right] =\text{plim}%
_{\left(N,T\right)\rightarrow \infty }\frac{1-\tau_{T}}{1-\frac{s}{T}}\frac{1%
}{\zeta_{NT}^{2}}D_{N\alpha }\geq \left( 1-\tau \right) D_{\alpha }>0,
\end{equation*}%
where $D_{\alpha }:=$plim$_{\left( N,T\right) \rightarrow \infty }\frac{1}{%
\zeta_{NT}^{2}}D_{N\alpha }>0$ by Assumption \ref{ass:6}(i). This implies
that 
\begin{equation}
\mathbb{P}\left( \hat{T}_{1}<T_{1}\right) \leq \mathbb{P}\left( \exists
s<T_{1},L(s)-L(T_{1})<0\right) \rightarrow 0.  \label{B.13}
\end{equation}
}

{\small By analogous arguments, we prove (ii) in the following part. When $%
s>T_{1}$, we have 
\begin{align*}
\bar{\dot{\Theta}}_{j,i}^{(1s)}& =\frac{1}{s}\sum_{t=1}^{s}\dot{\Theta}%
_{j,it}=\frac{1}{s}\sum_{t=1}^{s}\left[ \Theta _{j,it}^{0}+\left( \dot{\Theta%
}_{j,it}-\Theta _{j,it}^{0}\right) \right] =\frac{T_{1}}{s}\alpha
_{g_{i}^{(1)},j}^{(1)}+\frac{s-T_{1}}{s}\alpha _{g_{i}^{(2)},j}^{(2)}+\bar{%
\Delta}_{s,i}(j), \\
\bar{\dot{\Theta}}_{j,i}^{(2s)}& =\frac{1}{T-s}\sum_{t=s+1}^{T}\dot{\Theta}%
_{j,it}=\frac{1}{T-s}\sum_{t=s+1}^{T}\left[ \Theta _{j,it}^{0}+\left( \dot{%
\Theta}_{j,it}-\Theta _{j,it}^{0}\right) \right] =\alpha
_{g_{i}^{(2)},j}^{(2)}+\bar{\Delta}_{s_{+},i}(j).
\end{align*}%
It follows that 
\begin{align*}
\dot{\Theta}_{j,it}-\bar{\dot{\Theta}}_{j,i}^{(1s)}=& \left\{ 
\begin{array}{ll}
\frac{s-T_{1}}{s}(\alpha _{g_{i}^{(1)},j}^{(1)}-\alpha
_{g_{i}^{(2)},j}^{(2)})+\Delta _{it}(j)-\bar{\Delta}_{s,i}(j)\text{ } & 
\text{if }1\leq t\leq T_{1} \\ 
\frac{T_{1}}{s}(\alpha _{g_{i}^{(2)},j}^{(2)}-\alpha
_{g_{i}^{(1)},j}^{(1)})+\Delta _{it}(j)-\bar{\Delta}_{s,i}(j) & \text{if }%
T_{1}+1\leq t\leq s%
\end{array}%
\right. ,\text{ and} \\
\dot{\Theta}_{j,it}-\bar{\dot{\Theta}}_{j,i}^{(2s)}=& \Delta _{it}(j)-\bar{%
\Delta}_{s_{+},i}(j),\quad s<t\leq T.
\end{align*}%
As in (\ref{B.2}), we obtain that 
\begin{align*}
\sum_{t=1}^{s}\left[ \dot{\Theta}_{j,it}-\bar{\dot{\Theta}}_{j,i}^{(1s)}%
\right] ^{2}& =\sum_{t=1}^{T_{1}}\left[ \frac{s-T_{1}}{s}(\alpha
_{g_{i}^{(1)},j}^{(1)}-\alpha _{g_{i}^{(2)},j}^{(2)})+\Delta _{it}(j)-\bar{%
\Delta}_{s,i}(j)\right] ^{2} \\
& +\sum_{t=T_{1}+1}^{s}\left[ \frac{T_{1}}{s}(\alpha
_{g_{i}^{(2)},j}^{(2)}-\alpha _{g_{i}^{(1)},j}^{(1)})+\Delta _{it}(j)-\bar{%
\Delta}_{s,i}(j)\right] ^{2} \\
& =\frac{T_{1}\left( s-T_{1}\right) ^{2}}{s^{2}}(\alpha
_{g_{i}^{(1)},j}^{(1)}-\alpha _{g_{i}^{(2)},j}^{(2)})^{2}+\sum_{t=1}^{T_{1}}%
\left[ \Delta _{it}(j)-\bar{\Delta}_{s,i}(j)\right] ^{2} \\
& +2\frac{s-T_{1}}{s}(\alpha _{g_{i}^{(1)},j}^{(1)}-\alpha
_{g_{i}^{(2)},j}^{(2)})\sum_{t=1}^{T_{1}}\left[ \Delta _{it}(j)-\bar{\Delta}%
_{s,i}(j)\right]  \\
& +\frac{\left( s-T_{1}\right) T_{1}^{2}}{s^{2}}(\alpha
_{g_{i}^{(2)},j}^{(2)}-\alpha
_{g_{i}^{(1)},j}^{(1)})^{2}+\sum_{t=T_{1}+1}^{s}\left[ \Delta _{it}(j)-\bar{%
\Delta}_{s,i}(j)\right] ^{2} \\
& +2\frac{T_{1}}{s}(\alpha _{g_{i}^{(2)},j}^{(2)}-\alpha
_{g_{i}^{(1)},j}^{(1)})\sum_{t=T_{1}+1}^{s}\left[ \Delta _{it}(j)-\bar{\Delta%
}_{s,i}(j)\right]  \\
& =\frac{\left( s-T_{1}\right) T_{1}}{s}(\alpha
_{g_{i}^{(1)},j}^{(1)}-\alpha _{g_{i}^{(2)},j}^{(2)})^{2}+\sum_{t=1}^{s}%
\left[ \Delta _{it}(j)-\bar{\Delta}_{s,i}(j)\right] ^{2} \\
& +2\frac{T_{1}\left( s-T_{1}\right) }{s}(\alpha
_{g_{i}^{(1)},j}^{(1)}-\alpha _{g_{i}^{(2)},j}^{(2)})\left[ \frac{1}{T_{1}}%
\sum_{t=s+1}^{T_{1}}\Delta _{it}(j)-\frac{1}{s-T_{1}}\sum_{t=T_{1}+1}^{T}%
\Delta _{it}(j)\right] 
\end{align*}%
and $\sum_{t=s+1}^{T}[\dot{\Theta}_{j,it}-\bar{\dot{\Theta}}%
_{j,i}^{(2s)}]^{2}=\sum_{t=s+1}^{T}\left[ \Delta _{it}(j)-\bar{\Delta}%
_{s_{+},i}(j)\right] ^{2}.$ It follows that 
\begin{align*}
L(s)-L(T_{1})& =\frac{T_{1}\left( s-T_{1}\right) }{sT}\frac{1}{pN}\sum_{i\in
\lbrack N]}\left\Vert \alpha _{g_{i}^{(1)}}^{(1)}-\alpha
_{g_{i}^{(2)}}^{(2)}\right\Vert ^{2} \\
& +\frac{1}{pNT}\sum_{j\in \lbrack p]}\sum_{i\in \lbrack N]}\left\{
\sum_{t=1}^{s}\left[ \Delta _{it}(j)-\bar{\Delta}_{s,i}(j)\right]
^{2}-\sum_{t=1}^{T_{1}}\left[ \Delta _{it}(j)-\bar{\Delta}_{T_{1},i}\left(
j\right) \right] ^{2}\right\}  \\
& +\frac{1}{pNT}\sum_{j\in \lbrack p]}\sum_{i\in \lbrack N]}\left\{
\sum_{t=s+1}^{T}\left[ \Delta _{it}(j)-\bar{\Delta}_{s_{+},i}(j)\right]
^{2}-\sum_{t=T_{1}+1}^{T}\left[ \Delta _{it}(j)-\bar{\Delta}%
_{T_{1+},i}\left( j\right) \right] ^{2}\right\}  \\
& +2\frac{T_{1}\left( s-T_{1}\right) }{sT}\frac{1}{pN}\sum_{j\in \lbrack
p]}\sum_{i\in \lbrack N]}(\alpha _{g_{i}^{(1)},j}^{(1)}-\alpha
_{g_{i}^{(2)},j}^{(2)})\left[ \frac{1}{T_{1}}\sum_{t=s+1}^{T_{1}}\Delta
_{it}(j)-\frac{1}{s-T_{1}}\sum_{t=T_{1}+1}^{T}\Delta _{it}(j)\right]  \\
& :=B_{1}\left( s\right) +B_{2}\left( s\right) +B_{3}\left( s\right)
+B_{4}\left( s\right) ,
\end{align*}%
where $B_{4}\left( s\right) $ parallels $A_{4}\left( s\right) +A_{5}\left(
s\right) $ in (\ref{B.5}). Let $\bar{\kappa}_{s}=\frac{s-T_{1}}{T}\in \left[ 
\frac{1}{T},1-\tau _{T}\right] .$ Following the analyses of $A_{\ell }\left(
s\right) $'s, we can readily show that 
\begin{equation*}
B_{1}\left( s\right) =\bar{\kappa}_{s}\frac{T}{s}\frac{T_{1}}{T}D_{N\alpha }=%
\bar{\kappa}_{s}O_{p}(\zeta _{NT}^{2}),\text{ }B_{\ell }\left( s\right) =%
\bar{\kappa}_{s}O_{p}\left( \eta _{N,2}^{2}\right) =\bar{\kappa}%
_{s}o_{p}(\zeta _{NT}^{2})\text{ for }\ell =2,3,\text{ }
\end{equation*}%
and $B_{4}\left( s\right) =\bar{\kappa}_{s}O_{p}\left( \eta _{N,2}\zeta
_{NT}\right) =\bar{\kappa}_{s}o_{p}(\zeta _{NT}^{2}).$ It follows that for
any $s>T_{1},$ 
\begin{equation*}
\text{plim}_{\left( N,T\right) \rightarrow \infty }\frac{1}{\bar{\kappa}%
_{s}\zeta _{NT}^{2}}\left[ L(s)-L(T_{1})\right] =\text{plim}_{\left(
N,T\right) \rightarrow \infty }\frac{T_{1}T}{Ts}\frac{1}{\zeta _{NT}^{2}}%
D_{N\alpha }\geq \tau D_{\alpha }>0\text{.}
\end{equation*}%
This implies that 
\begin{equation}
\mathbb{P(}\hat{T}_{1}>T_{1})\leq \mathbb{P}\left( \exists
s>T_{1},L(s)-L(T_{1})<0\right) \rightarrow 0.  \label{B.14}
\end{equation}%
Combining (\ref{B.13}) and (\ref{B.14}), we conclude that $\mathbb{P(}\hat{T}%
_{1}=T_{1})\rightarrow 1$. $\blacksquare $ }

\subsection{{\protect\small Proof of Theorem \protect\ref{Thm3}}}

{\small 
By Theorem \ref{Thm2}, $\mathbb{P(}\hat{T}_{1}=T_{1})\rightarrow 1$. It
follows that we can prove Theorem \ref{Thm3} by conditioning on the event
that }${\small \{}${\small $\hat{T}_{1}=T_{1}\}$. Below we prove the theorem
under the event that $\{\hat{T}_{1}=T_{1}\}.$ }

{\small Define $\dot{\Theta}_{j,i}^{0,(1)}=(\dot{\Theta}_{j,i1},\cdots ,\dot{%
\Theta}_{j,iT_{1}})^{\prime }$, $\dot{\Theta}_{j,i}^{0,(2)}=(\dot{\Theta}%
_{j,i,T_{1}+1},\cdots ,\dot{\Theta}_{j,iT})^{\prime }$, $\dot{\beta}%
_{i}^{0,(1)}=\frac{1}{\sqrt{T_{1}}}(\dot{\Theta}_{1,i}^{0,(1)\prime },\cdots
,\dot{\Theta}_{p,i}^{0,(1)\prime })^{\prime },$ and $\dot{\beta}_{i}^{0,(2)}=%
\frac{1}{\sqrt{T_{2}}}(\dot{\Theta}_{1,i}^{0,(2)\prime },\cdots ,\dot{\Theta}%
_{p,i}^{0,(2)\prime })^{\prime }$. Noted that in the definitions of $\dot{%
\beta}_{i}^{0,(1)}$ and $\dot{\beta}_{i}^{0,(2)}$ we use the true break date 
$T_{1}$ rather than the estimated one compared to $\dot{\beta}_{i}^{(1)}$
and $\dot{\beta}_{i}^{(2)}$ defined in Step 4. As in \eqref{kmeans_obj} and %
\eqref{group estimates}, we further define 
\begin{align}
& \left\{ \dot{a}_{k,m}^{0,(\ell )}\right\} _{k\in \lbrack m]}=\argmin%
_{\left\{ a_{k}^{(\ell )}\right\} _{k\in \lbrack m]}}\frac{1}{N}\sum_{i\in
\lbrack N]}\min_{k\in \lbrack m]}\left\Vert \dot{\beta}_{i}^{0,(\ell
)}-a_{k}^{(\ell )}\right\Vert ^{2},  \label{group estimates_truebreak} \\
& \hat{g}_{i,m}^{0,(\ell )}=\argmin_{k\in \lbrack m]}\left\Vert \dot{\beta}%
_{i}^{(\ell )}-\dot{a}_{k,m}^{0,(\ell )}\right\Vert ,\quad \forall i\in
\lbrack N].
\end{align}%
}

{\small (i) In the case of $m=K^{(\ell )}$, Theorem \ref{Thm3}(i.a) is from
the combination of Lemma \ref{Lem:Kmeans_null} for the consistency of the
membership estimates via K-means algorithm and Theorem \ref{Thm2} for the
consistency of the break point estimator. }

{\small Next, we show (i.c). Recall that $z_{\varsigma }$ is the critical
value at significance level $\varsigma $ calculated from the maximum of $m$
independent $\chi ^{2}(1)$ random variables. By the definition of the STK
algorithm, we observe that 
\begin{equation*}
\mathbb{P(}\hat{K}^{(\ell )}\leq K^{(\ell )})\geq \mathbb{P}(\hat{\Gamma}%
_{K^{(\ell )}}^{(\ell )}\leq z_{\varsigma }),
\end{equation*}%
which leads to the fact that (i.c) holds as long as we can show (i.b). This
is because, under (i.b), we have 
\begin{equation*}
\mathbb{P}(\hat{\Gamma}_{K^{(\ell )}}^{(\ell )}\leq z_{\varsigma })\geq
1-\varsigma +o(1).
\end{equation*}%
}

{\small Now, we focus on (i.b). Notice that $\hat{\Gamma}_{k,K^{(\ell
)}}^{(\ell )}$ depends on the K-means classification result, i.e., the
estimated group membership $\hat{G}_{k,K^{(\ell )}}^{(\ell )}$ for $k\in
\lbrack K^{(\ell )}]$. From Theorem \ref{Thm3}(i.1), we notice that we can
change the estimated group membership $\hat{G}_{k,K^{(\ell )}}^{(\ell )}$ to
the true group membership $G_{k}^{(\ell )}$, and the replacement has only an 
asymptotically negligible effect. Recall that $\mathcal{T}_{1}=[T_{1}]$ and $%
\mathcal{T}_{2}=[T]\backslash \lbrack T_{1}]$. Define $\mathcal{T}_{1,-1}=%
\mathcal{T}_{1}\backslash \{T_{1}\}$, $\mathcal{T}_{2,-1}=\mathcal{T}%
_{2}\backslash \{T\}$, $\mathcal{T}_{1,j}=\left\{ 1+j,\cdots ,T_{1}\right\} $
and $\mathcal{T}_{2,j}=\left\{ T_{1}+1+j,\cdots ,T\right\} $ for some
specific $j\in \mathcal{T}_{\ell ,-1}$. Let 
\begin{equation*}
\left( \left\{ \hat{\theta}_{i,k,K^{(\ell )}}^{0,(\ell )}\right\} _{i\in
G_{k}^{(\ell )}},\hat{\Lambda}_{k,K^{(\ell )}}^{0,(\ell )},\left\{ \hat{f}%
_{t,k,K^{(\ell )}}^{0,(\ell )}\right\} _{t\in \mathcal{T}_{\ell }}\right) =%
\argmin_{\left\{ \theta _{i},\lambda _{i}\right\} _{i\in G_{k}^{(\ell
)}},\left\{ f_{t}\right\} _{t\in \mathcal{T}_{\ell }}}\sum_{i\in
G_{k}^{(\ell )}}\sum_{t\in \mathcal{T}_{\ell }}\left( Y_{it}-X_{it}^{\prime
}\theta _{i}-\lambda _{i}^{\prime }f_{t}\right) ^{2},
\end{equation*}%
$\hat{F}_{k,K^{(1)}}^{0,(1)}=(\hat{f}_{1,k,K^{(1)}},\cdots ,\hat{f}%
_{T_{1},k,K^{(1)}})^{\prime },$ $\hat{F}_{k,K^{(2)}}^{0,(2)}=(\hat{f}%
_{T_{1}+1,k,K^{(1)}},\cdots ,\hat{f}_{T,k,K^{(2)}})^{\prime }$, $\hat{\Lambda%
}_{k,K^{(\ell )}}^{(\ell )}=\{\hat{\lambda}_{i,k,K^{(\ell )}}^{(\ell
)}\}_{i\in G_{k}^{(\ell )}}$, and $(\hat{z}_{it}^{0,(\ell )})^{\prime }$
denote the $t$-th row of $M_{\hat{F}_{k,K^{(\ell )}}^{0,(\ell
)}}X_{i}^{(\ell )}.$ Further define 
\begin{align*}
& \hat{\bar{\theta}}_{k,K^{(\ell )}}^{0,(\ell )}=\frac{1}{|G_{k}^{(\ell )}|}%
\sum_{i\in G_{k}^{(\ell )}}\hat{\theta}_{i,k,K^{(\ell )}}^{0,(\ell )},\quad
\quad \hat{S}_{ii,k,K^{(\ell )}}^{0,(\ell )}=\frac{1}{T_{\ell }}%
(X_{i}^{(\ell )})^{\prime }M_{\hat{F}_{k,K^{(\ell )}}^{0,(\ell
)}}X_{i}^{(\ell )}, \\
& \hat{\Omega}_{i,k,K^{(\ell )}}^{0,(\ell )}=\frac{1}{T_{\ell }}\sum_{t\in 
\mathcal{T}_{\ell }}\hat{z}_{it}^{0,(\ell )}\hat{z}_{it}^{0,(\ell )\prime }%
\hat{e}_{it}^{2}+\frac{1}{T_{\ell }}\sum_{j\in \mathcal{T}_{\ell
,-1}}k(j,L)\sum_{t\in \mathcal{T}_{\ell ,j}}[\hat{z}_{it}^{0,(\ell )}\hat{z}%
_{i,t+j}^{0,(\ell )\prime }\hat{e}_{it}\hat{e}_{i,t+j}+\hat{z}%
_{i,t-j}^{0,(\ell )}\hat{z}_{it}^{0,(\ell )\prime }\hat{e}_{i,t-j}\hat{e}%
_{it}], \\
& \hat{a}_{ii,k,K^{(\ell )}}^{0,(\ell )}=\hat{\lambda}_{i,k,K^{(\ell
)}}^{(\ell )\prime }\left( \frac{1}{|G_{k}^{(\ell )}|}\hat{\Lambda}%
_{k,K^{(\ell )}}^{(\ell )\prime }\hat{\Lambda}_{k,K^{(\ell )}}^{(\ell
)}\right) ^{-1}\hat{\lambda}_{i,k,K^{(\ell )}}^{(\ell )}.
\end{align*}%
Then $\forall k\in \lbrack K^{(\ell )}]$, we can define 
\begin{equation*}
\hat{\Gamma}_{k,K^{(\ell )}}^{0,(\ell )}=\sqrt{\left\vert G_{k}^{(\ell
)}\right\vert }\frac{\frac{1}{\left\vert G_{k}^{(\ell )}\right\vert }%
\sum_{i\in G_{k}^{(\ell )}}\hat{\mathbb{S}}_{i,k,K^{(\ell )}}^{0,(\ell )}-p}{%
\sqrt{2p}},
\end{equation*}%
where 
\begin{equation*}
\hat{\mathbb{S}}_{i,k,K^{(\ell )}}^{0,(\ell )}=T_{\ell }(\hat{\theta}%
_{i,k,K^{(\ell )}}^{0,(\ell )}-\hat{\bar{\theta}}_{k,K^{(\ell )}}^{0,(\ell
)})^{\prime }\hat{S}_{ii,k,K^{(\ell )}}^{0,(\ell )}(\hat{\Omega}%
_{i,k,K^{(\ell )}}^{0,(\ell )})^{-1}\hat{S}_{ii,k,K^{(\ell )}}^{0,(\ell )}(%
\hat{\theta}_{i,k,K^{(\ell )}}^{0,(\ell )}-\hat{\bar{\theta}}_{k,K^{(\ell
)}}^{0,(\ell )})\left( 1-\hat{a}_{ii,k,K^{(\ell )}}^{0,(\ell
)}/|G_{k}^{(\ell )}|\right) ^{2}.
\end{equation*}%
By Lemma \ref{Lem:size}, we notice that $\hat{\Gamma}_{k,K^{(\ell
)}}^{0,(\ell )}\rightsquigarrow \mathbb{N}(0,1)$ owing to the fact that the
slope coefficient $\alpha _{k}^{(\ell )}$ is homogeneous across $i\in
G_{k}^{(\ell )}$ $\forall k\in \lbrack K^{(\ell )}]$. Furthermore, $\{\hat{%
\Gamma}_{k,K^{(\ell )}}^{0,(\ell )},$ $k\in \lbrack K^{(\ell )}]\}$ are
asymptotically independent under Assumption \ref{ass:1}(i). It follows that 
\begin{equation*}
\hat{\Gamma}_{K^{(\ell )}}^{(\ell )}=\max_{k\in \lbrack K^{(\ell )}]}\left( 
\hat{\Gamma}_{k,K^{(\ell )}}^{(\ell )}\right) ^{2}=\max_{k\in \lbrack
m]}\left( \hat{\Gamma}_{k,K^{(\ell )}}^{0,(\ell )}\right)
^{2}+o_{p}(1)\rightsquigarrow \mathcal{Z},
\end{equation*}%
where $\mathcal{Z}$ is the maximum of $m$ independent $\chi ^{2}\left(
1\right) $ random variables. Then Theorem \ref{Thm3}(i.b) follows. }

{\small (ii) When $m<K^{(\ell )},$ Theorem \ref{Thm3}(i.1) does not hold and
we can not change the estimated group membership $\hat{\mathcal{G}}%
_{K^{(\ell )}}^{(\ell )}$ to the true group membership $\mathcal{G}^{(\ell
)} $. To get around of this issue, we define the \textquotedblleft pseudo
groups\textquotedblright . For $m<K^{(\ell )}$, let $\mathbb{G}_{m}^{(\ell
)}:=\{G_{1,m}^{(\ell )},\cdots ,G_{m,m}^{(\ell )}\}$ such that $%
[N]=G_{1,m}^{(\ell )}\cup \cdots \cup G_{m,m}^{(\ell )}$, which indicates
one possible partition of the set }${\small [N]}${\small . We further define 
$\mathcal{G}_{m}^{(\ell )}$ to be the collection of all possible $\mathbb{G}%
_{m}^{(\ell )}$. }

{\small By Theorem \ref{Thm3}(i.c), we can conclude that $\mathbb{P}\left( 
\hat{K}^{(\ell )}\neq K^{(\ell )}\right) \leq \varsigma +o(1)$ provided we
can show that $\hat{\Gamma}_{m}^{(\ell )}\rightarrow \infty $ when $%
m<K^{(\ell )}$. By Lemma \ref{Lem:NSP}, we notice that $\hat{\mathcal{G}}%
_{m}^{(\ell )}\in \mathbb{G}_{m}^{(\ell )}$ w.p.a.1. Conditioning on the
event $\{\hat{\mathcal{G}}_{m}^{(\ell )}\in \mathbb{G}_{m}^{(\ell )}\}\cap \{%
\hat{T}_{1}=T_{1}\}$, we have 
\begin{equation*}
\hat{\Gamma}_{m}^{(\ell )}>\min_{\mathcal{G}_{m}^{(\ell )}\in \mathbb{G}%
_{m}^{(\ell )}}\hat{\Gamma}_{m}^{0,(\ell )}(\mathcal{G}_{m}^{(\ell
)}):=\min_{\mathcal{G}_{m}^{(\ell )}\in \mathbb{G}_{m}^{(\ell )}}\left\{
\max_{k\in \lbrack m]}\left[ \hat{\Gamma}_{k,m}^{0,(\ell )}(G_{k,m}^{(\ell
)})\right] ^{2}\right\} ,
\end{equation*}%
where 
\begin{equation*}
\hat{\Gamma}_{k,m}^{0,(\ell )}(G_{k,m}^{(\ell )})=\sqrt{\left\vert
G_{k,m}^{(\ell )}\right\vert }\frac{\frac{1}{\left\vert G_{k,m}^{(\ell
)}\right\vert }\sum_{i\in G_{k,m}^{(\ell )}}\hat{\mathbb{S}}%
_{i,k,m}^{0,(\ell )}-p}{\sqrt{2p}},
\end{equation*}%
and $\hat{\mathbb{S}}_{i,k,m}^{0,(\ell )}$ is defined similarly to $\hat{%
\mathbb{S}}_{i,k,K^{(\ell )}}^{0,(\ell )}$ in the proof of (i). }

{\small Owing to the fact that $|\mathbb{G}_{m}^{(\ell )}|=m^{K^{(\ell )}}$
which is a constant since $K^{(\ell )}$ is a constant, we can show that $%
\hat{\Gamma}_{m}^{(\ell )}\rightarrow \infty $ by showing that $\hat{\Gamma}%
_{m}^{0,(\ell )}(\mathcal{G}_{m}^{(\ell )})\rightarrow \infty $ for any
possible realization $\mathcal{G}_{m}^{(\ell )}$. Under the case when $%
m<K^{(\ell )}$, there exists at least one $k\in \lbrack m]$ such that the
slope coefficient is not homogeneous across $i\in G_{k,m}^{(\ell )}$. Assume
that $G_{k,m}^{(\ell )}$ contains $n$ true groups, i.e., $G_{k,m}^{(\ell
)}=G_{k_{1}}^{(\ell )}\cup \cdots \cup G_{k_{n}}^{(\ell )}$ for $%
k_{1},\cdots ,k_{n}\in \lbrack K^{(\ell )}]$ and $k_{1}\neq \cdots \neq k_{n}
$. Then for $i\in G_{k,m}^{(\ell )}$, we have 
\begin{align*}
\theta _{i}^{0,(\ell )}& =\sum_{s=1}^{n}\alpha _{k_{s}}^{(\ell )}\mathbf{1}%
\{i\in G_{k_{s}}^{(\ell )}\}=\frac{1}{n}\sum_{s^{\ast }=1}^{n}\alpha
_{k_{s^{\ast }}}^{(\ell )}+\sum_{s=1}^{n}\left( \frac{n-1}{n}\alpha
_{k_{s}}^{(\ell )}-\frac{1}{n}\sum_{s^{\ast }\in \left[ n\right] ,s^{\ast
}\neq s}\alpha _{k_{s^{\ast }}}^{(\ell )}\right) \mathbf{1}\{i\in
G_{k_{s}}^{(\ell )}\} \\
& =\frac{1}{n}\sum_{s^{\ast }=1}^{n}\alpha _{k_{s^{\ast }}}^{(\ell
)}+\sum_{s=1}^{n}\frac{1}{n}\sum_{s^{\ast }\in \left[ n\right] ,s^{\ast
}\neq s}(\alpha _{k_{s}}^{(\ell )}-\alpha _{k_{s^{\ast }}}^{(\ell )})\mathbf{%
1}\{i\in G_{k_{s}}^{(\ell )}\}:=\bar{\theta}_{n}^{0,(\ell )}+c_{i}^{(\ell )}
\end{align*}%
such that 
\begin{eqnarray*}
\frac{T_{\ell }}{\sqrt{N}}\sum_{i\in \lbrack N]}\left\Vert c_{i}^{(\ell
)}\right\Vert ^{2}/(\log N)^{1/2} &=&\frac{T_{\ell }}{\sqrt{N}}\sum_{s=1}^{n}%
\frac{N_{k_{s}}^{(\ell )}}{n}\left\Vert \sum_{s^{\ast }\in \left[ n\right]
,s^{\ast }\neq s}(\alpha _{k_{s}}^{(\ell )}-\alpha _{k_{s^{\ast }}}^{(\ell
)})\right\Vert ^{2}/(\log N)^{1/2} \\
&=&\frac{T_{\ell }}{\sqrt{N}n}\sum_{s=1}^{n}N_{k_{s}}^{(\ell )}\left\Vert 
\frac{\sum_{s^{\ast }\in \left[ n\right] ,s^{\ast }\neq s}\alpha
_{k_{s^{\ast }}}^{(\ell )}}{n-1}-\alpha _{k_{s}}^{(\ell )}\right\Vert
^{2}/(\log N)^{1/2}\rightarrow \infty, 
\end{eqnarray*}%
by Assumption \ref{ass:7}(iii). Then }${\small |}${\small $\hat{\Gamma}%
_{k,m}^{0,(\ell )}(G_{k,m}^{(\ell )})|/(\log N)^{1/2}\rightarrow \infty $
for some $k\in \lbrack m]$ by Lemma \ref{Lem:power}. By the definition of $%
\hat{\Gamma}_{m}^{0,(\ell )}(\mathcal{G}_{m}^{(\ell )})$, we have $\hat{%
\Gamma}_{m}^{0,(\ell )}(\mathcal{G}_{m}^{(\ell )})/(\log N)^{1/2}\rightarrow
\infty $, which yields $\hat{\Gamma}_{m}^{(\ell )}/\log N\rightarrow \infty $
w.p.a.1 for $m<K^{(\ell )}$ and $\mathbb{P(}\hat{K}^{(\ell )}\neq K^{(\ell
)})\leq \varsigma +o(1)$ as $z_{\varsigma }$ diverges to infinity at rate $%
\log N$ as }${\small \varsigma =\varsigma }_{N}{\small \rightarrow 0}$ 
{\small at some rate }$N^{-c}${\small \ for some }${\small c>0}.${\small $%
\quad \blacksquare $ }

\subsection{{\protect\small Proof of Theorem \protect\ref{Thm4}}}

{\small To show Theorem \ref{Thm4}, we can directly derive the asymptotic
distribution for the oracle estimator $\hat{\alpha}_{k}^{*(\ell)}$ by
combining Theorems \ref{Thm2} and \ref{Thm3}. }

{\small The asymptotic distribution theory for the linear panel model with
IFEs has already been studied in the literature; see \cite{bai2009panel}, 
\cite{moon2017dynamic} and \cite{lu2016shrinkage} for instance. However, 
\cite{bai2009panel} rules out dynamic panels. \cite{moon2017dynamic} allow
dynamic panels and assume the independence over both $i$ and $t$ for the
error term. For the dynamic linear panel model under Assumptions \ref{ass:1*} and \ref{ass:2}--\ref{ass:10}, 
Theorem \ref{Thm4} extends Theorem 4.3 in \cite{moon2017dynamic} to allow for multiple groups. }

{\small Below, we follow the arguments in \cite{moon2017dynamic} and sketch
the proof to allow the serial correlation of error terms in non-dynamic
panels.\footnote{%
It is well known that one cannot allow for serially correlated errors in
dynamic panels in general in order to avoid issues with endogeneity.} To proceed, let $%
\mathbb{C}_{NT,k}^{(\ell )}$ be the $p$-vector with $j$-th entry being $%
\mathbb{C}_{NT,k,j}^{(\ell )}=\mathbb{C}_{1}(\Lambda _{k}^{0,(\ell
)},F^{0,(\ell )},\mathbb{X}_{j,k}^{(\ell )},E_{k}^{(\ell )})+\mathbb{C}%
_{2}(\Lambda _{k}^{0,(\ell )},F^{0,(\ell )},\mathbb{X}_{j,k}^{(\ell
)},E_{k}^{(\ell )})$, where 
\begin{align*}
& \mathbb{C}_{1}(\Lambda _{k}^{0,(\ell )},F^{0,(\ell )},\mathbb{X}%
_{j,k}^{(\ell )},E_{k}^{(\ell )})=\frac{1}{\sqrt{N_{k}^{(\ell )}T_{\ell }}}%
tr\left( M_{F^{0,(\ell )}}E_{k}^{(\ell )\prime }M_{\Lambda _{k}^{0,(\ell )}}%
\mathbb{X}_{j,k}^{(\ell )}\right) ,\text{ and} \\
& \mathbb{C}_{2}(\Lambda _{k}^{0,(\ell )},F^{0,(\ell )},\mathbb{X}%
_{j,k}^{(\ell )},E_{k}^{(\ell )}) \\
& =-\frac{1}{\sqrt{N_{k}^{(\ell )}T_{\ell }}}tr\left( E_{k}^{(\ell
)}M_{F^{0,(\ell )}}E_{k}^{(\ell )\prime }M_{\Lambda _{k}^{0,(\ell )}}\mathbb{%
X}_{j,k}^{(\ell )}F^{0,(\ell )}\left( F^{0,(\ell )\prime }F^{0,(\ell
)}\right) ^{-1}\left( \Lambda _{k}^{0,(\ell )\prime }\Lambda _{k}^{0,(\ell
)}\right) ^{-1}\Lambda _{k}^{0,(\ell )\prime }\right) \\
& -\frac{1}{\sqrt{N_{k}^{(\ell )}T_{\ell }}}tr\left( E_{k}^{(\ell )\prime
}M_{\Lambda _{k}^{0,(\ell )}}E_{k}^{(\ell )}M_{F^{0,(\ell )}}\mathbb{X}%
_{j,k}^{(\ell )\prime }\Lambda _{k}^{0,(\ell )}\left( \Lambda _{k}^{0,(\ell
)\prime }\Lambda _{k}^{0,(\ell )}\right) ^{-1}\left( F^{0,(\ell )\prime
}F^{0,(\ell )}\right) ^{-1}\left( \Lambda _{k}^{0,(\ell )\prime }\Lambda
_{k}^{0,(\ell )}\right) ^{-1}F^{0,(\ell )\prime }\right) \\
& -\frac{1}{\sqrt{N_{k}^{(\ell )}T_{\ell }}}tr\left( E_{k}^{(\ell )\prime
}M_{\Lambda _{k}^{0,(\ell )}}\mathbb{X}_{j,k}^{(\ell )}M_{F^{0,(\ell
)}}E_{k}^{(\ell )\prime }\Lambda _{k}^{0,(\ell )}\left( \Lambda
_{k}^{0,(\ell )\prime }\Lambda _{k}^{0,(\ell )}\right) ^{-1}\left(
F^{0,(\ell )\prime }F^{0,(\ell )}\right) ^{-1}\left( \Lambda _{k}^{0,(\ell
)\prime }\Lambda _{k}^{0,(\ell )}\right) ^{-1}F^{0,(\ell )\prime }\right) .
\end{align*}%
By Lemma \ref{Lem:consis_homo}, we have $\sqrt{N_{k}^{(\ell )}T_{\ell }}(%
\hat{\alpha}_{k}^{(\ell )}-\alpha _{k}^{(\ell )})=\mathbb{W}_{NT,k}^{(\ell
)-1}\mathbb{C}_{NT,k}^{(\ell )}+o_{p}(1)$. }

{\small In \cite{moon2017dynamic}, the asymptotic distribution is derived
mainly relying on their Lemmas B.1 and B.2. Lemma B.2 is the standard
central limit theorem, which also holds under our Assumption \ref{ass:1}.
For Lemma B.1, we need to extend it to allow for serially correlated errors
in non-dynamic panels in Lemma \ref{Lem:distri_homo}. Hence, by Lemma \ref%
{Lem:distri_homo} and following arguments analogous to those in the proof of
Theorem 4.3 (\cite{moon2017dynamic}), for a specific $\ell \in \{1,2\}$ and $%
k\in \lbrack K^{(\ell )}]$, we can readily show that 
\begin{equation*}
\mathbb{W}_{NT,k}^{(\ell )}\sqrt{N_{k}^{(\ell )}T_{\ell }}(\hat{\alpha}%
_{k}^{(\ell )}-\alpha _{k}^{(\ell )})-\mathbb{B}_{NT,k}^{(\ell
)}\rightsquigarrow \mathcal{N(}0,\Omega _{k}^{(\ell )}),
\end{equation*}%
which yields the final distributional results in Theorem \ref{Thm4} by
stacking all subgroups of parameter estimators into a large vector and
resorting to the Cram\'er-Wold device. }

\subsection{{\protect\small Proof of Theorem \protect\ref{Thm5}}}

{\small Recall that $\dot{v}_{t,j}^{\ast }:=\frac{\dot{v}_{t,j}}{\left\Vert 
\dot{v}_{t,j}\right\Vert }$, $\dot{v}_{t}^{\ast }=\left( \dot{v}_{t,1}^{\ast
\prime },\cdots ,\dot{v}_{t,p}^{\ast \prime }\right) ^{\prime }$, $%
v_{t,j}^{\ast }=\frac{O_{j}v_{t,j}^{0}}{\left\Vert
O_{j}v_{t,j}^{0}\right\Vert }$ and $v_{t}^{\ast }=\left( v_{t,1}^{\ast
\prime },\cdots ,v_{t,p}^{\ast \prime }\right) ^{\prime }$. With the fact
that 
\begin{align*}
\frac{\dot{v}_{t,j}}{\left\Vert \dot{v}_{t,j}\right\Vert }-\frac{%
O_{j}v_{t,j}^{0}}{\left\Vert O_{j}v_{t,j}^{0}\right\Vert }& =\frac{\dot{v}%
_{t,j}\left\Vert O_{j}v_{t,j}^{0}\right\Vert -O_{j}v_{t,j}^{0}\left\Vert 
\dot{v}_{t,j}\right\Vert }{\left\Vert \dot{v}_{t,j}\right\Vert \left\Vert
O_{j}v_{t,j}^{0}\right\Vert } \\
& =\frac{\left( \dot{v}_{t,j}-O_{j}v_{t,j}^{0}\right) \left\Vert
O_{j}v_{t,j}^{0}\right\Vert +O_{j}v_{t,j}^{0}\left( \left\Vert
O_{j}v_{t,j}^{0}\right\Vert -\left\Vert \dot{v}_{t,j}\right\Vert \right) }{%
\left\Vert \dot{v}_{t,j}\right\Vert \left\Vert O_{j}v_{t,j}^{0}\right\Vert },
\end{align*}%
it follows that 
\begin{equation*}
\max_{t}\left\Vert \dot{v}_{t}^{\ast }-v_{t}^{\ast }\right\Vert \leq
p\max_{j\in \lbrack p],t\in \lbrack T]}\left\Vert \dot{v}_{t,j}^{\ast
}-v_{t,j}^{\ast }\right\Vert \leq 2p\max_{j\in \lbrack p],t\in \lbrack T]}%
\frac{\left\Vert \dot{v}_{t,j}-O_{j}v_{t,j}^{0}\right\Vert }{\left\Vert \dot{%
v}_{t,j}\right\Vert }=O_{p}(\eta _{N,2}),
\end{equation*}%
where the last line is by Lemma \ref{Lem:bounded u&v}(i) and Theorem \ref%
{Thm1}(ii). }

\section{\protect\small Technical Lemmas}

{\small \label{sec:lem} }

\begin{lemma}
{\small \label{Lem:matrix Bern} Consider a matrix sequence $\left\{
A_{i},i\in \lbrack N]\right\} $ whose elements are symmetric matrices with
dimension $d$. Suppose $\left\{ A_{i},i\in \lbrack N]\right\} $ is
independent with $\mathbb{E}\left( A_{i}\right) =0$ and $\left\Vert
A_{i}\right\Vert _{op}\leq M$ a.s.. Let $\sigma ^{2}=\left\Vert \sum_{i\in
\lbrack N]}\mathbb{E}\left( A_{i}^{2}\right) \right\Vert _{op}.$ Then for
all $t>0$, we have 
\begin{equation*}
\mathbb{P}\left( \left\Vert \sum_{i\in \lbrack N]}A_{i}\right\Vert
_{op}>t\right) \leq d\cdot \exp \left\{ -\frac{t^{2}/2}{\sigma ^{2}+Mt/3}%
\right\} .
\end{equation*}%
}
\end{lemma}

\begin{proof}
{\small Lemma \ref{Lem:matrix Bern} states a matrix Bernstein inequality;
see Theorem 1.3 in \cite{tropp2011user}. }
\end{proof}

\begin{lemma}
{\small \label{Lem:matrix op} Consider a specific matrix $A\in \mathbb{R}%
^{N\times T}$ whose rows (denoted as $A_{i}^{\prime }$) are independent
random vectors in $\mathbb{R}^{T}$ with $\mathbb{E}A_{i}=0$ and $\Sigma _{i}=%
\mathbb{E}\left( A_{i}A_{i}^{\prime }\right) $. Suppose $\max_{i}\left\Vert
A_{i}\right\Vert \leq \sqrt{m}$ almost surely and $\max_{i}\left\Vert \Sigma
_{i}\right\Vert _{op}\leq M$ for some positive constant $M$. Then for every $%
t>0$, with probability $1-2T\exp \left( -c_{1}t^{2}\right) $, we have 
\begin{equation*}
\left\Vert A\right\Vert _{op}\leq \sqrt{NM}+t\sqrt{m+M},
\end{equation*}%
where $c_{1}$ is an absolute constant. }
\end{lemma}

\begin{proof}
{\small The proof follows arguments like those used in the proof of Theorem 5.41
in \cite{vershynin2010introduction}. Define $Z_{i}:=\frac{1}{N}\left(
A_{i}A_{i}^{\prime }-\Sigma _{i}\right) \in \mathbb{R}^{T\times T}$, and we
notice that $\left( Z_{1},\cdots ,Z_{N}\right) $ is an independent sequence
with $\mathbb{E}\left( Z_{i}\right) =0$. To use the matrix Bernstein's
inequality, we analyze $\left\Vert Z_{i}\right\Vert _{op}$ and $\left\Vert
\sum_{i\in \lbrack N]}\mathbb{E}\left( X_{i}^{2}\right) \right\Vert _{op}$
as follows: 
\begin{equation}
\left\Vert Z_{i}\right\Vert _{op}\leq \frac{1}{N}\left( \left\Vert
A_{i}A_{i}^{\prime }\right\Vert _{op}+\left\Vert \Sigma _{i}\right\Vert
_{op}\right) \leq \frac{1}{N}\left( \left\Vert A_{i}\right\Vert
_{2}^{2}+\left\Vert \Sigma _{i}\right\Vert _{op}\right) \leq \frac{m+M}{N}%
\quad a.s.  \label{Lem2.1}
\end{equation}%
uniformly over $i$. Moreover, note that 
\begin{equation*}
\mathbb{E}\left[ \left( A_{i}A_{i}^{\prime }\right) ^{2}\right] =\mathbb{E}%
\left[ \left\Vert A_{i}\right\Vert _{2}A_{i}A_{i}^{\prime }\right] \leq
m\Sigma _{i}\text{ }
\end{equation*}%
and 
\begin{equation*}
Z_{i}^{2}=\frac{1}{N^{2}}\left[ \left( A_{i}A_{i}^{\prime }\right)
^{2}-A_{i}A_{i}^{\prime }\Sigma _{i}-\Sigma _{i}A_{i}A_{i}^{\prime }+\Sigma
_{i}^{2}\right] .
\end{equation*}%
We then obtain that 
\begin{align*}
\left\Vert \mathbb{E}\left( Z_{i}^{2}\right) \right\Vert _{op}& =\left\Vert 
\mathbb{E}\left\{ \frac{1}{N^{2}}\left[ \left( A_{i}A_{i}^{\prime }\right)
^{2}-\Sigma _{i}^{2}\right] \right\} \right\Vert _{op}\leq \frac{1}{N^{2}}%
\left\{ \left\Vert \mathbb{E}\left[ \left( A_{i}A_{i}^{\prime }\right) ^{2}%
\right] \right\Vert _{op}+\left\Vert \Sigma _{i}\right\Vert _{op}^{2}\right\}
\\
& \leq \frac{1}{N^{2}}\left( m\left\Vert \Sigma _{i}\right\Vert
_{op}+\left\Vert \Sigma _{i}\right\Vert _{op}^{2}\right) \,\leq \frac{%
mM+M^{2}}{N^{2}}\text{ }\quad a.s.
\end{align*}%
uniformly over $i$, and 
\begin{equation}
\left\Vert \sum_{i\in \lbrack N]}\mathbb{E}\left( Z_{i}^{2}\right)
\right\Vert _{op}\leq N\max_{i}\left\Vert \mathbb{E}\left( Z_{i}^{2}\right)
\right\Vert _{op}\leq \frac{mM+M^{2}}{N}\quad a.s.  \label{Lem2.2}
\end{equation}%
Define $\varepsilon =\max (\sqrt{M}\delta ,\delta ^{2})$ with $\delta =t%
\sqrt{\frac{m+M}{N}}$. Combining (\ref{Lem2.1}) and (\ref{Lem2.2}), by
matrix Bernstein's inequality, we have 
\begin{align*}
& \mathbb{P}\left\{ \left\Vert \frac{1}{N}\left( A^{\prime }A-\sum_{i\in
\lbrack N]}\Sigma _{i}\right) \right\Vert _{op}\geq \varepsilon \right\} =%
\mathbb{P}\left( \left\Vert \sum_{i\in \lbrack N]}Z_{i}\right\Vert _{op}\geq
\varepsilon \right) \\
& \leq 2T\exp \left\{ -c\min \left( \frac{\varepsilon ^{2}}{\frac{mM+M^{2}}{N%
}},\frac{\varepsilon }{\frac{m+M}{N}}\right) \right\} \leq 2T\exp \left\{
-c\min \left( \frac{\varepsilon ^{2}}{M},\varepsilon \right) \frac{N}{m+M}%
\right\} \\
& \leq 2T\exp \left\{ -\frac{c\delta ^{2}N}{m+M}\right\} =2T\exp \left\{
-c_{1}t^{2}\right\} ,
\end{align*}%
for some positive constant $c$, where the third inequality is due to the
fact that 
\begin{align*}
\min \left( \frac{\varepsilon ^{2}}{M},\varepsilon \right) & =\min \left(
\max \left( \delta ^{2},\delta ^{4}/M\right) ,\max \left( \sqrt{M}\delta
,\delta \right) \right) \\
& =\left\{ \begin{aligned}
&\min\left(\delta^{2},\sqrt{M}\delta\right)=\delta^{2},\quad
\text{if}\quad\delta^{2}\geq \frac{\delta^{4}}{M},\\
&\min\left(\delta^{4}/M,\delta^{2}\right)=\delta^{2},\quad\text{if}\quad%
\delta^{2}< \frac{\delta^{4}}{M}. \end{aligned}\right.
\end{align*}%
It implies that 
\begin{equation}
\left\Vert \frac{1}{N}A^{\prime }A-\frac{1}{N}\sum_{i\in \lbrack N]}\Sigma
_{i}\right\Vert _{op}\leq \max \left( \sqrt{M}\delta ,\delta ^{2}\right)
\label{Lem2.3}
\end{equation}%
with probability $1-\exp \left( -ct^{2}\right) $. Combining the fact that $%
\left\Vert \Sigma _{i}\right\Vert _{op}\leq M$ uniformly over $i$ and (\ref%
{Lem2.3}), we show that 
\begin{align*}
\frac{1}{N}\left\Vert A\right\Vert _{op}^{2}& =\left\Vert \frac{1}{N}%
A^{\prime }A\right\Vert _{op}\leq \left\Vert \frac{1}{N}\sum_{i\in \lbrack
N]}\Sigma _{i}\right\Vert _{op}+\left\Vert \frac{1}{N}A^{\prime }A-\frac{1}{N%
}\sum_{i\in \lbrack N]}\Sigma _{i}\right\Vert _{op} \\
& \leq \max_{i}\left\Vert \Sigma _{i}\right\Vert _{op}+\sqrt{M}\delta
+\delta ^{2}\leq M+\sqrt{M}t\sqrt{\frac{m+M}{N}}+t^{2}\frac{m+M}{N} \\
& \leq \left( \sqrt{M}+t\sqrt{\frac{m+M}{N}}\right) ^{2},
\end{align*}%
and the result follows: $\left\Vert A\right\Vert _{op}\leq \sqrt{NM}+t\sqrt{%
m+M}.$ }
\end{proof}

\begin{lemma}
{\small \label{Lem:score op} Recall that $X_{j}=\left\{ X_{j,it}\right\} $
and $E=\left\{ e_{it}\right\} $. Under Assumption \ref{ass:1}, $\forall j\in
\lbrack p]$, we have $\left\Vert X_{j}\odot E\right\Vert _{op}=O_{p}(\sqrt{N}%
+\sqrt{T\log T})$ and $\left\Vert E\right\Vert _{op}=O_{p}(\sqrt{N}+\sqrt{%
T\log T})$. }
\end{lemma}

\begin{proof}
{\small We focus on $\left\Vert X_{j}\odot E\right\Vert _{op}$ as the result
for $\left\Vert E\right\Vert _{op}$ can be derived in the same manner. We
first note that, conditional on $\left\{ V_{j}^{0}\right\} _{j\in \lbrack
p]\cup \{0\}}$, the rows of $X_{j}\odot E$ are independent across $i$.
Denote the $i$-th row of $X_{j}\odot E$ as $A_{i}^{\prime }=X_{j,i}^{\prime
}\odot E_{i}^{\prime }$, where $X_{j,i}^{\prime }$ and $E_{i}^{\prime }$
being the $i$-th row of matrix $X_{j}$ and $E$, respectively. Recall that $%
\mathscr{D}$ is the minimum $\sigma $-field generated by $\left\{
V_{j}^{0}\right\} _{j\in \lbrack p]\cup \{0\}}$. In addition, for the $t$-th
element of $A_{i}$, we have 
\begin{equation*}
\mathbb{E}\left[ X_{j,it}e_{it}|\mathscr{D}\right] =\mathbb{E}\left\{
X_{j,it}\mathbb{E}\left[ e_{it}|\mathscr{D},X_{it}\}\right] \big|\mathscr{D}%
\right\} =0,
\end{equation*}%
where the second equality holds by Assumption \ref{ass:1}(ii). Therefore, to
apply Lemma \ref{Lem:matrix op} conditionally on $\mathscr{D}$, we only need
to upper bound $\left\Vert A_{i}\right\Vert $ and }${\small ||}${\small $%
\mathbb{E}[A_{i}A_{i}^{\prime }\big|\mathscr{D}]||_{op}$. }

{\small First, under Assumption \ref{ass:1}, we have $\frac{1}{T}\sum_{t\in
\lbrack T]}\left( X_{j,it}e_{it}\right) ^{2}\leq C~a.s.$ by Assumption \ref%
{ass:1}(iv), which implies 
\begin{equation}
||A_{i}||=\left\Vert X_{j,i}\odot E_{i}\right\Vert \leq C\sqrt{T}~a.s.
\label{Lem3.1}
\end{equation}%
}

{\small Second, let $\Sigma _{i}=\mathbb{E}\left\{ \left[ \left(
X_{j,i}\odot E_{i}\right) \left( X_{j,i}\odot E_{i}\right) ^{\prime }\right] %
\big|\mathscr{D}\right\} $ with $\left( t,s\right) $ element being $\mathbb{E%
}\left( X_{j,it}X_{j,is}e_{it}e_{is}\big|\mathscr{D}\right) .$ Let $%
\left\Vert \cdot \right\Vert _{1}$ and $\left\Vert \cdot \right\Vert
_{\infty }$ denote the norms induced by the 1-norms and }$\infty ${\small %
-norms, respectively: , 
\begin{equation*}
\left\Vert \Sigma _{i}\right\Vert _{1}=\max_{s\in \lbrack T]}\sum_{t\in
\lbrack T]}\left\vert \mathbb{E}\left( X_{j,it}X_{j,is}e_{it}e_{is}\big|%
\mathscr{D}\right) \right\vert \text{ and }\left\Vert \Sigma _{i}\right\Vert
_{\infty }=\max_{t}\sum_{s\in \lbrack T]}\left\vert \mathbb{E}\left(
X_{j,it}X_{j,is}e_{it}e_{is}\big|\mathscr{D}\right) \right\vert .
\end{equation*}%
By Davydov's inequality for conditional strong mixing sequence (e.g., Lemma
4.3 in \cite{su2013testing}), we can show that 
\begin{align*}
& \max_{s\in \lbrack T]}\sum_{t\in \lbrack T]}\left\vert \mathbb{E}\left(
X_{j,it}X_{j,is}e_{it}e_{is}\big|\mathscr{D}\right) \right\vert =\max_{s\in
\lbrack T]}\sum_{t\in \lbrack T]}\left\vert Cov\left(
X_{j,it}e_{it},X_{j,is}e_{is}\big|\mathscr{D}\right) \right\vert  \\
& \lesssim \max_{s\in \lbrack T]}\sum_{t\in \lbrack T]}\left\{ \mathbb{E}%
\left[ \left\vert X_{j,it}e_{it}\right\vert ^{q}\big|\mathscr{D}\right]
\right\} ^{1/q}\left\{ \mathbb{E}\left[ \left\vert
X_{j,is}e_{is}\}\right\vert ^{q}\big|\mathscr{D}\right] \right\}
^{1/q}\alpha \left( t-s\right) ^{\left( q-2\right) /q} \\
& \leq \max_{i,t}\left\{ \mathbb{E}\left[ \left\vert
X_{j,it}e_{j,it}\right\vert ^{q}\big|\mathscr{D}\right] \right\}
^{2/q}\max_{s\in \lbrack T]}\sum_{t\in \lbrack T]}\left[ \alpha \left(
t-s\right) \right] ^{\left( q-2\right) /q}\leq c_{2}~a.s.,
\end{align*}%
where $c_{2}$ is a positive constant which does not depend on $i$.
Similarly, we have 
\begin{equation*}
\max_{t}\sum_{s\in \lbrack T]}\left\vert \mathbb{E}\left(
X_{j,it}X_{j,is}e_{it}e_{is}\big|\mathscr{D}\right) \right\vert \leq
c_{2}~a.s.
\end{equation*}%
Therefore, by Corollary 2.3.2 in \cite{golub1996matrix}, we have 
\begin{equation}
\max_{i}\left\Vert \Sigma _{i}\right\Vert _{op}\leq \sqrt{\left\Vert \Sigma
_{i}\right\Vert _{1}\left\Vert \Sigma _{i}\right\Vert _{\infty }}\leq
c_{2}~a.s.  \label{Lem3.2}
\end{equation}%
Combining (\ref{Lem3.1})-(\ref{Lem3.2}) and using Lemma \ref{Lem:matrix op}
with $t=\sqrt{\log T}$, we obtain the desired result. }
\end{proof}

{\small \bigskip }

{\small Recall that $\mathcal{R}(C_{1}):=\left\{ \left\{
\Delta_{\Theta_{j}}\right\}_{j\in [p]\cup \{0\}}\in \mathbb{R}^{N\times
T\times\left( p+1\right) }:\operatornamewithlimits{\sum}\limits_{j\in
[p]\cup \{0\}}\left\Vert \mathcal{P}_{j}^{\bot
}(\Delta_{\Theta_{j}})\right\Vert_{\ast }\leq C_{1}\operatornamewithlimits{%
\sum}\limits_{j\in[p]\cup \{0\}}\left\Vert \mathcal{P}_{j}(\Delta_{%
\Theta_{j}})\right\Vert_{\ast }\right\} .$ }

\begin{lemma}
{\small \label{Lem:RS} Suppose Assumptions \ref{ass:1}--\ref{ass:3} hold.
Then $\{\tilde{\Delta}_{\Theta _{j}}\}_{j\in \lbrack p]\cup \{0\}}\in 
\mathcal{R}(3)$ w.p.a.1. }
\end{lemma}

\begin{proof}
{\small Let }${\small A}^{{\small c}}$ {\small denote the complement of }$%
{\small A.}$ {\small Define event 
\begin{equation*}
\mathscr{A}_{1,N}(c_{3})=\left\{ \left\Vert E\right\Vert _{op}\leq c_{3}(%
\sqrt{N}\vee \sqrt{T\log T}),\left\Vert X_{j}\odot E\right\Vert _{op}\leq
c_{3}(\sqrt{N}\vee \sqrt{T\log T}),\forall j\in \lbrack p]\right\} .
\end{equation*}%
Then there exists a positive constant $c_{3}$ such that $\mathbb{P}(\mathscr{A}%
_{1,N}^{c}(c_{3}))\leq \epsilon $ for any $\epsilon >0$ by Lemma \ref%
{Lem:score op}. Under event $\mathscr{A}_{1,N}(c_{3})$, by the definition of 
$\tilde{\Theta}_{j}$ in (\ref{obj}), we notice that 
\begin{equation}
0\leq \frac{1}{NT}\left\Vert Y-\Theta _{0}^{0}-\sum_{j\in \lbrack
p]}X_{j}\odot \Theta _{j}^{0}\right\Vert ^{2}-\frac{1}{NT}\left\Vert Y-%
\tilde{\Theta}_{0}-\sum_{j\in \lbrack p]}X_{j}\odot \tilde{\Theta}%
_{j}\right\Vert ^{2}+\sum_{j\in \lbrack p]\cup \{0\}}\nu _{j}\left(
\left\Vert \Theta _{j}^{0}\right\Vert _{\ast }-\left\Vert \tilde{\Theta}%
_{j}\right\Vert _{\ast }\right)  \label{L1.1}
\end{equation}%
and 
\begin{align}
& \frac{1}{NT}\left\Vert Y-\Theta _{0}^{0}-\sum_{j\in \lbrack p]}X_{j}\odot
\Theta _{j}^{0}\right\Vert ^{2}-\frac{1}{NT}\left\Vert Y-\tilde{\Theta}%
_{0}-\sum_{j\in \lbrack p]}X_{j}\odot \tilde{\Theta}_{j}\right\Vert ^{2} 
\notag  \label{L1.2} \\
& =\frac{1}{NT}\sum_{i\in \lbrack N]}\sum_{t\in \lbrack T]}\left\{
e_{it}^{2}-\left[ e_{it}-\left( \tilde{\Delta}_{\Theta _{0},it}+\sum_{j\in
\lbrack p]}X_{j,it}\tilde{\Delta}_{\Theta _{j},it}\right) \right]
^{2}\right\}  \notag \\
& =\frac{2}{NT}tr\left( E^{\prime }\tilde{\Delta}_{\Theta _{0}}\right)
+\sum_{j\in \lbrack p]}\frac{2}{NT}tr\left( \left( E\odot X_{j}\right)
^{\prime }\tilde{\Delta}_{\Theta _{j}}\right) -\frac{1}{NT}\sum_{i\in
\lbrack N]}\sum_{t\in \lbrack T]}\left( \tilde{\Delta}_{\Theta
_{0},it}+\sum_{j\in \lbrack p]}X_{j,it}\tilde{\Delta}_{\Theta
_{j},it}\right) ^{2}  \notag \\
& \leq \frac{2}{NT}\left\vert tr\left( E^{\prime }\tilde{\Delta}_{\Theta
_{0}}\right) \right\vert +\sum_{j\in \lbrack p]}\frac{2}{NT}\left\vert
tr\left( \left( E\odot X_{j}\right) ^{\prime }\tilde{\Delta}_{\Theta
_{j}}\right) \right\vert  \notag \\
& \leq \frac{2}{NT}\left\Vert E\right\Vert _{op}\left\Vert \tilde{\Delta}%
_{\Theta _{0}}\right\Vert _{\ast }+\sum_{j\in \lbrack p]}\frac{2}{NT}%
\left\Vert E\odot X_{j}\right\Vert _{op}\left\Vert \tilde{\Delta}_{\Theta
_{j}}\right\Vert _{\ast }  \notag \\
& \leq 2c_{3}\sum_{j\in \lbrack p]\cup \{0\}}\frac{\sqrt{N}\vee \sqrt{T\log T%
}}{NT}\left\Vert \tilde{\Delta}_{\Theta _{j}}\right\Vert _{\ast },
\end{align}%
where the second inequality holds by the fact that $tr(AB)\leq
||A||_{op}||B||_{\ast },$ and the last inequality is by the definition of
event $\mathscr{A}_{1,N}$. }

{\small Combining (\ref{L1.1}) and (\ref{L1.2}), we have 
\begin{equation}
0\leq \sum_{j\in \lbrack p]\cup \{0\}}\left\{ \frac{2c_{3}(\sqrt{N}\vee 
\sqrt{T\log T})}{NT}\left\Vert \tilde{\Delta}_{\Theta _{j}}\right\Vert
_{\ast }+\nu _{j}\left( \left\Vert \Theta _{j}^{0}\right\Vert _{\ast
}-\left\Vert \tilde{\Theta}_{j}\right\Vert _{\ast }\right) \right\} ~w.p.a.1.
\label{L1.3}
\end{equation}%
Besides, we can show that 
\begin{align}
\left\Vert \tilde{\Theta}_{j}\right\Vert _{\ast }& =\left\Vert \tilde{\Delta}%
_{\Theta _{j}}+\Theta _{j}^{0}\right\Vert _{\ast }=\left\Vert \Theta
_{j}^{0}+\mathcal{P}_{j}^{\bot }(\tilde{\Delta}_{\Theta _{j}})+\mathcal{P}%
_{j}(\tilde{\Delta}_{\Theta _{j}})\right\Vert _{\ast }  \notag  \label{L1.4}
\\
& \geq \left\Vert \Theta _{j}^{0}+\mathcal{P}_{j}^{\bot }(\tilde{\Delta}%
_{\Theta _{j}})\right\Vert _{\ast }-\left\Vert \mathcal{P}_{j}(\tilde{\Delta}%
_{\Theta _{j}})\right\Vert _{\ast }=\left\Vert \Theta _{j}^{0}\right\Vert
_{\ast }+\left\Vert \mathcal{P}_{j}^{\bot }(\tilde{\Delta}_{\Theta
_{j}})\right\Vert _{\ast }-\left\Vert \mathcal{P}_{j}(\tilde{\Delta}_{\Theta
_{j}})\right\Vert _{\ast },
\end{align}%
where the second equality holds by Lemma D.2(i) in \cite%
{chernozhukov2019inference}, the first inequality is by the triangle inequality
and the last equality is by the construction of the spaces $\mathcal{P}%
_{j}^{\bot }$ and $\mathcal{P}_{j}$. Then combining (\ref{L1.3}) and (\ref%
{L1.4}), w.p.a.1, we have 
\begin{equation*}
\sum_{j\in \lbrack p]\cup \{0\}}\nu _{j}\left\Vert \tilde{\Theta}%
_{j}\right\Vert _{\ast }\leq \sum_{j\in \lbrack p]\cup \{0\}}\left\{ \nu
_{j}\left\Vert \Theta _{j}^{0}\right\Vert _{\ast }+2c_{3}\sum_{j\in \lbrack
p]\cup \{0\}}\frac{(\sqrt{N}\vee \sqrt{T})}{NT}\left\Vert \tilde{\Delta}%
_{\Theta _{j}}\right\Vert _{\ast }\right\}
\end{equation*}%
and 
\begin{align*}
\sum_{j\in \lbrack p]\cup \{0\}}\nu _{j}\left\{ \left\Vert \mathcal{P}%
_{j}^{\bot }\left( \tilde{\Delta}_{\Theta _{j}}\right) \right\Vert _{\ast
}-\left\Vert \mathcal{P}_{j}\left( \tilde{\Delta}_{\Theta _{j}}\right)
\right\Vert _{\ast }\right\} & \leq 2c_{3}\sum_{j\in \lbrack p]\cup \{0\}}%
\frac{(\sqrt{N}\vee \sqrt{T\log T})}{NT}\left\Vert \tilde{\Delta}_{\Theta
_{j}}\right\Vert _{\ast } \\
& =2c_{3}\sum_{j\in \lbrack p]\cup \{0\}}\frac{(\sqrt{N}\vee \sqrt{T\log T})%
}{NT}\left\{ \left\Vert \mathcal{P}_{j}\left( \tilde{\Delta}_{\Theta
_{j}}\right) \right\Vert _{\ast }+\left\Vert \mathcal{P}_{j}^{\bot }\left( 
\tilde{\Delta}_{\Theta _{j}}\right) \right\Vert _{\ast }\right\} ,
\end{align*}%
If we set $\nu _{j}=\frac{4c_{3}(\sqrt{N}\vee \sqrt{T\log T})}{NT}$, we
obtain the final result $\sum_{j\in \lbrack p]\cup \{0\}}\left\Vert \mathcal{%
P}_{j}^{\bot }(\tilde{\Delta}_{\Theta _{j}})\right\Vert _{\ast }\leq
3\sum_{j\in \lbrack p]\cup \{0\}}\left\Vert \mathcal{P}_{j}(\tilde{\Delta}%
_{\Theta _{j}})\right\Vert _{\ast }.$}
\end{proof}

\begin{lemma}
{\small \label{Lem:Bern_mds} Consider a sequence of random variables $%
\left\{ B_{i},i\in \lbrack n]\right\} $. }

\begin{itemize}
\item[(i)] {\small Suppose $B_{i},$ $i\in \lbrack n],$ are independent with $%
\mathbb{E}\left( B_{i}\right) =0$ and $\max_{i\in \lbrack n]}\left\vert
B_{i}\right\vert \leq M$ a.s. Let $\sigma ^{2}=\sum_{i\in \lbrack n]}\mathbb{%
E}\left( B_{i}^{2}\right) $. Then for all $t>0$, we have 
\begin{equation*}
\mathbb{P}\left( \left\vert \sum_{i\in \lbrack n]}B_{i}\right\vert >t\right)
\leq \exp \left\{ -\frac{t^{2}/2}{\sigma ^{2}+Mt/3}\right\} .
\end{equation*}%
}

\item[(ii)] {\small Suppose $\left\{ B_{i},i\in \lbrack n]\right\} $ is an
m.d.s. with $\mathbb{E}_{i-1}\left( B_{i}\right) =0$ and $\max_{i\in \lbrack
n]}\left\vert B_{i}\right\vert \leq M$ a.s., where $\mathbb{E}_{i-1}$
denotes $\mathbb{E}\left( \cdot |\mathscr{F}_{i-1}\right) $, where $\left\{ %
\mathscr{F}_{i}:i\leq n\right\} $ denotes the filtration that is clear from
the context. Let $\left\vert \sum_{i\in \lbrack n]}\mathbb{E}_{i-1}\left(
B_{i}^{2}\right) \right\vert \leq \sigma ^{2}$. Then for all $t>0$, we have 
\begin{equation*}
\mathbb{P}\left( \left\vert \sum_{i\in \lbrack n]}B_{i}\right\vert >t\right)
\leq \exp \left\{ -\frac{t^{2}/2}{\sigma ^{2}+Mt/3}\right\} .
\end{equation*}%
}
\end{itemize}
\end{lemma}

\begin{proof}
{\small Lemma \ref{Lem:Bern_mds}(i) and (ii) are Bernstein inequality for
the partial sum of an independent sequence and the Freedman inequality for
the partial sum of an m.d.s., which are respectively stated in Lemma 2.2.9 
\cite{vaart1996weak} and Theorem 1.1 \cite{tropp2011user}. }
\end{proof}

\begin{lemma}
{\small \label{Lem:Bern_mixing} Let $\left\{ \Upsilon _{t},t\geq 1\right\} $
be a zero-mean strong mixing process, not necessarily stationary, with the
mixing coefficients satisfying $\alpha (z)\leq c_{\alpha }{\small \vartheta }%
^{z}$ for some $c_{\alpha }>0$ and }${\small \vartheta }${\small $\in (0,1)$%
. If $\sup_{t\in \lbrack T]}\left\vert \Upsilon _{t}\right\vert \leq M_{T}$,
then there exists a constant $c_{4}$ depending on $c_{\alpha }$ and }$%
{\small \vartheta }${\small \ such that for any $T\geq 2$ and $\varepsilon >0
$, }

\begin{itemize}
\item[(i)] {\small $\mathbb{P}\left\{
\left\vert\sum_{t=1}^{T}\Upsilon_{t}\right\vert >\varepsilon \right\} \leq
\exp\left\{ -\frac{c_{4}\varepsilon ^{2}}{M_{T}^{2}T+\varepsilon
M_{T}\left(\log T\right)\left( \log \log T\right) }\right\} $, }

\item[(ii)] {\small $\mathbb{P}\left\{
\left\vert\sum_{t=1}^{T}\Upsilon_{t}\right\vert >\varepsilon \right\} \leq
\exp\left\{ -\frac{c_{4}\varepsilon ^{2}}{\upsilon_{0}^{2}T+M_{T}^{2}+%
\varepsilon M_{T}\left(\log T\right) ^{2}}\right\} $, }
\end{itemize}

\noindent {\small where $\upsilon _{0}^{2}=\sup_{t\in \lbrack T]}\left[
Var(\Upsilon _{t})+2\sum_{s>t}\left\vert Cov(\Upsilon _{t},\Upsilon
_{s})\right\vert \right] $. }
\end{lemma}

\begin{proof}
{\small The proof is the same as that of Theorems 1 and 2 in \cite%
{merlevede2009bernstein} with the condition $\alpha (a)\leq
\exp\left\{-2ca\right\} $ for some $c>0.$ Here we can set $c=-\log \gamma $
if $c_{\alpha }\geq 1$ and $c=-\log (\gamma /c_{\alpha })$ otherwise. }
\end{proof}

\begin{lemma}
{\small \label{Lem:bounded u&v} Suppose Assumptions \ref{ass:1}--\ref{ass:4}
hold, for $j\in \left\{ 0,\cdots ,p\right\} $, we have }

\begin{itemize}
\item[(i)] {\small $\max_{i}\left\Vert u_{i,j}^{0}\right\Vert \leq M$ and $%
\max_{t}\left\Vert v_{t,j}^{0}\right\Vert \leq \frac{M}{\sigma _{K_{j},j}}%
\leq \frac{M}{c_{\sigma }}$, }

\item[(ii)] {\small $\max_{t}\left\Vert O_{j}^{\prime }\tilde{v}%
_{t,j}\right\Vert \leq \frac{2M}{\sigma _{K_{j},j}}\leq \frac{2M}{c_{\sigma }%
}$ w.p.a.1, }

\item[(iii)] {\small $\max_{i}\frac{1}{T}\sum_{t\in \lbrack T]}\left\Vert 
\tilde{\phi}_{it}\right\Vert ^{2}\leq \frac{4M^{2}}{c_{\sigma }^{2}}\left(
1+pC\right) $ w.p.a.1, }

\item[(iv)] {\small $\max_{i}\frac{1}{T}\sum_{t\in \lbrack T]}\left\Vert 
\tilde{\phi}_{it}-\phi _{it}^{0}\right\Vert ^{2}=O_{p}(\eta
_{N,1}^{2}(NT)^{2/q})$. }
\end{itemize}
\end{lemma}

\begin{proof}
{\small (i) Recall that $\frac{1}{\sqrt{NT}}\Theta _{j}^{0}=\mathcal{U}%
_{j}^{0}\Sigma _{j}^{0}\mathcal{V}_{j}^{0\prime },$ $U_{j}^{0}=\sqrt{N}%
\mathcal{U}_{j}^{0}\Sigma _{j}^{0}$ and $V_{j}=\sqrt{T}\mathcal{V}_{j}^{0}$.
Let }${\small [A]}_{{\small i\cdot }}${\small \ and }${\small [A]}_{{\small %
\cdot t}}$ {\small denote the }${\small i}${\small -th row and }${\small t}$%
{\small -th column of }${\small A,}$ {\small respectively.} {\small Note
that 
\begin{equation}
\frac{1}{\sqrt{T}}\Theta _{j}^{0}\mathcal{V}_{j}^{0}=\sqrt{N}\mathcal{U}%
_{j}^{0}\Sigma _{j}^{0}=U_{j}^{0},\text{ and }\frac{1}{\sqrt{N}}\mathcal{U}%
_{j}^{0\prime }\Theta _{j}^{0}=\sqrt{T}\Sigma _{j}^{0}\mathcal{V}%
_{j}^{0\prime }=\Sigma _{j}^{0}V_{j}^{0\prime }.  \label{Lem12.1a}
\end{equation}%
Hence, it's easy to see that 
\begin{equation*}
\left\Vert u_{i,j}^{0}\right\Vert =\frac{1}{\sqrt{T}}\left\Vert \left[
\Theta _{j}^{0}\mathcal{V}_{j}^{0}\right] _{i\cdot }\right\Vert \leq \frac{1%
}{\sqrt{T}}\left\Vert \left[ \Theta _{j}\right] _{i\cdot }\right\Vert \leq M,
\end{equation*}%
where the first inequality is due to the fact that $\mathcal{V}_{j}$ is the
unitary matrix and the last inequality holds by Assumption \ref{ass:2}.
Since the upper bound $M$ is not dependent on $i$, this result holds
uniformly. Analogously, we see that 
\begin{equation*}
\left\Vert v_{t,j}^{0}\right\Vert \leq \frac{1}{\sqrt{N}}c_{\sigma
}^{-1}\left\Vert \left[ \mathcal{U}_{j}^{0\prime }\Theta _{j}^{0}\right]
_{\cdot t}\right\Vert \leq \frac{1}{\sqrt{N}}c_{\sigma }^{-1}\left\Vert %
\left[ \Theta _{j}^{0}\right] _{\cdot t}\right\Vert \leq \frac{M}{c_{\sigma }%
}.
\end{equation*}%
}

{\small (ii) As in (\ref{Lem12.1a}), we have 
\begin{equation*}
\frac{1}{\sqrt{N}}\tilde{\mathcal{U}}_{j}^{(1)\prime }\tilde{\Theta}_{j}=%
\sqrt{T}\tilde{\Sigma}_{j}^{(1)}\tilde{\mathcal{V}}_{j}^{(1)\prime }=\tilde{%
\Sigma}_{j}^{(1)}\tilde{V}_{j}^{(1)\prime },
\end{equation*}%
and 
\begin{equation*}
\left\Vert O_{j}^{\prime }\tilde{v}_{t,j}\right\Vert \leq \frac{1}{\sqrt{N}}%
\frac{1}{\tilde{\sigma}_{K_{j},j}^{(1)}}\left\Vert \left[ \tilde{\mathcal{U}}%
_{j}^{(1)\prime }\tilde{\Theta}_{j}\right] _{\cdot t}\right\Vert \leq \frac{1%
}{\sqrt{N}}\frac{1}{\tilde{\sigma}_{K_{j},j}^{(1)}}\left\Vert \left[ \tilde{%
\Theta}_{j}\right] _{\cdot t}\right\Vert \leq \frac{2M}{c_{\sigma }},
\end{equation*}%
where the last inequality holds due to the constrained optimization in (\ref%
{obj}) and the fact that $\max_{k\in \lbrack K_{j}]}|\tilde{\sigma}%
_{k,j}^{-1}-\sigma _{k,j}^{-1}|\leq \sigma _{K_{j},j}^{-1}$ w.p.a.1. }

{\small (iii) Note that 
\begin{eqnarray*}
\max_{i}\frac{1}{T}\sum_{t\in \lbrack T]}\left\Vert \tilde{\phi}%
_{it}\right\Vert ^{2} &\leq &\max_{i}\left\{ \frac{1}{T}\sum_{t\in \lbrack
T]}\left\Vert O_{0}^{\prime }\tilde{v}_{t,0}\right\Vert ^{2}+\sum_{j\in
\lbrack p]}\frac{1}{T}\sum_{t\in \lbrack T]}\left\Vert O_{j}^{\prime }\tilde{%
v}_{t,j}\right\Vert ^{2}\left\vert X_{j,it}\right\vert ^{2}\right\} \\
&\leq &\max_{t\in \lbrack T],j\in \lbrack p]\cup \left\{ 0\right\}
}\left\Vert O_{j}^{\prime }\tilde{v}_{t,j}\right\Vert ^{2}\left\{
1+\max_{i}\sum_{j\in \lbrack p]}\frac{1}{T}\sum_{t\in \lbrack T]}\left\vert
X_{j,it}\right\vert ^{2}\right\} \leq \frac{4M^{2}}{c_{\sigma }^{2}}\left(
1+pC\right)
\end{eqnarray*}%
where the last inequality holds by Lemma \ref{Lem:bounded u&v}(ii). }

{\small (iv) Note that 
\begin{align*}
& \max_{i}\frac{1}{T}\sum_{t\in \lbrack T]}\left\Vert \tilde{\phi}_{it}-\phi
_{it}^{0}\right\Vert ^{2}\leq \frac{1}{T}\sum_{t\in \lbrack T]}\left\Vert 
\tilde{O}_{0}^{(1)\prime }\tilde{v}_{t,0}-v_{t,0}^{0}\right\Vert
^{2}+p\max_{t\in \lbrack T],j\in \lbrack p]}\frac{1}{T}\sum_{t\in \lbrack
T]}\left\vert X_{j,it}\right\vert ^{2}\left\Vert \tilde{O}_{j}^{(1)\prime }%
\tilde{v}_{t,j}-v_{t,j}^{0}\right\Vert ^{2} \\
& \lesssim \frac{1}{T}\sum_{t\in \lbrack T]}\left\Vert \tilde{O}%
_{0}^{(1)\prime }\tilde{v}_{t,0}-v_{t,0}^{0}\right\Vert
^{2}+p(NT)^{2/q}\max_{j\in \left[ p\right] }\frac{1}{T}\sum_{t\in \lbrack
T]}\left\Vert \tilde{O}_{j}^{(1)\prime }\tilde{v}_{t,j}-v_{t,j}^{0}\right%
\Vert ^{2}\quad \\
& =\frac{1}{T}\left\Vert O_{0}\tilde{V}_{0}-V_{0}^{0}\right\Vert
^{2}+p(NT)^{2/q}\max_{j\in \lbrack p]}\frac{1}{T}\left\Vert O_{j}\tilde{V}%
_{j}-V_{j}^{0}\right\Vert ^{2}=O_{p}(\eta _{N,1}^{2}(NT)^{2/q}),
\end{align*}%
where the second inequality is by Assumption \ref{ass:1}(v) and the last
equality holds by Theorem \ref{Thm1}(ii). }
\end{proof}

\begin{lemma}
{\small \label{Lem:lambda_phi} Under Assumptions \ref{ass:1}--\ref{ass:5},
we have $\min_{i\in \lbrack N]}\lambda _{\min }(\tilde{\Phi}_{i})\geq \frac{%
c_{\phi }}{2}~w.p.a.1$, and $\min_{t\in \lbrack T]}\lambda _{\min }(\tilde{%
\Psi}_{t})\geq \frac{c_{\phi }}{2}~w.p.a.1$. }
\end{lemma}

\begin{proof}
{\small Recall that $\Phi _{i}=\frac{1}{T}\sum_{t=1}^{T}\phi _{it}^{0}\phi
_{it}^{0\prime }$ and $\tilde{\Phi}_{i}=\frac{1}{T}\sum_{t=1}^{T}\tilde{\phi}%
_{it}\tilde{\phi}_{it}^{\prime },$ where 
\begin{equation*}
\phi _{it}^{0}=[v_{t,0}^{0\prime },v_{t,1}^{0\prime }X_{1,it},\cdots
,v_{t,p}^{0\prime }X_{p,it}]^{\prime }\text{ and }\tilde{\phi}_{it}=\left[
\left( O_{0}^{\prime }\tilde{v}_{t,0}\right) ^{\prime },\left( O_{1}^{\prime
}\tilde{v}_{t,1}X_{1,it}\right) ^{\prime },\cdots ,\left( O_{p}^{\prime }%
\tilde{v}_{t,p}X_{p,it}\right) ^{\prime }\right] ^{\prime }.
\end{equation*}%
Uniformly over $i\in \lbrack N]$, it is clear that 
\begin{align*}
\left\Vert \tilde{\Phi}_{i}-\Phi _{i}\right\Vert _{{}}& \lesssim \frac{4M}{%
c_{\sigma }T}\sum_{t=1}^{T}\left\Vert O_{0}^{\prime }\tilde{v}%
_{t,0}-v_{t,0}^{0}\right\Vert +\frac{4M}{c_{\sigma }T}\sum_{j=1}^{p}%
\sum_{t=1}^{T}\left\Vert O_{j}^{\prime }\tilde{v}_{t,j}-v_{t,j}^{0}\right%
\Vert \left\vert X_{j,it}\right\vert \\
& \leq \frac{4M}{c_{\sigma }}\frac{1}{\sqrt{T}}\left\Vert O_{0}^{\prime }%
\tilde{V}_{0}-V_{0}^{0}\right\Vert +\frac{4M^{2}}{c_{\sigma }}\sum_{j=1}^{p}%
\frac{1}{\sqrt{T}}\left\Vert O_{j}^{\prime }\tilde{V}_{j}-V_{j}^{0}\right%
\Vert \left( \frac{1}{T}\sum_{t\in \lbrack T]}\left\vert X_{j,it}\right\vert
^{2}\right) ^{1/2}=O_{p}\left( \eta _{N,1}\right) ,
\end{align*}%
where the last equality holds by Lemma \ref{Lem:bounded u&v}(i) and
Assumption \ref{ass:1}(iv). It follows that 
\begin{equation*}
\min_{i\in \lbrack N]}\lambda _{\min }(\tilde{\Phi}_{i})\geq \min_{i\in
\lbrack N]}\lambda _{\min }(\Phi _{i})-O\left( \eta _{N,1}\right) \geq \frac{%
c_{\phi }}{2},\quad \text{w.p.a.1}.
\end{equation*}%
Analogously, we can establish the lower bound of $\lambda _{\min }(\tilde{%
\Psi}_{t})$. }
\end{proof}

\begin{lemma}
{\small \label{Lem:Kmeans_null} Under Assumptions \ref{ass:1}--\ref{ass:7},
we have $\max_{i}\mathbf{1}\{\hat{g}_{i,K^{(\ell )}}^{0,(\ell )}\neq
g_{i}^{(\ell )}\}=0~$w.p.a.1$,$ where $\hat{g}_{i,K^{(\ell )}}^{0,(\ell )}$
is defined in \eqref{group estimates_truebreak}. }
\end{lemma}

\begin{proof}
{\small The above lemma holds by Theorem 2.3 in \cite{su2020strong} provided
we can verify the conditions in their Assumption 4. Let $\alpha _{k}^{(\ell
)}=(\alpha _{k,1}^{(\ell )},\cdots ,\alpha _{k,p}^{(\ell )})^{\prime }$.
Then we have 
\begin{equation*}
\beta _{i}^{0,(\ell )}=\frac{1}{\sqrt{T_{\ell }}}\sum_{k\in \lbrack K^{(\ell
)}]}\alpha _{k}^{(\ell )}\otimes \iota _{T_{\ell }}\mathbf{1}\{g_{i}^{(\ell
)}=k\}
\end{equation*}%
and 
\begin{equation}
\max_{k\in \lbrack K^{(\ell )}]}\left\Vert \frac{1}{\sqrt{T_{\ell }}}\alpha
_{k}^{(\ell )}\otimes \iota _{T_{\ell }}\right\Vert =\max_{k\in \lbrack
K^{(\ell )}]}\frac{1}{\sqrt{T_{\ell }}}\sqrt{T_{\ell }\sum_{j=1}^{p}(\alpha
_{k,j}^{(\ell )})^{2}}\leq \sqrt{p}\max_{k\in \lbrack K^{(\ell )}],j\in
\lbrack p]}|\alpha _{k,j}^{(\ell )}|\leq \sqrt{p}M,  \label{Lem:Kmeans_1}
\end{equation}%
where the last inequality is due to Assumption \ref{ass:2}. }

{\small Second, with $\Theta _{j,i}^{0,(1)}=\left( \Theta _{j,i1}^{0},\cdots
,\Theta _{j,iT_{1}}^{0}\right) ^{\prime }$ and $\Theta _{j,i}^{0,(2)}=\left(
\Theta _{j,i,T_{1}+1}^{0},\cdots ,\Theta _{j,iT}^{0}\right) ^{\prime }$, we
observe that 
\begin{align}
& \max_{i}\left\Vert \dot{\beta}_{i}^{0,(\ell )}-\beta _{i}^{0,(\ell
)}\right\Vert =\frac{1}{\sqrt{T_{\ell }}}\max_{i}\left\Vert \dot{\Theta}%
_{i}^{(\ell )}-\Theta _{i}^{0,(\ell )}\right\Vert =\frac{1}{\sqrt{T_{\ell }}}%
\max_{i}\sqrt{\sum_{j=1}^{p}\sum_{t\in \mathcal{T}_{\ell }}(\dot{\Theta}%
_{j,it}^{(\ell )}-\Theta _{j,it}^{0,(\ell )})^{2}}  \notag \\
& \leq \sqrt{p}\max_{j\in \lbrack p],i\in \lbrack N],t\in \lbrack
T]}\left\vert \dot{\Theta}_{j,it}-\Theta _{j,it}^{0}\right\vert \leq
c_{5}\eta _{N,2}\quad \text{w.p.a.1,}  \label{Lem:Kmeans_2}
\end{align}%
with $c_{5}$ being some positive large enough constant, and the last
inequality holds by Theorem \ref{Thm1}(iii). }

{\small Third, we also observe that 
\begin{equation}
\min_{1\leq k_{s}<k_{s^{\ast }}\leq K^{(\ell )}}\frac{1}{\sqrt{T_{\ell }}}%
\left\Vert \alpha _{k_{s}}^{(\ell )}\otimes \iota _{T_{\ell }}-\alpha
_{k_{s^{\ast }}}^{(\ell )}\otimes \iota _{T_{\ell }}\right\Vert =\min_{1\leq
k_{s}<k_{s^{\ast }}\leq K^{(\ell )}}\sqrt{\sum_{j=1}^{p}(\alpha
_{k_{s},j}^{(\ell )}-\alpha _{k_{s^{\ast }},j}^{(\ell )})^{2}}\geq C_{5},
\label{Lem:Kmeans_3}
\end{equation}%
where the last inequality holds by Assumption \ref{ass:7}(i). }

{\small Combining (\ref{Lem:Kmeans_1}), (\ref{Lem:Kmeans_2}) and (\ref%
{Lem:Kmeans_3}), we obtain that $\mathbb{P}\left( \max_{i}\mathbf{1}\{\hat{g}%
_{i,K^{(\ell )}}^{0,(\ell )}\neq g_{i}^{(\ell )}\}=0\right) \rightarrow 1$
once we can ensure Assumption 4.3 in \cite{su2020strong} holds with $%
c_{1n}=C_{5}$, $c_{2n}=c_{5}\eta _{N,2}$, $K=K^{(1)}$, and with their $c_{1}$
and $M$ being replaced by $\underline{c}$ and and $\sqrt{p}M$ here. Under
Assumption \ref{ass:7}, Assumption 4.3 in \cite{su2020strong} holds. This
completes the proof of the lemma. }
\end{proof}

\vspace{2mm}
{\small To study the NSP property of our group structure estimator, we
introduce some notation in the following definition. }

\begin{definition}
{\small \label{def:2} Fix $K^{(\ell )}>1$ and $1<m\leq K^{(\ell )}.$ Define
a $K^{(\ell )}\times p$ matrix $\alpha ^{(\ell )}=(\alpha _{1}^{(\ell
)},\cdots ,\alpha _{K}^{(\ell )})^{\prime }$. Let $d_{K^{(\ell )}}(\alpha
^{(\ell )})$ be the minimum pairwise distance of all $K^{(\ell )}$ rows and $%
\alpha _{k}^{(\ell )}$ and $\alpha _{l}^{(\ell )}$ be the pair that
satisfies $||\alpha _{k}^{(\ell )}-\alpha _{l}^{(\ell )}||=d_{K^{(\ell
)}}(\alpha ^{(\ell )})$ (if this holds for multiple pairs, pick the first
pair in lexicographical order). Remove row $l$ from matrix $\alpha
^{(\ell )}$ and let $d_{K^{(\ell )}-1}(\alpha ^{(\ell )})$ be the minimum
pairwise distance for the remaining $(K^{(\ell )}-1)$ rows. Repeat this step
and define $d_{K^{(\ell )}-2}(\alpha ^{(\ell )}),\cdots ,d_{2}(\alpha
^{(\ell )})$ recursively. }
\end{definition}

\begin{lemma}
{\small \label{Lem:NSP} Recall that $\hat{\mathcal{G}}_{m}^{(\ell )}$ is the
estimated group structure from K-means algorithm with $m$ groups. Under
Assumptions \ref{ass:1}-\ref{ass:7} and the event $\{\hat{T}_{1}=T_{1}\}$,
w.p.a.1, for each $1<m<K^{(\ell )}$, $\hat{\mathcal{G}}_{m}^{(\ell )}$
enjoys the NSP defined in Definition \ref{def:1}. }
\end{lemma}

\begin{proof}
{\small By Theorem 4.1 in \cite{Jin2022Optimal}, Lemma \ref{Lem:NSP} is
proved if we ensure all conditions in their Theorem 4.1 hold. We now apply
their Theorem 4.1 with $\hat{x}_{i}=\dot{\beta}_{i}^{0,(\ell )}$, $%
x_{i}=\beta _{i}^{0,(\ell )}$ and $u_{k}=\frac{1}{\sqrt{T_{\ell }}}\alpha
_{k}^{(\ell )}\otimes \iota _{T_{\ell }}$ for $k\in \lbrack K^{(\ell )}]$.
By the definition of $d_{m}\left( \alpha ^{(\ell )}\right) $ in Definition %
\ref{def:2}, we notice that $d_{m}\left( \alpha ^{(\ell )}\right) \geq
d_{K^{(\ell )}}\left( \alpha ^{(\ell )}\right) $ such that $d_{K^{(\ell
)}}\left( \alpha ^{(\ell )}\right) \geq C_{5}$ by Assumption \ref{ass:7}(i).
With \eqref{Lem:Kmeans_2} shown above and Assumption \ref{ass:2}, we have 
\begin{equation*}
\max_{k\in \lbrack K^{(\ell )}]}\left\Vert u_{k}\right\Vert \leq M,\quad
\max_{i}\left\Vert \hat{x}_{i}-x_{i}\right\Vert =O_{p}(\eta _{N,2}),
\end{equation*}%
which satisfy the Theorem 4.1 in \cite{Jin2022Optimal}, i.e., $\max_{k\in
\lbrack K^{(\ell )}]}\left\Vert u_{k}\right\Vert \lesssim d_{m}\left( \alpha
^{(\ell )}\right) $ and $\max_{i}\left\Vert \hat{x}_{i}-x_{i}\right\Vert $ $%
\lesssim d_{m}\left( \alpha ^{(\ell )}\right) $. Consequently, it leads to
the NSP of $\hat{\mathcal{G}}_{m}^{(\ell )}$ for $1<m<K^{(\ell )}$ w.p.a.1
under the event }${\small \{}${\small $\hat{T}_{1}=T_{1}\}$. }
\end{proof}

\begin{lemma}
{\small \label{Lem:consis_homo} Under Assumptions \ref{ass:1}, \ref{ass:6}%
(ii), \ref{ass:7}(ii), \ref{ass:8} and \ref{ass:10}(i)-(iii), for $\ell
=\{1,2\}$ and $k\in \lbrack K^{(\ell )}]$, we have $\hat{\alpha}_{k}^{(\ell
)}\overset{p}{\longrightarrow }\alpha _{k}^{(\ell )}$ and $\sqrt{%
N_{k}^{(\ell )}T_{\ell }}(\hat{\alpha}_{k}^{(\ell )}-\alpha _{k}^{(\ell )})=%
\mathbb{W}_{NT,k}^{(\ell )-1}\mathbb{C}_{NT,k}^{(\ell )}+o_{p}(1)$. }
\end{lemma}

\begin{proof}
{\small The result in the lemma combines those in Theorem 4.1 and Corollary
4.2 in \cite{moon2017dynamic} under their Assumptions 1-4. Hence, we only
need to verify the conditions in their Assumptions 2 and 3 since our
Assumptions \ref{ass:8} and \ref{ass:10}(ii)-(iii) are the same as their
Assumptions 1 and 4. }

{\small Notice that the Assumption 2 in \cite{moon2017dynamic} holds if we
can show that 
\begin{equation*}
\frac{1}{N_{k}^{(\ell )}T_{\ell }}\sum_{i\in G_{k}^{(\ell )}}\sum_{t\in 
\mathcal{T}_{\ell }}X_{j,it}e_{it}\overset{p}{\longrightarrow }0,~\forall
k\in \lbrack K^{(\ell )}],~\ell \in \{1,2\}.
\end{equation*}%
Fix a specific $k$ and $\ell $. We can show that 
\begin{align}
& \mathbb{E}\left( \frac{1}{N_{k}^{(\ell )}T_{\ell }}\sum_{i\in G_{k}^{(\ell
)}}\sum_{t\in \mathcal{T}_{\ell }}X_{j,it}e_{it}\bigg|\mathscr{D}\right) ^{2}
\notag  \label{Lem:B12_1} \\
& =\frac{1}{\left( N_{k}^{(\ell )}T_{\ell }\right) ^{2}}\sum_{i_{1}\in
G_{k}^{(\ell )}}\sum_{i_{2}\in G_{k}^{(\ell )}}\sum_{t_{1}\in \mathcal{T}%
_{\ell }}\sum_{t_{2}\in \mathcal{T}_{\ell }}\mathbb{E}\left(
X_{j,i_{1}t_{1}}X_{j,i_{2}t_{2}}e_{i_{1}t_{1}}e_{i_{2}t_{2}}|\mathscr{D}%
\right)   \notag \\
& =\frac{1}{\left( N_{k}^{(\ell )}T_{\ell }\right) ^{2}}\sum_{i\in
G_{k}^{(\ell )}}\sum_{t_{1}\in \mathcal{T}_{\ell }}\sum_{t_{2}\in \mathcal{T}%
_{\ell }}\mathbb{E}\left( X_{j,it_{1}}X_{j,it_{2}}e_{it_{1}}e_{it_{2}}|%
\mathscr{D}\right)   \notag \\
& =\frac{1}{\left( N_{k}^{(\ell )}T_{\ell }\right) ^{2}}\sum_{i\in
G_{k}^{(\ell )}}\sum_{t\in \mathcal{T}_{\ell }}\mathbb{E}\left(
X_{j,it}^{2}e_{it}^{2}|\mathscr{D}\right) +\frac{2}{\left( N_{k}^{(\ell
)}T_{\ell }\right) ^{2}}\sum_{i\in G_{k}^{(\ell )}}\sum_{t_{1}\in \mathcal{T}%
_{\ell }}\sum_{t_{2}\in \mathcal{T}_{\ell },t_{2}>t_{1}}\mathbb{E}\left(
X_{j,it_{1}}X_{j,it_{2}}e_{it_{1}}e_{it_{2}}|\mathscr{D}\right)   \notag \\
& \leq \frac{M}{N_{k}^{(\ell )}T_{\ell }}+\frac{16}{\left( N_{k}^{(\ell
)}T_{\ell }\right) ^{2}}\max_{i\in G_{k}^{(\ell )}}\max_{t\in \mathcal{T}%
_{L}}\left( \mathbb{E}\left\vert X_{j,it}e_{it}\right\vert ^{q}\right)
^{2/q}\sum_{i\in G_{k}^{(\ell )}}\sum_{t_{1}\in \mathcal{T}_{\ell
}}\sum_{t_{2}\in \mathcal{T}_{\ell },t_{2}>t_{1}}\left[ \alpha (t_{2}-t_{1})%
\right] ^{1-2/q}  \notag \\
& =O\left( (NT)^{-1}\right) ,
\end{align}%
where the second equality holds by Assumption \ref{ass:1}(i) with the
conditional independence sequence for $i_{1}\neq i_{2}$, the first
inequality combines Assumption \ref{ass:1}(ii), (iii), (v), and the
Davydov's inequality for strong mixing sequence in Lemma 4.3, \cite%
{su2013testing}, and the last equality is by Assumption \ref{ass:1}(iii),
(v), Assumption \ref{ass:6}(ii) and Assumption \ref{ass:7}(ii). Following
this, it yields that 
\begin{equation*}
\frac{1}{N_{k}^{(\ell )}T_{\ell }}\sum_{i\in G_{k}^{(\ell )}}\sum_{t\in 
\mathcal{T}_{\ell }}X_{j,it}e_{it}=O_{p}((NT)^{-1/2}).
\end{equation*}%
}

{\small By similar arguments to those used in the proof of Lemma \ref{Lem:score op}%
, we can show that 
\begin{equation}
\left\Vert E_{k}^{(\ell )}\right\Vert _{op}=O_{p}(\sqrt{N}+\sqrt{T\log T}),
\label{Lem:B12_2}
\end{equation}%
which, in conjunction with Assumption \ref{ass:10}(i), implies that
Assumption $3^{\ast }$ in \cite{moon2017dynamic} is satisfied. }
\end{proof}

{\small For $j\in \lbrack p]$, recall that $X_{j,i}^{(1)}=\left(
X_{j,i1},\cdots ,X_{j,iT_{1}}\right) ^{\prime }$, $X_{j,i}^{(2)}=\left(
X_{j,i(T_{1}+1)},\cdots ,X_{j,iT}\right) ^{\prime }$, $e_{i}^{(1)}=\left(
e_{i1},\cdots ,e_{iT_{1}}\right) ^{\prime }$, $e_{i}^{(2)}=\left(
e_{i(T_{1}+1)},\cdots ,e_{iT}\right) ^{\prime }$, $\tilde{X}_{j,it}=X_{j,it}-%
\mathbb{E}\left( X_{j,it}|\mathscr{D}\right) $. Besides, let $\mathbb{X}%
_{j,k}^{(\ell )}\in \mathbb{R}^{N_{k}^{(\ell )}\times T_{\ell }}$ and $%
E_{k}^{(\ell )}\in \mathbb{R}^{N_{k}^{(\ell )}\times T_{\ell }}$ denote the
regressor and error matrix for subgroup $k\in \lbrack K^{(\ell )}]$ with a
typical row being $X_{j,i}^{(\ell )}$ and $e_{i}^{(\ell )}$, respectively.
For $\ell \in \{1,2\}$ and $k\in \lbrack K^{(\ell )}]$, we also define 
\begin{equation*}
\bar{\mathbb{X}}_{j,k}^{(\ell )}=\mathbb{E}(\mathbb{X}_{j,k}^{(\ell )}\big|%
\mathscr{D}),\quad \tilde{\mathbb{X}}_{j,k}^{(\ell )}=\mathbb{X}%
_{j,k}^{(\ell )}-\bar{\mathbb{X}}_{j,k}^{(\ell )},\quad \mathfrak{X}%
_{j,k}^{(\ell )}=M_{\Lambda _{k}^{0,(\ell )}}\bar{\mathbb{X}}_{j,k}^{(\ell
)}M_{F^{0,(\ell )}}+\tilde{\mathbb{X}}_{j,k}^{(\ell )},
\end{equation*}%
with $\mathfrak{X}_{j,k,it}^{(\ell )}$ being each entry of $\mathfrak{X}%
_{j,k}^{(\ell )}$. Further let $\mathfrak{X}_{k,it}^{(\ell )}=(\mathfrak{X}%
_{1,k,it}^{(\ell )},\cdots ,\mathfrak{X}_{p,k,it}^{(\ell )})^{\prime }$. }

\begin{lemma}
{\small \label{Lem:distri_homo} Under Assumptions \ref{ass:1}, \ref{ass:2}, %
\ref{ass:6}(ii), \ref{ass:7}(ii), \ref{ass:8} and \ref{ass:10}, for $j\in[p]$%
, $\ell\in\{1,2\}$ and $k\in[K^{(\ell)}]$, we have }

\begin{itemize}
\item[(i)] {\small $\frac{1}{\sqrt{N_{k}^{(\ell)}T_{\ell}}}%
tr\left(P_{F^{0,(\ell)}}E_{k}^{(\ell)\prime}P_{\Lambda_{k}^{0,(\ell)}}\tilde{%
\mathbb{X}}_{j,k}^{(\ell)}\right)=o_{p}(1)$, }

\item[(ii)] {\small $\frac{1}{\sqrt{N_{k}^{(\ell)}T_{\ell}}}%
tr\left(P_{\Lambda_{k}^{0,(\ell)}} E_{k}^{(\ell)\prime}\tilde{\mathbb{X}}%
_{j,k}^{(\ell)}\right)=o_{p}(1)$, }

\item[(iii)] {\small $\frac{1}{\sqrt{N_{k}^{(\ell)}T_{\ell}}}%
tr\left\{P_{F^{0,(\ell)}}\left[ E_{k}^{(\ell)\prime}\tilde{\mathbb{X}}%
_{j,k}^{(\ell)}-\mathbb{E}\left(E_{k}^{(\ell)\prime}\tilde{\mathbb{X}}%
_{j,k}^{(\ell)} \big|\mathscr{D}\right)\right]\right\}=o_{p}(1)$, }

\item[(iv)] {\small $\frac{1}{\sqrt{N_{k}^{(\ell)}T_{\ell}}}tr\left[%
E_{k}^{(\ell)}P_{F^{0,(\ell)}}E_{k}^{(\ell)\prime}M_{\Lambda_{k}^{0,(\ell)}}%
\mathbb{X}_{j,k}^{(\ell)}
F^{0,(\ell)}\left(F^{0,(\ell)\prime}F^{0,(\ell)}\right)^{-1}\left(%
\Lambda_{k}^{0,(\ell)\prime}\Lambda_{k}^{0,(\ell)}\right)^{-1}
\Lambda_{k}^{0,(\ell)\prime} \right]=o_{p}(1)$, }

\item[(v)] {\small $\frac{1}{\sqrt{N_{k}^{(\ell)}T_{\ell}}}tr\left[%
E_{k}^{(\ell)\prime} P_{\Lambda_{k}^{0,(\ell)}}E_{k}^{(\ell)}M_{F^{0,(\ell)}}%
\mathbb{X}_{j,k}^{(\ell)\prime}\Lambda_{k}^{0,(\ell)}\left(\Lambda_{k}^{0,(%
\ell)\prime}\Lambda_{k}^{0,(\ell)}
\right)^{-1}\left(F^{0,(\ell)\prime}F^{0,(\ell)}
\right)^{-1}F^{0,(\ell)\prime} \right]=o_{p}(1)$, }

\item[(vi)] {\small $\frac{1}{\sqrt{N_{k}^{(\ell)}T_{\ell}}}tr\left[%
E_{k}^{(\ell)\prime}M_{\Lambda_{k}^{0,(\ell)}}\mathbb{X}_{j,k}^{(%
\ell)}M_{F^{0,(\ell)}}E_{k}^{(\ell)\prime}\Lambda_{k}^{0,(\ell)}\left(%
\Lambda_{k}^{0,(\ell)\prime}\Lambda_{k}^{0,(\ell)}
\right)^{-1}\left(F^{0,(\ell)\prime}F^{0,(\ell)}
\right)^{-1}F^{0,(\ell)\prime} \right]=o_{p}(1)$, }

\item[(vii)] {\small $\frac{1}{\sqrt{N_{k}^{(\ell)}T_{\ell}}}tr\left\{\left[%
E_{k}^{(\ell)}E_{k}^{(\ell)\prime}-\mathbb{E}\left(E_{k}^{(\ell)}E_{k}^{(%
\ell)\prime} \big|\mathscr{D}\right) \right]M_{\Lambda_{k}^{0,(\ell)}}%
\mathbb{X}_{j,k}^{(\ell)}
F^{0,(\ell)}\left(F^{0,(\ell)\prime}F^{0,(\ell)}\right)^{-1}\left(%
\Lambda_{k}^{0,(\ell)\prime}\Lambda_{k}^{0,(\ell)}\right)^{-1}
\Lambda_{k}^{0,(\ell)\prime} \right\}=o_{p}(1)$, }

\item[(viii)] {\small $\frac{1}{\sqrt{N_{k}^{(\ell)}T_{\ell}}}tr\left\{\left[%
E_{k}^{(\ell)\prime}E_{k}^{(\ell)}-\mathbb{E}\left(E_{k}^{(\ell)%
\prime}E_{k}^{(\ell)} \big|\mathscr{D}\right) \right]M_{\Lambda_{k}^{0,(%
\ell)}}\mathbb{X}_{j,k}^{(\ell)\prime}\Lambda_{k}^{0,(\ell)}\left(%
\Lambda_{k}^{0,(\ell)\prime}\Lambda_{k}^{0,(\ell)}
\right)^{-1}\left(F^{0,(\ell)\prime}F^{0,(\ell)}
\right)^{-1}F^{0,(\ell)\prime}\right\}=o_{p}(1)$, }

\item[(ix)] {\small $\frac{1}{N_{k}^{(\ell)}T_{\ell}}\sum_{i\in
G_{k}^{(\ell)}}\sum_{t\in\mathcal{T}_{\ell}}\left[e_{it}^{2}\mathfrak{X}%
_{k,it}^{(\ell)}\mathfrak{X}_{k,it}^{(\ell)\prime}-\mathbb{E}\left(e_{it}^{2}%
\mathfrak{X}_{k,it}^{(\ell)}\mathfrak{X}_{k,it}^{(\ell)\prime}\big|%
\mathscr{D}\right)\right]=o_{p}(1)$, }

\item[(x)] {\small $\frac{1}{N_{k}^{(\ell)}T_{\ell}}\sum_{i\in
G_{k}^{(\ell)}}\sum_{t\in\mathcal{T}_{\ell}}e_{it}^{2}\left(\mathfrak{X}%
_{k,it}^{(\ell)}\mathfrak{X}_{k,it}^{(\ell)\prime}-\mathcal{X}_{it}\mathcal{X%
}_{it}^{\prime}\right)=o_{p}(1)$. }
\end{itemize}
\end{lemma}

\begin{proof}
{\small (i) We first show that $\left\Vert F^{0,(\ell )\prime }E_{k}^{(\ell
)\prime }\Lambda _{k}^{0,(\ell )}\right\Vert =O_{p}(\sqrt{NT})$. Note that 
\begin{align*}
& \mathbb{E}\left[ \left( \frac{\left\Vert F^{0,(\ell )\prime }E_{k}^{(\ell
)\prime }\Lambda _{k}^{0,(\ell )}\right\Vert _{{}}}{\sqrt{N_{k}^{(\ell
)}T_{\ell }}}\right) ^{2}\Bigg|\mathscr{D}\right] =\frac{1}{N_{k}^{(\ell
)}T_{\ell }}\mathbb{E}\left[ \left( \sum_{i\in G_{k}^{(\ell )}}\sum_{t\in 
\mathcal{T}_{\ell }}e_{it}f_{t}^{0\prime }\lambda _{i}^{0}\right) ^{2}\Bigg|%
\mathscr{D}\right] \\
& =\frac{1}{N_{k}^{(\ell )}T_{\ell }}\sum_{i_{1}\in G_{k}^{(\ell
)}}\sum_{i_{2}\in G_{k}^{(\ell )}}\sum_{t_{1}\in \mathcal{T}_{\ell
}}\sum_{t_{2}\in \mathcal{T}_{\ell }}\mathbb{E}\left(
e_{i_{1}t_{1}}e_{i_{2}t_{2}}|\mathscr{D}\right) f_{t_{1}}^{0\prime }\lambda
_{i_{1}}^{0}\lambda _{i_{2}}^{0\prime }f_{t_{2}}^{0} \\
& \leq \max_{i\in G_{k}^{(\ell )}}\left\Vert \lambda _{i}^{0}\right\Vert
_{2}^{2}\max_{t\in \mathcal{T}_{\ell }}\left\Vert f_{t}^{0}\right\Vert
_{2}^{2}\frac{1}{N_{k}^{(\ell )}T_{\ell }}\sum_{i\in G_{k}^{(\ell
)}}\sum_{t_{1}\in \mathcal{T}_{\ell }}\sum_{t_{2}\in \mathcal{T}_{\ell
}}\left\vert \mathbb{E}\left( e_{it_{1}}e_{it_{2}}|\mathscr{D}\right)
\right\vert \\
& \lesssim \frac{1}{N_{k}^{(\ell )}T_{\ell }}\sum_{i\in G_{k}^{(\ell
)}}\sum_{t\in \mathcal{T}_{\ell }}\left\vert Var\left( e_{it}|\mathscr{D}%
\right) \right\vert +\frac{2}{N_{k}^{(\ell )}T_{\ell }}\sum_{i\in
G_{k}^{(\ell )}}\sum_{t_{1}\in \mathcal{T}_{\ell }}\sum_{t_{2}\in \mathcal{T}%
_{\ell },t_{2}>t_{1}}\left\vert Cov\left( e_{it_{1}},e_{it_{2}}|\mathscr{D}%
\right) \right\vert \\
& =O(1)~a.s.,
\end{align*}%
where the fourth line is by Lemma \ref{Lem:bounded u&v}(i) and the last line
combines Assumption \ref{ass:1}(v) and Davydov's inequality for conditional
strong mixing sequences, similarly as \eqref{Lem:B12_1}. It follows that 
\begin{align}
\left\Vert P_{F^{0,(\ell )}}E_{k}^{(\ell )\prime }P_{\Lambda _{k}^{0,(\ell
)}}\right\Vert & \leq \left\Vert F^{0,(\ell )}\right\Vert \left\Vert
F^{0,(\ell )\prime }F^{0,(\ell )}\right\Vert \left\Vert F^{0,(\ell )\prime
}E_{k}^{(\ell )\prime }\Lambda _{k}^{0,(\ell )}\right\Vert \left\Vert
(\Lambda _{k}^{0,(\ell )\prime }\Lambda _{k}^{0,(\ell )})^{-1}\right\Vert
\left\Vert \Lambda _{k}^{0,(\ell )\prime }\right\Vert  \notag
\label{Lem:B13_1} \\
& =O(T^{1/2})O_{p}(T^{-1})O_{p}(\sqrt{NT})O_{p}(N^{-1})O(N^{1/2})=O_{p}(1),
\end{align}%
where the first equality holds by Assumptions \ref{ass:2} and \ref{ass:8}. }

{\small Moreover, we have 
\begin{equation}
\left\Vert P_{\Lambda _{k}^{0,(\ell )}}\tilde{\mathbb{X}}_{j,k}^{(\ell
)}\right\Vert \leq \left\Vert \Lambda _{k}^{0,(\ell )}\right\Vert \left\Vert
(\Lambda _{k}^{0,(\ell )\prime }\Lambda _{k}^{0,(\ell )})^{-1}\right\Vert
\left\Vert \Lambda _{k}^{0,(\ell )\prime }\tilde{\mathbb{X}}_{j,k}^{(\ell
)}\right\Vert =O(N^{1/2})O_{p}(N^{-1})O_{p}(\sqrt{NT})=O_{p}(T^{1/2}),
\label{Lem:B13_2}
\end{equation}%
where the first equality holds by Assumptions \ref{ass:2} and \ref{ass:8}(i)
and the fact that 
\begin{align*}
\mathbb{E}\left( \left\Vert \Lambda _{k}^{0,(\ell )\prime }\tilde{\mathbb{X}}%
_{j,k}^{(\ell )}\right\Vert ^{2}\big|\mathscr{D}\right) &
=\sum_{r=1}^{r_{0}}\sum_{t\in \mathcal{T}_{\ell }}\mathbb{E}\left[ \left(
\sum_{i\in G_{k}^{(\ell )}}\lambda _{i,r}^{0}\tilde{X}_{j,it}\right) ^{2}%
\Bigg|\mathscr{D}\right] \\
& =\sum_{r=1}^{r_{0}}\sum_{t\in \mathcal{T}_{\ell }}\sum_{i\in G_{k}^{(\ell
)}}\sum_{i^{\ast }\in G_{k}^{(\ell )}}\lambda _{i,r}^{0}\lambda _{i^{\ast
},r}^{0}\mathbb{E}[\tilde{X}_{j,it}\tilde{X}_{j,i^{\ast }t}\big|\mathscr{D}]
\\
& =\sum_{r=1}^{r_{0}}\sum_{t\in \mathcal{T}_{\ell }}\sum_{i\in G_{k}^{(\ell
)}}\left( \lambda _{i,r}^{0}\right) ^{2}\mathbb{E}\left[ (\tilde{X}%
_{j,it})^{2}\big|\mathscr{D}\right] =O_{p}(NT).
\end{align*}%
Then we are ready to show that 
\begin{align*}
\left\vert \frac{1}{\sqrt{N_{k}^{(\ell )}T_{\ell }}}tr\left( P_{F^{0,(\ell
)}}E_{k}^{(\ell )\prime }P_{\Lambda _{k}^{0,(\ell )}}\tilde{\mathbb{X}}%
_{j,k}^{(\ell )}\right) \right\vert & =\left\vert \frac{1}{\sqrt{%
N_{k}^{(\ell )}T_{\ell }}}tr\left( P_{F^{0,(\ell )}}E_{k}^{(\ell )\prime
}P_{\Lambda _{k}^{0,(\ell )}}P_{\Lambda _{k}^{0,(\ell )}}\tilde{\mathbb{X}}%
_{j,k}^{(\ell )}\right) \right\vert \\
& \leq \frac{1}{\sqrt{N_{k}^{(\ell )}T_{\ell }}}\left\Vert P_{F^{0,(\ell
)}}E_{k}^{(\ell )\prime }P_{\Lambda _{k}^{0,(\ell )}}\right\Vert \left\Vert
P_{\Lambda _{k}^{0,(\ell )}}\tilde{\mathbb{X}}_{j,k}^{(\ell )}\right\Vert \\
& =\frac{1}{\sqrt{N_{k}^{(\ell )}T_{\ell }}}%
O_{p}(1)O_{p}(T^{1/2})=O(N^{-1/2})=o_{p}(1).
\end{align*}%
}

{\small (ii) Let $\left[ A\right] _{jl}$ denote the $\left( j,l\right) $-th
element of $A.$ Note that 
\begin{align*}
& \left\vert \frac{1}{\sqrt{N_{k}^{(\ell )}T_{\ell }}}tr\left( P_{\Lambda
_{k}^{0,(\ell )}}E_{k}^{(\ell )\prime }\tilde{\mathbb{X}}_{j,k}^{(\ell
)}\right) \right\vert \\
& =\left\vert \sum_{j_{1},j_{2}=1}^{r_{0}}\left[ \left( \frac{1}{%
N_{k}^{(\ell )}}\Lambda _{k}^{0,(\ell )\prime }\Lambda _{k}^{0,(\ell
)}\right) ^{-1}\right] _{j_{1}j_{2}}\frac{1}{N_{k}^{(\ell )}\sqrt{%
N_{k}^{(\ell )}T_{\ell }}}\sum_{t\in \mathcal{T}_{\ell }}\sum_{i_{1}\in
G_{k}^{(\ell )}}\sum_{i_{2}\in G_{k}^{(\ell )}}\lambda
_{i_{1},j_{1}}^{0}\lambda _{i_{2},j_{2}}^{0}e_{i_{1}t}\tilde{X}%
_{j,i_{2}t}^{(\ell )}\right\vert \\
& \lesssim \max_{j_{1},j_{2}\in \lbrack r_{0}]}\left\vert \frac{1}{%
N_{k}^{(\ell )}\sqrt{N_{k}^{(\ell )}T_{\ell }}}\sum_{t\in \mathcal{T}_{\ell
}}\sum_{i_{1}\in G_{k}^{(\ell )}}\sum_{i_{2}\in G_{k}^{(\ell )}}\lambda
_{i_{1},j_{1}}^{0}\lambda _{i_{2},j_{2}}^{0}e_{i_{1}t}\tilde{X}%
_{j,i_{2}t}^{(\ell )}\right\vert =O_{p}(N^{-1/2}),
\end{align*}%
where the last line holds by the fact that 
\begin{align*}
& \mathbb{E}\left( \left\vert \frac{1}{N_{k}^{(\ell )}\sqrt{N_{k}^{(\ell
)}T_{\ell }}}\sum_{t\in \mathcal{T}_{\ell }}\sum_{i_{1}\in G_{k}^{(\ell
)}}\sum_{i_{2}\in G_{k}^{(\ell )}}\lambda _{i_{1},j_{1}}^{0}\lambda
_{i_{2},j_{2}}^{0}e_{i_{1}t}\tilde{X}_{j,i_{2}t}^{(\ell )}\right\vert ^{2}%
\Bigg|\mathscr{D}\right) \\
& =\frac{1}{\left( N_{k}^{(\ell )}\right) ^{3}T_{\ell }}\sum_{t\in \mathcal{T%
}_{\ell }}\sum_{s\in \mathcal{T}_{\ell }}\sum_{i_{1}\in G_{k}^{(\ell
)}}\sum_{i_{2}\in G_{k}^{(\ell )}}\sum_{m_{1}\in G_{k}^{(\ell
)}}\sum_{m_{2}\in G_{k}^{(\ell )}}\lambda _{i_{1},j_{1}}^{0}\lambda
_{i_{2},j_{2}}^{0}\lambda _{m_{1},j_{1}}^{0}\lambda _{m_{2},j_{2}}^{0}%
\mathbb{E}\left( e_{i_{1}t}\tilde{X}_{j,i_{2}t}^{(\ell )}e_{m_{1}s}\tilde{X}%
_{j,m_{2}s}^{(\ell )}\big|\mathscr{D}\right) \\
& =\frac{1}{\left( N_{k}^{(\ell )}\right) ^{3}T_{\ell }}\sum_{t\in \mathcal{T%
}_{\ell }}\sum_{s\in \mathcal{T}_{\ell }}\sum_{i_{1}\in G_{k}^{(\ell
)}}\sum_{i_{2}\in G_{k}^{(\ell )}}\left( \lambda _{i_{1},j_{1}}^{0}\right)
^{2}\left( \lambda _{i_{2},j_{2}}^{0}\right) ^{2}\mathbb{E}\left(
e_{i_{1}t}e_{i_{1}s}\tilde{X}_{j,i_{2}t}^{(\ell )}\tilde{X}%
_{j,i_{2}s}^{(\ell )}\big|\mathscr{D}\right) \\
& =\frac{1}{\left( N_{k}^{(\ell )}\right) ^{3}T_{\ell }}\sum_{t\in \mathcal{T%
}_{\ell }}\sum_{i_{1}\in G_{k}^{(\ell )}}\sum_{i_{2}\in G_{k}^{(\ell
)}}\left( \lambda _{i_{1},j_{1}}^{0}\right) ^{2}\left( \lambda
_{i_{2},j_{2}}^{0}\right) ^{2}\mathbb{E}\left( e_{i_{1}t}^{2}\left( \tilde{X}%
_{j,i_{2}t}^{(\ell )}\right) ^{2}\bigg|\mathscr{D}\right) \\
& +\frac{2}{\left( N_{k}^{(\ell )}\right) ^{3}T_{\ell }}\sum_{t\in \mathcal{T%
}_{\ell }}\sum_{s\in \mathcal{T}_{\ell },s>t}\sum_{i_{1}\in G_{k}^{(\ell
)}}\sum_{i_{2}\in G_{k}^{(\ell )}}\left( \lambda _{i_{1},j_{1}}^{0}\right)
^{2}\left( \lambda _{i_{2},j_{2}}^{0}\right) ^{2}\mathbb{E}\left(
e_{i_{1}t}e_{i_{1}s}\tilde{X}_{j,i_{2}t}^{(\ell )}\tilde{X}%
_{j,i_{2}s}^{(\ell )}\big|\mathscr{D}\right) \\
& =O_{p}(N^{-1}),
\end{align*}%
where the second equality is by Assumption \ref{ass:1}(i) and the last line
holds by Assumption \ref{ass:1}(iii) and (v), and Davydov's inequality. }

{\small (iii) Define $\zeta _{j,its}^{(\ell )}:=e_{it}\tilde{\mathbb{X}}%
_{j,is}^{(\ell )}-\mathbb{E}(e_{it}\tilde{\mathbb{X}}_{j,is}^{(\ell )}\big|%
\mathscr{D})$. As above, we have 
\begin{align*}
& \mathbb{E}\left\{ \left\vert \frac{1}{T_{\ell }\sqrt{N_{k}^{(\ell
)}T_{\ell }}}\sum_{i_{1}\in G_{k}^{(\ell )}}\sum_{t_{1}\in \mathcal{T}_{\ell
}}\sum_{t_{2}\in \mathcal{T}_{\ell
}}f_{t_{1},j_{1}}^{0}f_{t_{2},j_{2}}^{0}\zeta _{j,i_{1}t_{1}t_{2}}^{(\ell
)}\right\vert ^{2}\Bigg|\mathscr{D}\right\} \\
& =\frac{1}{T_{\ell }^{3}N_{k}^{(\ell )}}\sum_{i_{1}\in G_{k}^{(\ell
)}}\sum_{i_{2}\in G_{k}^{(\ell )}}\sum_{t_{1}\in \mathcal{T}_{\ell
}}\sum_{t_{2}\in \mathcal{T}_{\ell }}\sum_{s_{1}\in \mathcal{T}_{\ell
}}\sum_{s_{2}\in \mathcal{T}_{\ell
}}f_{t_{1},j_{1}}^{0}f_{t_{2},j_{2}}^{0}f_{s_{1},j_{1}}^{0}f_{s_{2},j_{2}}^{0}%
\mathbb{E}\left( \zeta _{j,i_{1}t_{1}t_{2}}^{(\ell )}\zeta
_{j,i_{2}s_{1}s_{2}}^{(\ell )}\big|\mathscr{D}\right\} \\
& \lesssim \frac{1}{T_{\ell }^{3}N_{k}^{(\ell )}}\sum_{i_{1}\in G_{k}^{(\ell
)}}\sum_{i_{2}\in G_{k}^{(\ell )}}\sum_{t_{1}\in \mathcal{T}_{\ell
}}\sum_{t_{2}\in \mathcal{T}_{\ell }}\sum_{s_{1}\in \mathcal{T}_{\ell
}}\sum_{s_{2}\in \mathcal{T}_{\ell }}\left\vert Cov\left( e_{i_{1}t_{1}}%
\tilde{X}_{j,i_{1}t_{2}}^{(\ell )},e_{i_{2}s_{1}}\tilde{X}%
_{j,i_{2}s_{2}}^{(\ell )}\big|\mathscr{D}\right) \right\vert \\
& =\frac{1}{T_{\ell }^{3}N_{k}^{(\ell )}}\sum_{i\in G_{k}^{(\ell
)}}\sum_{t_{1}\in \mathcal{T}_{\ell }}\sum_{t_{2}\in \mathcal{T}_{\ell
}}\sum_{s_{1}\in \mathcal{T}_{\ell }}\sum_{s_{2}\in \mathcal{T}_{\ell
}}\left\vert Cov\left( e_{it_{1}}\tilde{X}_{j,it_{2}}^{(\ell )},e_{is_{1}}%
\tilde{X}_{j,is_{2}}^{(\ell )}\big|\mathscr{D}\right) \right\vert
=O_{p}(T^{-1}),
\end{align*}%
where the last equality holds by Assumption \ref{ass:10}(iv). It follows
that 
\begin{align*}
& \left\vert \frac{1}{\sqrt{N_{k}^{(\ell )}T_{\ell }}}tr\left\{
P_{F^{0,(\ell )}}\left[ E_{k}^{(\ell )\prime }\tilde{\mathbb{X}}%
_{j,k}^{(\ell )}-\mathbb{E}\left( E_{k}^{(\ell )\prime }\tilde{\mathbb{X}}%
_{j,k}^{(\ell )}\big|\mathscr{D}\right) \right] \right\} \right\vert \\
& =\left\vert \sum_{j_{1},j_{2}=1}^{r_{0}}\left[ \left( \frac{1}{T_{\ell }}%
F^{0,(\ell )\prime }F^{0,(\ell )}\right) ^{-1}\right] _{j_{1}j_{2}}\frac{1}{%
T_{\ell }\sqrt{N_{k}^{(\ell )}T_{\ell }}}\sum_{i_{1}\in G_{k}^{(\ell
)}}\sum_{t_{1}\in \mathcal{T}_{1}}\sum_{t_{2}\in \mathcal{T}%
_{1}}f_{t_{1},j_{1}}^{0}f_{t_{2},j_{2}}^{0}\zeta _{j,i_{1}t_{1}t_{2}}^{(\ell
)}\right\vert =O_{p}(T^{-1/2}).
\end{align*}%
}

{\small (iv) As in \cite{moon2017dynamic}, it is clear that 
\begin{align*}
& \left\vert \frac{1}{\sqrt{N_{k}^{(\ell )}T_{\ell }}}tr\left[ E_{k}^{(\ell
)}P_{F^{0,(\ell )}}E_{k}^{(\ell )\prime }M_{\Lambda _{k}^{0,(\ell )}}\mathbb{%
X}_{j,k}^{(\ell )}F^{0,(\ell )}\left( F^{0,(\ell )\prime }F^{0,(\ell
)}\right) ^{-1}\left( \Lambda _{k}^{0,(\ell )\prime }\Lambda _{k}^{0,(\ell
)}\right) ^{-1}\Lambda _{k}^{0,(\ell )\prime }\right] \right\vert \\
& \lesssim \frac{1}{\sqrt{N_{k}^{(\ell )}T_{\ell }}}\left\Vert P_{\Lambda
_{k}^{0,(\ell )}}E_{k}^{(\ell )}P_{F^{0,(\ell )}}\right\Vert \left\Vert
E_{k}^{(\ell )}\right\Vert _{op}\left\Vert \mathbb{X}_{j,k}^{(\ell
)}\right\Vert \left\Vert F^{0,(\ell )}\left( F^{0,(\ell )\prime }F^{0,(\ell
)}\right) ^{-1}\left( \Lambda _{k}^{0,(\ell )\prime }\Lambda _{k}^{0,(\ell
)}\right) ^{-1}\Lambda _{k}^{0,(\ell )\prime }\right\Vert _{{}} \\
& =\frac{1}{\sqrt{N_{k}^{(\ell )}T_{\ell }}}O_{p}(1)O_{p}\left( \sqrt{N}+%
\sqrt{T\log T}\right) O_{p}((NT)^{1/2})O_{p}((NT)^{-1/2})=o_{p}(1),
\end{align*}%
where the last line combines \eqref{Lem:B12_2},\eqref{Lem:B13_1}, the fact
that $||\mathbb{X}_{j,k}^{(\ell )}||=O_{p}((NT)^{1/2})$ by Assumption \ref%
{ass:8}(ii), and $||F^{0,(\ell )}(F^{0,(\ell )\prime }F^{0,(\ell
)})^{-1}(\Lambda _{k}^{0,(\ell )\prime }\Lambda _{k}^{0,(\ell
)})^{-1}\Lambda _{k}^{0,(\ell )\prime }||=O_{p}((NT)^{-1/2})$ by Assumptions %
\ref{ass:2} and \ref{ass:8}(i). }

{\small (v) The proof of (v) is analogous to that of (iv) and is omitted for
brevity. }

{\small (vi) First, we note that 
\begin{align*}
\mathbb{E}\left( \left\Vert \Lambda _{k}^{0,(\ell )\prime }E_{k}^{(\ell )}%
\mathbb{X}_{j,k}^{(\ell )\prime }\right\Vert \bigg|\mathscr{D}\right) & =%
\mathbb{E}\left[ \sum_{j_{1}=1}^{r_{0}}\sum_{m\in G_{k}^{(\ell )}}\left(
\sum_{i_{1}\in G_{k}^{(\ell )}}\sum_{t_{1}\in \mathcal{T}_{\ell }}\lambda
_{i_{1},j_{1}}^{0}e_{i_{1}t_{1}}X_{j,mt_{1}}\right) ^{2}\Bigg|\mathscr{D}%
\right] \\
& =\sum_{j_{1}=1}^{r_{0}}\sum_{m\in G_{k}^{(\ell )}}\sum_{i_{1}\in
G_{k}^{(\ell )}}\sum_{i_{2}\in G_{k}^{(\ell )}}\sum_{t_{1}\in \mathcal{T}%
_{\ell }}\sum_{t_{2}\in \mathcal{T}_{\ell }}\lambda
_{i_{1},j_{1}}^{0}\lambda _{i_{2},j_{1}}^{0}\mathbb{E}\left(
e_{i_{1}t_{1}}X_{j,mt_{1}}e_{i_{2}t_{2}}X_{j,mt_{2}}|\mathscr{D}\right) \\
& =\sum_{j_{1}=1}^{r_{0}}\sum_{m\in G_{k}^{(\ell )}}\sum_{i_{1}\in
G_{k}^{(\ell )}}\sum_{t_{1}\in \mathcal{T}_{\ell }}\sum_{t_{2}\in \mathcal{T}%
_{\ell }}\left( \lambda _{i_{1},j_{1}}^{0}\right) ^{2}\mathbb{E}\left(
e_{i_{1}t_{1}}X_{j,mt_{1}}e_{i_{1}t_{2}}X_{j,mt_{2}}|\mathscr{D}\right) \\
& \lesssim \sum_{m\in G_{k}^{(\ell )}}\sum_{i_{1}\in G_{k}^{(\ell
)}}\sum_{t\in \mathcal{T}_{\ell }}\mathbb{E}\left(
e_{i_{1}t}^{2}X_{j,mt}^{2}|\mathscr{D}\right) \\
& +2\sum_{m\in G_{k}^{(\ell )}}\sum_{i_{1}\in G_{k}^{(\ell )}}\sum_{t_{1}\in 
\mathcal{T}_{\ell }}\sum_{t_{2}\in \mathcal{T}_{\ell
},t_{2}>t_{1}}\left\vert Cov\left(
e_{i_{1}t_{1}}X_{j,mt_{1}},e_{i_{1}t_{2}}X_{j,mt_{2}}|\mathscr{D}\right)
\right\vert \\
& =O_{p}(N^{2}T),
\end{align*}%
which leads to the result that 
\begin{equation*}
\left\Vert P_{\Lambda _{k}^{0,(\ell )}}E_{k}^{(\ell )}\mathbb{X}%
_{j,k}^{(\ell )\prime }\right\Vert \leq \left\Vert \Lambda _{k}^{0,(\ell
)}\right\Vert \left\Vert (\Lambda _{k}^{0,(\ell )\prime }\Lambda
_{k}^{0,(\ell )})^{-1}\right\Vert \left\Vert \Lambda _{k}^{0,(\ell )\prime
}E_{k}^{(\ell )}\mathbb{X}_{j,k}^{(\ell )\prime }\right\Vert
=O(N^{-1/2})O_{p}(N\sqrt{T})=O_{p}(\sqrt{NT}).
\end{equation*}%
As in the proof of part (iv), it yields that 
\begin{align*}
& \left\vert \frac{1}{\sqrt{N_{k}^{(\ell )}T_{\ell }}}tr\left[ E_{k}^{(\ell
)\prime }M_{\Lambda _{k}^{0,(\ell )}}\mathbb{X}_{j,k}^{(\ell )}M_{F^{0,(\ell
)}}E_{k}^{(\ell )\prime }\Lambda _{k}^{0,(\ell )}\left( \Lambda
_{k}^{0,(\ell )\prime }\Lambda _{k}^{0,(\ell )}\right) ^{-1}\left(
F^{0,(\ell )\prime }F^{0,(\ell )}\right) ^{-1}F^{0,(\ell )\prime }\right]
\right\vert \\
& \leq \left\vert \frac{1}{\sqrt{N_{k}^{(\ell )}T_{\ell }}}tr\left[
E_{k}^{(\ell )\prime }M_{\Lambda _{k}^{0,(\ell )}}\mathbb{X}_{j,k}^{(\ell
)}E_{k}^{(\ell )\prime }P_{\Lambda _{k}^{0,(\ell )}}\Lambda _{k}^{0,(\ell
)}\left( \Lambda _{k}^{0,(\ell )\prime }\Lambda _{k}^{0,(\ell )}\right)
^{-1}\left( F^{0,(\ell )\prime }F^{0,(\ell )}\right) ^{-1}F^{0,(\ell )\prime
}\right] \right\vert \\
& +\left\vert \frac{1}{\sqrt{N_{k}^{(\ell )}T_{\ell }}}tr\left[ E_{k}^{(\ell
)\prime }M_{\Lambda _{k}^{0,(\ell )}}\mathbb{X}_{j,k}^{(\ell )}P_{F^{0,(\ell
)}}E_{k}^{(\ell )\prime }P_{\Lambda _{k}^{0,(\ell )}}\Lambda _{k}^{0,(\ell
)}\left( \Lambda _{k}^{0,(\ell )\prime }\Lambda _{k}^{0,(\ell )}\right)
^{-1}\left( F^{0,(\ell )\prime }F^{0,(\ell )}\right) ^{-1}F^{0,(\ell )\prime
}\right] \right\vert \\
& \lesssim \frac{1}{\sqrt{N_{k}^{(\ell )}T_{\ell }}}\left\Vert E_{k}^{(\ell
)}\right\Vert _{op}\left\Vert P_{\Lambda _{k}^{0,(\ell )}}E_{k}^{(\ell )}%
\mathbb{X}_{j,k}^{(\ell )\prime }\right\Vert \left\Vert \Lambda
_{k}^{0,(\ell )}\left( \Lambda _{k}^{0,(\ell )\prime }\Lambda _{k}^{0,(\ell
)}\right) ^{-1}\left( F^{0,(\ell )\prime }F^{0,(\ell )}\right)
^{-1}F^{0,(\ell )\prime }\right\Vert \\
& +\frac{1}{\sqrt{N_{k}^{(\ell )}T_{\ell }}}\left\Vert E_{k}^{(\ell
)}\right\Vert _{op}\left\Vert \mathbb{X}_{j,k}^{(\ell )}\right\Vert
\left\Vert \Lambda _{k}^{0,(\ell )}\left( \Lambda _{k}^{0,(\ell )\prime
}\Lambda _{k}^{0,(\ell )}\right) ^{-1}\left( F^{0,(\ell )\prime }F^{0,(\ell
)}\right) ^{-1}F^{0,(\ell )\prime }\right\Vert \left\Vert P_{\Lambda
_{k}^{0,(\ell )}}E_{k}^{(\ell )}P_{F^{0,(\ell )}}\right\Vert \\
& =o_{p}(1).
\end{align*}%
}

{\small (vii) For this statement, we sketch the proof because \cite%
{lu2016shrinkage} have already proved a similar result. 
\begin{align*}
& \left\vert \frac{1}{\sqrt{N_{k}^{(\ell )}T_{\ell }}}tr\left\{ \left[
E_{k}^{(\ell )}E_{k}^{(\ell )\prime }-\mathbb{E}\left( E_{k}^{(\ell
)}E_{k}^{(\ell )\prime }\big|\mathscr{D}\right) \right] M_{\Lambda
_{k}^{0,(\ell )}}\mathbb{X}_{j,k}^{(\ell )}F^{0,(\ell )}\left( F^{0,(\ell
)\prime }F^{0,(\ell )}\right) ^{-1}\left( \Lambda _{k}^{0,(\ell )\prime
}\Lambda _{k}^{0,(\ell )}\right) ^{-1}\Lambda _{k}^{0,(\ell )\prime
}\right\} \right\vert \\
& \lesssim \frac{1}{\left( N_{k}^{(\ell )}\right) ^{3/2}}\left\Vert \Lambda
_{k}^{0,(\ell )\prime }\frac{1}{T_{\ell }}\left[ E_{k}^{(\ell )}E_{k}^{(\ell
)\prime }-\mathbb{E}(E_{k}^{(\ell )}E_{k}^{(\ell )\prime }\big|\mathscr{D})%
\right] M_{\Lambda _{k}^{0,(\ell )}}\mathbb{X}_{j,k}^{(\ell )}\right\Vert
=o_{p}(1),
\end{align*}%
where the last equality holds by the fact that 
\begin{equation*}
\left( N_{k}^{(\ell )}\right) ^{-3/2}\mathbb{E}\left\{ \left\Vert \Lambda
_{k}^{0,(\ell )\prime }\frac{1}{T_{\ell }}\left[ E_{k}^{(\ell )}E_{k}^{(\ell
)\prime }-\mathbb{E}\left( E_{k}^{(\ell )}E_{k}^{(\ell )\prime }\big|%
\mathscr{D}\right) \right] M_{\Lambda _{k}^{0,(\ell )}}\mathbb{X}%
_{j,k}^{(\ell )}\right\Vert \bigg|\mathscr{D}\right\} =o_{p}(1)
\end{equation*}%
which follows by similar arguments as used in the proof of Lemma D.3(vi) in 
\cite{lu2016shrinkage}. }

{\small (viii) Analogous to the previous statement, we have 
\begin{equation*}
\left( T_{\ell }\right) ^{-3/2}\mathbb{E}\left\{ \left\Vert F^{0,(\ell
)\prime }\frac{1}{N_{k}^{(\ell )}}\left[ E_{k}^{(\ell )\prime }E_{k}^{(\ell
)}-\mathbb{E}\left( E_{k}^{(\ell )\prime }E_{k}^{(\ell )}\big|\mathscr{D}%
\right) \right] M_{\Lambda _{k}^{0,(\ell )}}\mathbb{X}_{j,k}^{(\ell )\prime
}\right\Vert _{{}}\bigg|\mathscr{D}\right\} =o_{p}(1)
\end{equation*}%
by similar arguments as used in the proof of Lemma D.4(iii) in \cite%
{lu2016shrinkage}. Then we are ready to show that 
\begin{align*}
& \left\vert \frac{1}{\sqrt{N_{k}^{(\ell )}T_{\ell }}}tr\left\{ \left[
E_{k}^{(\ell )\prime }E_{k}^{(\ell )}-\mathbb{E}\left( E_{k}^{(\ell )\prime
}E_{k}^{(\ell )}\big|\mathscr{D}\right) \right] M_{\Lambda _{k}^{0,(\ell )}}%
\mathbb{X}_{j,k}^{(\ell )\prime }\Lambda _{k}^{0,(\ell )}\left( \Lambda
_{k}^{0,(\ell )\prime }\Lambda _{k}^{0,(\ell )}\right) ^{-1}\left(
F^{0,(\ell )\prime }F^{0,(\ell )}\right) ^{-1}F^{0,(\ell )\prime }\right\}
\right\vert \\
& \lesssim \left( T_{\ell }\right) ^{-3/2}\left\Vert F^{0,(\ell )\prime }%
\frac{1}{N_{k}^{(\ell )}}\left[ E_{k}^{(\ell )\prime }E_{k}^{(\ell )}-%
\mathbb{E}\left( E_{k}^{(\ell )\prime }E_{k}^{(\ell )}\big|\mathscr{D}%
\right) \right] M_{\Lambda _{k}^{0,(\ell )}}\mathbb{X}_{j,k}^{(\ell )\prime
}\right\Vert =o_{p}(1).
\end{align*}%
}

{\small (ix) This statement can be proved owing to the fact that the second
moment of the term on the left side of the equality conditioning on $%
\mathscr{D}$ is $O_{p}(N^{-1})$. See the proof of Lemma B.1(i) for detail. }

{\small (x) Similarly to \eqref{Lem:B13_2}, we can also show that $||\tilde{%
\mathbb{X}}_{j,k}^{(\ell )}P_{F^{0,(\ell )}}||=O_{p}(N^{1/2})$. Then,
following the same arguments as used in the proof of Lemma B.1(j) in \cite%
{moon2017dynamic}, we can finish the proof. }
\end{proof}

\section{\protect\small Estimation of Panels with IFEs and Heterogeneous
Slopes}

{\small \label{sec:panel_IFE} For $\forall i\in \mathcal{N}:=\{n_{1},\cdots
,n_{n}\}$ and $t\in \lbrack T]$, consider the model 
\begin{equation}
Y_{it}=\left\{ \begin{aligned}
&\lambda_{i}^{0\prime}f_{t}^{0}+X_{it}^{\prime}\theta_{i}^{0,(1)}+e_{it},~t%
\in\{1,\cdots,T_{1}\},\\
&\lambda_{i}^{0\prime}f_{t}^{0}+X_{it}^{\prime}\theta_{i}^{0,(2)}+e_{it},~t%
\in\{T_{1}+1,\cdots,T\}. \end{aligned}\right.  \label{PIFE_model}
\end{equation}%
Here $\mathcal{N}$ is a subset of $[N]$ and $n\asymp N$. To distinguish from
the notation $\Lambda ^{0}$ in the paper, we define $\Lambda
_{n}^{0}:=\left( \lambda _{n_{1}}^{0},\cdots ,\lambda _{n_{n}}^{0}\right)
^{\prime }$. }

{\small Let $X_{i}^{(1)}=\left( X_{i1},\cdots ,X_{iT_{1}}\right) ^{\prime }$%
, $X_{i}^{(2)}=\left( X_{i(T_{1}+1)},\cdots ,X_{iT}\right) ^{\prime }$, $%
e_{i}^{(1)}=\left( e_{i1},\cdots ,e_{iT_{1}}\right) ^{\prime }$, $%
e_{i}^{(2)}=\left( e_{i(T_{1}+1)},\cdots ,e_{iT}\right) ^{\prime }$, $%
F^{0,(1)}=\left( f_{1}^{0},\cdots ,f_{T_{1}}^{0}\right) ^{\prime }$, and $%
F^{0,(2)}=\left( f_{T_{1}+1}^{0},\cdots ,f_{T}^{0}\right) ^{\prime }$. To
estimate $\theta _{i}^{0,(\ell )}$, $\lambda _{i}^{0}$ and $f_{t}^{0}$, we
follow the lead of \cite{bai2009panel} and consider the PCA for
heterogeneous panels. For $\forall \ell \in \{1,2\}$, let 
\begin{equation}
\left( \left\{ \hat{\theta}_{i}^{(\ell )}\right\} _{i\in \mathcal{N}},\hat{F}%
^{(\ell )}\right) =\argmin_{F^{(\ell )},\left\{ \theta _{i}\right\} _{i\in 
\mathcal{N}}}\frac{1}{nT_{\ell }}\sum_{_{i\in \mathcal{N}}}\left(
Y_{i}^{(\ell )}-X_{i}^{(\ell )}\theta _{i}\right) ^{\prime }M_{F^{(\ell
)}}\left( Y_{i}^{(\ell )}-X_{i}^{(\ell )}\theta _{i}\right) ,
\label{PIFE_obj}
\end{equation}%
where $T_{2}=T-T_{1}$, $W_{i}^{(1)}=\left( W_{i1},\cdots ,W_{iT_{1}}\right)
^{\prime }$, $W_{i}^{(2)}=\left( W_{i(T_{1}+1)},\cdots ,W_{iT}\right)
^{\prime }$ for $W_{i}$ denotes $Y_{i}$ or $X_{i}$, $F^{(\ell )}$ is any $%
T_{\ell }\times r_{0}$ matrix such that $\frac{F^{(\ell )\prime }F^{(\ell )}%
}{T_{\ell }}=I_{r_{0}}$ and $M_{F^{(\ell )}}=I_{T_{\ell }}-\frac{F^{(\ell
)}F^{(\ell )\prime }}{T_{\ell }}$. Note that we consider the concentrated
objective function here by concentrating out the factor loadings. The
solutions to the minimization problem in \eqref{PIFE_obj} solve the
following nonlinear system of equations: 
\begin{align}
& \hat{\theta}_{i}^{(\ell )}=\left( X_{i}^{(\ell )\prime }M_{\hat{F}^{(\ell
)}}X_{i}^{(\ell )}\right) ^{-1}X_{i}^{(\ell )\prime }M_{\hat{F}^{(\ell
)}}Y_{i}^{(\ell )},  \label{PIFE_est_theta} \\
& \left[ \frac{1}{nT_{\ell }}\sum_{_{i\in \mathcal{N}}}\left( Y_{i}^{(\ell
)}-X_{i}^{(\ell )}\hat{\theta}_{i}^{(\ell )}\right) ^{\prime }\left(
Y_{i}^{(\ell )}-X_{i}^{(\ell )}\hat{\theta}_{i}^{(\ell )}\right) \right] 
\hat{F}^{(\ell )}=\hat{F}^{(\ell )}\hat{V}_{NT}^{(\ell )},
\label{PIFE_est_F}
\end{align}%
where $\hat{V}_{NT}^{(\ell )}$ is a diagonal matrix that contains the $r_{0}$
largest eigenvalues of the matrix in the square brackets in (\ref{PIFE_est_F}%
). Let $\hat{\lambda}_{i}^{(\ell )}=\frac{1}{T}\hat{F}^{(\ell )\prime
}(Y_{i}^{(\ell )}-X_{i}^{(\ell )}\hat{\theta}_{i}^{(\ell )}),$ which are
estimates of $\lambda _{i}^{0}.$ Let $\hat{\Lambda}_{n}^{(\ell )}:=(\hat{%
\lambda}_{n_{1}}^{(\ell )},\cdots ,\hat{\lambda}_{n_{n}}^{(\ell )})^{\prime
},$ and $\hat{a}_{ii}^{(\ell )}:=\hat{\lambda}_{i}^{(\ell )\prime }\left( 
\frac{\hat{\Lambda}_{n}^{(\ell )\prime }\hat{\Lambda}_{n}^{(\ell )}}{n}%
\right) ^{-1}\hat{\lambda}_{i}^{(\ell )}$. }

{\small Let $\theta _{i}^{0,(\ell )}=\bar{\theta}^{0,(\ell )}+c_{i}^{(\ell
)},$ where $\bar{\theta}^{0,(\ell )}=\frac{1}{n}\sum_{i\in \mathcal{N}%
}\theta _{i}^{0,(\ell )}$. Here, we consider testing the slope homogeneity
for $i\in \mathcal{N}.$ The null and alternative hypotheses are respectively
given by 
\begin{equation*}
H_{0}:c_{i}^{(\ell )}=0\text{ }\forall i\in \mathcal{N}\text{ and }%
H_{1}:c_{i}^{(\ell )}\neq 0~\text{for some }i\in \mathcal{N}.
\end{equation*}%
Following \cite{pesaran2008testing} and \cite{ando2016panel}, we define 
\begin{equation}
\hat{\Gamma}^{(\ell )}=\sqrt{n}\cdot \frac{\frac{1}{n}\sum_{i\in
\{n_{1}\cdots ,n_{n}\}}\hat{\mathbb{S}}_{i}^{(\ell )}-p}{\sqrt{2p}}
\label{PIFE_test}
\end{equation}%
where 
\begin{align*}
& \hat{\mathbb{S}}_{i}^{(\ell )}=T_{\ell }(\hat{\theta}_{i}^{(\ell )}-\hat{%
\bar{\theta}}^{(\ell )})^{\prime }\hat{S}_{ii}^{(\ell )}(\hat{\Omega}%
_{i}^{(\ell )})^{-1}\hat{S}_{ii}^{(\ell )}(\hat{\theta}_{i}^{(\ell )}-\hat{%
\bar{\theta}}^{(\ell )})(1-\hat{a}_{ii}^{(\ell )}/n)^{2},\quad \hat{\bar{%
\theta}}^{(\ell )}=\frac{1}{n}\sum_{i\in \mathcal{N}}\hat{\theta}_{i}^{(\ell
)}, \\
& M_{\hat{F}^{(\ell )}}=I_{T_{\ell }}-\frac{\hat{F}^{(\ell )}\hat{F}^{(\ell
)\top }}{T_{\ell }},\quad \hat{S}_{ii}^{(\ell )}=\frac{X_{i}^{(\ell )\prime
}M_{\hat{F}^{(\ell )}}X_{i}^{(\ell )}}{T_{\ell }},\quad (\hat{\mathfrak{x}}%
_{it}^{(\ell )})^{\prime }\text{ is the $t$-th row of $M_{\hat{F}^{(\ell
)}}X_{i}^{(\ell )}$,} \\
& \hat{\Omega}_{i}^{(\ell )}=\frac{1}{T_{\ell }}\sum_{t\in \mathcal{T}_{\ell
}}\hat{\mathfrak{x}}_{it}^{(\ell )}\hat{\mathfrak{x}}_{it}^{(\ell )\prime }%
\hat{e}_{it}^{2}+\frac{1}{T_{\ell }}\sum_{j\in \mathcal{T}_{\ell
,-1}}k(j/S_{T})\sum_{t\in \mathcal{T}_{\ell ,j}}[\hat{\mathfrak{x}}%
_{it}^{(\ell )}\hat{\mathfrak{x}}_{i,t+j}^{(\ell )\prime }\hat{e}_{it}\hat{e}%
_{i,t+j}+\hat{\mathfrak{x}}_{i,t-j}^{(\ell )}\hat{\mathfrak{x}}_{it}^{(\ell
)\prime }\hat{e}_{i,t-j}\hat{e}_{it}]
\end{align*}%
and recall that $\mathcal{T}_{1}=[T_{1}]$, $\mathcal{T}_{2}=[T]\backslash
\lbrack T_{1}]$, $\mathcal{T}_{1,-1}=\mathcal{T}_{1}\backslash \{T_{1}\}$, $%
\mathcal{T}_{2,-1}=\mathcal{T}_{2}\backslash \{T\}$, $\mathcal{T}%
_{1,j}=\left\{ 1+j,\cdots ,T_{1}\right\} $, and $\mathcal{T}_{2,j}=\left\{
T_{1}+1+j,\cdots ,T\right\} $ for some specific $j\in \mathcal{T}_{\ell ,-1}$%
. }

{\small In the next section, we study the asymptotic distribution of $\hat{%
\theta}_{i}^{(\ell )},$ the uniform convergence rates for the estimators of
factors and factor loadings, and the asymptotic behavior for $\hat{\Gamma}%
^{(\ell )}$ under $H_{0}$ and $H_{1}$, respectively. }

\section{\protect\small Lemmas for Panel IFEs Model with Heterogeneous Slope}

{\small \label{sec:proof_IFE} Below we derive the asymptotic distribution
for the slope estimators in our heterogeneous panel models which allow for
dynamics. To allow the dynamic panel, we focus on Assumption \ref{ass:1*}
where the error process is an m.d.s.. If we focus on Assumption \ref{ass:1},
we can obtain similar results by using Davydov's inequality for strong
mixing errors. Here we skip the analyses for static panels with serially
correlated errors for brevity. Let $M$ be a generic large positive constant
and $\mathcal{F}^{(\ell )}:=\left\{ F^{(\ell )}\in \mathbb{R}^{T_{\ell
}\times r_{0}}:\frac{F^{(\ell )\prime }F^{(\ell )}}{T_{\ell }}%
=I_{r_{0}}\right\} $. }

\begin{lemma}
{\small \label{Lem:D1} Under Assumptions \ref{ass:1*}, \ref{ass:2} and \ref%
{ass:8}, we have }

\begin{itemize}
\item[(i)] {\small $\left\vert \frac{1}{nT_{\ell }}\sum_{i\in \mathcal{N}%
}e_{i}^{(\ell )\prime }P_{F^{0,(\ell )}}e_{i}^{(\ell )}\right\vert =o_{p}(1)$%
, }

\item[(ii)] {\small $\sup_{F^{(\ell )}\in \mathcal{F}^{(\ell )}}\left\vert 
\frac{1}{nT_{\ell }}\sum_{i\in \mathcal{N}}e_{i}^{(\ell )\prime }P_{F^{(\ell
)}}e_{i}^{(\ell )}\right\vert =o_{p}(1)$, }

\item[(iii)] {\small $\sup_{F^{(\ell )}\in \mathcal{F}^{(\ell )}}\left\vert 
\frac{1}{nT_{\ell }}\sum_{i\in \mathcal{N}}\lambda _{i}^{0\prime }F^{0,(\ell
)\prime }M_{F^{(\ell )}}e_{i}^{(\ell )}\right\vert =o_{p}(1)$, }

\item[(iv)] {\small $\sup_{\left\{ \max_{i}\left\Vert \theta _{i}\right\Vert
_{\max }\leq M\right\} ,F^{(\ell )}\in \mathcal{F}^{(\ell )}}\left\vert 
\frac{1}{n}\sum_{i\in \mathcal{N}}(\theta _{i}-\theta _{i}^{0,(\ell
)})^{\prime }\frac{X_{i}^{(\ell )\prime }M_{F^{(\ell )}}e_{i}^{(\ell )}}{%
T_{\ell }}\right\vert =o_{p}(1)$. }
\end{itemize}
\end{lemma}

\begin{proof}
{\small (i) We notice that 
\begin{equation*}
\left\vert \frac{1}{nT_{\ell }}\sum_{i\in \mathcal{N}}e_{i}^{(\ell )\prime
}P_{F^{0,(\ell )}}e_{i}^{(\ell )}\right\vert \leq \frac{1}{T_{\ell }}\left( 
\frac{1}{n}\sum_{i\in \mathcal{N}}\left\Vert \frac{1}{\sqrt{T_{\ell }}}%
\sum_{t\in \mathcal{T}_{\ell }}f_{t}^{0}e_{it}\right\Vert ^{2}\right)
\left\Vert \left( \frac{F^{0,(\ell )\prime }F^{0,(\ell )}}{T_{\ell }}\right)
^{-1}\right\Vert \lesssim \frac{1}{T_{\ell }}\left( \frac{1}{n}\sum_{i\in 
\mathcal{N}}\left\Vert \frac{1}{\sqrt{T_{\ell }}}\sum_{t\in \mathcal{T}%
_{\ell }}f_{t}^{0}e_{it}\right\Vert ^{2}\right) .
\end{equation*}%
Recall that $\mathscr{D}$ denotes the minimum $\sigma $-fields generated by $%
\left\{ V_{j}^{0}\right\} _{j\in \lbrack p]\cup \{0\}}$. Furthermore, we
observe that 
\begin{align}
\mathbb{E}\left( \frac{1}{n}\sum_{i\in \mathcal{N}}\left\Vert \frac{1}{\sqrt{%
T_{\ell }}}\sum_{t\in \mathcal{T}_{\ell }}f_{t}^{0}e_{it}\right\Vert ^{2}%
\Bigg|\mathscr{D}\right) & \leq \frac{1}{nT_{\ell }}\sum_{i\in \mathcal{N}%
}\sum_{t\in \mathcal{T}_{\ell }}\sum_{s\in \mathcal{T}_{\ell }}\left\Vert 
\mathbb{E}\left( f_{t}^{0}f_{s}^{0\prime }e_{it}e_{is}\big|\mathscr{D}%
\right) \right\Vert   \notag \\
& \leq \frac{1}{nT_{\ell }}\sum_{i\in \mathcal{N}}\sum_{t\in \mathcal{T}%
_{\ell }}\sum_{s\in \mathcal{T}_{\ell }}\left\Vert f_{t}^{0}\right\Vert
\left\Vert f_{s}^{0\prime }\right\Vert \left\vert \mathbb{E}\left(
e_{it}e_{is}\big|\mathscr{D}\right) \right\vert   \notag \\
& \lesssim \frac{1}{nT_{\ell }}\sum_{i\in \mathcal{N}}\sum_{t\in \mathcal{T}%
_{\ell }}\mathbb{E}(e_{it}^{2}\big|\mathscr{D})\leq M~a.s.,  \label{Lem:D1_1}
\end{align}%
%
%
%
%
%
%
%
%
%
%
%
%
%
where the fourth line holds by the boundedness of factors shown in Lemma \ref%
{Lem:bounded u&v}(i) and the conditional independence of $e_{it}$ under
Assumption \ref{ass:1*}(i) and (iii). It follows that $\frac{1}{n}\sum_{i\in 
\mathcal{N}}||\frac{1}{\sqrt{T_{\ell }}}\sum_{t\in \mathcal{T}_{\ell
}}f_{t}^{0}e_{it}||^{2}=O_{p}(1)$ and }${\small |}${\small $\frac{1}{%
nT_{\ell }}\sum_{i\in \mathcal{N}}e_{i}^{(\ell )\prime }P_{F^{0,(\ell
)}}e_{i}^{(\ell )}|=O_{p}(T^{-1})$. }

{\small (ii) Noting that $P_{F^{(\ell )}}=F^{(\ell )}\left( F^{(\ell )\prime
}F^{(\ell )}\right) ^{-1}F^{(\ell )\prime }=T^{-1}F^{(\ell )}F^{(\ell
)\prime }$ for $F^{(\ell )}\in \mathcal{F}^{(\ell )},$ we have $%
\sup_{F^{(\ell )}\in \mathcal{F}^{(\ell )}}|\frac{1}{nT_{\ell }}\sum_{i\in 
\mathcal{N}}$ $e_{i}^{(\ell )\prime }P_{F^{(\ell )}}e_{i}^{(\ell )}|$ $\leq
\sup_{F^{(\ell )}\in \mathcal{F}^{(\ell )}}\frac{1}{n}\sum_{i\in \mathcal{N}%
}||\frac{1}{T_{\ell }}\sum_{t\in \mathcal{T}_{\ell }}f_{t}e_{it}||^{2}.$
Next, 
\begin{align}
& \sup_{F^{(\ell )}\in \mathcal{F}^{(\ell )}}\frac{1}{n}\sum_{i\in \mathcal{N%
}}\left\Vert \frac{1}{T_{\ell }}\sum_{t\in \mathcal{T}_{\ell
}}f_{t}e_{it}\right\Vert ^{2}=\sup_{F^{(\ell )}\in \mathcal{F}^{(\ell
)}}tr\left( \frac{1}{n}\sum_{i\in \mathcal{N}}\frac{1}{T_{\ell }^{2}}%
\sum_{t\in \mathcal{T}_{\ell }}\sum_{s\in \mathcal{T}_{\ell
}}f_{t}f_{s}^{\prime }e_{it}e_{is}\right)   \notag  \label{Lem:D1_2} \\
& =\sup_{F^{(\ell )}\in \mathcal{F}^{(\ell )}}tr\left\{ \frac{1}{T_{\ell
}^{2}}\sum_{t\in \mathcal{T}_{\ell }}\sum_{s\in \mathcal{T}_{\ell
}}f_{t}f_{s}^{\prime }\frac{1}{n}\sum_{i\in \mathcal{N}}\left[ e_{it}e_{is}-%
\mathbb{E}\left( e_{it}e_{is}\big|\mathscr{D}\right) \right] \right\}
+\sup_{F^{(\ell )}\in \mathcal{F}^{(\ell )}}tr\left\{ \frac{1}{T_{\ell }^{2}}%
\sum_{t\in \mathcal{T}_{\ell }}\sum_{s\in \mathcal{T}_{\ell
}}f_{t}f_{s}^{\prime }\frac{1}{n}\sum_{i\in \mathcal{N}}\mathbb{E}\left(
e_{it}e_{is}\big|\mathscr{D}\right) \right\} .
\end{align}%
For the first term on the second line of \eqref{Lem:D1_2}, we have $%
\max_{t,s}|\frac{1}{n}\sum_{i\in \mathcal{N}}[e_{it}e_{is}-\mathbb{E}%
(e_{it}e_{is}|\mathscr{D})]|=O_{p}(\sqrt{(\log T)/N})$ by conditional
Bernstein's inequality for independent sequence combining the fact that $%
e_{it}e_{is}$ is independent across $i$ given $\mathscr{D}$ by Assumption %
\ref{ass:1*}(i). Then 
\begin{align}
& \sup_{F^{(\ell )}\in \mathcal{F}^{(\ell )}}tr\left\{ \frac{1}{T_{\ell }^{2}%
}\sum_{t\in \mathcal{T}_{\ell }}\sum_{s\in \mathcal{T}_{\ell
}}f_{t}f_{s}^{\prime }\frac{1}{n}\sum_{i\in \mathcal{N}}\left[ e_{it}e_{is}-%
\mathbb{E}\left( e_{it}e_{is}\big|\mathscr{D}\right) \right] \right\}  
\notag \\
& =O_{p}(\sqrt{(\log T)/N})\sup_{F^{(\ell )}\in \mathcal{F}^{(\ell )}}\left( 
\frac{1}{T_{\ell }}\sum_{t\in \mathcal{T}_{\ell }}\left\Vert
f_{t}\right\Vert \right) ^{2}=O_{p}(\sqrt{(\log T)/N}).  \label{Lem:D1_3}
\end{align}%
For the second term on the second line of \eqref{Lem:D1_2}, we have 
\begin{align}
& \sup_{F^{(\ell )}\in \mathcal{F}^{(\ell )}}tr\left\{ \frac{1}{T_{\ell }^{2}%
}\sum_{t\in \mathcal{T}_{\ell }}\sum_{s\in \mathcal{T}_{\ell
}}f_{t}f_{s}^{\prime }\frac{1}{n}\sum_{i\in \mathcal{N}}\mathbb{E}\left(
e_{it}e_{is}\big|\mathscr{D}\right) \right\}   \notag \\
& \leq \sup_{F^{(\ell )}\in \mathcal{F}^{(\ell )}}\frac{1}{nT_{\ell }^{2}}%
\sum_{i\in \mathcal{N}}\sum_{t\in \mathcal{T}_{\ell }}\sum_{s\in \mathcal{T}%
_{\ell }}\left\Vert f_{t}\right\Vert \left\Vert f_{s}\right\Vert \mathbb{E}%
\left( e_{it}e_{is}\big|\mathscr{D}\right)   \notag \\
& \lesssim \frac{1}{nT_{\ell }^{2}}\sum_{i\in \mathcal{N}}\sum_{t\in 
\mathcal{T}_{\ell }}\left\vert \mathbb{E}\left( e_{it}^{2}\big|\mathscr{D}%
\right) \right\vert =O_{p}(T^{-1}),  \label{Lem:D1_4}
\end{align}%
where the first inequality is by Cauchy's inequality, the third line is by
the definition of $\mathcal{F}^{(\ell )}$ and similar arguments as in %
\eqref{Lem:D1_1}. }

{\small Combining \eqref{Lem:D1_2}-\eqref{Lem:D1_4}, we have shown that $%
\sup_{F^{(\ell )}\in \mathcal{F}^{(\ell )}}|\frac{1}{nT_{\ell }}\sum_{i\in 
\mathcal{N}}e_{i}^{(\ell )\prime }P_{F^{(\ell )}}e_{i}^{(\ell )}|=O_{p}(%
\sqrt{(\log T)/N}).$ }

{\small (iii) Note that 
\begin{equation*}
\sup_{F^{(\ell )}\in \mathcal{F}^{(\ell )}}\left\vert \frac{1}{nT_{\ell }}%
\sum_{i\in \mathcal{N}}\lambda _{i}^{0\prime }F^{0,(\ell )\prime
}M_{F^{(\ell )}}e_{i}^{(\ell )}\right\vert \leq \left\vert \frac{1}{nT_{\ell
}}\sum_{i\in \mathcal{N}}\lambda _{i}^{0\prime }F^{0,(\ell )\prime
}e_{i}^{(\ell )}\right\vert +\sup_{F^{(\ell )}\in \mathcal{F}^{(\ell
)}}\left\vert \frac{1}{nT_{\ell }}\sum_{i\in \mathcal{N}}\lambda
_{i}^{0\prime }F^{0,(\ell )\prime }P_{F^{(\ell )}}e_{i}^{(\ell )}\right\vert
.
\end{equation*}%
We show the convergence rate for the two terms on the right side of above
inequality. For the first term, we note that $\mathbb{E}\left( \lambda
_{i}^{0\prime }f_{t}^{0}e_{it}\big|\mathscr{D}\right) =0$ and $e_{it}$ is
independent across $i$ and strong mixing across $t$ given $\mathscr{D}$.
Then we have 
\begin{equation*}
\left\vert \frac{1}{nT_{\ell }}\sum_{i\in \mathcal{N}}\lambda _{i}^{0\prime
}F^{0,(\ell )\prime }e_{i}^{(\ell )}\right\vert =\left\vert \frac{1}{%
nT_{\ell }}\sum_{i\in \mathcal{N}}\sum_{t\in \mathcal{T}_{\ell }}\lambda
_{i}^{0\prime }f_{t}^{0}e_{it}\right\vert =O_{p}((NT)^{-1/2}),
\end{equation*}%
by Lemma \ref{Lem:Bern_mixing}(ii). For the second term, we note that 
\begin{align*}
& \sup_{F^{(\ell )}\in \mathcal{F}^{(\ell )}}\left\vert \frac{1}{nT_{\ell }}%
\sum_{i\in \mathcal{N}}\lambda _{i}^{0\prime }F^{0,(\ell )\prime
}P_{F^{(\ell )}}e_{i}^{(\ell )}\right\vert  \\
& =\sup_{F^{(\ell )}\in \mathcal{F}^{(\ell )}}\left\vert \frac{1}{nT_{\ell }}%
\sum_{i\in \mathcal{N}}\lambda _{i}^{0\prime }F^{0,(\ell )\prime }F^{(\ell
)}\left( F^{(\ell )\prime }F^{(\ell )}\right) ^{-1}\sum_{t\in \mathcal{T}%
_{\ell }}f_{t}e_{it}\right\vert  \\
& \leq \sup_{F^{(\ell )}\in \mathcal{F}^{(\ell )}}\sqrt{\frac{1}{n}%
\sum_{i\in \mathcal{N}}\left\Vert \lambda _{i}^{0}\right\Vert ^{2}}\sqrt{%
\sup_{F^{(\ell )}\in \mathcal{F}^{(\ell )}}\frac{1}{n}\sum_{i\in \mathcal{N}%
}\left\Vert \frac{1}{T_{\ell }}\sum_{t\in \mathcal{T}_{\ell
}}f_{t}e_{it}\right\Vert ^{2}}\sup_{F^{(\ell )}\in \mathcal{F}^{(\ell
)}}\left\Vert \left( \frac{F^{(\ell )\prime }F^{(\ell )}}{T_{\ell }}\right)
^{-1}\right\Vert \left\Vert \frac{F^{0,(\ell )\prime }F^{(\ell )}}{T_{\ell }}%
\right\Vert  \\
& \lesssim \sqrt{\sup_{F^{(\ell )}\in \mathcal{F}^{(\ell )}}\left\Vert \frac{%
1}{T_{\ell }}\sum_{t\in \mathcal{T}_{\ell }}f_{t}e_{it}\right\Vert ^{2}}%
=O_{p}[\left( (\log T)/N\right) ^{1/4}]=o_{p}(1),
\end{align*}%
where the third line is by Cauchy's inequality and the last line is by
arguments in \eqref{Lem:D1_3} and \eqref{Lem:D1_4}. Combining the above
results completes the proof. }

{\small (iv) We first observe that 
\begin{align}
& \mathbb{E}\left( \frac{1}{n}\sum_{i\in \mathcal{N}}\left\Vert \frac{1}{%
\sqrt{T_{\ell }}}\sum_{t\in \mathcal{T}_{\ell }}X_{it}e_{it}\right\Vert
^{2}\right) =\frac{1}{n}\sum_{i\in \mathcal{N}}\mathbb{E}\left( \left\Vert 
\frac{1}{\sqrt{T_{\ell }}}\sum_{t\in \mathcal{T}_{\ell
}}X_{it}e_{it}\right\Vert ^{2}\right)   \notag  \label{Lem:D1_5} \\
& \leq \frac{1}{nT_{\ell }}\sum_{i\in \mathcal{N}}\sum_{t\in \mathcal{T}%
_{\ell }}\sum_{s\in \mathcal{T}_{\ell }}\left\Vert \mathbb{E}\left(
X_{it}X_{is}^{\prime }e_{it}e_{is}\right) \right\Vert =\frac{1}{nT_{\ell }}%
\sum_{i\in \mathcal{N}}\sum_{t\in \mathcal{T}_{\ell }}\left\Vert \mathbb{E}%
\left( X_{it}X_{it}^{\prime }e_{it}^{2}\right) \right\Vert \leq M~a.s.,
\end{align}%
where the second equality is by Assumption \ref{ass:1*}(ii) and the law of
iterated expectations, and the last inequality is by Assumption \ref{ass:1*}%
(v). It follows that 
\begin{align*}
& \sup_{\left\{ \max_{i}\left\Vert \theta _{i}\right\Vert _{\max }\leq
M\right\} ,F^{(\ell )}\in \mathcal{F}^{(\ell )}}\left\Vert \frac{1}{n}%
\sum_{i\in \mathcal{N}}\left( \theta _{i}-\theta _{i}^{0,(\ell )}\right)
^{\prime }\frac{X_{i}^{(\ell )\prime }M_{F^{(\ell )}}e_{i}^{(\ell )}}{%
T_{\ell }}\right\Vert _{{}} \\
& \leq \sup_{\left\{ \max_{i}\left\Vert \theta _{i}\right\Vert _{\max }\leq
M\right\} }\left\vert \frac{1}{n}\sum_{i\in \mathcal{N}}\left( \theta
_{i}-\theta _{i}^{0,(\ell )}\right) ^{\prime }\frac{X_{i}^{\prime
}e_{i}^{(\ell )}}{T_{\ell }}\right\vert +\sup_{\left\{ \max_{i}\left\Vert
\theta _{i}\right\Vert _{\max }\leq M\right\} ,F^{(\ell )}\in \mathcal{F}%
^{(\ell )}}\left\vert \frac{1}{n}\sum_{i\in \mathcal{N}}\left( \theta
_{i}-\theta _{i}^{0,(\ell )}\right) ^{\prime }\frac{X_{i}^{(\ell )\prime
}P_{F^{(\ell )}}e_{i}^{(\ell )}}{T_{\ell }}\right\vert  \\
& \leq \sup_{\left\{ \max_{i}\left\Vert \theta _{i}\right\Vert _{\max }\leq
M\right\} }\frac{1}{\sqrt{T_{\ell }}}\left( \frac{1}{n}\sum_{i\in \mathcal{N}%
}\left\Vert \theta _{i}-\theta _{i}^{0,(\ell )}\right\Vert ^{2}\right)
^{1/2}\left( \frac{1}{n}\sum_{i\in \mathcal{N}}\left\Vert \frac{1}{\sqrt{%
T_{\ell }}}\sum_{t\in \mathcal{T}_{\ell }}X_{it}e_{it}\right\Vert
^{2}\right) ^{1/2} \\
& +\sup_{\left\{ \max_{i}\left\Vert \theta _{i}\right\Vert _{\max }\leq
M\right\} ,F^{(\ell )}\in \mathcal{F}^{(\ell )}}\frac{1}{n}\sum_{i\in 
\mathcal{N}}\left\Vert \theta _{i}-\theta _{i}^{0,(\ell )}\right\Vert
\left\Vert \frac{X_{i}^{(\ell )\prime }F^{(\ell )}}{T_{\ell }}\right\Vert
\left\Vert \frac{1}{T_{\ell }}\sum_{t\in \mathcal{T}_{\ell
}}f_{t}e_{it}\right\Vert \left\Vert \left( \frac{F^{(\ell )\prime }F^{(\ell
)}}{T_{\ell }}\right) ^{-1}\right\Vert  \\
& \lesssim O_{p}(T^{-1/2})+\sup_{F^{(\ell )}\in \mathcal{F}^{(\ell )}}\left( 
\frac{1}{n}\sum_{i\in \mathcal{N}}\left\Vert \frac{X_{i}^{(\ell )\prime
}F^{(\ell )}}{T_{\ell }}\right\Vert ^{2}\right) ^{1/2}\sup_{F^{(\ell )}\in 
\mathcal{F}^{(\ell )}}\left( \frac{1}{n}\sum_{i\in \mathcal{N}}\left\Vert 
\frac{1}{T_{\ell }}\sum_{t\in \mathcal{T}_{\ell }}f_{t}e_{it}\right\Vert
^{2}\right) ^{1/2} \\
& =O_{p}(T^{-1/2})+O_{p}[\left( (\log T)/N\right) ^{1/4}]=o_{p}(1),
\end{align*}%
where the second inequality is by Cauchy's inequality, the fifth line holds
by the fact that both $\theta _{i}$ and $\theta _{i}^{0,(\ell )}$ are
bounded and $\frac{1}{n}\sum_{i\in \mathcal{N}}||\frac{1}{\sqrt{T_{\ell }}}%
\sum_{t\in \mathcal{T}_{\ell }}X_{it}e_{it}||^{2}=O_{p}(1)$ by %
\eqref{Lem:D1_5}, and the last line is due to \eqref{Lem:D1_2} and the fact
that 
\begin{equation*}
\sup_{F^{(\ell )}\in \mathcal{F}^{(\ell )}}\left( \frac{1}{n}\sum_{i\in 
\mathcal{N}}\left\Vert \frac{X_{i}^{(\ell )\prime }F^{(\ell )}}{T_{\ell }}%
\right\Vert ^{2}\right) ^{1/2}\lesssim \max_{i\in \mathcal{N}}\frac{1}{%
T_{\ell }}\sum_{t\in \mathcal{T}_{\ell }}\left\Vert X_{it}\right\Vert
=O_{p}(1).
\end{equation*}%
by Assumption \ref{ass:8}(ii). }
\end{proof}

\begin{lemma}
{\small \label{Lem:theta_const} Under Assumptions \ref{ass:1*}, \ref{ass:2}
and \ref{ass:8}, we have $\hat{\theta}_{i}^{(\ell )}-\theta _{i}^{0,(\ell )}%
\overset{p}{\longrightarrow }0$ and $||P_{\hat{F}^{(\ell )}}-P_{F^{0,(\ell
)}}||\overset{p}{\longrightarrow }0$. }
\end{lemma}

\begin{proof}
{\small Let 
\begin{equation*}
S_{NT}\left( \left\{ \theta _{i}\right\} ,F^{(\ell )}\right) =\frac{1}{%
nT_{\ell }}\sum_{_{i\in \mathcal{N}}}(Y_{i}^{(\ell )}-X_{i}^{(\ell )}\theta
_{i})^{\prime }M_{F^{(\ell )}}(Y_{i}^{(\ell )}-X_{i}^{(\ell )}\theta _{i})-%
\frac{1}{nT_{\ell }}\sum_{_{i\in \mathcal{N}}}e_{i}^{(\ell )\prime
}M_{F^{0,(\ell )}}e_{i}^{(\ell )}.
\end{equation*}%
Recall from \eqref{PIFE_obj} that $(\{\hat{\theta}_{i}^{(\ell )}\},\hat{F}%
^{(\ell )})$ is the minimizer of $S_{NT}\left( \left\{ \theta _{i}\right\}
,F^{(\ell )}\right) $. By \eqref{PIFE_model} and Lemma \ref{Lem:D1}, we have 
\begin{align*}
S_{NT}(\left\{ \theta _{i}\right\} ,F^{(\ell )})& =\tilde{S}_{NT}(\left\{
\theta _{i}\right\} _{\forall i},F^{(\ell )})+\frac{2}{n}\sum_{i\in \mathcal{%
N}}(\theta _{i}-\theta _{i}^{0,(\ell )})^{\prime }\frac{X_{i}^{(\ell )\prime
}M_{F^{(\ell )}}e_{i}^{(\ell )}}{T_{\ell }} \\
& +\frac{2}{nT_{\ell }}\sum_{i\in \mathcal{N}}\lambda _{i}^{0\prime
}F^{0,(\ell )\prime }M_{F^{(\ell )}}e_{i}^{(\ell )}+\frac{1}{nT_{\ell }}%
\sum_{i\in \mathcal{N}}e_{i}^{(\ell )\prime }P_{F^{(\ell )}}e_{i}^{(\ell )}-%
\frac{1}{nT_{\ell }}\sum_{i\in \mathcal{N}}e_{i}^{(\ell )\prime
}P_{F^{0,(\ell )}}e_{i}^{(\ell )} \\
& =\tilde{S}_{NT}(\left\{ \theta _{i}\right\} ,F^{(\ell )})+o_{p}(1),
\end{align*}%
where 
\begin{align*}
\tilde{S}_{NT}(\left\{ \theta _{i}\right\} ,F^{(\ell )})& =\frac{1}{n}%
\sum_{i\in \mathcal{N}}(\theta _{i}-\theta _{i}^{0,(\ell )})^{\prime }\frac{%
X_{i}^{(\ell )\prime }M_{F^{(\ell )}}e_{i}^{(\ell )}}{T_{\ell }}(\theta
_{i}-\theta _{i}^{0,(\ell )})+\frac{1}{nT_{\ell }}\sum_{i\in \mathcal{N}%
}\lambda _{i}^{0\prime }F^{0,(\ell )\prime }M_{F^{(\ell )}}F^{0,(\ell
)}\lambda _{i}^{0} \\
& +\frac{2}{n}\sum_{i\in \mathcal{N}}(\theta _{i}-\theta _{i}^{0,(\ell
)})^{\prime }\frac{X_{i}^{(\ell )\prime }M_{F^{(\ell )}}e_{i}^{(\ell )}}{%
T_{\ell }}F^{0,(\ell )}\lambda _{i}^{0}.
\end{align*}%
}

{\small Following \cite{song2013essays} and \cite{bai2009panel}, we can show
that $\tilde{S}_{NT}\left( \left\{ \theta _{i}\right\} _{\forall i},F^{(\ell
)}\right) $ is uniquely minimized at $(\{\theta _{i}^{0,(\ell
)}\},F^{0,(\ell )}H^{(\ell )}),$ where $H^{(\ell )}$ is a rotation matrix.
Hence, we conclude that $\hat{\theta}_{i}^{(\ell )}-\theta _{i}^{0,(\ell )}%
\overset{p}{\longrightarrow }0$. Following the proof of Proposition 1 of 
\cite{bai2009panel}, we can show that $||P_{\hat{F}^{(\ell )}}-P_{F^{0,(\ell
)}}||\overset{p}{\longrightarrow }0$. }
\end{proof}

{\small Let $B_{N}$ denote the uniform convergence rate for $\hat{\theta}%
_{i}^{(\ell )}$. That is, $\max_{i\in \mathcal{N}}||\hat{\theta}_{i}^{(\ell
)}-\theta _{i}^{0,(\ell )}||=O_{p}(B_{N})$. }

\begin{lemma}
{\small \label{Lem:Fhat} Under Assumptions \ref{ass:1*}, \ref{ass:2} and \ref%
{ass:8}, we have $\frac{1}{\sqrt{T_{\ell }}}\left\Vert \hat{F}^{(\ell
)}-F^{0,(\ell )}H^{(\ell )}\right\Vert =O_{p}(B_{N}+\frac{1}{\sqrt{N\wedge T}%
}),$ where $H^{(\ell )}:=\left( \frac{\Lambda _{n}^{0\prime }\Lambda _{n}^{0}%
}{n}\right) \left( \frac{F^{0,(\ell )\prime }\hat{F}^{(\ell )}}{T_{\ell }}%
\right) V_{NT}^{(\ell )-1}$. }
\end{lemma}

\begin{proof}
{\small Recall that $V_{NT}^{(\ell )}$ is the diagonal matrix that contains
the eigenvalues of $\frac{1}{nT_{\ell }}\sum_{i\in \mathcal{N}}(Y_{i}^{(\ell
)}-X_{i}^{(\ell )}\theta _{i})^{\prime }(Y_{i}^{(\ell )}-X_{i}^{(\ell
)}\theta _{i})$ along its diagonal line. By inserting \eqref{PIFE_model}
into \eqref{PIFE_est_F}, we obtain that 
\begin{equation}
\hat{F}^{(\ell )}V_{NT}^{(\ell )}-F^{0,(\ell )}\left( \frac{\Lambda
_{n}^{0\prime }\Lambda _{n}^{0}}{n}\right) \left( \frac{F^{0,(\ell )\prime }%
\hat{F}^{(\ell )}}{T_{\ell }}\right) =\sum_{m\in \lbrack 8]}J_{m}^{(\ell )},
\label{Lem:Fhat_1}
\end{equation}%
where%
\begin{align*}
& J_{1}^{(\ell )}=\frac{1}{nT_{\ell }}\sum_{i\in \mathcal{N}}X_{i}^{(\ell
)}(\theta _{i}^{0,(\ell )}-\hat{\theta}_{i}^{(\ell )})(\theta _{i}^{0,(\ell
)}-\hat{\theta}_{i}^{(\ell )})^{\prime }X_{i}^{(\ell )\prime }\hat{F}^{(\ell
)},J_{2}^{(\ell )}=\frac{1}{nT_{\ell }}\sum_{i\in \mathcal{N}}X_{i}^{(\ell
)}(\theta _{i}^{0,(\ell )}-\hat{\theta}_{i}^{(\ell )})\lambda _{i}^{0\prime
}F^{0,(\ell )\prime }\hat{F}^{(\ell )}, \\
& J_{3}^{(\ell )}=\frac{1}{nT_{\ell }}\sum_{i\in \mathcal{N}}X_{i}^{(\ell
)}(\theta _{i}^{0,(\ell )}-\hat{\theta}_{i}^{(\ell )})e_{i}^{(\ell )\prime }%
\hat{F}^{(\ell )},\text{ }J_{4}^{(\ell )}=\frac{1}{nT_{\ell }}\sum_{i\in 
\mathcal{N}}F^{0,(\ell )}\lambda _{i}^{0}(\theta _{i}^{0,(\ell )}-\hat{\theta%
}_{i}^{(\ell )})^{\prime }X_{i}^{(\ell )\prime }\hat{F}^{(\ell )}, \\
& J_{5}^{(\ell )}=\frac{1}{nT_{\ell }}\sum_{i\in \mathcal{N}}e_{i}^{(\ell
)}(\theta _{i}^{0,(\ell )}-\hat{\theta}_{i}^{(\ell )})^{\prime }X_{i}^{(\ell
)\prime }\hat{F}^{(\ell )},\quad J_{6}^{(\ell )}=\frac{1}{nT_{\ell }}%
\sum_{i\in \mathcal{N}}F^{0,(\ell )}\lambda _{i}^{0}e_{i}^{(\ell )\prime }%
\hat{F}^{(\ell )}, \\
& J_{7}^{(\ell )}=\frac{1}{nT_{\ell }}\sum_{i\in \mathcal{N}}e_{i}^{(\ell
)}\lambda _{i}^{0\prime }F^{0,(\ell )\prime }\hat{F}^{(\ell )},\quad \text{%
and}\quad J_{8}^{(\ell )}=\frac{1}{nT_{\ell }}\sum_{i\in \mathcal{N}%
}e_{i}^{(\ell )}e_{i}^{(\ell )\prime }\hat{F}^{(\ell )}.
\end{align*}%
We show the convergence rate for $J_{m}^{(\ell )}~\forall m\in \lbrack 8]$
below.. }

{\small For $J_{1}^{(\ell )}$, we notice that 
\begin{equation}
\frac{1}{\sqrt{T_{\ell }}}\left\Vert J_{1}^{(\ell )}\right\Vert \leq
\max_{i\in \mathcal{N}}\left\Vert \hat{\theta}_{i}^{(\ell )}-\theta
_{i}^{0,(\ell )}\right\Vert ^{2}\frac{1}{\sqrt{T_{\ell }}}\left\Vert \hat{F}%
^{(\ell )}\right\Vert \frac{1}{nT_{\ell }}\sum_{i\in \mathcal{N}}\left\Vert
X_{i}^{(\ell )}\right\Vert ^{2}=O_{p}(B_{N}^{2}),  \label{Lem:Fhat_J1}
\end{equation}%
where the equality holds by Assumption \ref{ass:8}(ii) and normalization of
the factor vector. Similarly, we have 
\begin{align*}
& \frac{1}{\sqrt{T_{\ell }}}\left\Vert J_{2}^{(\ell )}\right\Vert \leq
\max_{i\in \mathcal{N}}\left\Vert \hat{\theta}_{i}^{(\ell )}-\theta
_{i}^{0,(\ell )}\right\Vert \max_{i\in \mathcal{N}}\left\Vert \lambda
_{i}^{0}\right\Vert _{2}\frac{\left\Vert F^{0,(\ell )}\right\Vert }{\sqrt{%
T_{\ell }}}\frac{\left\Vert \hat{F}^{(\ell )}\right\Vert }{\sqrt{T_{\ell }}}%
\frac{1}{n\sqrt{T_{\ell }}}\sum_{i\in \mathcal{N}}\left\Vert X_{i}^{(\ell
)}\right\Vert =O_{p}(B_{N}), \\
& \frac{1}{\sqrt{T_{\ell }}}\left\Vert J_{3}^{(\ell )}\right\Vert \leq
\max_{i\in \mathcal{N}}\left\Vert \hat{\theta}_{i}^{(\ell )}-\theta
_{i}^{0,(\ell )}\right\Vert \frac{\left\Vert \hat{F}^{(\ell )}\right\Vert }{%
\sqrt{T_{\ell }}}\sqrt{\frac{1}{nT_{\ell }}\sum_{i\in \mathcal{N}}\left\Vert
X_{i}^{(\ell )}\right\Vert ^{2}}\sqrt{\frac{1}{nT_{\ell }}\sum_{i\in 
\mathcal{N}}\left\Vert e_{i}^{(\ell )}\right\Vert ^{2}}=O_{p}(B_{N}), \\
& \frac{1}{\sqrt{T_{\ell }}}\left\Vert J_{6}^{(\ell )}\right\Vert \leq \frac{%
1}{\sqrt{n}}\frac{\left\Vert F^{0,(\ell )}\right\Vert }{\sqrt{T_{\ell }}}%
\frac{\left\Vert \hat{F}^{(\ell )}\right\Vert }{\sqrt{T_{\ell }}}\sqrt{\frac{%
1}{T_{\ell }}\sum_{t\in \mathcal{T}_{\ell }}\left\Vert \frac{1}{\sqrt{n}}%
\sum_{i\in \mathcal{N}}\lambda _{i}^{0}e_{it}\right\Vert ^{2}}=O_{p}\left(
N^{-1/2}\right)  \\
& \frac{1}{\sqrt{T_{\ell }}}\left\Vert J_{8}^{(\ell )}\right\Vert \leq \frac{%
\left\Vert \hat{F}^{(\ell )}\right\Vert }{\sqrt{T_{\ell }}}\frac{1}{nT_{\ell
}}\left\Vert \sum_{i\in \mathcal{N}}e_{i}^{(\ell )}e_{i}^{(\ell )\prime
}\right\Vert =O_{p}((N\wedge T)^{-1/2}),
\end{align*}%
where the third line is by the fact that $\frac{1}{T_{\ell }}\sum_{t\in 
\mathcal{T}_{\ell }}\left\Vert \frac{1}{\sqrt{n}}\sum_{i\in \mathcal{N}%
}\lambda _{i}^{0}e_{it}\right\Vert =O_{p}(1)$ by similar arguments as in %
\eqref{Lem:D1_1} and the last line is due to the fact that 
\begin{align}
& \mathbb{E}\left( \left\Vert \frac{1}{nT_{\ell }}\sum_{i\in \mathcal{N}%
}e_{i}^{(\ell )}e_{i}^{(\ell )\prime }\right\Vert ^{2}\Bigg|\mathscr{D}%
\right) =\frac{1}{(nT_{\ell })^{2}}\sum_{i\in \mathcal{N}}\sum_{i^{\ast }\in 
\mathcal{N}}\sum_{t\in \mathcal{T}_{\ell }}\sum_{t^{\ast }\in \mathcal{T}%
_{\ell }}\mathbb{E}\left( e_{it}e_{it^{\ast }}e_{i^{\ast }t}e_{i^{\ast
}t^{\ast }}\big|\mathscr{D}\right)   \notag  \label{Lem:Fhat_2} \\
& =\frac{1}{(nT_{\ell })^{2}}\sum_{i\in \mathcal{N}}\sum_{t\in \mathcal{T}%
_{\ell }}\sum_{t^{\ast }\in \mathcal{T}_{\ell }}\mathbb{E}\left(
e_{it}^{2}e_{it^{\ast }}^{2}\big|\mathscr{D}\right) +\frac{1}{(nT_{\ell
})^{2}}\sum_{i\in \mathcal{N}}\sum_{i^{\ast }\neq i}\sum_{t\in \mathcal{T}%
_{\ell }}\sum_{t^{\ast }\in \mathcal{T}_{\ell }}\mathbb{E}\left(
e_{it}e_{it^{\ast }}\big|\mathscr{D}\right) \mathbb{E}\left( e_{i^{\ast
}t}e_{i^{\ast }t^{\ast }}\big|\mathscr{D}\right)   \notag \\
& =\frac{1}{(nT_{\ell })^{2}}\sum_{i\in \mathcal{N}}\sum_{t\in \mathcal{T}%
_{\ell }}\mathbb{E}\left( e_{it}^{4}|\mathscr{D}\right) +\frac{1}{(nT_{\ell
})^{2}}\sum_{i\in \mathcal{N}}\sum_{t\in \mathcal{T}_{\ell }}\sum_{t^{\ast
}\in \mathcal{T}_{\ell },t^{\ast }\neq t}\mathbb{E}\left( e_{it}^{2}|%
\mathscr{D}\right) \mathbb{E}\left( e_{it^{\ast }}^{2}|\mathscr{D}\right)  
\notag \\
& +\frac{1}{(nT_{\ell })^{2}}\sum_{i\in \mathcal{N}}\sum_{i^{\ast }\neq
i}\sum_{t\in \mathcal{T}_{\ell }}\mathbb{E}\left( e_{it}^{2}|\mathscr{D}%
\right) \mathbb{E}\left( e_{i^{\ast }t}^{2}|\mathscr{D}\right)   \notag \\
& =O\left( (N)^{-1}+(T)^{-1}\right) ~a.s.
\end{align}%
by Assumption \ref{ass:1*}(i), (ii), (iii), and (v). Besides, we have }$%
{\small T}_{{\small \ell }}^{{\small -1/2}}{\small ||J}_{{\small 4}}^{(%
{\small \ell })}{\small ||}${\small $=O_{p}(B_{N}),$ }${\small T}_{{\small %
\ell }}^{{\small -1/2}}{\small ||J}_{{\small 5}}^{({\small \ell })}{\small %
||=O}_{p}{\small (B}_{N}{\small )}${\small , and }${\small T}_{{\small \ell }%
}^{{\small -1/2}}{\small ||J}_{{\small 4}}^{({\small \ell })}{\small ||}$%
{\small $=O_{p}(N^{-1/2})$ by similar analyses as used for $J_{2}^{(\ell )}$%
, $J_{3}^{(\ell )}$ and $J_{6}^{(\ell )}$, respectively. }

{\small Combining the above arguments, premultiplying both sides of %
\eqref{Lem:Fhat_1} by $\hat{F}^{(\ell )\prime }$ and using the fact that $%
\hat{F}^{(\ell )\prime }\hat{F}^{(\ell )}=T_{\ell }I_{r}$, we have 
\begin{equation}
\frac{1}{T_{\ell }}\left\Vert \hat{F}^{(\ell )}V_{NT}^{(\ell )}-F^{0,(\ell )}%
\frac{\Lambda _{n}^{0\prime }\Lambda _{n}^{0}}{n}\frac{F^{0,(\ell )\prime }%
\hat{F}^{(\ell )}}{T_{\ell }}\right\Vert =O_{p}(B_{N})+O_{p}((N\wedge
T)^{-1/2}),  \label{Lem:Fhat_3}
\end{equation}%
and 
\begin{equation}
V_{NT}^{(\ell )}=\frac{F^{0,(\ell )\prime }\hat{F}^{(\ell )}}{T_{\ell }}%
\frac{\Lambda _{n}^{0\prime }\Lambda _{n}^{0}}{n}\frac{F^{0,(\ell )\prime }%
\hat{F}^{(\ell )}}{T_{\ell }}+\frac{\hat{F}^{(\ell )\prime }}{\sqrt{T_{\ell }%
}}\frac{\sum_{m\in \lbrack 8]}J_{m}^{(\ell )}}{\sqrt{T_{\ell }}}=\frac{%
F^{0,(\ell )\prime }\hat{F}^{(\ell )}}{T_{\ell }}\frac{\Lambda _{n}^{0\prime
}\Lambda _{n}^{0}}{n}\frac{F^{0,(\ell )\prime }\hat{F}^{(\ell )}}{T_{\ell }}%
+o_{p}(1).  \notag
\end{equation}%
Then $V_{NT}^{(\ell )}$ is invertible and $||V_{NT}^{(\ell )}||=O_{p}(1)$.
By the definition of $H^{(\ell )}$ and \eqref{Lem:Fhat_3}, we have $\frac{1}{%
\sqrt{T_{\ell }}}||\hat{F}^{(\ell )}-F^{0,(\ell )}H^{(\ell
)}||=O_{p}(B_{N}+(N\wedge T)^{-1/2})$. }
\end{proof}

\begin{lemma}
{\small \label{Lem:thetahat_pre} Under Assumptions \ref{ass:1*}, \ref{ass:2}
and \ref{ass:8}, we have }

\begin{itemize}
\item[(i)] {\small $\frac{X_{i}^{(\ell )\prime }M_{\hat{F}^{(\ell
)}}X_{i^{\ast }}^{(\ell )}}{T_{\ell }}=\frac{X_{i}^{(\ell )\prime
}M_{F^{0,(\ell )}}X_{i^{\ast }}^{(\ell )}}{T_{\ell }}+O_{p}\left( B_{N}+%
\frac{1}{\sqrt{N\wedge T}}\right) $ uniformly in $i,i^{\ast }\in \mathcal{N}$%
, }

\item[(ii)] {\small $\frac{X_{i}^{(\ell )\prime }M_{\hat{F}^{(\ell
)}}e_{i}^{(\ell )}}{T_{\ell }}=\frac{X_{i}^{(\ell )\prime }M_{F^{0,(\ell
)}}e_{i}^{(\ell )}}{T_{\ell }}+O_{p}\left( B_{N}^{2}+B_{N}\sqrt{\frac{\log N%
}{T}}+\frac{1}{N\wedge T}+\sqrt{\frac{\log N}{(N\wedge T)T}}\right) $
uniformly in $i\in \mathcal{N}$, }

\item[(iii)] {\small $\frac{1}{nT_{\ell }^{2}}\sum_{i\in \mathcal{N}%
}\left\Vert e_{i}^{(\ell )\prime }\hat{F}^{(\ell )}\right\Vert
^{2}=O_{p}\left( B_{N}^{2}+\frac{1}{N\wedge T}\right) $, }

\item[(iv)] {\small $\max_{i\in \mathcal{N}}\left\Vert \frac{1}{T_{\ell }}%
X_{i}^{(\ell )\prime }M_{\hat{F}^{(\ell )}}F^{0,(\ell )}\right\Vert
=O_{p}\left( B_{N}+\frac{1}{\sqrt{N\wedge T}}\right) $, }

\item[(v)] {\small $\frac{1}{nT_{\ell }^{2}}\sum_{i^{\ast }\in \mathcal{N}%
}\left\Vert X_{i}^{(\ell )\prime }e_{i^{\ast }}^{(\ell )}\right\Vert
^{2}=O_{p}\left( \frac{\log N}{T}\right) $ uniformly in $i\in \mathcal{N}$, }

\item[(vi)] {\small $\left\Vert \frac{1}{nT_{\ell }}\sum_{i\in \mathcal{N}%
}\lambda _{i}^{0}e_{i}^{(\ell )\prime }\hat{F}^{(\ell )}\right\Vert
=O_{p}\left( \frac{B_{N}}{\sqrt{N}}+\frac{1}{N}+\frac{1}{\sqrt{NT}}\right) $%
, }

\item[(vii)] {\small $\max_{i\in \mathcal{N}}\left\Vert \frac{1}{nT_{\ell }}%
\sum_{i^{\ast }\in \mathcal{N}}X_{i}^{(\ell )\prime }M_{\hat{F}^{(\ell
)}}e_{i^{\ast }}^{(\ell )}\lambda _{i^{\ast }}^{0\prime }\right\Vert
=O_{p}\left( \frac{B_{N}}{\sqrt{N}}+\frac{1}{N}+\sqrt{\frac{\log N}{NT}}%
\right) $, }

\item[(viii)] {\small $\max_{i\in \mathcal{N}}\left\Vert \frac{1}{nT_{\ell
}^{2}}\sum_{i^{\ast }\in \mathcal{N}}X_{i}^{(\ell )\prime }M_{\hat{F}^{(\ell
)}}e_{i^{\ast }}^{(\ell )}e_{i^{\ast }}^{(\ell )\prime }\hat{F}^{(\ell
)}\right\Vert =O_{p}\left( B_{N}^{2}+B_{N}\sqrt{\frac{\log N}{T}}+\frac{%
\sqrt{\log N}}{N\wedge T}\right) $. }
\end{itemize}
\end{lemma}

\begin{proof}
{\small (i) Let }${\small \delta }_{{\small 1}}^{(\ell )}{\small :=}${\small 
$(T_{\ell }^{-1}\hat{F}^{(\ell )\prime }\hat{F}^{(\ell )})^{-1}-(T_{\ell
}^{-1}F^{0,(\ell )\prime }F^{0,(\ell )})^{-1}$ and }${\small \delta }_{%
{\small 2}}^{(\ell )}{\small :=}$$T_{\ell }^{-1/2}(\hat{F}^{(\ell )\prime
}-F^{0,(\ell )\prime }H^{(\ell )})${\small $.$} {\small Noting that $M_{\hat{%
F}^{(\ell )}}=I_{T_{\ell }}-\hat{F}^{(\ell )}(\hat{F}^{(\ell )\prime }\hat{F}%
^{(\ell )})^{-1}\hat{F}^{(\ell )\prime }$ and $M_{F^{0,(\ell )}}=I_{T_{\ell
}}-F^{0,(\ell )}(F^{0,(\ell )\prime }F^{0,(\ell )})^{-1}F^{0,(\ell )\prime }$%
, we can show that 
\begin{align}
M_{F^{0,(\ell )}}-M_{\hat{F}^{(\ell )}}& =\frac{\hat{F}^{(\ell )}}{\sqrt{%
T_{\ell }}}\left( \frac{\hat{F}^{(\ell )\prime }\hat{F}^{(\ell )}}{T_{\ell }}%
\right) ^{-1}\frac{\hat{F}^{(\ell )\prime }}{\sqrt{T_{\ell }}}-\frac{%
F^{0,(\ell )}}{\sqrt{T_{\ell }}}\left( \frac{F^{0,(\ell )\prime }F^{0,(\ell
)}}{T_{\ell }}\right) ^{-1}\frac{F^{0,(\ell )\prime }}{\sqrt{T_{\ell }}} 
\notag \\
& =\delta _{2}^{(\ell )}\delta _{1}^{(\ell )}\delta _{2}^{(\ell )\prime
}+\delta _{2}^{(\ell )}\delta _{1}^{(\ell )}\left( \frac{F^{0,(\ell
)}H^{(\ell )}}{\sqrt{T_{\ell }}}\right) ^{\prime }+\delta _{2}^{(\ell
)}\left( \frac{F^{0,(\ell )\prime }F^{0,(\ell )}}{T_{\ell }}\right)
^{-1}\delta _{2}^{(\ell )\prime }+\frac{F^{0,(\ell )}H^{(\ell )}}{\sqrt{%
T_{\ell }}}\delta _{1}^{(\ell )}\delta _{2}^{(\ell )\prime }  \notag \\
& +\delta _{2}^{(\ell )}\left( \frac{F^{0,(\ell )\prime }F^{0,(\ell )}}{%
T_{\ell }}\right) ^{-1}\left( \frac{F^{0,(\ell )}H^{(\ell )}}{\sqrt{T_{\ell }%
}}\right) ^{\prime }+\frac{F^{0,(\ell )}H^{(\ell )}}{\sqrt{T_{\ell }}}\delta
_{1}^{(\ell )}\left( \frac{F^{0,(\ell )}H^{(\ell )}}{\sqrt{T_{\ell }}}%
\right) ^{\prime }  \notag \\
& +\frac{F^{0,(\ell )}H^{(\ell )}}{\sqrt{T_{\ell }}}\left( \frac{F^{0,(\ell
)\prime }F^{0,(\ell )}}{T_{\ell }}\right) ^{-1}\delta _{2}^{(\ell )\prime }.
\label{Lem4:MF_1}
\end{align}%
By Lemma \ref{Lem:Fhat}, Assumption \ref{ass:8}, the normalization for the
factor space, and the fact that 
\begin{equation}
\left\Vert \delta _{1}^{(\ell )}\right\Vert \leq \frac{\left\Vert \hat{F}%
^{(\ell )}-F^{0,(\ell )}H^{(\ell )}\right\Vert ^{2}}{T_{\ell }}+2\frac{%
\left\Vert F^{0,(\ell )}H^{(\ell )}\right\Vert \left\Vert \hat{F}^{(\ell
)}-F^{0,(\ell )}H^{(\ell )}\right\Vert }{T_{\ell }}=O_{p}(B_{N}+(N\wedge
T)^{-1/2}),  \label{Lem4:MF_2}
\end{equation}%
we can readily show that 
\begin{equation}
\left\Vert M_{F^{0,(\ell )}}-M_{\hat{F}^{(\ell )}}\right\Vert
=O_{p}(B_{N}+(N\wedge T)^{-1/2}).  \label{Lem4:MF_3}
\end{equation}%
Then $\max_{i,i^{\ast }\in \mathcal{N}}\left\Vert \frac{X_{i}^{(\ell )\prime
}M_{\hat{F}^{(\ell )}}X_{i^{\ast }}^{(\ell )}}{T_{\ell }}-\frac{X_{i}^{(\ell
)\prime }M_{F^{0,(\ell )}}X_{i^{\ast }}^{(\ell )}}{T_{\ell }}\right\Vert
\leq \max_{i\in \mathcal{N}}\frac{1}{T_{\ell }}\left\Vert X_{i}^{(\ell
)}\right\Vert ^{2}\left\Vert M_{F^{0,(\ell )}}-M_{\hat{F}^{(\ell
)}}\right\Vert =O_{p}(B_{N}+(N\wedge T)^{-1/2}),$ where recall that $%
\mathcal{N}:=\{n_{1},\cdots ,n_{n}\}$. }

{\small (ii) By \eqref{Lem4:MF_1}, we notice that 
\begin{align*}
& \max_{i\in \mathcal{N}}\left\Vert \frac{X_{i}^{(\ell )\prime }M_{\hat{F}%
^{(\ell )}}e_{i}^{(\ell )}}{\sqrt{T_{\ell }}}-\frac{X_{i}^{(\ell )\prime
}M_{F^{0,(\ell )}}e_{i}^{(\ell )}}{\sqrt{T_{\ell }}}\right\Vert \\
& =\sqrt{T_{\ell }}\max_{i\in \mathcal{N}}\frac{\left\Vert X_{i}^{(\ell
)}\right\Vert }{\sqrt{T_{\ell }}}\frac{\left\Vert e_{i}^{(\ell )}\right\Vert 
}{\sqrt{T_{\ell }}}O_{p}\left[ (B_{N}+(N\wedge
T)^{-1/2})^{3}+(B_{N}+(N\wedge T)^{-1/2})^{2}\right] \\
& +\sqrt{T_{\ell }}\max_{i\in \mathcal{N}}\frac{\left\Vert X_{i}^{(\ell
)}\right\Vert }{\sqrt{T_{\ell }}}\frac{\left\Vert \hat{F}^{(\ell
)}-F^{0,(\ell )}H^{(\ell )}\right\Vert }{\sqrt{T_{\ell }}}\left\Vert \left( 
\frac{F^{0,(\ell )\prime }F^{0,(\ell )}}{T_{\ell }}\right) ^{-1}\right\Vert 
\frac{\left\Vert \left( F^{0,(\ell )}H^{(\ell )}\right) ^{\prime
}e_{i}^{(\ell )}\right\Vert }{T_{\ell }} \\
& +\sqrt{T_{\ell }}\max_{i\in \mathcal{N}}\frac{\left\Vert X_{i}^{(\ell
)}\right\Vert }{\sqrt{T_{\ell }}}\frac{\left\Vert F^{0,(\ell )}H^{(\ell
)}\right\Vert }{\sqrt{T_{\ell }}}\left\Vert \delta _{1}^{(\ell )}\right\Vert 
\frac{\left\Vert \left( F^{0,(\ell )}H^{(\ell )}\right) ^{\prime
}e_{i}^{(\ell )}\right\Vert }{T_{\ell }} \\
& =\sqrt{T_{\ell }}\left[ O_{p}(B_{N}^{2}+(N\wedge
T)^{-1})+O_{p}(B_{N}+(N\wedge T)^{-1/2})O_{p}(\sqrt{(\log N)/T})\right] ,
\end{align*}%
where the last line holds by combining Assumption \ref{ass:8}(ii), %
\eqref{Lem4:MF_2}, Lemma \ref{Lem:Fhat} and the fact that 
\begin{equation*}
\max_{i\in \mathcal{N}}\frac{\left\Vert \left( F^{0,(\ell )}H^{(\ell
)}\right) ^{\prime }e_{i}^{(\ell )}\right\Vert }{T_{\ell }}\lesssim
\max_{i\in \mathcal{N}}\left\Vert \frac{1}{T_{\ell }}\sum_{t\in \mathcal{T}%
_{\ell }}f_{t}^{0}e_{it}\right\Vert =O_{p}(\sqrt{(\log N)/T}).
\end{equation*}%
We will show the last equality by using the Bernstein's inequality in Lemma %
\ref{Lem:Bern_mds}(i). }

{\small Note that 
\begin{align}
& \max_{i\in \mathcal{N},t\in \mathcal{T}_{\ell }}\left\Vert Var\left(
e_{it}f_{t}^{0}\big|\mathscr{D}\right) \right\Vert =\max_{i\in \mathcal{N}%
,t\in \mathcal{T}_{\ell }}\left\Vert \mathbb{E}\left( e_{it}^{2}\big|%
\mathscr{D}\right) f_{t}^{0}f_{t}^{0\prime }\right\Vert =O_{p}(1)\quad \text{%
and}  \notag \\
& \max_{i\in \mathcal{N},t\in \mathcal{T}_{\ell }}\left\Vert
e_{it}f_{t}^{0}\right\Vert =O_{p}\left( (NT)^{1/q}\right) ,  \label{Lem:var}
\end{align}%
where the second line is by Assumption \ref{ass:1*}(v) and the last line is
by Assumption \ref{ass:1*}(v). Define events $\mathscr{A}_{4,N}(M)=\left\{
\max_{i\in \mathcal{N},t\in \mathcal{T}_{\ell }}\left\Vert
e_{it}f_{t}^{0}\right\Vert \leq M(NT)^{1/q}\right\} $ and $\mathscr{A}%
_{4,N,i}(M)=\left\{ \max_{t\in \mathcal{T}_{\ell }}\left\Vert
e_{it}f_{t}^{0}\right\Vert \leq M(NT)^{1/q}\right\} $ for a large enough
constant $M$. Then for some large positive constants $c_{6}$ and $c_{7}$, we
have $\mathbb{P}\left( \mathscr{A}_{4,N}^{c}(M)\right) \rightarrow 0$ and 
\begin{align}
& \mathbb{P}\left( \max_{i\in \mathcal{N}}\left\Vert \frac{1}{T_{\ell }}%
\sum_{t\in \mathcal{T}_{\ell }}e_{it}f_{t}^{0}\right\Vert >c_{6}\sqrt{\frac{%
\log N}{T}}\right)   \notag  \label{Lem:emp_Bern} \\
& \leq \mathbb{P}\left( \max_{i\in \mathcal{N}}\left\Vert \frac{1}{T_{\ell }}%
\sum_{t\in \mathcal{T}_{\ell }}e_{it}f_{t}^{0}\right\Vert >c_{6}\sqrt{\frac{%
\log N}{T}},\mathscr{A}_{4,N}(M)\right) +\mathbb{P}\left( \mathscr{A}%
_{4,N}^{c}(M)\right)   \notag \\
& \leq \sum_{i\in \mathcal{N}}\mathbb{P}\left( \left\Vert \frac{1}{T_{\ell }}%
\sum_{t\in \mathcal{T}_{\ell }}e_{it}f_{t}^{0}\right\Vert >c_{6}\sqrt{\frac{%
\log N}{T}},\mathscr{A}_{4,N}(M)\right) +\mathbb{P}\left( \mathscr{A}%
_{4,N}^{c}(M)\right)   \notag \\
& \leq \sum_{i\in \mathcal{N}}\mathbb{P}\left( \left\Vert \frac{1}{T_{\ell }}%
\sum_{t\in \mathcal{T}_{\ell }}e_{it}f_{t}^{0}\right\Vert >c_{6}\sqrt{\frac{%
\log N}{T}},\mathscr{A}_{4,N,i}(M)\right) +\mathbb{P}\left( \mathscr{A}%
_{4,N}^{c}(M)\right)   \notag \\
& \leq \sum_{i\in \mathcal{N}}\mathbb{E}\mathbb{P}\left( \left\Vert \frac{1}{%
T_{\ell }}\sum_{t\in \mathcal{T}_{\ell }}e_{it}f_{t}^{0}\right\Vert >c_{6}%
\sqrt{\frac{\log N}{T}}\bigg|\mathscr{D}\right) \mathbf{1}\left\{ \mathscr{A}%
_{4,N,i}(M)\right\} +\mathbb{P}\left( \mathscr{A}_{4,N}^{c}(M)\right)  
\notag \\
& \leq \sum_{i\in \mathcal{N}}\exp \left\{ -\frac{c_{4}c_{6}^{2}T\log N/2}{%
c_{7}T+c_{6}M\sqrt{T\log N}(NT)^{1/q}(\log T)^{2}/3}\right\} +o(1)  \notag \\
& =o(1),
\end{align}%
where the last inequality holds by Lemma \ref{Lem:Bern_mds}(i), %
\eqref{Lem:var}, and the definition of event $\mathscr{A}_{4,N,i}$, and the
last line holds by Assumption \ref{ass:1*}(vi) and the fact that $q>8$. }

{\small (iii) By the fact that 
\begin{align}
& \frac{1}{nT_{\ell }^{2}}\sum_{i\in \mathcal{N}}\left\Vert e_{i}^{(\ell
)\prime }\hat{F}^{(\ell )}\right\Vert ^{2}\leq \frac{1}{nT_{\ell }^{2}}%
\sum_{i\in \mathcal{N}}\left\Vert e_{i}^{(\ell )\prime }F^{0,(\ell
)}H^{(\ell )}\right\Vert ^{2}+\frac{1}{nT_{\ell }^{2}}\sum_{i\in \mathcal{N}%
}\left\Vert e_{i}^{(\ell )\prime }\left( \hat{F}^{(\ell )}-F^{0,(\ell
)}H^{(\ell )}\right) \right\Vert ^{2}  \notag \\
& \lesssim \frac{1}{nT_{\ell }^{2}}\sum_{i\in \mathcal{N}}\left\Vert
e_{i}^{(\ell )\prime }F^{0,(\ell )}\right\Vert ^{2}+\frac{1}{nT_{\ell }}%
\sum_{i\in \mathcal{N}}\left\Vert e_{i}^{(\ell )\prime }\right\Vert ^{2}%
\frac{1}{T_{\ell }}\left\Vert \hat{F}^{(\ell )}-F^{0,(\ell )}H^{(\ell
)}\right\Vert ^{2}  \notag \\
& \leq \frac{1}{nT_{\ell }^{2}}\sum_{i\in \mathcal{N}}\left\Vert
e_{i}^{(\ell )\prime }F^{0,(\ell )}\right\Vert ^{2}+O_{p}(B_{N}^{2}+(N\wedge
T)^{-1})=O_{p}(B_{N}^{2}+(N\wedge T)^{-1}),
\end{align}%
where the last inequality holds by Assumption \ref{ass:8}(ii) and Lemma \ref%
{Lem:Fhat}, and last equality holds by the fact that 
\begin{equation*}
\frac{1}{nT_{\ell }^{2}}\sum_{i\in \mathcal{N}}\left\Vert e_{i}^{(\ell
)\prime }F^{0,(\ell )}\right\Vert _{2}^{2}=\frac{1}{nT_{\ell }^{2}}%
\sum_{i\in \mathcal{N}}\left\Vert \sum_{t\in \mathcal{T}_{\ell
}}e_{it}f_{t}^{0}\right\Vert _{2}^{2}=\frac{1}{T_{\ell }}\left[ \frac{1}{n}%
\sum_{i\in \mathcal{N}}\left\Vert \frac{1}{\sqrt{T_{\ell }}}\sum_{t\in 
\mathcal{T}_{\ell }}e_{it}f_{t}^{0}\right\Vert _{2}^{2}\right] =O_{p}(T^{-1})
\end{equation*}%
by \eqref{Lem:D1_1}.}

{\small 
}

{\small 
}

{\small (iv) Noting that $M_{\hat{F}^{(\ell )}}\hat{F}^{(\ell )}=0$, we have 
\begin{align*}
\max_{i\in \mathcal{N}}\left\Vert \frac{1}{T_{\ell }}X_{i}^{(\ell )\prime
}M_{\hat{F}^{(\ell )}}F^{0,(\ell )}\right\Vert & =\max_{i\in \mathcal{N}%
}\left\Vert \frac{1}{T_{\ell }}X_{i}^{(\ell )\prime }M_{\hat{F}^{(\ell
)}}\left( F^{0,(\ell )}-\hat{F}^{(\ell )}H^{(\ell )-1}\right) \right\Vert \\
& \leq \max_{i\in \mathcal{N}}\left\Vert \frac{1}{\sqrt{T_{\ell }}}%
X_{i}^{(\ell )}M_{\hat{F}^{(\ell )}}\right\Vert O_{p}(B_{N}+(N\wedge
T)^{-1/2})=O_{p}(B_{N}+(N\wedge T)^{-1/2}),
\end{align*}%
where the last inequality is by Lemma \ref{Lem:Fhat} and the last equality
is by the fact that 
\begin{equation}
\max_{i\in \mathcal{N}}\left\Vert \frac{1}{\sqrt{T_{\ell }}}X_{i}^{(\ell
)\prime }M_{\hat{F}^{(\ell )}}\right\Vert \leq \max_{i\in \mathcal{N}%
}\left\Vert \frac{1}{\sqrt{T_{\ell }}}X_{i}^{(\ell )\prime }\right\Vert
=O_{p}(1)  \label{Lem:pre_iv_1}
\end{equation}%
by Assumption \ref{ass:8}(ii). }

{\small (v) Note that $\frac{1}{nT_{\ell }^{2}}\sum_{i^{\ast }\in \mathcal{N}%
}\left\Vert X_{i}^{(\ell )\prime }e_{i^{\ast }}^{(\ell )}\right\Vert ^{2}=%
\frac{1}{n}\sum_{i^{\ast }\in \mathcal{N}}\left\Vert \frac{1}{T_{\ell }}%
\sum_{t\in \mathcal{T}_{\ell }}X_{it}e_{i^{\ast }t}\right\Vert ^{2}\leq
\max_{i,i^{\ast }\in \mathcal{N}}\left\Vert \frac{1}{T_{\ell }}\sum_{t\in 
\mathcal{T}_{\ell }}X_{it}e_{i^{\ast }t}\right\Vert ^{2}.$ Under Assumptions %
\ref{ass:1*} and Assumption \ref{ass:8}(ii) 
\begin{align*}
& \max_{i,i^{\ast }\in \mathcal{N},t\in \mathcal{T}_{\ell }}\left\Vert
X_{it}e_{i^{\ast }t}\right\Vert =O_{p}((NT)^{1/q}), \\
& \max_{i,i^{\ast }\in \mathcal{N}}\left\Vert \sum_{t\in \mathcal{T}_{\ell }}%
\mathbb{E}(X_{it}X_{it}^{\prime }e_{i^{\ast }t}^{2}|\mathscr{G}%
_{t-1})\right\Vert \leq \max_{i^{\ast }\in \mathcal{N},t}\mathbb{E}%
(e_{i^{\ast }t}^{2}|\mathscr{G}_{t-1})\max_{i\in \mathcal{N}}\sum_{t\in 
\mathcal{T}_{\ell }}\left\Vert X_{it}\right\Vert ^{2}\leq c_{8}T~a.s.
\end{align*}%
Define events $\mathscr{A}_{5,N}(M)=\{\max_{i,i^{\ast }\in \mathcal{N},t\in 
\mathcal{T}_{\ell }}\left\Vert X_{it}e_{i^{\ast }t}\right\Vert \leq
M(NT)^{1/q}\}$ and $\mathscr{A}_{5,N,i,i^{\ast }}(M)=\{\max_{t\in \mathcal{T}%
_{\ell }}\left\Vert X_{it}e_{i^{\ast }t}\right\Vert $ $\leq M(NT)^{1/q}\}$
for a large enough constant $M$ such that $\mathbb{P}\left( \mathscr{A}%
_{5,N}^{c}(M)\right) \rightarrow 0$. Then we have 
\begin{align}
& \mathbb{P}\left( \max_{i,i^{\ast }\in \mathcal{N}}\left\Vert \frac{1}{%
T_{\ell }}\sum_{t\in \mathcal{T}_{\ell }}X_{it}e_{i^{\ast }t}\right\Vert
>c_{6}\sqrt{\frac{\log N}{T}}\right)   \notag  \label{Lem:emp_freedman} \\
& \leq \mathbb{P}\left( \max_{i,i^{\ast }\in \mathcal{N}}\left\Vert \frac{1}{%
T_{\ell }}\sum_{t\in \mathcal{T}_{\ell }}X_{it}e_{i^{\ast }t}\right\Vert
>c_{6}\sqrt{\frac{\log N}{T}},\mathscr{A}_{5,N}(M)\right) +\mathbb{P}\left( %
\mathscr{A}_{5,N}^{c}(M)\right)   \notag \\
& \leq \sum_{i\in \mathcal{N}}\sum_{i^{\ast }\in \mathcal{N}}\mathbb{P}%
\left( \left\Vert \frac{1}{T_{\ell }}\sum_{t\in \mathcal{T}_{\ell
}}X_{it}e_{i^{\ast }t}\right\Vert >c_{6}\sqrt{\frac{\log N}{T}},\mathscr{A}%
_{5,N}(M)\right) +\mathbb{P}\left( \mathscr{A}_{5,N}^{c}(M)\right)   \notag
\\
& \leq \sum_{i\in \mathcal{N}}\sum_{i^{\ast }\in \mathcal{N}}\mathbb{P}%
\left( \left\Vert \frac{1}{T_{\ell }}\sum_{t\in \mathcal{T}_{\ell
}}X_{it}e_{i^{\ast }t}\right\Vert >c_{6}\sqrt{\frac{\log N}{T}},\mathscr{A}%
_{5,N,i,i^{\ast }}(M)\right) +\mathbb{P}\left( \mathscr{A}%
_{5,N}^{c}(M)\right)   \notag \\
& \leq \sum_{i\in \mathcal{N}}\sum_{i^{\ast }\in \mathcal{N}}\exp \left\{ 
\frac{-c_{6}^{2}T\log N/2}{c_{8}T+Mc_{6}(NT)^{1/q}\sqrt{T\log N}/3}\right\}
+o(1)  \notag \\
& =o(1),
\end{align}%
where the last inequality holds by Lemma \ref{Lem:Bern_mds}(ii) and the last
line is by Assumption \ref{ass:1*}(vi). }

{\small (vi) Noted that 
\begin{align*}
& \mathbb{E}\left[ \left\Vert \frac{1}{\sqrt{nT_{\ell }}}\sum_{i\in \mathcal{%
N}}\lambda _{i}^{0}e_{i}^{(\ell )\prime }F^{0,(\ell )}\right\Vert ^{2}\bigg|%
\mathscr{D}\right] =\frac{1}{nT_{\ell }}\sum_{i\in \mathcal{N}}\sum_{i^{\ast
}\in \mathcal{N}}\sum_{t\in \mathcal{T}_{\ell }}\sum_{t^{\ast }\in \mathcal{T%
}_{\ell }}\mathbb{E}\left( \lambda _{i}^{0}f_{t}^{0\prime }f_{t^{\ast
}}^{0}\lambda _{i^{\ast }}^{0\prime }e_{it}e_{i^{\ast }t^{\ast }}|\mathscr{D}%
\right)  \\
& \leq \max_{i\in \mathcal{N}}\left\Vert \lambda _{i}^{0}\right\Vert
_{2}^{2}\max_{t\in \mathcal{T}_{\ell }}\left\Vert f_{t}^{0}\right\Vert
_{2}^{2}\frac{1}{nT_{\ell }}\sum_{i\in \mathcal{N}}\sum_{i^{\ast }\in 
\mathcal{N}}\sum_{t\in \mathcal{T}_{\ell }}\sum_{t^{\ast }\in \mathcal{T}%
_{\ell }}\left\vert \mathbb{E}\left( e_{it}e_{i^{\ast }t^{\ast }}|\mathscr{D}%
\right) \right\vert  \\
& \lesssim \frac{1}{nT_{\ell }}\sum_{i\in \mathcal{N}}\sum_{t\in \mathcal{T}%
_{\ell }}\left\vert \mathbb{E}\left( e_{it}^{2}|\mathscr{D}\right)
\right\vert =O(1)~a.s.,
\end{align*}%
where the last line holds by Lemma \ref{Lem:bounded u&v}(i) and Assumption %
\ref{ass:1*}. Similarly as above, we can also show that $\mathbb{E}\left[
\left\Vert \frac{1}{\sqrt{nT_{\ell }}}\sum_{i\in \mathcal{N}}\lambda
_{i}^{0}e_{i}^{(\ell )\prime }\right\Vert _{{}}^{2}\bigg|\mathscr{D}\right]
=O_{p}(1)$. Then 
\begin{equation*}
\left\Vert \frac{1}{\sqrt{nT_{\ell }}}\sum_{i\in \mathcal{N}}\lambda
_{i}^{0}e_{i}^{(\ell )\prime }F^{0,(\ell )}\right\Vert =O_{p}(1)\text{ and }%
\left\Vert \frac{1}{\sqrt{nT_{\ell }}}\sum_{i\in \mathcal{N}}\lambda
_{i}^{0}e_{i}^{(\ell )\prime }\right\Vert =O_{p}(1).
\end{equation*}%
Furthermore, we have 
\begin{align*}
\left\Vert \frac{1}{nT_{\ell }}\sum_{i\in \mathcal{N}}\lambda
_{i}^{0}e_{i}^{(\ell )\prime }\hat{F}^{(\ell )}\right\Vert & \leq \frac{1}{%
\sqrt{n}}\frac{1}{\sqrt{T_{\ell }}}\left\Vert \hat{F}^{(\ell )}-F^{0,(\ell
)}H^{(\ell )}\right\Vert \left\Vert \frac{1}{\sqrt{nT_{\ell }}}\sum_{i\in 
\mathcal{N}}\lambda _{i}^{0}e_{i}^{(\ell )\prime }\right\Vert +\frac{1}{%
\sqrt{nT_{\ell }}}\left\Vert \frac{1}{\sqrt{nT_{\ell }}}\sum_{i\in \mathcal{N%
}}\lambda _{i}^{0}e_{i}^{(\ell )\prime }F^{0,(\ell )}\right\Vert \left\Vert
H^{(\ell )}\right\Vert  \\
& =O_{p}\left( \frac{B_{N}}{\sqrt{N}}+\frac{1}{N+\sqrt{NT}}\right)
+O_{p}\left( \frac{1}{\sqrt{NT}}\right) =O_{p}\left( \frac{B_{N}}{\sqrt{N}}+%
\frac{1}{N}+\frac{1}{\sqrt{NT}}\right) .
\end{align*}%
}

{\small (vii) We first notice that 
\begin{equation}
\max_{i\in \mathcal{N}}\left\Vert \frac{1}{nT_{\ell }}\sum_{i^{\ast }\in 
\mathcal{N}}X_{i}^{(\ell )\prime }e_{i^{\ast }}^{(\ell )}\lambda _{i^{\ast
}}^{0\prime }\right\Vert =\max_{i\in \mathcal{N}}\left\Vert \frac{1}{%
nT_{\ell }}\sum_{i^{\ast }\in \mathcal{N}}\sum_{t\in \mathcal{T}_{\ell
}}X_{it}e_{i^{\ast }t}\lambda _{i^{\ast }}^{0\prime }\right\Vert
=O_{p}\left( \sqrt{\frac{\log N}{NT}}\right)   \label{Lem:emp_freedman2}
\end{equation}%
by similar arguments as used to obtain \eqref{Lem:emp_freedman}. This
result, in conjunction with Lemma \ref{Lem:thetahat_pre}(vi), implies that 
\begin{align*}
& \max_{i\in \mathcal{N}}\left\Vert \frac{1}{nT_{\ell }}\sum_{i^{\ast }\in 
\mathcal{N}}X_{i}^{(\ell )\prime }M_{\hat{F}^{(\ell )}}e_{i^{\ast }}^{(\ell
)}\lambda _{i^{\ast }}^{0\prime }\right\Vert  \\
& \leq \max_{i\in \mathcal{N}}\left\Vert \frac{1}{nT_{\ell }}\sum_{i^{\ast
}\in \mathcal{N}}X_{i}^{(\ell )\prime }e_{i^{\ast }}^{(\ell )}\lambda
_{i^{\ast }}^{0\prime }\right\Vert +\max_{i\in \mathcal{N}}\left\Vert \frac{1%
}{nT_{\ell }}\sum_{i^{\ast }\in \mathcal{N}}X_{i}^{(\ell )\prime }\frac{\hat{%
F}^{(\ell )}\hat{F}^{(\ell )\prime }}{T_{\ell }}e_{i^{\ast }}^{(\ell
)}\lambda _{i^{\ast }}^{0\prime }\right\Vert  \\
& \leq \max_{i\in \mathcal{N}}\left\Vert \frac{1}{nT_{\ell }}\sum_{i^{\ast
}\in \mathcal{N}}X_{i}^{(\ell )\prime }e_{i^{\ast }}^{(\ell )}\lambda
_{i^{\ast }}^{0\prime }\right\Vert +\max_{i\in \mathcal{N}}\frac{\left\Vert
X_{i}^{(\ell )}\right\Vert _{{}}}{\sqrt{T_{\ell }}}\frac{\hat{F}^{(\ell )}}{%
\sqrt{T_{\ell }}}\left\Vert \frac{1}{nT_{\ell }}\sum_{i^{\ast }\in \mathcal{N%
}}\hat{F}^{(\ell )\prime }e_{i^{\ast }}^{(\ell )}\lambda _{i^{\ast
}}^{0\prime }\right\Vert  \\
& =O_{p}\left( \sqrt{\frac{\log N}{NT}}\right) +O_{p}\left( \frac{B_{N}}{%
\sqrt{N}}+\frac{1}{N}+\frac{1}{\sqrt{NT}}\right) =O_{p}\left( \frac{B_{N}}{%
\sqrt{N}}+\frac{1}{N}+\sqrt{\frac{\log N}{NT}}\right) .
\end{align*}%
}

{\small (viii) By \eqref{Lem:emp_freedman} and Lemma \ref{Lem:thetahat_pre}%
(iii), we have 
\begin{align*}
& \max_{i\in \mathcal{N}}\left\Vert \frac{1}{nT_{\ell }^{2}}\sum_{i^{\ast
}\in \mathcal{N}}X_{i}^{(\ell )\prime }M_{\hat{F}^{(\ell )}}e_{i^{\ast
}}^{(\ell )}e_{i^{\ast }}^{(\ell )\prime }\hat{F}^{(\ell )}\right\Vert  \\
& \leq \max_{i\in \mathcal{N}}\left\Vert \frac{1}{nT_{\ell }^{2}}%
\sum_{i^{\ast }\in \mathcal{N}}X_{i}^{(\ell )\prime }e_{i^{\ast }}^{(\ell
)}e_{i^{\ast }}^{(\ell )\prime }\hat{F}^{(\ell )}\right\Vert +\max_{i\in 
\mathcal{N}}\left\Vert \frac{1}{nT_{\ell }^{2}}\sum_{i^{\ast }\in \mathcal{N}%
}X_{i}^{(\ell )\prime }\frac{\hat{F}^{(\ell )}\hat{F}^{(\ell )\prime }}{%
T_{\ell }}e_{i^{\ast }}^{(\ell )}e_{i^{\ast }}^{(\ell )\prime }\hat{F}%
^{(\ell )}\right\Vert  \\
& \leq \max_{i\in \mathcal{N}}\sqrt{\frac{1}{n}\sum_{i^{\ast }\in \mathcal{N}%
}\left\Vert \frac{1}{T_{\ell }}\sum_{t\in \mathcal{T}_{\ell
}}X_{it}e_{i^{\ast }t}\right\Vert _{2}^{2}}\sqrt{\frac{1}{nT_{\ell }^{2}}%
\sum_{i\in \mathcal{N}}\left\Vert e_{i}^{(\ell )\prime }\hat{F}^{(\ell
)}\right\Vert _{2}^{2}}+\frac{1}{nT_{\ell }^{2}}\sum_{i\in \mathcal{N}%
}\left\Vert e_{i}^{(\ell )\prime }\hat{F}^{(\ell )}\right\Vert
_{2}^{2}\max_{i\in \mathcal{N}}\frac{\left\Vert X_{i}^{(\ell )}\right\Vert }{%
\sqrt{T_{\ell }}}\frac{\left\Vert \hat{F}^{(\ell )}\right\Vert }{\sqrt{%
T_{\ell }}} \\
& =O_{p}\left( \sqrt{\frac{\log N}{T}}\right) O_{p}\left( B_{N}+\frac{1}{%
\sqrt{N\wedge T}}\right) +O_{p}\left( B_{N}^{2}+\frac{1}{N\wedge T}\right)
=O_{p}\left( B_{N}^{2}+B_{N}\sqrt{\frac{\log N}{T}}+\frac{\sqrt{\log N}}{%
N\wedge T}\right) .
\end{align*}%
}
\end{proof}

{\small Define 
\begin{align*}
& \xi_{i}^{0,(\ell )}:=\frac{X_{i}^{(\ell )\prime
}M_{F^{0,(\ell)}}e_{i}^{(\ell )}}{T_{\ell }},\quad S_{ii^{\ast }}^{0,(\ell
)}:=\frac{X_{i}^{(\ell )\prime }M_{F^{0,(\ell )}}X_{i^{\ast }}^{(\ell )}}{%
T_{\ell }},\quad a_{ii^{\ast }}^{0}:=\lambda_{i}^{0\prime }\left( \frac{%
\Lambda_{n}^{0\prime }\Lambda_{n}^{0}}{n}\right) ^{-1}\lambda_{i^{\ast
}}^{0}, \\
& G_{ii^{\ast }}^{0,(\ell )}:=S_{ii^{\ast }}^{0,(\ell
)}a_{ii^{\ast}}^{0},\quad \text{and }\Omega_{i}^{0,(\ell
)}:=Var(\xi_{i}^{0,(\ell )}).
\end{align*}
}

\begin{lemma}
{\small \label{Lem:heter_slope} Under Assumptions \ref{ass:1*}, \ref{ass:2}
and \ref{ass:8}, we have }

\begin{itemize}
\item[(i)] {\small $\mathbb{E}(S_{ii}^{0,(\ell )}|\mathscr{D})(\hat{\theta}%
_{i}^{(\ell )}-\theta _{i}^{0,(\ell )})(1-\frac{a_{ii}^{0}}{n})=\xi
_{i}^{0,(\ell )}+\mathcal{R}_{i}^{(\ell )}$ such that $\max_{i\in \mathcal{N}%
}||\mathcal{R}_{i}^{(\ell )}||=O_{p}(\log N/(N\wedge T))$, }

\item[(ii)] {\small $\sqrt{T_{\ell }}(\Omega _{i}^{0,(\ell )})^{-1/2}\mathbb{%
E}(S_{ii}^{0,(\ell )}|\mathscr{D})(\hat{\theta}_{i}^{(\ell )}-\theta
_{i}^{0,(\ell )})(1-\frac{a_{ii}^{0}}{n})\rightsquigarrow \mathbb{N}(0,1)$, }

\item[(iii)] {\small $\max_{i\in \mathcal{N}}||\hat{\theta}_{i}^{(\ell
)}-\theta _{i}^{0,(\ell )}||=O_{p}(\sqrt{(\log N)/T})$. }
\end{itemize}
\end{lemma}

\begin{proof}
{\small (i) Noting from \eqref{PIFE_est_theta} that $\hat{\theta}_{i}^{(\ell
)}=(\hat{S}_{i}^{\left( \ell \right) })^{-1}T_{\ell }^{-1}X_{i}^{(\ell
)\prime }M_{\hat{F}^{(\ell )}}Y_{i}^{(\ell )}$ with }${\small \hat{S}}_{%
{\small ii}}^{\left( {\small \ell }\right) }{\small =}${\small $T_{\ell
}^{-1}X_{i}^{(\ell )\prime }M_{\hat{F}^{(\ell )}}X_{i}^{(\ell )}$, we have 
\begin{align}
& \hat{\theta}_{i}^{(\ell )}-\theta _{i}^{0,(\ell )}=(\hat{S}_{ii}^{\left(
\ell \right) })^{-1}\left[ \frac{1}{T_{\ell }}X_{i}^{(\ell )\prime }M_{\hat{F%
}^{(\ell )}}e_{i}^{(\ell )}+\frac{1}{T_{\ell }}X_{i}^{(\ell )\prime }M_{\hat{%
F}^{(\ell )}}F^{0,(\ell )}\lambda _{i}^{0}\right]  \notag \\
& =(\hat{S}_{ii}^{\left( \ell \right) })^{-1}\frac{X_{i}^{(\ell )\prime }M_{%
\hat{F}^{(\ell )}}e_{i}^{(\ell )}}{T_{\ell }}+(\hat{D}_{i}^{\left( \ell
\right) })^{-1}\frac{1}{T_{\ell }}X_{i}^{(\ell )\prime }M_{\hat{F}^{(\ell )}}%
\left[ \hat{F}^{(\ell )}H^{(\ell )-1}-\sum_{m\in \lbrack 8]}J_{m}^{(\ell
)}\left( \frac{F^{0,(\ell )\prime }\hat{F}^{(\ell )}}{T_{\ell }}\right)
^{-1}\left( \frac{\Lambda _{n}^{0\prime }\Lambda _{n}^{0}}{n}\right) ^{-1}%
\right] \lambda _{i}^{0}  \notag \\
& =(\hat{S}_{ii}^{\left( \ell \right) })^{-1}\frac{X_{i}^{(\ell )\prime }M_{%
\hat{F}^{(\ell )}}e_{i}^{(\ell )}}{T_{\ell }}-(\hat{D}_{i}^{\left( \ell
\right) })^{-1}\frac{1}{T_{\ell }}X_{i}^{(\ell )\prime }M_{\hat{F}^{(\ell
)}}\sum_{m\in \lbrack 8]}J_{m}^{(\ell )}\left( \frac{F^{0,(\ell )\prime }%
\hat{F}^{(\ell )}}{T_{\ell }}\right) ^{-1}\left( \frac{\Lambda _{n}^{0\prime
}\Lambda _{n}^{0}}{n}\right) ^{-1}\lambda _{i}^{0},  \label{Lem:thetahat_1}
\end{align}%
where the second equality is from \eqref{Lem:Fhat_1}. Note that 
\begin{equation*}
\max_{i\in \mathcal{N}}\left\Vert \frac{1}{T_{\ell }}X_{i}^{(\ell )\prime
}M_{\hat{F}^{(\ell )}}\sum_{m\in \lbrack 8]}J_{m}^{(\ell )}\left( \frac{%
F^{0,(\ell )\prime }\hat{F}^{(\ell )}}{T_{\ell }}\right) ^{-1}\left( \frac{%
\Lambda _{n}^{0\prime }\Lambda _{n}^{0}}{n}\right) ^{-1}\lambda
_{i}^{0}\right\Vert \lesssim \sum_{m\in \lbrack 8]}\max_{i\in \mathcal{N}%
}\left\Vert \frac{1}{T_{\ell }}X_{i}^{(\ell )\prime }M_{\hat{F}^{(\ell
)}}J_{m}^{(\ell )}\right\Vert =:\sum_{m\in \lbrack 8]}II_{m},
\end{equation*}%
by Lemmas \ref{Lem:bounded u&v}(i) and \ref{Lem:Fhat}, and the normalization
of the factor and factor loadings. Hence, it suffices to show the uniform
convergence rate }${\small II}_{{\small m}}${\small $~$\ for $m\in \lbrack
8]\backslash \{2\}$. The term associated with }${\small II}_{{\small 2}}$ 
{\small needs to be kept.}

{\small For }${\small II}_{{\small 1}},${\small \ we have $II_{1}\leq
\max_{i\in \mathcal{N}}\left\Vert \frac{1}{\sqrt{T_{\ell }}}X_{i}^{(\ell
)\prime }\right\Vert \left\Vert J_{1}^{(\ell )}\right\Vert
=O_{p}(B_{N}^{2})\ $by \eqref{Lem:Fhat_J1} and \eqref{Lem:pre_iv_1}. Next,
noting that 
\begin{align*}
II_{2,i}& :=\frac{1}{T_{\ell }}X_{i}^{(\ell )\prime }M_{\hat{F}^{(\ell
)}}J_{2}^{(\ell )}\left( \frac{F^{0,(\ell )\prime }\hat{F}^{(\ell )}}{%
T_{\ell }}\right) ^{-1}\left( \frac{\Lambda _{n}^{0\prime }\Lambda _{n}^{0}}{%
n}\right) ^{-1}\lambda _{i}^{0} \\
& =\frac{1}{T_{\ell }}X_{i}^{(\ell )\prime }M_{\hat{F}^{(\ell )}}\frac{1}{%
nT_{\ell }}\sum_{i^{\ast }\in \mathcal{N}}X_{i^{\ast }}^{(\ell )}\left(
\theta _{i^{\ast }}^{0,(\ell )}-\hat{\theta}_{i^{\ast }}^{(\ell )}\right)
\lambda _{i^{\ast }}^{0\prime }F^{0,(\ell )\prime }\hat{F}^{(\ell )}\left( 
\frac{F^{0,(\ell )\prime }\hat{F}^{(\ell )}}{T_{\ell }}\right) ^{-1}\left( 
\frac{\Lambda _{n}^{0\prime }\Lambda _{n}^{0}}{n}\right) ^{-1}\lambda
_{i}^{0} \\
& =\frac{1}{n}\sum_{i^{\ast }\in \mathcal{N}}\frac{X_{i}^{(\ell )\prime }M_{%
\hat{F}^{(\ell )}}X_{i^{\ast }}^{(\ell )}}{T_{\ell }}\left( \theta _{i^{\ast
}}^{0,(\ell )}-\hat{\theta}_{i^{\ast }}^{(\ell )}\right) \lambda _{i^{\ast
}}^{0\prime }\left( \frac{\Lambda _{n}^{0\prime }\Lambda _{n}^{0}}{n}\right)
^{-1}\lambda _{i}^{0},
\end{align*}%
we have $\max_{i\in \mathcal{N}}\left\Vert II_{2,i}\right\Vert =O_{p}(B_{N}),
$ and this term will be kept in the linear expansion for $\hat{\theta}%
_{i}^{(\ell )}$. For }${\small II}_{{\small 3}},${\small \ we have 
\begin{align*}
II_{3}& =\max_{i\in \mathcal{N}}\left\Vert \frac{1}{T_{\ell }}X_{i}^{(\ell
)\prime }M_{\hat{F}^{(\ell )}}\frac{1}{nT_{\ell }}\sum_{i^{\ast }\in 
\mathcal{N}}X_{i^{\ast }}^{(\ell )}\left( \theta _{i^{\ast }}^{0,(\ell )}-%
\hat{\theta}_{i^{\ast }}^{(\ell )}\right) e_{i^{\ast }}^{(\ell )\prime }\hat{%
F}^{(\ell )}\right\Vert  \\
& \leq \max_{i\in \mathcal{N}}\left\Vert \frac{1}{\sqrt{T_{\ell }}}%
X_{i}^{(\ell )\prime }\right\Vert \max_{i\in \mathcal{N}}\left\Vert \theta
_{i^{\ast }}^{0,(\ell )}-\hat{\theta}_{i^{\ast }}^{(\ell )}\right\Vert \sqrt{%
\frac{1}{nT_{\ell }}\sum_{i^{\ast }\in \mathcal{N}}\left\Vert X_{i^{\ast
}}^{(\ell )}\right\Vert ^{2}}\sqrt{\frac{1}{nT_{\ell }^{2}}\sum_{i^{\ast
}\in \mathcal{N}}\left\Vert e_{i^{\ast }}^{(\ell )\prime }\hat{F}^{(\ell
)}\right\Vert } \\
& =O_{p}(B_{N}^{2}+B_{N}(N\wedge T)^{-1/2})
\end{align*}%
by \eqref{Lem:pre_iv_1}, Assumption \ref{ass:8}(ii) and Lemma \ref%
{Lem:thetahat_pre}(iii). For }${\small II}_{{\small 4}},${\small \ we have 
\begin{align*}
II_{4}& =\max_{i\in \mathcal{N}}\left\Vert \frac{1}{T_{\ell }}X_{i}^{(\ell
)\prime }M_{\hat{F}^{(\ell )}}\frac{1}{nT_{\ell }}\sum_{i^{\ast }\in 
\mathcal{N}}F^{0,(\ell )}\lambda _{i^{\ast }}^{0}\left( \theta _{i^{\ast
}}^{0,(\ell )}-\hat{\theta}_{i^{\ast }}^{(\ell )}\right) ^{\prime
}X_{i^{\ast }}^{(\ell )\prime }\hat{F}^{(\ell )}\right\Vert  \\
& \leq \max_{i\in \mathcal{N}}\left\Vert \frac{1}{T_{\ell }}X_{i}^{(\ell
)\prime }M_{\hat{F}^{(\ell )}}F^{0,(\ell )}\right\Vert \max_{i\in \mathcal{N}%
}\left\Vert \lambda _{i}^{0}\right\Vert \max_{i\in \mathcal{N}}\left\Vert
\theta _{i}^{0,(\ell )}-\hat{\theta}_{i}^{(\ell )}\right\Vert \frac{1}{n%
\sqrt{T_{\ell }}}\sum_{i^{\ast }\in \mathcal{N}}\left\Vert X_{i^{\ast
}}^{(\ell )}\right\Vert \frac{\left\Vert \hat{F}^{(\ell )}\right\Vert _{{}}}{%
\sqrt{T_{\ell }}} \\
& =O_{p}(B_{N}+(N\wedge T)^{-1/2})O_{p}(B_{N})=O_{p}(B_{N}^{2}+B_{N}(N\wedge
T)^{-1/2}),
\end{align*}%
where the last line holds by Lemma \ref{Lem:thetahat_pre}(iv), the
normalization of factors and the fact that $\frac{1}{n\sqrt{T_{\ell }}}%
\sum_{i\in \mathcal{N}}||X_{i}^{(\ell )}||=\frac{1}{n}\sum_{i\in \mathcal{N}}%
\sqrt{\frac{1}{T_{\ell }}\sum_{t\in \mathcal{T}_{\ell }}\left\Vert
X_{it}\right\Vert ^{2}}=O_{p}(1)$ by Assumption \ref{ass:8}(ii). For }$%
{\small II}_{{\small 5}},${\small \ we have 
\begin{align*}
II_{5}& =\max_{i\in \mathcal{N}}\left\Vert \frac{1}{T_{\ell }}X_{i}^{(\ell
)\prime }M_{\hat{F}^{(\ell )}}\frac{1}{nT_{\ell }}\sum_{i^{\ast }\in 
\mathcal{N}}e_{i^{\ast }}^{(\ell )}\left( \theta _{i^{\ast }}^{0,(\ell )}-%
\hat{\theta}_{i^{\ast }}^{(\ell )}\right) ^{\prime }X_{i^{\ast }}^{(\ell
)\prime }\hat{F}^{(\ell )}\right\Vert  \\
& \lesssim \max_{i\in \mathcal{N}}\left\Vert \frac{1}{nT_{\ell }\sqrt{%
T_{\ell }}}\sum_{i^{\ast }\in \mathcal{N}}X_{i}^{(\ell )\prime }e_{i^{\ast
}}^{(\ell )}\left( \theta _{i^{\ast }}^{0,(\ell )}-\hat{\theta}_{i^{\ast
}}^{(\ell )}\right) ^{\prime }X_{i^{\ast }}^{(\ell )\prime }\right\Vert  \\
& +\max_{i\in \mathcal{N}}\left\Vert \frac{1}{T_{\ell }}X_{i}^{(\ell )\prime
}\hat{F}^{(\ell )}\frac{1}{n\sqrt{T_{\ell }}}\sum_{i^{\ast }\in \mathcal{N}}%
\frac{1}{T_{\ell }}\hat{F}^{(\ell )\prime }e_{i^{\ast }}^{(\ell )}\left(
\theta _{i^{\ast }}^{0,(\ell )}-\hat{\theta}_{i^{\ast }}^{(\ell )}\right)
^{\prime }X_{i^{\ast }}^{(\ell )\prime }\right\Vert  \\
& \leq \max_{i\in \mathcal{N}}\left\Vert \theta _{i}^{0,(\ell )}-\hat{\theta}%
_{i}^{(\ell )}\right\Vert _{2}\sqrt{\frac{1}{nT_{\ell }^{2}}\sum_{i^{\ast
}\in \mathcal{N}}\left\Vert X_{i}^{(\ell )\prime }e_{i^{\ast }}^{(\ell
)}\right\Vert ^{2}}\sqrt{\frac{1}{nT_{\ell }}\sum_{i^{\ast }\in \mathcal{N}%
}\left\Vert X_{i^{\ast }}^{(\ell )\prime }\right\Vert ^{2}} \\
& +\max_{i\in \mathcal{N}}\left\Vert \theta _{i}^{0,(\ell )}-\hat{\theta}%
_{i}^{(\ell )}\right\Vert _{2}\max_{i\in \mathcal{N}}\frac{\left\Vert
X_{i}^{(\ell )}\right\Vert }{\sqrt{T_{\ell }}}\frac{\left\Vert \hat{F}%
^{(\ell )}\right\Vert }{\sqrt{T_{\ell }}}\sqrt{\frac{1}{nT_{\ell }^{2}}%
\sum_{i^{\ast }\in \mathcal{N}}\left\Vert e_{i^{\ast }}^{(\ell )\prime }\hat{%
F}^{(\ell )}\right\Vert ^{2}}\sqrt{\frac{1}{nT_{\ell }}\sum_{i^{\ast }\in 
\mathcal{N}}\left\Vert X_{i^{\ast }}^{(\ell )}\right\Vert ^{2}} \\
& =O_{p}(B_{N})O_{p}(\sqrt{(\log N)/T})+O_{p}(B_{N})O_{p}\left(
B_{N}+(N\wedge T)^{-1/2}\right) =O_{p}(B_{N}^{2}+B_{N}\sqrt{(\log
N)/(N\wedge T)}),
\end{align*}%
where the last line holds by Lemma \ref{Lem:thetahat_pre}(iii), \ref%
{Lem:thetahat_pre}(v) and the fact that $\frac{1}{nT_{\ell }}\sum_{i^{\ast
}\in \mathcal{N}}||X_{i^{\ast }}^{(\ell )}||^{2}=O_{p}(1)$ under Assumption %
\ref{ass:8}(ii). Next, 
\begin{align*}
II_{6}& =\max_{i\in \mathcal{N}}\left\Vert \frac{1}{T_{\ell }}X_{i}^{(\ell
)\prime }M_{\hat{F}^{(\ell )}}\frac{1}{nT_{\ell }}\sum_{i^{\ast }\in 
\mathcal{N}}F^{0,(\ell )}\lambda _{i^{\ast }}^{0}e_{i^{\ast }}^{(\ell
)\prime }\hat{F}^{(\ell )}\right\Vert \leq \max_{i\in \mathcal{N}}\left\Vert 
\frac{1}{T_{\ell }}X_{i}^{(\ell )\prime }M_{\hat{F}^{(\ell )}}F^{0,(\ell
)}\right\Vert \left\Vert \frac{1}{nT_{\ell }}\sum_{i^{\ast }\in \mathcal{N}%
}\lambda _{i^{\ast }}^{0}e_{i^{\ast }}^{(\ell )\prime }\hat{F}^{(\ell
)}\right\Vert  \\
& =O_{p}((B_{N}+(N\wedge T)^{-1/2})O_{p}(B_{N}N^{-1/2}+N^{-1}+(NT)^{-1/2})
\end{align*}%
by Lemma \ref{Lem:thetahat_pre}(iv) and \ref{Lem:thetahat_pre}(vi). By} 
{\small Lemma \ref{Lem:thetahat_pre}(vii), 
\begin{align}
II_{7}& =\max_{i\in \mathcal{N}}\left\Vert \frac{1}{T_{\ell }}X_{i}^{(\ell
)\prime }M_{\hat{F}^{(\ell )}}\frac{1}{nT_{\ell }}\sum_{i^{\ast }\in 
\mathcal{N}}e_{i^{\ast }}^{(\ell )}\lambda _{i^{\ast }}^{0\prime }F^{0,(\ell
)\prime }\hat{F}^{(\ell )}\right\Vert   \notag  \label{Lem:heter_3} \\
& \lesssim \max_{i\in \mathcal{N}}\left\Vert \frac{1}{nT_{\ell }}%
\sum_{i^{\ast }\in \mathcal{N}}X_{i}^{(\ell )\prime }M_{\hat{F}^{(\ell
)}}e_{i^{\ast }}^{(\ell )}\lambda _{i^{\ast }}^{0\prime }\right\Vert
=O_{p}\left( B_{N}N^{-1/2}+N^{-1}+(NT/\log N)^{-1/2}\right) .
\end{align}%
By Lemma \ref{Lem:thetahat_pre}(viii), 
\begin{equation*}
II_{8}=\max_{i\in \mathcal{N}}\left\Vert \frac{1}{T_{\ell }}X_{i}^{(\ell
)\prime }M_{\hat{F}^{(\ell )}}\frac{1}{nT_{\ell }}\sum_{i^{\ast }\in 
\mathcal{N}}e_{i^{\ast }}^{(\ell )}e_{i^{\ast }}^{(\ell )\prime }\hat{F}%
^{(\ell )}\right\Vert =O_{p}\left( B_{N}^{2}+B_{N}\sqrt{(\log N)/T}+\sqrt{%
\log N}(N\wedge T)^{-1}\right) .
\end{equation*}%
}

{\small In sum, by \eqref{Lem:thetahat_1} and the above analyses, we have 
\begin{align}
\hat{\theta}_{i}^{(\ell )}-\theta _{i}^{0,(\ell )}& =(\hat{S}_{ii}^{\left(
\ell \right) })^{-1}\frac{X_{i}^{(\ell )\prime }M_{\hat{F}^{(\ell
)}}e_{i}^{(\ell )}}{T_{\ell }}+(\hat{S}_{ii}^{\left( \ell \right) })^{-1}%
\frac{1}{n}\sum_{i^{\ast }\in \mathcal{N}}\frac{X_{i}^{(\ell )\prime }M_{%
\hat{F}^{(\ell )}}X_{i^{\ast }}^{(\ell )}}{T_{\ell }}\left( \hat{\theta}%
_{i^{\ast }}^{(\ell )}-\theta _{i^{\ast }}^{0,(\ell )}\right) \lambda
_{i^{\ast }}^{0\prime }\left( \frac{\Lambda _{n}^{0\prime }\Lambda _{n}^{0}}{%
n}\right) ^{-1}\lambda _{i}^{0}  \notag \\
& +O_{p}\left( B_{N}^{2}+B_{N}\sqrt{\frac{\log N}{N\wedge T}}+\frac{\sqrt{%
\log N}}{N\wedge T}\right) 
\end{align}%
uniformly in $i\in \mathcal{N}$. Then by \eqref{Lem:thetahat_1} and Lemma %
\ref{Lem:thetahat_pre}(i)-(ii), we have 
\begin{align*}
\hat{\theta}_{i}^{(\ell )}-\theta _{i}^{0,(\ell )}& =\left( S_{ii}^{0,(\ell
)}\right) ^{-1}\frac{X_{i}^{(\ell )\prime }M_{F^{0,(\ell )}}e_{i}^{(\ell )}}{%
T_{\ell }}+\left( S_{ii}^{0,(\ell )}\right) ^{-1}\frac{1}{n}\sum_{i^{\ast
}\in \mathcal{N}}\frac{X_{i}^{(\ell )\prime }M_{F^{0,(\ell )}}X_{i^{\ast
}}^{(\ell )}}{T_{\ell }}\left( \hat{\theta}_{i^{\ast }}^{(\ell )}-\theta
_{i^{\ast }}^{0,(\ell )}\right) \lambda _{i^{\ast }}^{0\prime }\left( \frac{%
\Lambda _{n}^{0\prime }\Lambda _{n}^{0}}{n}\right) ^{-1}\lambda _{i}^{0} \\
& +O_{p}\left( B_{N}^{2}+B_{N}\sqrt{\frac{\log N}{N\wedge T}}+\frac{\sqrt{%
\log N}}{N\wedge T}\right) .
\end{align*}%
where recall that }${\small S}_{{\small ii}^{\ast }}^{{\small 0,}({\small %
\ell })}{\small :=}${\small $T_{\ell }^{-1}X_{i}^{(\ell )\prime
}M_{F^{0,(\ell )}}X_{i^{\ast }}^{(\ell )}.$} {\small Let $\mathfrak{x}%
_{it}^{0,(\ell )}$ be the $t$-th row of matrix $M_{F^{0,(\ell
)}}X_{i}^{(\ell )}$ and note that $\mathfrak{x}_{it}^{0,(\ell )}$ is strong
mixing across $t$ and independent across $i$ conditional on $\mathscr{D}$ by
Assumption \ref{ass:1*}(i), (iii). Then we can show that 
\begin{equation*}
S_{ii^{\ast }}^{0,(\ell )}-\mathbb{E}(S_{ii^{\ast }}^{0,(\ell )}|\mathscr{D}%
)=\frac{1}{T_{\ell }}\sum_{t\in \mathcal{T}_{\ell }}[\mathfrak{x}%
_{it}^{0,(\ell )}\mathfrak{x}_{i^{\ast }t}^{0,(\ell )\prime }-\mathbb{E}(%
\mathfrak{x}_{it}^{0,(\ell )}\mathfrak{x}_{i^{\ast }t}^{0,(\ell )\prime }|%
\mathscr{D})]=O_{p}\left( (T/\log N)^{-1/2}\right) 
\end{equation*}%
uniformly over $i,i^{\ast }\in \mathcal{N}$ by similar arguments as in %
\eqref{Lem:emp_Bern}. Then by the fact that 
\begin{equation*}
\max_{i,i^{\ast }\in \mathcal{N}}\left\Vert \frac{X_{i}^{(\ell )\prime
}M_{F^{0,(\ell )}}e_{i}^{(\ell )}}{T_{\ell }}\right\Vert =\max_{i,i^{\ast
}\in \mathcal{N}}\left\Vert \frac{1}{T_{\ell }}\sum_{t\in \mathcal{T}_{\ell
}}X_{it}e_{it}\right\Vert +\max_{i,i^{\ast }\in \mathcal{N}}\left\Vert \frac{%
X_{i}^{(\ell )\prime }F^{0,(\ell )}}{T_{\ell }}\right\Vert \left\Vert \frac{1%
}{T_{\ell }}\sum_{t\in \mathcal{T}_{\ell }}f_{t}^{0}e_{it}\right\Vert
=O_{p}\left( (T/\log N)^{-1/2}\right) 
\end{equation*}%
by \eqref{Lem:emp_Bern} and \eqref{Lem:emp_freedman}, we obtain that 
\begin{align}
\hat{\theta}_{i}^{(\ell )}-\theta _{i}^{0,(\ell )}& =\left[ \mathbb{E}\left(
S_{ii}^{0,(\ell )}\big|\mathscr{D}\right) \right] ^{-1}\xi _{i}^{0,(\ell )}+%
\left[ \mathbb{E}\left( S_{ii}^{0,(\ell )}\big|\mathscr{D}\right) \right]
^{-1}\frac{1}{n}\sum_{i^{\ast }\in \mathcal{N}}\mathbb{E}\left( G_{ii^{\ast
}}^{0,(\ell )}\big|\mathscr{D}\right) \left( \hat{\theta}_{i^{\ast }}^{(\ell
)}-\theta _{i^{\ast }}^{0,(\ell )}\right)   \notag  \label{Lem:thetahat_3} \\
& +O_{p}\left( B_{N}^{2}+B_{N}\sqrt{\frac{\log N}{N\wedge T}}+\frac{\sqrt{%
\log N}}{N\wedge T}\right) .
\end{align}%
}

{\small For the second term on the right side of \eqref{Lem:thetahat_3}, we
observe that 
\begin{align*}
& \left[ \mathbb{E}\left( S_{ii}^{0,(\ell )}\big|\mathscr{D}\right) \right]
^{-1}\frac{1}{n}\sum_{i^{\ast }\in \mathcal{N}}\mathbb{E}\left( G_{ii^{\ast
}}^{0,(\ell )}\big|\mathscr{D}\right) \left( \hat{\theta}_{i^{\ast }}^{(\ell
)}-\theta _{i^{\ast }}^{0,(\ell )}\right)  \\
& =\left[ \mathbb{E}\left( S_{ii}^{0,(\ell )}\big|\mathscr{D}\right) \right]
^{-1}\mathbb{E}\left( G_{ii}^{0,(\ell )}\big|\mathscr{D}\right) \frac{1}{n}%
\left( \hat{\theta}_{i}^{(\ell )}-\theta _{i}^{0,(\ell )}\right) +\left[ 
\mathbb{E}\left( S_{ii}^{0,(\ell )}\big|\mathscr{D}\right) \right] ^{-1}%
\frac{1}{n}\sum_{i^{\ast }\neq i}\mathbb{E}\left( G_{ii^{\ast }}^{0,(\ell )}%
\big|\mathscr{D}\right) \left( \hat{\theta}_{i^{\ast }}^{(\ell )}-\theta
_{i^{\ast }}^{0,(\ell )}\right)  \\
& =\frac{a_{ii}^{0}}{n}\left( \hat{\theta}_{i}^{(\ell )}-\theta
_{i}^{0,(\ell )}\right) +\left[ \mathbb{E}\left( S_{ii}^{0,(\ell )}\big|%
\mathscr{D}\right) \right] ^{-1}\frac{1}{n}\sum_{i^{\ast }\neq i}\mathbb{E}%
\left( G_{ii^{\ast }}^{0,(\ell )}\big|\mathscr{D}\right) \left( \hat{\theta}%
_{i^{\ast }}^{(\ell )}-\theta _{i^{\ast }}^{0,(\ell )}\right) .
\end{align*}%
By \eqref{Lem:thetahat_3}, it's clear that 
\begin{align*}
\frac{1}{n}\sum_{i^{\ast }\neq i}\mathbb{E}\left( G_{ii^{\ast }}^{0,(\ell )}%
\big|\mathscr{D}\right) \left( \hat{\theta}_{i^{\ast }}^{(\ell )}-\theta
_{i^{\ast }}^{0,(\ell )}\right) & =\frac{1}{n}\sum_{i^{\ast }\neq i}\mathbb{E%
}\left( G_{ii^{\ast }}^{0,(\ell )}\big|\mathscr{D}\right) \left[ \mathbb{E}%
\left( S_{i^{\ast }i^{\ast }}^{0,(\ell )}\big|\mathscr{D}\right) \right]
^{-1}\xi _{i^{\ast }}^{0,(\ell )} \\
& +\frac{1}{n}\sum_{i^{\ast }\neq i}\mathbb{E}\left( G_{ii^{\ast }}^{0,(\ell
)}\big|\mathscr{D}\right) \left[ \mathbb{E}\left( S_{i^{\ast }i^{\ast
}}^{0,(\ell )}\big|\mathscr{D}\right) \right] ^{-1}\frac{1}{n}%
\sum_{j=n_{1}}^{n_{n}}\mathbb{E}\left( G_{i^{\ast }j}^{0,(\ell )}\big|%
\mathscr{D}\right) \left( \hat{\theta}_{j}^{(\ell )}-\theta _{j}^{0,(\ell
)}\right)  \\
& +O_{p}\left( B_{N}^{2}+B_{N}\sqrt{\frac{\log N}{N\wedge T}}+\frac{\sqrt{%
\log N}}{N\wedge T}\right) 
\end{align*}%
where the second term on the right side of the above equality gives the
recursive form and shrinks to zero quickly owing to the $\frac{1}{n^{k}}$
term, and we only need to show the rate of the first term, i.e., 
\begin{align*}
\frac{1}{n}\sum_{i^{\ast }\neq i}\mathbb{E}\left( G_{ii^{\ast }}^{0,(\ell )}%
\big|\mathscr{D}\right) \left[ \mathbb{E}\left( S_{i^{\ast }i^{\ast
}}^{0,(\ell )}\big|\mathscr{D}\right) \right] ^{-1}\xi _{i^{\ast }}^{0,(\ell
)}& =\frac{1}{nT_{\ell }}\sum_{i^{\ast }\neq i}\sum_{t\in \mathcal{T}_{\ell
}}\mathbb{E}\left( G_{ii^{\ast }}^{0,(\ell )}\big|\mathscr{D}\right) \left[ 
\mathbb{E}\left( S_{i^{\ast }i^{\ast }}^{0,(\ell )}\big|\mathscr{D}\right) %
\right] ^{-1}\mathfrak{x}_{i^{\ast }t}e_{i^{\ast }t} \\
& =O_{p}\left( \sqrt{\frac{\log N}{NT}}\right) ~\text{uniformly over $i\in 
\mathcal{N}$,}
\end{align*}%
similarly to the result in \eqref{Lem:emp_freedman2}. This yields 
\begin{equation*}
\max_{i\in \mathcal{N}}\left\Vert \frac{1}{n}\sum_{i^{\ast }\neq i}\mathbb{E}%
\left( G_{ii^{\ast }}^{0,(\ell )}\big|\mathscr{D}\right) \left( \hat{\theta}%
_{i^{\ast }}^{(\ell )}-\theta _{i^{\ast }}^{0,(\ell )}\right) \right\Vert
=O_{p}\left( B_{N}^{2}+B_{N}\sqrt{\frac{\log N}{N\wedge T}}+\frac{\sqrt{\log
N}}{N\wedge T}\right) 
\end{equation*}%
and further gives 
\begin{equation*}
\left( \hat{\theta}_{i}^{(\ell )}-\theta _{i}^{0,(\ell )}\right) \left( 1-%
\frac{a_{ii}^{0}}{n}\right) =\left[ \mathbb{E}\left( S_{ii}^{0,(\ell )}\big|%
\mathscr{D}\right) \right] ^{-1}\xi _{i}^{0,(\ell )}+O_{p}\left(
B_{N}^{2}+B_{N}\sqrt{\frac{\log N}{N\wedge T}}+\frac{\sqrt{\log N}}{N\wedge T%
}\right) 
\end{equation*}%
with $B_{N}=\max_{i\in \mathcal{N}}||\hat{\theta}_{i}^{(\ell )}-\theta
_{i}^{0,(\ell )}||=O_{p}(\sqrt{(\log N)/T})$. Finally, we obtain that 
\begin{equation*}
\mathbb{E}\left( S_{ii}^{0,(\ell )}\big|\mathscr{D}\right) \left( \hat{\theta%
}_{i}^{(\ell )}-\theta _{i}^{0,(\ell )}\right) \left( 1-\frac{a_{ii}^{0}}{n}%
\right) =\xi _{i}^{0,(\ell )}+\mathcal{R}_{i}^{(\ell )}
\end{equation*}%
such that $\max_{i\in \mathcal{N}}\left\Vert \mathcal{R}_{i}^{(\ell
)}\right\Vert _{2}=O_{p}\left( \frac{\log N}{N\wedge T}\right) $. }

{\small (ii) Given the definition of $\Omega _{i}^{0,(\ell )}$ and by the
central limit theorem for m.d.s., we can easily obtain (ii). }

{\small (iii) The proof has already been done in the proof of (i). 
}
\end{proof}

\begin{lemma}
{\small \label{Lem:factor_uniform} Under Assumptions \ref{ass:1*}, \ref%
{ass:2} and \ref{ass:8}, we have }

\begin{itemize}
\item[(i)] {\small $\frac{1}{\sqrt{T_{\ell }}}\left\Vert \hat{F}^{(\ell
)}-F^{0,(\ell )}H^{(\ell )}\right\Vert =O_{p}\left( \sqrt{\frac{\log N}{%
N\wedge T}}\right) $, }

\item[(ii)] {\small $\left\Vert M_{\hat{F}^{(\ell )}}-M_{F^{0,(\ell
)}}\right\Vert =O_{p}\left( \sqrt{\frac{\log N}{N\wedge T}}\right) $, }

\item[(iii)] {\small $\max_{i\in \mathcal{N}}\left\Vert \hat{\lambda}%
_{i}^{(\ell )}-H^{(\ell )-1}\lambda _{i}^{0}\right\Vert =O_{p}\left( \sqrt{%
\frac{\log N}{N\wedge T}}\right) $, }

\item[(iv)] {\small $\max_{t}\left\Vert \hat{f}_{t}^{(\ell )}-H^{(\ell
)\prime }f_{t}^{0}\right\Vert =O_{p}\left( \sqrt{\frac{\log (N\vee T)}{%
N\wedge T}}\right) $. }
\end{itemize}
\end{lemma}

\begin{proof}
{\small (i) We obtain the result by combining Lemma \ref{Lem:Fhat} and Lemma %
\ref{Lem:heter_slope}(iii). }

{\small (ii) We obtain the result by combining \eqref{Lem4:MF_3} and Lemma %
\ref{Lem:heter_slope}(iii). }

{\small (iii) Recall that 
\begin{align*}
\hat{\lambda}_{i}^{(\ell )}& =\left( \hat{F}^{(\ell )\prime }\hat{F}^{(\ell
)}\right) ^{-1}\hat{F}^{(\ell )\prime }\left( Y_{i}^{(\ell )}-X_{i}^{(\ell
)\prime }\hat{\theta}_{i}^{(\ell )}\right) =\frac{1}{T_{\ell }}\hat{F}%
^{(\ell )\prime }\left[ Y_{i}^{(\ell )}-X_{i}^{(\ell )\prime }\theta
_{i}^{0,(\ell )}-X_{i}^{(\ell )\prime }\left( \hat{\theta}_{i}^{(\ell
)}-\theta _{i}^{0,(\ell )}\right) \right] \\
& =\frac{1}{T_{\ell }}\hat{F}^{(\ell )\prime }\left[ F^{0,(\ell )}\lambda
_{i}^{0}+e_{i}^{(\ell )}-X_{i}^{(\ell )\prime }\left( \hat{\theta}%
_{i}^{(\ell )}-\theta _{i}^{0,(\ell )}\right) \right] \\
& =\frac{1}{T_{\ell }}\hat{F}^{(\ell )\prime }\left( F^{0,(\ell )}-\hat{F}%
^{(\ell )}H^{(\ell )-1}\right) \lambda _{i}^{0}+\frac{1}{T_{\ell }}\hat{F}%
^{(\ell )\prime }\hat{F}^{(\ell )}H^{(\ell )-1}\lambda _{i}^{0}+\frac{1}{%
T_{\ell }}\hat{F}^{(\ell )\prime }e_{i}^{(\ell )}-\frac{1}{T_{\ell }}\hat{F}%
^{(\ell )\prime }X_{i}^{(\ell )\prime }\left( \hat{\theta}_{i}^{(\ell
)}-\theta _{i}^{0,(\ell )}\right) \\
& =H^{(\ell )-1}\lambda _{i}^{0}+\frac{1}{T_{\ell }}\hat{F}^{(\ell )\prime
}\left( F^{0,(\ell )}-\hat{F}^{(\ell )}H^{(\ell )-1}\right) \lambda _{i}^{0}+%
\frac{1}{T_{\ell }}\hat{F}^{(\ell )\prime }e_{i}^{(\ell )}-\frac{1}{T_{\ell }%
}\hat{F}^{(\ell )\prime }X_{i}^{(\ell )\prime }\left( \hat{\theta}%
_{i}^{(\ell )}-\theta _{i}^{0,(\ell )}\right)
\end{align*}%
where the second and fifth equalities are by the normalization that $\frac{%
\hat{F}^{(\ell )\prime }\hat{F}^{(\ell )}}{T_{\ell }}=I_{r_{0}}$. It follows
that 
\begin{align*}
\hat{\lambda}_{i}^{(\ell )}-H^{(\ell )-1}\lambda _{i}^{0}& =\frac{1}{T_{\ell
}}\hat{F}^{(\ell )\prime }\left( F^{0,(\ell )}-\hat{F}^{(\ell )}H^{(\ell
)-1}\right) \lambda _{i}^{0}+\frac{1}{T_{\ell }}\hat{F}^{(\ell )\prime
}e_{i}^{(\ell )}-\frac{1}{T_{\ell }}\hat{F}^{(\ell )\prime }X_{i}^{(\ell
)\prime }\left( \hat{\theta}_{i}^{(\ell )}-\theta _{i}^{0,(\ell )}\right) \\
& =\frac{1}{T_{\ell }}\hat{F}^{(\ell )\prime }\left( F^{0,(\ell )}-\hat{F}%
^{(\ell )}H^{(\ell )-1}\right) \lambda _{i}^{0}+\frac{1}{T_{\ell }}\left( 
\hat{F}^{(\ell )}-F^{0,(\ell )}H^{(\ell )}\right) ^{\prime }e_{i}^{(\ell )}
\\
& +H^{(\ell )\prime }\frac{1}{T_{\ell }}F^{0,(\ell )}e_{i}^{(\ell )}-\frac{1%
}{T_{\ell }}\hat{F}^{(\ell )\prime }X_{i}^{(\ell )\prime }\left( \hat{\theta}%
_{i}^{(\ell )}-\theta _{i}^{0,(\ell )}\right) :=I_{i,1}^{(\ell
)}+I_{i,2}^{(\ell )}+I_{i,3}^{(\ell )}-I_{i,4}^{(\ell )}.
\end{align*}%
}

{\small First, by Lemmas \ref{Lem:factor_uniform}(i) and \ref{Lem:bounded
u&v}(i) and Assumption \ref{ass:8}(ii), 
\begin{eqnarray*}
\max_{i\in \mathcal{N}}\left\Vert I_{i,1}^{(\ell )}\right\Vert &\leq
&\left\Vert \frac{\hat{F}^{(\ell )}}{\sqrt{T_{\ell }}}\right\Vert \frac{%
\left\Vert F^{0,(\ell )}-\hat{F}^{(\ell )}H^{(\ell )-1}\right\Vert }{\sqrt{%
T_{\ell }}}\max_{i\in \mathcal{N}}\left\Vert \lambda _{i}^{0}\right\Vert
=O_{p}\left( \sqrt{\frac{\log N}{N\wedge T}}\right) ,\text{ and} \\
\max_{i\in \mathcal{N}}\left\Vert I_{i,2}^{(\ell )}\right\Vert &\leq &\frac{1%
}{\sqrt{T_{\ell }}}\left\Vert \hat{F}^{(\ell )}-F^{0,(\ell )}H^{(\ell
)}\right\Vert \frac{\left\Vert e_{i}^{(\ell )}\right\Vert }{\sqrt{T_{\ell }}}%
=O_{p}\left( \sqrt{\frac{\log N}{N\wedge T}}\right) .
\end{eqnarray*}%
Similarly $\max_{i\in \mathcal{N}}||I_{i,3}^{(\ell )}||=O_{p}(\sqrt{(\log
N)/T})$ by \eqref{Lem:emp_Bern}. Now, 
\begin{equation*}
\max_{i\in \mathcal{N}}\left\Vert I_{i,4}^{(\ell )}\right\Vert \leq
\left\Vert \frac{\hat{F}^{(\ell )}}{\sqrt{T_{\ell }}}\right\Vert \frac{%
\left\Vert X_{i}^{(\ell )}\right\Vert }{\sqrt{T_{\ell }}}\max_{i\in \mathcal{%
N}}\left\Vert \hat{\theta}_{i}^{(\ell )}-\theta _{i}^{0,(\ell )}\right\Vert
=O_{p}\left( \sqrt{\frac{\log N}{T}}\right) ,
\end{equation*}%
by Lemma \ref{Lem:heter_slope}(iii). Combining the above results yields $%
\max_{i\in \mathcal{N}}||\hat{\lambda}_{i}-H^{(\ell )-1}\lambda _{i}||$ $%
=O_{p}(\sqrt{(\log N)/(N\wedge T)}).$}

{\small (iv) Recall from \eqref{Lem:Fhat_1} that $\hat{F}^{(\ell )\prime
}-H^{(\ell )\prime }F^{\prime }=V_{NT}^{(\ell )-1}\sum_{m\in \lbrack
8]}J_{m}^{(\ell )\prime },$ where $J_{m}^{(\ell )},~m\in \lbrack 8],$ are
defined in the proof of Lemma \ref{Lem:Fhat}. Let $J_{m,t}^{(\ell )}$ be the 
$t$-th column of $V_{NT}^{(\ell )-1}J_{m,t}^{(\ell )}$ for $m\in \lbrack 8]$%
. We observe that $\hat{f}_{t}-H^{(\ell )\prime }f_{t}^{0}$ is the $t$-th
column of $\hat{F}^{(\ell )\prime }-H^{(\ell )\prime }F^{\prime }$. It
remains to show the convergence rate for $J_{m,t}^{(\ell )},m\in \lbrack 8]$%
. }

{\small For $J_{1,t}^{(\ell )}$, we notice that 
\begin{align*}
\max_{t\in \mathcal{T}_{\ell }}\left\Vert J_{1,t}^{(\ell )}\right\Vert &
=\max_{t\in \mathcal{T}_{\ell }}\left\Vert V_{NT}^{(\ell )-1}\hat{F}^{(\ell
)\prime }\left[ \frac{1}{NT_{\ell }}\sum_{i\in \mathcal{N}}X_{i}^{(\ell
)}\left( \theta _{i}^{0,(\ell )}-\hat{\theta}_{i}^{(\ell )}\right) \left(
\theta _{i}^{0,(\ell )}-\hat{\theta}_{i}^{(\ell )}\right) ^{\prime }X_{it}%
\right] \right\Vert  \\
& \leq \left\Vert V_{NT}^{(\ell )-1}\right\Vert \frac{\left\Vert \hat{F}%
^{(\ell )}\right\Vert }{\sqrt{T_{\ell }}}\max_{i\in \mathcal{N}}\left\Vert 
\hat{\theta}_{i}^{(\ell )}-\theta _{i}^{0,(\ell )}\right\Vert ^{2}\max_{i\in 
\mathcal{N}}\frac{\left\Vert X_{i}^{(\ell )}\right\Vert }{\sqrt{T_{\ell }}}%
\max_{t\in \mathcal{T}_{\ell }}\frac{1}{n}\sum_{i\in \mathcal{N}}\left\Vert
X_{it}\right\Vert  \\
& \lesssim \max_{i\in \mathcal{N}}\left\Vert \hat{\theta}_{i}^{(\ell
)}-\theta _{i}^{0,(\ell )}\right\Vert ^{2}=O_{p}((\log N)/T),
\end{align*}%
by Lemma \ref{Lem:heter_slope}(iii) and Assumption \ref{ass:8}(ii).
Similarly, by Lemma \ref{Lem:bounded u&v}(i), 
\begin{align*}
\max_{t\in \mathcal{T}_{\ell }}\left\Vert J_{2,t}^{(\ell )}\right\Vert &
=\max_{t\in \mathcal{T}_{\ell }}\left\Vert V_{NT}^{(\ell )-1}\hat{F}^{(\ell
)\prime }F^{0,(\ell )}\left[ \frac{1}{nT_{\ell }}\sum_{i\in \mathcal{N}%
}\lambda _{i}^{0}\left( \theta _{i}^{0,(\ell )}-\hat{\theta}_{i}^{(\ell
)}\right) ^{\prime }X_{it}\right] \right\Vert  \\
& \lesssim \max_{i\in \mathcal{N}}\left\Vert \hat{\theta}_{i}^{(\ell
)}-\theta _{i}^{0,(\ell )}\right\Vert =O_{p}(\sqrt{(\log N)/T}), \\
\max_{t\in \mathcal{T}_{\ell }}\left\Vert J_{3,t}^{(\ell )}\right\Vert &
=\max_{t\in \mathcal{T}_{\ell }}\left\Vert V_{NT}^{(\ell )-1}\hat{F}^{(\ell
)\prime }\left[ \frac{1}{nT_{\ell }}\sum_{i\in \mathcal{N}}e_{i}^{(\ell
)}\left( \theta _{i}^{0,(\ell )}-\hat{\theta}_{i}^{(\ell )}\right) ^{\prime
}X_{it}\right] \right\Vert  \\
& \lesssim \max_{i\in \mathcal{N}}\left\Vert \hat{\theta}_{i}^{(\ell
)}-\theta _{i}^{0,(\ell )}\right\Vert _{2}=O_{p}(\sqrt{(\log N)/T}) \\
\max_{t\in \mathcal{T}_{\ell }}\left\Vert J_{4,t}^{(\ell )}\right\Vert &
=\max_{t\in \mathcal{T}_{\ell }}\left\Vert V_{NT}^{(\ell )-1}\hat{F}^{(\ell
)\prime }\left[ \frac{1}{nT_{\ell }}\sum_{i\in \mathcal{N}}X_{i}^{(\ell
)}\left( \theta _{i}^{0,(\ell )}-\hat{\theta}_{i}^{(\ell )}\right) \lambda
_{i}^{0\prime }\right] f_{t}^{0}\right\Vert  \\
& \lesssim \max_{i\in \mathcal{N}}\left\Vert \hat{\theta}_{i}^{(\ell
)}-\theta _{i}^{0,(\ell )}\right\Vert _{2}=O_{p}(\sqrt{(\log N)/T}), \\
\max_{t\in \mathcal{T}_{\ell }}\left\Vert J_{5,t}^{(\ell )}\right\Vert _{2}&
=\max_{t\in \mathcal{T}_{\ell }}\left\Vert V_{NT}^{(\ell )-1}\hat{F}^{(\ell
)\prime }\left[ \frac{1}{nT_{\ell }}\sum_{i\in \mathcal{N}}X_{i}^{(\ell
)}\left( \theta _{i}^{0,(\ell )}-\hat{\theta}_{i}^{(\ell )}\right) e_{it}%
\right] \right\Vert  \\
& \lesssim \max_{i\in \mathcal{N}}\left\Vert \hat{\theta}_{i}^{(\ell
)}-\theta _{i}^{0,(\ell )}\right\Vert _{2}=O_{p}(\sqrt{(\log N)/T}).
\end{align*}%
Next, 
\begin{equation*}
\max_{t\in \mathcal{T}_{\ell }}\left\Vert J_{6,t}^{(\ell )}\right\Vert
=\max_{t\in \mathcal{T}_{\ell }}\left\Vert V_{NT}^{(\ell )-1}\hat{F}^{(\ell
)\prime }\left[ \frac{1}{nT_{\ell }}\sum_{i\in \mathcal{N}}e_{i}^{(\ell
)}\lambda _{i}^{0\prime }\right] f_{t}^{0}\right\Vert \lesssim \frac{1}{n%
\sqrt{T_{\ell }}}\left\Vert \sum_{i\in \mathcal{N}}e_{i}^{(\ell )}\lambda
_{i}^{0\prime }\right\Vert =O_{p}(N^{-1/2}),
\end{equation*}%
by the fact that 
\begin{equation*}
\mathbb{E}\left( \frac{1}{nT_{\ell }}\left\Vert \sum_{i\in \mathcal{N}%
}e_{i}^{(\ell )}\lambda _{i}^{0\prime }\right\Vert ^{2}\Bigg|\mathscr{D}%
\right) =\mathbb{E}\left( \frac{1}{T_{\ell }}\sum_{t\in \mathcal{T}_{\ell
}}\left\Vert \frac{1}{\sqrt{n}}\sum_{i\in \lbrack N]}e_{it}\lambda
_{i}^{0}\right\Vert ^{2}\Bigg|\mathscr{D}\right) =O_{p}(1)
\end{equation*}%
with the same manner as \eqref{Lem:D1_1}. Similarly, 
\begin{equation*}
\max_{t\in \mathcal{T}_{\ell }}\left\Vert J_{7,t}^{(\ell )}\right\Vert
=\max_{t\in \mathcal{T}_{\ell }}\left\Vert V_{NT}^{(\ell )-1}\hat{F}^{(\ell
)\prime }F^{0,(\ell )}\left[ \frac{1}{nT_{\ell }}\sum_{i\in \mathcal{N}%
}\lambda _{i}^{0}e_{it}\right] \right\Vert \lesssim \max_{t\in \mathcal{T}%
_{\ell }}\left\Vert \frac{1}{n}\sum_{i\in \mathcal{N}}\lambda
_{i}^{0}e_{it}\right\Vert =O_{p}(\sqrt{(\log T)/N}),
\end{equation*}%
by using the Bernstein's inequality for the independent sequence in Lemma %
\ref{Lem:Bern_mds}(i). For }${\small J}_{{\small 8,t}}^{({\small \ell })}$%
{\small , we have 
\begin{align*}
\max_{t\in \mathcal{T}_{\ell }}\left\Vert J_{8,t}^{(\ell )}\right\Vert &
\lesssim \frac{1}{\sqrt{n}}\max_{t\in \mathcal{T}_{\ell }}\left\Vert \frac{1%
}{\sqrt{nT_{\ell }}}\sum_{i\in \mathcal{N}}e_{i}^{(\ell )}e_{it}\right\Vert =%
\frac{1}{\sqrt{n}}\sqrt{\max_{t\in \mathcal{T}_{\ell }}\frac{1}{T_{\ell }}%
\sum_{s\in \lbrack T_{\ell }]}\left( \frac{1}{\sqrt{n}}\sum_{i\in \mathcal{N}%
}e_{is}e_{it}\right) ^{2}} \\
& \leq \frac{1}{\sqrt{n}}\sqrt{\frac{1}{T_{\ell }}\sum_{t\in \mathcal{T}%
_{\ell }}\sum_{s\in \lbrack T_{\ell }]}\left( \frac{1}{\sqrt{n}}\sum_{i\in 
\mathcal{N}}e_{is}e_{it}\right) ^{2}}=O_{p}(N^{-1/2}),
\end{align*}%
where the last equality holds by the fact that $\mathbb{E}[\frac{1}{T_{\ell }%
}\sum_{t\in \mathcal{T}_{\ell }}\sum_{s\in \lbrack T_{\ell }]}(\frac{1}{%
\sqrt{N}}\sum_{i\in \mathcal{N}}e_{is}e_{it})^{2}]=O(1)$ by %
\eqref{Lem:Fhat_2}. In sum, we have $\max_{t}||\hat{f}_{t}-H^{(\ell )\prime
}f_{t}^{0}||=O_{p}(\sqrt{(\log N\vee T)/(N\wedge T)})$. }
\end{proof}

\begin{lemma}
{\small \label{Lem:variance_estimates} Under Assumptions \ref{ass:1*}, \ref%
{ass:2} and \ref{ass:8}, we have 
\begin{equation*}
\max_{i\in \mathcal{N}}\left\Vert \hat{S}_{ii}^{(\ell )}-\mathbb{E}\left(
S_{ii}^{0,(\ell )}\big|\mathscr{D}\right) \right\Vert =O_{p}\left( \sqrt{%
\frac{\log N\vee T}{N\wedge T}}\right) ~\text{and}~\max_{i\in \mathcal{N}%
}\left\Vert \hat{\Omega}_{i}^{(\ell )}-\Omega _{i}^{0,(\ell )}\right\Vert
=o_{p}(1).
\end{equation*}%
}
\end{lemma}

\begin{proof}
{\small Recall that $S_{ii}^{0,(\ell )}=\frac{X_{i}^{(\ell )\prime
}M_{F^{0,(\ell )}}X_{i}^{(\ell )}}{T_{\ell }}$ and $\hat{S}_{ii}^{(\ell )}=%
\frac{X_{i}^{(\ell )\prime }M_{\hat{F}^{(\ell )}}X_{i}^{(\ell )}}{T_{\ell }}$%
. Combining \eqref{Lem:heter_3}, Lemma \ref{Lem:thetahat_pre}(i) and Lemma %
\ref{Lem:heter_slope}(iii), we have 
\begin{equation*}
\max_{i\in \mathcal{N}}\left\Vert \hat{S}_{ii}^{(\ell )}-\mathbb{E}\left(
S_{ii}^{0,(\ell )}\big|\mathscr{D}\right) \right\Vert =O_{p}\left( \sqrt{%
\frac{\log N\vee T}{N\wedge T}}\right) .
\end{equation*}%
}

{\small Recall that $\mathfrak{x}_{it}^{(\ell )\prime }$ is the $t$-th row
of $M_{F^{0,(\ell )}}X_{i}^{(\ell )}$ and let $\hat{\mathfrak{x}}%
_{it}^{(\ell )\prime }$ be the $t$-th row of $M_{\hat{F}^{(\ell
)}}X_{i}^{(\ell )}$, respectively. Under Assumption \ref{ass:1*}(iii), we
have $\hat{\Omega}_{i}^{(\ell )}=\frac{1}{T_{\ell }}\sum_{t\in \mathcal{T}%
_{\ell }}\hat{\mathfrak{x}}_{it}^{(\ell )}\hat{\mathfrak{x}}_{it}^{(\ell
)\prime }\hat{e}_{it}^{2}\quad $and $\Omega _{i}^{0,(\ell )}=\frac{1}{%
T_{\ell }}\sum_{t\in \mathcal{T}_{\ell }}\mathbb{E}(\mathfrak{x}_{it}^{(\ell
)}\mathfrak{x}_{it}^{(\ell )\prime }e_{it}^{2}).$ It remains to show 
\begin{equation}
\max_{i\in \mathcal{N}}\left\Vert \frac{1}{T_{\ell }}\sum_{t\in \mathcal{T}%
_{\ell }}\left[ \hat{\mathfrak{x}}_{it}^{(\ell )}\hat{\mathfrak{x}}%
_{it}^{(\ell )\prime }\hat{e}_{it}^{2}-\mathbb{E}(\mathfrak{x}_{it}^{(\ell )}%
\mathfrak{x}_{it}^{(\ell )\prime }e_{it}^{2})\right] \right\Vert =o_{p}(1).
\label{Lem:var_wts1}
\end{equation}%
From the definitions of $\mathfrak{x}_{it}^{(\ell )}$ and $\hat{\mathfrak{x}}%
_{it}^{(\ell )}$, we notice that $\mathfrak{x}_{it}^{(\ell )}=X_{it}-\frac{1%
}{T_{\ell }}X_{i}^{(\ell )\prime }F^{0,(\ell )}f_{t}^{0}$ and $\hat{%
\mathfrak{x}}_{it}^{(\ell )}=X_{it}-\frac{1}{T_{\ell }}X_{i}^{(\ell )\prime }%
\hat{F}^{(\ell )}\hat{f}_{t}^{0}$, which gives 
\begin{equation}
\mathfrak{x}_{it}^{(\ell )}-\hat{\mathfrak{x}}_{it}^{(\ell )}=\frac{1}{%
T_{\ell }}X_{i}^{(\ell )\prime }(\hat{F}^{(\ell )}\hat{f}_{t}^{0}-F^{0,(\ell
)}f_{t}^{0}).  \label{Lem:var_5}
\end{equation}%
Note that 
\begin{align*}
\max_{t\in \mathcal{T}_{\ell }}\frac{1}{\sqrt{T_{\ell }}}\left\Vert \hat{F}%
^{(\ell )}\hat{f}_{t}^{0}-F^{0,(\ell )}f_{t}^{0}\right\Vert _{2}& \leq \frac{%
\left\Vert \hat{F}^{(\ell )}-F^{0,(\ell )}H^{(\ell )}\right\Vert }{\sqrt{%
T_{\ell }}}\max_{t}\left\Vert \hat{f}_{t}^{(\ell )}-H^{(\ell )\prime
}f_{t}^{0}\right\Vert +\frac{\left\Vert F^{0,(\ell )}H^{(\ell )}\right\Vert 
}{\sqrt{T_{\ell }}}\max_{t}\left\Vert \hat{f}_{t}^{(\ell )}-H^{(\ell )\prime
}f_{t}^{0}\right\Vert \\
& +\frac{\left\Vert \hat{F}^{(\ell )}-F^{0,(\ell )}H^{(\ell )}\right\Vert }{%
\sqrt{T_{\ell }}}\max_{t}\left\Vert H^{(\ell )\prime }f_{t}^{0}\right\Vert
=O_{p}\left( \sqrt{\frac{\log N\vee T}{N\wedge T}}\right) ,
\end{align*}
by Lemma \ref{Lem:factor_uniform}(i) and (iv) and Lemma \ref{Lem:bounded u&v}%
(i). Then by \eqref{Lem:var_5} and Assumption \ref{ass:8}(ii), we have 
\begin{equation}
\max_{i\in \mathcal{N},t\in \mathcal{T}_{\ell }}\left\Vert \mathfrak{x}%
_{it}^{(\ell )}-\hat{\mathfrak{x}}_{it}^{(\ell )}\right\Vert \leq \frac{1}{%
\sqrt{T_{\ell }}}\left\Vert X_{i}\right\Vert O_{p}\left( \sqrt{\frac{\log
N\vee T}{N\wedge T}}\right) =O_{p}\left( \sqrt{\frac{\log N\vee T}{N\wedge T}%
}\right) .  \label{Lem:var_6}
\end{equation}%
}

{\small Next, for $i\in \mathcal{N},t\in \mathcal{T}_{\ell }$, note that 
\begin{align*}
\hat{e}_{it}& =Y_{it}-X_{it}^{\prime }\hat{\theta}_{i}^{(\ell )}-\hat{\lambda%
}_{i}^{(\ell )\prime }\hat{f}_{t}^{(\ell )} \\
& =e_{it}-\bigg[X_{it}^{\prime }\left( \theta _{i}^{0,(\ell )}-\hat{\theta}%
_{i}^{(\ell )}\right) +\left( \hat{\lambda}_{i}^{(\ell )}-H^{(\ell
)-1}\lambda _{i}^{0}\right) ^{\prime }\left( \hat{f}_{t}^{(\ell )}-H^{(\ell
)\prime }f_{t}^{0}\right) +\left( \hat{\lambda}_{i}^{(\ell )}-H^{(\ell
)-1}\lambda _{i}^{0}\right) ^{\prime }H^{(\ell )\prime }f_{t}^{0} \\
& +\left( H^{(\ell )-1}\lambda _{i}^{0}\right) ^{\prime }\left( \hat{f}%
_{t}^{(\ell )}-H^{(\ell )\prime }f_{t}^{0}\right) \bigg].
\end{align*}%
Then 
\begin{align}
& \max_{i\in \mathcal{N},t\in \mathcal{T}_{\ell }}\left\vert \hat{e}%
_{it}-e_{it}\right\vert =O_{p}\left( \sqrt{\frac{\log N}{T}}%
(NT)^{1/q}\right) +O_{p}\left( \sqrt{\frac{\log N\vee T}{N\wedge T}}\right)
,\quad \text{and}  \notag  \label{Lem:var_ehat} \\
& \hat{e}_{it}^{2}-e_{it}^{2}=e_{it}\left( \hat{e}_{it}-e_{it}\right)
+\left( \hat{e}_{it}-e_{it}\right) ^{2}=e_{it}X_{it}^{\prime
}R_{1,it}+R_{2,it} \\
& \text{s.t.}\max_{i\in \mathcal{N},t\in \mathcal{T}_{\ell }}\left\Vert
R_{1,it}\right\Vert _{2}=O_{p}\left( \sqrt{\frac{\log N}{T}}\right)
,\max_{i\in \lbrack N],t\in \mathcal{T}_{\ell }}\left\vert
R_{2,it}\right\vert =O_{p}\left( \sqrt{\frac{\log N\vee T}{N\wedge T}}\right)
\notag
\end{align}%
by Lemmas \ref{Lem:heter_slope}(iii), \ref{Lem:factor_uniform}(iii), and \ref%
{Lem:factor_uniform}(iv). It follows that 
\begin{align}
\frac{1}{T_{\ell }}\sum_{t\in \mathcal{T}_{\ell }}\hat{\mathfrak{x}}%
_{it}^{(\ell )}\hat{\mathfrak{x}}_{it}^{(\ell )\prime }\hat{e}_{it}^{2}& =%
\frac{1}{T_{\ell }}\sum_{t\in \mathcal{T}_{\ell }}\mathfrak{x}_{it}^{(\ell )}%
\mathfrak{x}_{it}^{(\ell )\prime }e_{it}^{2}+\frac{1}{T_{\ell }}\sum_{t\in 
\mathcal{T}_{\ell }}\left( \mathfrak{x}_{it}^{(\ell )}-\hat{\mathfrak{x}}%
_{it}^{(\ell )}\right) \left( \mathfrak{x}_{it}^{(\ell )}-\hat{\mathfrak{x}}%
_{it}^{(\ell )}\right) ^{\prime }\left( \hat{e}_{it}^{2}-e_{it}^{2}\right) 
\notag  \label{Lem:var_7} \\
& +\frac{1}{T_{\ell }}\sum_{t\in \mathcal{T}_{\ell }}\left( \mathfrak{x}%
_{it}^{(\ell )}-\hat{\mathfrak{x}}_{it}^{(\ell )}\right) \left( \mathfrak{x}%
_{it}^{(\ell )}-\hat{\mathfrak{x}}_{it}^{(\ell )}\right) ^{\prime
}e_{it}^{2}+\frac{1}{T_{\ell }}\sum_{t\in \mathcal{T}_{\ell }}\mathfrak{x}%
_{it}^{(\ell )}\left( \mathfrak{x}_{it}^{(\ell )}-\hat{\mathfrak{x}}%
_{it}^{(\ell )}\right) ^{\prime }\left( \hat{e}_{it}^{2}-e_{it}^{2}\right) 
\notag \\
& +\frac{1}{T_{\ell }}\sum_{t\in \mathcal{T}_{\ell }}\left( \mathfrak{x}%
_{it}^{(\ell )}-\hat{\mathfrak{x}}_{it}^{(\ell )}\right) \mathfrak{x}%
_{it}^{(\ell )\prime }\left( \hat{e}_{it}^{2}-e_{it}^{2}\right) +\frac{1}{%
T_{\ell }}\sum_{t\in \mathcal{T}_{\ell }}\left( \mathfrak{x}_{it}^{(\ell )}-%
\hat{\mathfrak{x}}_{it}^{(\ell )}\right) \mathfrak{x}_{it}^{(\ell )\prime
}e_{it}^{2}  \notag \\
& +\frac{1}{T_{\ell }}\sum_{t\in \mathcal{T}_{\ell }}\mathfrak{x}%
_{it}^{(\ell )}\left( \mathfrak{x}_{it}^{(\ell )}-\hat{\mathfrak{x}}%
_{it}^{(\ell )}\right) ^{\prime }e_{it}^{2}+\frac{1}{T_{\ell }}\sum_{t\in 
\mathcal{T}_{\ell }}\mathfrak{x}_{it}^{(\ell )}\mathfrak{x}_{it}^{(\ell
)\prime }\left( \hat{e}_{it}^{2}-e_{it}^{2}\right)  \notag \\
& =\frac{1}{T_{\ell }}\sum_{t\in \mathcal{T}_{\ell }}\mathfrak{x}%
_{it}^{(\ell )}\mathfrak{x}_{it}^{(\ell )\prime }e_{it}^{2}+O_{p}\left( 
\sqrt{\frac{\log N\vee T}{N\wedge T}}\right) \quad \text{uniformly over $%
i\in \mathcal{N}$,}
\end{align}%
where the last line holds by \eqref{Lem:var_6}, \eqref{Lem:var_ehat}, and
Assumptions \ref{ass:8}(ii) and \ref{ass:1*}(iv). Using similar arguments as
used to derive \eqref{Lem:emp_freedman} by the Bernstein's inequality for
m.d.s., for a positive constant $c_{9}$, we have 
\begin{equation}
\mathbb{P}\left\{ \max_{i\in \mathcal{N}}\left\Vert \frac{1}{T_{\ell }}%
\sum_{t\in \mathcal{T}_{\ell }}\left[ \mathfrak{x}_{it}^{(\ell )}\mathfrak{x}%
_{it}^{(\ell )\prime }e_{it}^{2}-\mathbb{E}\left( \mathfrak{x}_{it}^{(\ell )}%
\mathfrak{x}_{it}^{(\ell )\prime }e_{it}^{2}\right) \right] \right\Vert
>c_{9}\sqrt{\frac{\log N}{T}}\right\} =o(1).  \label{Lem:var_8}
\end{equation}%
Combining \eqref{Lem:var_7} and \eqref{Lem:var_8}, we obtain %
\eqref{Lem:var_wts1}. 
}
\end{proof}

\begin{lemma}
{\small \label{Lem:size} Let $\hat{\Gamma}^{(\ell )}$ be as defined in %
\eqref{PIFE_test}. Under Assumptions \ref{ass:1*}, \ref{ass:2} and \ref%
{ass:8}, we have $\hat{\Gamma}^{(\ell )}\rightsquigarrow \mathbb{N}(0,1)$
under $H_{0}$. }
\end{lemma}

\begin{proof}
{\small Under the null that $\theta _{i}^{0,(\ell )}=\theta ^{0,(\ell )}$
for $\forall i\in \mathcal{N}$, we use Lemma \ref{Lem:heter_slope}(i) to
obtain that 
\begin{align}
\hat{\bar{\theta}}^{(\ell )}-\theta ^{0,(\ell )}& =\frac{1}{n}\sum_{i\in 
\mathcal{N}}\left[ \mathbb{E}\left( S_{ii}^{0,(\ell )}\big|\mathscr{D}%
\right) \right] ^{-1}\xi _{i}^{(\ell )}  \notag  \label{Lem:size_1} \\
& +\frac{1}{n^{2}}\sum_{i\in \mathcal{N}}a_{ii}^{0}\left[ \mathbb{E}\left(
S_{ii}^{0,(\ell )}\big|\mathscr{D}\right) \right] ^{-1}\left( \hat{\theta}%
_{i}^{(\ell )}-\theta _{i}^{0,(\ell )}\right) +O_{p}\left( \frac{\log N}{%
N\wedge T}\right) ,
\end{align}%
such that 
\begin{equation*}
\left\Vert \frac{1}{n^{2}}\sum_{i\in \mathcal{N}}a_{ii}^{0}(\hat{\theta}%
_{i}^{(\ell )}-\theta _{i}^{0,(\ell )})\right\Vert \leq \frac{1}{n}%
\max_{i\in \mathcal{N}}\left\vert a_{ii}^{(\ell )}\right\vert \left\Vert 
\hat{\theta}_{i}^{(\ell )}-\theta _{i}^{0,(\ell )}\right\Vert =O_{p}\left( 
\frac{\sqrt{\log N}}{N\sqrt{T}}\right),
\end{equation*}%
with Lemma \ref{Lem:heter_slope}(iii) and the fact that $\max_{i\in \mathcal{%
N}}\left\vert a_{ii}^{0}\right\vert =O(1)$. For the first term on the right
side of \eqref{Lem:size_1}, we have 
\begin{equation*}
\frac{1}{n}\sum_{i\in \mathcal{N}}\left[ \mathbb{E}\left( S_{ii}^{0,(\ell )}%
\big|\mathscr{D}\right) \right] ^{-1}\xi _{i}^{(\ell )}=\frac{1}{nT_{\ell }}%
\sum_{i\in \mathcal{N}}\sum_{t\in \mathcal{T}_{\ell }}\left[ \mathbb{E}%
\left( S_{ii}^{0,(\ell )}\big|\mathscr{D}\right) \right] ^{-1}\mathfrak{x}%
_{it}e_{it}=O_{p}(\frac{1}{\sqrt{nT}}),
\end{equation*}%
by the central limit theorem for m.d.s., which yields that 
\begin{equation}
\left\Vert \hat{\bar{\theta}}^{0,(\ell )}-\theta ^{0,(\ell )}\right\Vert
=O_{p}\left( \frac{\log N}{N\wedge T}\right) .  \label{Lem:size_2}
\end{equation}%
}

{\small Recall from \eqref{PIFE_test} that $\hat{\Gamma}^{(\ell )}=\sqrt{n}%
\cdot \frac{\frac{1}{n}\sum_{i\in \mathcal{N}}\hat{\mathbb{S}}_{i}^{(\ell
)}-p}{\sqrt{2p}}$. Note that 
\begin{align*}
\hat{\mathbb{S}}_{i}^{(\ell )}& =T_{\ell }\left( \hat{\theta}_{i}^{(\ell )}-%
\hat{\bar{\theta}}^{(\ell )}\right) ^{\prime }\hat{S}_{ii}^{(\ell )}\left( 
\hat{\Omega}_{i}^{(\ell )}\right) ^{-1}\hat{S}_{ii}^{(\ell )}\left( \hat{%
\theta}_{i}^{(\ell )}-\hat{\bar{\theta}}^{(\ell )}\right) \left( 1-\hat{a}%
_{ii}^{(\ell )}/n\right) ^{2} \\
& =T_{\ell }\left( \hat{\theta}_{i}^{(\ell )}-\theta ^{0,(\ell )}\right)
^{\prime }\hat{S}_{ii}^{(\ell )}\left( \hat{\Omega}_{i}^{(\ell )}\right)
^{-1}\hat{S}_{ii}^{(\ell )}\left( \hat{\theta}_{i}^{(\ell )}-\theta
^{0,(\ell )}\right) \left( 1-\hat{a}_{ii}^{(\ell )}/n\right) ^{2} \\
& +T_{\ell }\left( \hat{\bar{\theta}}^{(\ell )}-\theta ^{0,(\ell )}\right)
^{\prime }\hat{S}_{ii}^{(\ell )}\left( \hat{\Omega}_{i}^{(\ell )}\right)
^{-1}\hat{S}_{ii}^{(\ell )}\left( \hat{\bar{\theta}}^{(\ell )}-\theta
^{0,(\ell )}\right) \left( 1-\hat{a}_{ii}^{(\ell )}/n\right) ^{2} \\
& -2T_{\ell }\left( \hat{\theta}_{i}^{(\ell )}-\theta ^{0,(\ell )}\right)
^{\prime }\hat{S}_{ii}^{(\ell )}\left( \hat{\Omega}_{i}^{(\ell )}\right)
^{-1}\hat{S}_{ii}^{(\ell )}\left( \hat{\bar{\theta}}^{(\ell )}-\theta
^{0,(\ell )}\right) \left( 1-\hat{a}_{ii}^{(\ell )}/n\right) ^{2}:=\hat{%
\mathbb{S}}_{i,1}^{(\ell )}+\hat{\mathbb{S}}_{i,2}^{(\ell )}-\hat{\mathbb{S}}%
_{i,3}^{(\ell )}.
\end{align*}%
Below we show that $\frac{1}{\sqrt{n}}\sum_{i\in \mathcal{N}}\hat{\mathbb{S}}%
_{i,2}^{(\ell )}$ and $\frac{1}{\sqrt{n}}\sum_{i\in \mathcal{N}}\hat{\mathbb{%
S}}_{i,3}^{(\ell )}$ are smaller terms and $\frac{1}{\sqrt{n}}\sum_{i\in 
\mathcal{N}}\hat{\mathbb{S}}_{i,1}^{(\ell )}\rightsquigarrow \mathcal{N}(0,1)
$. }

{\small First, noted that 
\begin{align*}
& \left\vert \frac{1}{\sqrt{n}}\sum_{i\in \mathcal{N}}\hat{\mathbb{S}}%
_{i,2}^{(\ell )}\right\vert \leq \sqrt{n}T_{\ell }\left\Vert \hat{\bar{\theta%
}}^{(\ell )}-\theta ^{0,(\ell )}\right\Vert ^{2}\max_{i\in \mathcal{N}%
}\lambda _{\max }\left( \hat{S}_{ii}^{(\ell )}\left( \hat{\Omega}_{i}^{(\ell
)}\right) ^{-1}\hat{S}_{ii}^{(\ell )}\right) \max_{i\in \mathcal{N}}\left( 1-%
\hat{a}_{ii}^{(\ell )}/n\right) ^{2} \\
& =\sqrt{n}T_{\ell }O_{p}\left( \frac{(\log N)^{2}}{N^{2}\wedge T^{2}}%
\right) \max_{i\in \mathcal{N}}\left\Vert \hat{S}_{ii}^{(\ell )}\right\Vert
^{2}\max_{i\in \mathcal{N}}\left\Vert \hat{\Omega}_{i}^{(\ell )}\right\Vert
\max_{i\in \mathcal{N}}\left( 1-\hat{a}_{ii}^{(\ell )}/n\right) ^{2} \\
& =\sqrt{n}T_{\ell }O_{p}\left( \frac{(\log N)^{2}}{N^{2}\wedge T^{2}}%
\right) \left[ \max_{i\in \mathcal{N}}\left\Vert S_{ii}^{0,(\ell
)}\right\Vert ^{2}\max_{i\in \mathcal{N}}\left\Vert \Omega _{i}^{0,(\ell
)}\right\Vert \max_{i\in \mathcal{N}}\left( 1-a_{ii}^{0}/N\right)
^{2}+o_{p}(1)\right] =o_{p}(1),
\end{align*}%
where the first equality is by \eqref{Lem:size_2}, the second equality is by
Lemma \ref{Lem:variance_estimates} and the fact that $\max_{i\in \mathcal{N}%
}|\hat{a}_{ii}^{(\ell )}-a_{ii}^{0}|=o_{p}(1)$ owing to Lemma \ref%
{Lem:factor_uniform}(iii), and the last equality holds by the fact that $%
\max_{i\in \mathcal{N}}||S_{ii}^{0,(\ell )}||=O(1)$, $\max_{i\in \mathcal{N}%
}||\Omega _{i}^{0,(\ell )}||=O(1)$, and $\max_{i\in \mathcal{N}%
}|a_{ii}^{(\ell )}|=O(1)$ and Assumption \ref{ass:1*}(vi). }

{\small Second, by analogous arguments as used above, we have 
\begin{align*}
\left\vert \frac{1}{\sqrt{n}}\sum_{i\in \mathcal{N}}\hat{\mathbb{S}}%
_{i,3}^{(\ell )}\right\vert & \leq 2\sqrt{n}T_{\ell }\max_{i\in \mathcal{N}%
}\left\Vert \hat{\theta}_{i}^{(\ell )}-\theta ^{0,(\ell )}\right\Vert
\left\Vert \hat{\bar{\theta}}^{(\ell )}-\theta ^{0,(\ell )}\right\Vert
\max_{i\in \mathcal{N}}\lambda _{\max }\left( \hat{S}_{ii}^{(\ell )}(\hat{%
\Omega}_{i}^{(\ell )})^{-1}\hat{S}_{ii}^{(\ell )}\right) \max_{i\in \mathcal{%
N}}(1-\hat{a}_{ii}^{(\ell )}/n)^{2} \\
& =\sqrt{n}T_{\ell }O_{p}\left( \sqrt{\frac{\log N}{T}}\right) O_{p}\left( 
\frac{\log N}{N\wedge T}\right) =o_{p}(1).
\end{align*}%
At last for $\hat{\mathbb{S}}_{i,1}^{(\ell )}$, it's clear that 
\begin{align*}
\frac{1}{\sqrt{n}}\sum_{i\in \mathcal{N}}\hat{\mathbb{S}}_{i,1}^{(\ell )}& =%
\frac{1}{\sqrt{n}}\sum_{i\in \mathcal{N}}T_{\ell }(\hat{\theta}_{i}^{(\ell
)}-\theta ^{0,(\ell )})^{\prime }\hat{S}_{ii}^{(\ell )}(\hat{\Omega}%
_{i}^{(\ell )})^{-1}\hat{S}_{ii}^{(\ell )}(\hat{\theta}_{i}^{(\ell )}-\theta
^{0,(\ell )})(1-\hat{a}_{ii}^{(\ell )}/n)^{2}=\frac{1}{\sqrt{n}}\sum_{i\in 
\mathcal{N}}\hat{z}_{i}^{(\ell )}+o_{p}(1), \\
\text{where}\quad \hat{z}_{i}^{(\ell )}& =T_{\ell }(\hat{\theta}_{i}^{(\ell
)}-\theta ^{0,(\ell )})^{\prime }S_{ii}^{0,(\ell )}(\Omega _{i}^{0,(\ell
)})^{-1}S_{ii}^{0,(\ell )}(\hat{\theta}_{i}^{(\ell )}-\theta ^{0,(\ell
)})\left( 1-a_{ii}^{0}/n\right) ^{2}.
\end{align*}%
Then by the central limit theorem, we have $\hat{\Gamma}^{(\ell )}=\frac{1}{%
\sqrt{n}}\sum_{i\in \mathcal{N}}\frac{\hat{z}_{i}^{(\ell )}-p}{\sqrt{2p}}%
+o_{p}(1)\rightsquigarrow \mathbb{N}(0,1).$}
\end{proof}

\begin{lemma}
{\small \label{Lem:power} Under Assumptions \ref{ass:1*}, \ref{ass:2} and %
\ref{ass:8}, we have $|\hat{\Gamma}^{(\ell )}|/(\log N)^{1/2}\rightarrow
\infty $ under $H_{1}$ if $\frac{T_{\ell }}{\sqrt{n}}\sum_{i\in \mathcal{N}%
}||c_{i}^{(\ell )}||^{2}/(\log N)^{1/2}\rightarrow \infty $. }
\end{lemma}

\begin{proof}
{\small Noting that $\theta _{i}^{0,(\ell )}=\theta ^{0,(\ell
)}+c_{i}^{(\ell )}$, we have 
\begin{equation*}
\hat{\theta}_{i}^{(\ell )}-\hat{\bar{\theta}}^{(\ell )}=(\hat{\theta}%
_{i}^{(\ell )}-\theta _{i}^{0,(\ell )})-(\hat{\bar{\theta}}^{(\ell )}-\bar{%
\theta}^{0,(\ell )})+\theta _{i}^{0,(\ell )}-\theta ^{0,(\ell )}=(\hat{\theta%
}_{i}^{(\ell )}-\theta _{i}^{0,(\ell )})-(\hat{\bar{\theta}}^{(\ell )}-\bar{%
\theta}^{0,(\ell )})+c_{i}^{(\ell )}.
\end{equation*}%
Then%
\begin{align*}
\hat{\mathbb{S}}_{i}^{(\ell )}& =T_{\ell }(\hat{\theta}_{i}^{(\ell )}-\theta
_{i}^{0,(\ell )})^{\prime }\hat{S}_{ii}^{(\ell )}(\hat{\Omega}_{i}^{(\ell
)})^{-1}\hat{S}_{ii}^{(\ell )}(\hat{\theta}_{i}^{(\ell )}-\theta
_{i}^{0,(\ell )})(1-\hat{a}_{ii}^{(\ell )}/n)^{2} \\
& +T_{\ell }(\hat{\bar{\theta}}^{(\ell )}-\bar{\theta}^{0,(\ell )})^{\prime }%
\hat{S}_{ii}^{(\ell )}(\hat{\Omega}_{i}^{(\ell )})^{-1}\hat{S}_{ii}^{(\ell
)}(\hat{\bar{\theta}}^{(\ell )}-\bar{\theta}^{0,(\ell )})(1-\hat{a}%
_{ii}^{(\ell )}/n)^{2} \\
& -2T_{\ell }(\hat{\theta}_{i}^{(\ell )}-\theta _{i}^{0,(\ell )})^{\prime }%
\hat{S}_{ii}^{(\ell )}(\hat{\Omega}_{i}^{(\ell )})^{-1}\hat{S}_{ii}^{(\ell
)}(\hat{\bar{\theta}}^{(\ell )}-\bar{\theta}^{0,(\ell )})(1-\hat{a}%
_{ii}^{(\ell )}/n)^{2} \\
& +T_{\ell }c_{i}^{(\ell )\prime }\hat{S}_{ii}^{(\ell )}(\hat{\Omega}%
_{i}^{(\ell )})^{-1}\hat{S}_{ii}^{(\ell )}c_{i}^{(\ell )}(1-\hat{a}%
_{ii}^{(\ell )}/n)^{2}+2T_{\ell }(\hat{\theta}_{i}^{(\ell )}-\theta
_{i}^{0,(\ell )})^{\prime }\hat{S}_{ii}^{(\ell )}(\hat{\Omega}_{i}^{(\ell
)})^{-1}\hat{S}_{ii}^{(\ell )}c_{i}^{(\ell )}(1-\hat{a}_{ii}^{(\ell )}/n)^{2}
\\
& -2T_{\ell }(\hat{\bar{\theta}}^{(\ell )}-\bar{\theta}^{0,(\ell )})^{\prime
}\hat{S}_{ii}^{(\ell )}(\hat{\Omega}_{i}^{(\ell )})^{-1}\hat{S}_{ii}^{(\ell
)}c_{i}^{(\ell )}(1-\hat{a}_{ii}^{(\ell )}/n)^{2}:=\sum_{m=4}^{9}\hat{%
\mathbb{S}}_{i,m}^{(\ell )}.
\end{align*}%
In the proof of Lemma \ref{Lem:size}, we have already shown that 
\begin{equation*}
\frac{1}{\sqrt{n}}\sum_{i\in \mathcal{N}}\frac{\hat{\mathbb{S}}_{i,4}^{(\ell
)}-p}{\sqrt{2p}}\rightsquigarrow \mathbb{N}(0,1)\quad \text{and}\quad
\left\vert \frac{1}{\sqrt{n}}\sum_{i\in \mathcal{N}}\hat{\mathbb{S}}%
_{i,m}^{(\ell )}\right\vert =o_{p}(1),~m=5,6.
\end{equation*}%
As for $\hat{\mathbb{S}}_{i,7}^{(\ell )}$, we can show that 
\begin{align*}
\frac{1}{\sqrt{n}}\sum_{i\in \mathcal{N}}\hat{\mathbb{S}}_{i,7}^{(\ell )}& =%
\frac{T_{\ell }}{\sqrt{n}}\sum_{i\in \mathcal{N}}c_{i}^{(\ell )\prime }\hat{S%
}_{ii}^{(\ell )}(\hat{\Omega}_{i}^{(\ell )})^{-1}\hat{S}_{ii}^{(\ell
)}c_{i}^{(\ell )}(1-\hat{a}_{ii}^{(\ell )}/n)^{2} \\
& =\frac{T_{\ell }}{\sqrt{n}}\sum_{i\in \mathcal{N}}\left[ c_{i}^{(\ell
)\prime }S_{ii}^{0,(\ell )}(\Omega _{i}^{0,(\ell )})^{-1}S_{ii}^{0,(\ell
)}c_{i}^{(\ell )}\left( 1-a_{ii}^{0}/n\right) ^{2}+o_{p}(1)\right]  \\
& \geq \frac{T_{\ell }}{\sqrt{n}}\sum_{i\in \mathcal{N}}\left[ [\lambda
_{\max }(S_{ii}^{0,(\ell )-1}\Omega _{i}^{0,(\ell )}S_{ii}^{0,(\ell
)-1})]^{-1}\left\Vert c_{i}^{(\ell )}\right\Vert _{2}^{2}\left(
1-a_{ii}^{0}/n\right) ^{2}+o_{p}(1)\right]  \\
& \geq \frac{T_{\ell }}{\sqrt{n}}\sum_{i\in \mathcal{N}}\left[ \frac{%
\left\Vert c_{i}^{(\ell )}\right\Vert _{2}^{2}\left( 1-a_{ii}^{0}/n\right)
^{2}}{\left\Vert S_{ii}^{0,(\ell )-1}\right\Vert ^{2}\left\Vert \Omega
_{i}^{0,(\ell )}\right\Vert }+o_{p}(1)\right]  \\
& =\frac{T_{\ell }}{\sqrt{n}}\sum_{i\in \mathcal{N}}\left[ \frac{\left\Vert
c_{i}^{(\ell )}\right\Vert _{2}^{2}}{\left\Vert S_{ii}^{0,(\ell
)-1}\right\Vert ^{2}\left\Vert \Omega _{i}^{0,(\ell )}\right\Vert }+o_{p}(1)%
\right]  \\
& \geq \frac{1}{\max_{i\in \mathcal{N}}\left\Vert S_{ii}^{0,(\ell
)-1}\right\Vert ^{2}\left\Vert \Omega _{i}^{0,(\ell )}\right\Vert }\frac{%
T_{\ell }}{\sqrt{n}}\left[ \sum_{i\in \mathcal{N}}\left\Vert c_{i}^{(\ell
)}\right\Vert _{2}^{2}+o_{p}(1)\right] \rightarrow \infty \text{ at a rate
faster than }(\log N)^{1/2},
\end{align*}%
where the second line is by the uniform convergence of $\hat{S}_{ii}^{(\ell
)}$, $\hat{\Omega}_{i}^{(\ell )}$ and $\hat{a}_{ii}^{(\ell )}$, the fifth
line is by the fact that $\max_{i\in \mathcal{N}}|\frac{a_{ii}^{0}}{n}%
|=o_{p}(1).$ and the last line draws from the assumption that $\frac{T_{\ell
}}{\sqrt{n}}\sum_{i\in \mathcal{N}}||c_{i}^{(\ell )}||^{2}/(\log
N)^{1/2}\rightarrow \infty $. By Cauchy's inequality, we observe that 
\begin{align*}
& \left\vert \frac{1}{\sqrt{n}}\sum_{i\in \mathcal{N}}\hat{\mathbb{S}}%
_{i,8}^{(\ell )}\right\vert \leq 2\sqrt{\frac{1}{\sqrt{n}}\sum_{i\in 
\mathcal{N}}\hat{\mathbb{S}}_{i,4}^{(\ell )}}\sqrt{\frac{1}{\sqrt{n}}%
\sum_{i\in \mathcal{N}}\hat{\mathbb{S}}_{i,7}^{(\ell )}}=o_{p}\left( \frac{1%
}{\sqrt{n}}\sum_{i\in \mathcal{N}}\hat{\mathbb{S}}_{i,7}^{(\ell )}\right) ,
\\
& \left\vert \frac{1}{\sqrt{n}}\sum_{i\in \mathcal{N}}\hat{\mathbb{S}}%
_{i,9}^{(\ell )}\right\vert \leq 2\sqrt{\frac{1}{\sqrt{n}}\sum_{i\in 
\mathcal{N}}\hat{\mathbb{S}}_{i,5}^{(\ell )}}\sqrt{\frac{1}{\sqrt{n}}%
\sum_{i\in \mathcal{N}}\hat{\mathbb{S}}_{i,7}^{(\ell )}}=o_{p}\left( \frac{1%
}{\sqrt{n}}\sum_{i\in \mathcal{N}}\hat{\mathbb{S}}_{i,7}^{(\ell )}\right) .
\end{align*}%
Combining the above results yields that $|\hat{\Gamma}^{(\ell )}|/(\log
N)^{1/2}\rightarrow \infty $. }
\end{proof}

\section{\protect\small Algorithm for Nuclear Norm Regularization}

\label{sec:algorihm} 
{\small To solve the optimization problem in \eqref{obj}, there
are different algorithms in the literature such as the Alternating Direction
Method of Multipliers (ADMM) algorithm and the singular value thresholding
(SVT) procedure. \cite{wang2022low-rank} provide the ADMM algorithm based on
a quantile regression framework, which can be easily extended to the
linear conditional mean regression framework. In this section, we focus on
the SVT procedure for the case of low-rank estimation with two regressors.
The case of more than two regressors is self-evident.

We can iteratively use SVT estimation to obtain the nuclear norm regularized
regression estimates. Specifically, given $\Theta _{1}$ and $\Theta _{2}$,
we solve for $\Theta _{0}$ with 
\begin{equation*}
\Theta _{0}(\Theta _{1},\Theta _{2})=\operatornamewithlimits{\argmin}\limits%
_{\Theta _{0}}\left\Vert Y-X_{1}\odot \Theta _{1}-X_{2}\odot \Theta
_{2}-\Theta _{0}\right\Vert _{F}+\nu _{0}NT\left\Vert \Theta _{0}\right\Vert
_{\ast }.
\end{equation*}%
Given $\Theta _{0}$ and $\Theta _{2}$, we solve for $\Theta _{1}$ with 
\begin{equation*}
\Theta _{1}(\Theta _{0},\Theta _{2})=\operatornamewithlimits{\argmin}\limits%
_{\Theta _{1}}\left\Vert Y-\Theta _{0}-X_{2}\odot \Theta _{2}-X_{1}\odot
\Theta _{1}\right\Vert _{F}+\nu _{1}NT\left\Vert \Theta _{1}\right\Vert
_{\ast }.
\end{equation*}%
Given $\Theta _{0}$ and $\Theta _{1}$, we solve for $\Theta _{2}$ with 
\begin{equation*}
\Theta _{2}(\Theta _{0},\Theta _{1})=\operatornamewithlimits{\argmin}\limits%
_{\Theta _{2}}\left\Vert Y-\Theta _{0}-X_{1}\odot \Theta _{1}-X_{2}\odot
\Theta _{2}\right\Vert _{F}+\nu _{2}NT\left\Vert \Theta _{2}\right\Vert
_{\ast }.
\end{equation*}%
Specifically, the algorithm goes as follows:\newline
\textit{Step 1:} initialize $\Theta _{0}$, $\Theta _{1}$ and $\Theta _{1}$
to be $\Theta _{0}^{1}$, $\Theta _{1}^{1}$ and $\Theta _{1}^{1}$ and set $k=1
$.\newline
\textit{Step 2:} let 
\begin{align*}
& \Theta _{0}^{k+1}=S_{\frac{\nu _{0}NT}{2}}\left( Y-X_{1}\odot \Theta
_{1}^{k}-X_{2}\odot \Theta _{2}^{k}\right) , \\
& \Theta _{1}^{k+1}=S_{\frac{\tau \nu _{1}NT}{2}}\left( \Theta _{1}^{k}-\tau
X_{1}\odot \left( X_{1}\odot \Theta _{1}^{k}-Y+\Theta _{0}^{k+1}+X_{2}\odot
\Theta _{2}^{k}\right) \right) , \\
& \Theta _{2}^{k+1}=S_{\frac{\tau \nu _{2}NT}{2}}\left( \Theta _{2}^{k}-\tau
X_{2}\odot \left( X_{2}\odot \Theta _{2}^{k}-Y+\Theta _{0}^{k+1}+X_{1}\odot
\Theta _{1}^{k+1}\right) \right) , \\
& k=k+1,
\end{align*}%
where $\tau $ is the step size, and $S_{\lambda }(M)$ is the singular value
operator for any matrix $M$ and fixed parameter $\lambda $. By SVD, we have $%
M=U_{M}D_{M}V_{M}^{\prime }$. Define $D_{M,\lambda }$ by replacing the
diagonal entry $D_{M,ii}$ of $D_{M}$ by $\max (D_{M,ii}-\lambda ,0)$, and
then let $S_{\lambda }(M)=U_{M}D_{M,\lambda }V_{M}^{\prime }$.\newline
\textit{Step 3:} repeat step 2 until convergence.

We can follow \cite{chernozhukov2019inference}, which gives the expression
to pin down the step size $\tau$. In addition, Proposition 2.1 of \cite%
{chernozhukov2019inference} shows the convergence of the above algorithm.}

\end{document}